\documentclass[11pt]{article}
\usepackage[margin=1in]{geometry}

\usepackage[table,xcdraw]{xcolor}
\definecolor{bestavg}{RGB}{220, 240, 255}     
\definecolor{bestworst}{RGB}{255, 230, 230}   

\usepackage{subcaption, caption}
\usepackage{booktabs}
\usepackage{graphicx}

\usepackage{float}
\usepackage{natbib}
 \bibpunct[, ]{(}{)}{,}{a}{}{,}%
\usepackage{color}
\usepackage{algorithm}
\usepackage{algorithmic}

\usepackage{array}
\usepackage{multirow}
\usepackage{comment}
\usepackage[toc,page]{appendix}
\usepackage{appendix}

\usepackage{setspace}
\usepackage{hyperref}
\usepackage[utf8]{inputenc}
\usepackage{tikz}

\usepackage{tabularx}
\usepackage{hyperref}
\usepackage{url}

\usepackage{natbib}
 \bibpunct[, ]{(}{)}{,}{a}{}{,}%

\usepackage{algorithm}

 \def\rew{\textsc rew}

 \def\pleft{p_{\text{left}}}
  \def\pright{p_{\text{right}}}

 \def\pbal{p_b}

\usepackage{color}
\usepackage{algorithm}
\usepackage{algorithmic}
\usepackage{setspace}
\usepackage{bm}
\usepackage{graphicx}

\def\wll{\widetilde{\ll}}
\def\pc{p^{\star}}
\def\Rstar{R^{\star}}
\def\nex{x_{-1}}
\usepackage{array}
\usepackage{multirow}
\usepackage{comment}
\usepackage[toc,page]{appendix}
\usepackage{appendix}
\usepackage{setspace}
\usepackage{hyperref}
\usepackage{amsmath}
\usepackage{amssymb}
\usepackage{amsthm}

\DeclareMathOperator*{\argmin}{arg\,min}

\newcommand{\widebar}[1]{\overline{#1}}
\newcommand{\Halmos}{\hfill$\square$}

\newtheorem{proposition}{Proposition}
\newtheorem{example}{Example}

\newenvironment{APPENDICES}{\appendix}{}

\newtheorem{theorem}{Theorem}
\newtheorem{lemma}{Lemma}
\newtheorem{definition}{Definition}
\usepackage{tabularx}
\usepackage{hyperref}
\usepackage{url}

\usepackage{natbib}
 \bibpunct[, ]{(}{)}{,}{a}{}{,}%
\usepackage{algorithm}
\def\cons{\textsc{consis}}
\def\comp{\textsc{robust}}
\def\region{\mathcal R}

\def\cstar{C^{\star}(\region)}

\def\opt{\textsc{opt}}
\def\CP{\text{CP}}
\def\UCP{\underline{\text{CP}}}
\def\xmin{\bar x_u}
\def\V{\mathcal{V}}

\def\RL{R_{\text{left}}}
\def\RR{R_{\text{right}}}

\usepackage{enumitem}

\def\u{u}
\def\ll{l}
\def\pl{\text{PL}}

\usepackage{enumitem}

\usetikzlibrary{arrows.meta,positioning}
\usetikzlibrary{calc}
\usetikzlibrary{shapes}
\usetikzlibrary{patterns}
\definecolor{purple}{rgb}{0.4,0.2,1}
\definecolor{violet}{RGB}{204,0,204}
\definecolor{grey}{RGB}{128, 128, 128}

\begin{document}

\author{
Negin Golrezaei\thanks{Sloan School of Management, Massachusetts Institute of Technology. Email: \texttt{golrezaei@mit.edu}}
\and
Patrick Jaillet\thanks{Department of Electrical Engineering and Computer Science, Massachusetts Institute of Technology. Email: \texttt{jaillet@mit.edu}}
\and
Zijie Zhou\thanks{Department of Industrial Engineering and Decision Analytics, Hong Kong University of Science and Technology. Email: \texttt{jerryzhou@ust.hk}}
}

\title{Online Resource Allocation with Convex-set Machine-Learned Advice}
\maketitle
\begin{abstract}
    Decision-makers often have access to a machine-learned prediction about demand, referred to as advice, which can potentially be utilized in online decision-making processes for resource allocation. However, exploiting such advice poses challenges due to its potential inaccuracy. To address this issue, we propose a framework that enhances online resource allocation decisions with potentially unreliable machine-learned (ML) advice. We assume here that this advice is represented by a general convex uncertainty set for the demand vector.

We introduce a parameterized class of Pareto optimal online resource allocation algorithms that strike a balance between consistent and robust ratios. The consistent ratio measures the algorithm's performance (compared to the optimal hindsight solution) when the ML advice is accurate, while the robust ratio captures performance under an adversarial demand process when the advice is inaccurate. Specifically, in a C-Pareto optimal setting, we maximize the robust ratio while ensuring that the consistent ratio is at least C. Our proposed C-Pareto optimal algorithm is an adaptive protection level algorithm, which extends the classical fixed protection level algorithm introduced in \cite{littlewood2005special} and \cite{ball2009toward}. Solving a complex non-convex continuous optimization problem characterizes the adaptive protection level algorithm. To complement our algorithms, we present a simple method for computing the maximum achievable consistent ratio, which serves as an estimate for the maximum value of the ML advice. 

Additionally, we present numerical studies to evaluate the performance of our algorithm in comparison to benchmark algorithms. The results demonstrate that by  adjusting the parameter C, our algorithms effectively strike a balance between worst-case and average performance. In fact, they surpass the benchmark algorithms, including  those that solely depend on a single point estimate advice, rather than leveraging the advantages of an uncertainty set advice.
\end{abstract}


\maketitle

\section{Introduction}

The problem of allocating a limited inventory of a single resource to sequentially arriving requests can be examined within the framework of revenue management, a significant discipline in operations research. Originally developed within the airline industry, revenue management has gained widespread recognition and applicability across various sectors, including retail, hospitality, online advertising, and more.

In the context of resource allocation, revenue management aims to optimize a firm's revenue by implementing effective policies to control quantities. This concept finds application in various applications, such as flight seat allocation, online event ticketing, car rental inventory management, and hotel room reservation.
To illustrate this concept, let's consider an example involving an airline. Airlines often offer different fare classes to cater to various customer types, including price-sensitive leisure travelers and business travelers. Each fare class comes with distinct prices and additional perks, such as seat selection and flexibility in cancellation. In this scenario, the airline must strategically determine the optimal number of seats to allocate to customers from different fare classes in order to maximize their overall revenue.

However, firms face the challenge of making real-time decisions to allocate their limited resources to incoming demand while lacking precise knowledge of future demand. This challenge arises due to the inherent trade-off between allocating resources, such as seats, to low-reward demand, such as leisure travelers, and reserving resources for potential high-reward demand, such as business travelers.

To address this trade-off, researchers have extensively studied two main models. The first model is the adversarial arrival model (\cite{ball2009toward}), which assumes no forecast about demand is available. However, in practice, the demand process is typically not the worst-case scenario and may exhibit some level of predictability. Consequently, the resulting algorithms under this regime tend to be overly conservative, as demonstrated in our numerical studies in Section \ref{sec:simulations}. The second model is the stochastic model, which assumes perfect knowledge of the demand process, with the assumption that low-reward demand arrives before high-reward demand \citep{littlewood2005special}. However, demand prediction is often challenging, especially in new and non-stationary settings that arise due to factors like seasonality or natural crises.

In this work, we aim to bridge the gap between the two models by augmenting the adversarial model with a demand forecast in the form of a convex uncertainty set. This uncertainty set that we refer to as \emph{machine-learned advice} is obtained through data-driven robust optimization algorithms. Despite the success and ubiquity of these techniques across many domains, leveraging them in online decision-making, such as the aforementioned resource allocation problem, presents significant challenges. The key challenge lies in effectively managing errors and biases in the forecast that inevitably exist.

To account for this challenge, we go beyond having a point estimate for the demand vector (i.e., the number of customers of different types) as done in \cite{balseiro2022single}.  (For a detailed comparison between our work and \cite{balseiro2022single}, refer to Sections \ref{sec:otherrelated} and \ref{sec:simulations}.) Instead, we adopt the approach of utilizing an uncertainty set for the demand vector. This modeling choice  that also covers a point estimate as its special case offers several advantages:  

\begin{itemize}[leftmargin=*]\item First,
it allows us to harness biases that exist in a single point estimate, resulting in a robust algorithm that does not overfit to a single point estimate. Please refer to our numerical studies in Section \ref{sec:simulations} for comparison  between resource allocation algorithms with a single point estimate and those we propose  with an uncertainty set. 
\item Second, the framework of uncertainty sets allows for the consideration of inaccuracies within the set itself. In other words, the realized demand may fall outside the bounds defined by the uncertainty set. This level of flexibility greatly enhances the resilience of the allocation process, empowering it to effectively handle unforeseen variations in demand that may deviate from initial expectations.
\item Third, the introduction of the uncertainty set allows us to capture the variance in demand process and potential positive or negative correlations between different types of demand. For example, if the uncertainty set establishes bounds on the total number of demands, we can  account for the expected correlation between the number of high-reward and low-reward customers. In this case, a large influx of low-reward customers would suggest a corresponding decrease in the number of high-reward customers.
\end{itemize}

 \subsection{Our Contributions and Results}

\textbf{A New Resource Allocation Model with a Convex Uncertainty Set.} We introduce a novel online resource allocation model with machine-learned (ML) advice (Section \ref{sec:model}) that addresses the challenge of allocating $m$ identical units of a resource to arriving requests, categorized as either low-reward or high-reward. {\color{black}See extension of our algorithms to settings with more than two types in Section \ref{sec:multitype}. }

At the onset of the allocation period, the decision-maker receives ML advice in the form of a convex uncertainty set $\region$. This uncertainty set characterizes the total number of high-reward or low-reward requests expected to arrive during the allocation period. For instance, the advice may indicate that the total number of requests falls within a specific interval while ensuring that the total number of high-reward requests remains below a certain threshold. Importantly, the ML advice does not provide any information regarding the order of arrivals.

Our choice of utilizing a convex uncertainty set is motivated by related works in the robust optimization literature, such as \cite{bertsimas2009constructing} and  \cite{bertsimas2018data}. These studies leverage offline historical data to construct convex uncertainty sets. Additionally, publications like \cite{cheramin2021data, bertsimas2016reformulation, jalilvand2016robust} have explored the application of convex uncertainty sets in offline robust optimization problems. 

In this work, we do not make any assumptions about the accuracy of the ML advice. Instead, we present a class of algorithms that demonstrate robust performance regardless of the accuracy of the advice.

 \textbf{Pareto optimal Algorithms (Section \ref{sec:maxrobust}).} As previously mentioned, our objective is to incorporate ML advice in a robust manner, accounting for potential inaccuracies. To achieve this, we introduce two performance measures: consistent ratio and robust ratio, which are analogous to the traditional competitive ratio used in the analysis of online algorithms.

The consistent ratio of an algorithm represents the worst-case ratio of its expected reward to the optimal hindsight solution under any arrival sequence consistent with the ML advice. On the other hand, the robust ratio is the worst-case ratio of its expected reward to the optimal hindsight solution under any arrival sequence that is not consistent with the advice. The formal definitions of these ratios can be found in Section \ref{sec:model}. As one of our main contributions, we present a parameterized class of Pareto optimal algorithms that strike a balance between robust and consistent ratios.

Let $\cstar$ denote the maximum consistent ratio achievable by any algorithm under the ML advice $\region$ (see the formal definition in Section \ref{sec:mlconsistent}). For any $C\leq \cstar$, a $C$-Pareto optimal algorithm maximizes the robust ratio among all deterministic or randomized online algorithms, while ensuring that its consistent ratio is at least $C$. By adjusting the parameter $C$, we can emphasize achieving a higher consistent ratio when we have greater confidence in the accuracy of the advice. Conversely, decreasing $C$ reflects concern about potential inaccuracies, leading to a focus on obtaining a higher robust ratio.

Our main result fully characterizes $C$-Pareto optimal algorithms, demonstrating that they belong to a class of (adaptive) protection level algorithms (PLAs). An adaptive PLA extends the classical fixed protection level algorithms studied in seminal works of \cite{littlewood2005special, ball2009toward}. In fixed protection level algorithms, a certain amount of resources is reserved or protected for potential high-reward requests that may arrive later, with the protection level remaining fixed throughout the allocation period. In the adaptive version introduced formally in Section \ref{sec:adaptive}, the protection level is represented by a function that maps the total received low-reward demand so far to a protection level. In adaptive PLAs, the protection level can vary or decrease over time, as long as the reduction is not too steep (see the formal definition in Definition \ref{def:1}). Our designed $C$-Pareto optimal algorithms fall within this class, simplifying their implementation.


To fully characterize the $C$-Pareto optimal algorithm, we present an optimization problem (referred to as Problem \ref{prob:rbmax}) that optimizes over the protection level function, assuming the algorithm is a PLA. Importantly, our results demonstrate that this assumption is not restrictive, as the designed algorithm is optimal among all deterministic and randomized algorithms, not just the PLAs. However, solving the optimization problem is challenging due to its non-convex and continuous nature.

\begin{figure}[ht]
\centering
\resizebox{\textwidth}{!}{
\begin{tikzpicture}[node distance=2cm, font=\footnotesize]
    \node (problem) [draw, align=center] {Problem \eqref{prob:rbmax}};
    \node (right) [draw, below left of=problem, align=center, xshift=-2cm, yshift=-0.5cm] { \footnotesize{Left Problem} \\ (Problem \eqref{prob:rbmax} for $x\in [0, \bar x]$)};
    \node (left) [draw, below right of=problem, align=center, xshift=2cm, yshift=-0.5cm] {Right Problem \\ (Problem \eqref{prob:rbmax} for $x\ge \bar x$)};
    \node (algoRight) [draw, below of=right, align=center] {Algorithm \ref{alg:left}};
    \node (algoLeft) [draw, below of=left, align=center] {Algorithm \ref{alg:right}};
    \node (algoTrans) [draw, below right of=algoRight, align=center, xshift=2cm] {Algorithm \ref{alg:trans}};
    
    \draw[->] (problem) -- (right);
    \draw[->] (problem) -- (left);
    \draw[->] (right) -- (algoRight);
    \draw[->] (left) -- (algoLeft);
    \draw[->] (algoRight) -- (algoTrans);
    \draw[->] (algoLeft) -- (algoTrans);
\end{tikzpicture}
}
\caption{Decomposition of Problem \eqref{prob:rbmax} into the right and left problems. Here, $x$ is the total low-reward demand and $\bar x$ is the maximum total low-reward demand under the ML advice.\label{fig:diagram}}
\end{figure}

\textbf{Technical Contributions Regarding  $C$-Pareto optimal Algorithms.} 
One of our main contributions is the development of a polynomial time scheme to solve Problem \eqref{prob:rbmax} for any convex ML advice $\region$. We achieve this by transforming the consistency constraints into bounds on the protection level (PL) function and applying similar transformations for the robustness constraints. This allows us to decompose the problem into ``right" and ``left" subproblems. The optimal solution to the right problem coincides with the original problem's solution when the total low-reward demand exceeds a certain threshold, while the left problem considers cases where the demand is below the threshold. The optimal solutions to both subproblems can be fully characterized, resulting in an efficient algorithm for Problem \eqref{prob:rbmax} with guaranteed performance guarantees. See Figure \ref{fig:diagram} for an outline of our approach.

\textbf{Characterizing the Maximum Consistent Ratio $\cstar$ (Section \ref{sec:mlconsistent}).}
In Section \ref{sec:mlconsistent}, we introduce polynomial-time methods to compute the maximum consistent ratio $\cstar$, which represents the highest value achievable with the ML advice. 
For general convex ML regions, we propose a bisection method (Algorithm \ref{alg:bisection}) that provides an $\epsilon$-accurate estimate of $\cstar$ in polynomial time, where $\epsilon$ can be chosen within the range $(0, 1/2]$. When the ML region is a polyhedron, we present a faster approach (Algorithm \ref{alg:mlconsis}) based on enumerating the polyhedron vertices. This method computes $\cstar$ exactly by identifying the worst vertices (bottlenecks) of Problem \eqref{prob:cpmax1} that determine $\cstar$. 



\textbf{Numerical Studies (Section \ref{sec:simulations}).}
We perform numerical studies to assess the performance of our proposed algorithms when the ML advice is derived from a limited number of samples. Our findings demonstrate that, in the presence of ML advice, our algorithms enhance both the average and worst-case performance, surpassing other benchmark algorithms. One such benchmark pertains to resource allocation algorithms based on a single point estimate, which exhibit average and worst-case compatible ratios that are up to $14\%$ and $40\%$ lower than those achieved by our resource allocation algorithms considering convex uncertainty sets.

\subsection{Other Related Works} \label{sec:otherrelated}

{\textbf{Online Decision-making with ML Advice.}} Our class of algorithms 
contributes to  a recent literature on using ML advice in the online algorithm design. Examples include \cite{lykouris2018competitive} and \cite{rohatgi2020near}  for online caching problems, \cite{antoniadis2020secretary}  for online secretary problems, \cite{jin2022online} for online matching problems, \cite{lattanzi2020online} for online scheduling with job weight advice, and \cite{balseiro2022single} for an online resource allocation problem.\footnote{See  \citep{mahdian2012online, esfandiari2015online, hwang2021online, golrezaei2014real, golrezaei2022online} for other works 
that  explore the partially known demand models.} 

{\color{black}In the context of single-leg revenue management, the work most related to ours is \cite{balseiro2022single}. Specifically, \cite{balseiro2022single} focuses on single-point demand prediction (i.e., point ML advice) under any arbitrary number of types, presenting a class of Pareto-optimal algorithms. The algorithm achieves the highest level of consistency given a required level of competitiveness in the presence of point advice. To this end, hard instances are identified, and an LP is used to design an optimal algorithm that switches between two protection level algorithms with fixed protection levels. They show that this switching, although sacrificing certain practical properties like monotonicity of the allocation rule, is deemed necessary for optimality. They further focus on designing an optimal policy with a fixed protection level. They show that as the distance between realized fare sequences and advice increases, the performance of protection level policies degrades linearly as the distance increases.}

{\color{black} In our work, while we adopt the general
framework of a Pareto-optimal tradeoff between consistent ratio and robustness ratio proposed in \cite{balseiro2022single}, 
there are important distinctions.
Firstly, we shift our focus to ML advice in the form of an uncertainty set, rather than relying solely on single point estimates. This strategic shift aims to mitigate potential biases inherent in point estimates, such as outliers, asymmetric errors, or incomplete data. Additionally, uncertainty sets furnish valuable insights into demand characteristics, including variance and correlations between different demand types. Moreover, there are instances where obtaining a point estimate for the demand vector is unfeasible, such as when dealing with datasets that only provide the total number of arrivals without specifying the number of arrivals for each type. We delve deeper into some of these aspects in Section \ref{sec:simulations}, using numerical analyses.}

{\color{black}\textbf{Online Resource Allocation.} The allocation of scarce resources in an online setting has been a subject of extensive research. A large body of work has studied this problem under stochastic arrival assumptions (e.g., \cite{devanur2009adwords,feldman2010online,agrawal2014dynamic}). One central paradigm in this line of work  is to repeatedly re-solve a fluid LP as time and inventory evolve. Such re-solving heuristics have been extensively analyzed in the revenue management literature (see, e.g., \cite{jasin2013analysis,jasin2015performance, bumpensanti2020re, wang2022nonasymptotic,  chen2024improved}), with results showing bounded revenue loss and asymptotic optimality under stochastic models. Another prominent approach, particularly in the online algorithms community, is the primal--dual technique, which has been widely applied to design algorithms  in both stochastic and adversarial settings (e.g., \cite{mehta2007adwords,buchbinder2007online,golrezaei2014real}). 

In contrast to both the re-solving and primal--dual paradigms, our work adopts a robust optimization perspective. Rather than assuming a correct stochastic specification, we allow for potentially unreliable machine-learned advice, which we encode as a convex uncertainty set. We then design protection-level algorithms that guarantee Pareto tradeoffs between consistency (when the advice is accurate) and robustness (under adversarially generated arrivals outside the set). This builds on the protection-level framework, a well-established approach in single-leg revenue management, but extends it to advice-augmented robust online allocation.
}

\textbf{Single-leg Revenue Management.} In this work, we study the single-leg revenue management problem in the presence of ML advice. 
Single-leg revenue management is a well-established model in the field of revenue management. \cite{littlewood2005special} proposed an optimal policy for the single-leg revenue management problem involving two types of customers under stochastic arrival processes. \cite{brumelle1993airline} extended this problem to include multiple types of customers and designed an optimal policy using dynamic programming. \cite{ball2009toward} was the first to address the single-leg revenue management problem under adversarial arrival sequences. They proposed an optimal protection-level policy for the two types of customers case and then introduced the concept of ``nesting" to generalize to the multiple types case. See also \cite{jasin2015performance, hwang2021online,  ma2021policies, golrezaei2021upfront} for more recent works on single-leg revenue management. Our work contributes to this literature by  presenting a new model for single-leg revenue management that bridges the gap between the adversarial and stochastic models.

\section{Model}\label{sec:model}
Consider a scenario where a firm has been endowed with $m$ identical units of a divisible resource
 to allocate over $T$ rounds, but the number of rounds $T$ is unknown to the firm. In each round $t$, a request with size $s_t>0$ and type $z_t\in \{\ell, h\}$ arrives, where the size of the request $s$ demands at most $s$ units of the resource. \footnote{\color{black}Please refer to Section~\ref{sec:multitype} for an extension of our setting and algorithms to cases with more than two types.
} Requests can be categorized into two types based on the normalized reward or revenue they generate upon receiving one unit of the resource, namely low-reward and high-reward requests. (See our discussion in Section \ref{sec:extension} about handling more than two types of requests.)

Let the reward for type $z\in \{\ell, h\}$ upon receiving one unit of the resource be denoted as $r_{z}$. Without loss of generality, it is assumed that $0<r_{\ell} < r_{h}$.\footnote{Without loss of generality, one can normalize $r_h$ to one.} Upon the arrival of a request $(s_t, z_t)$, the firm observes the size and type of the request and must make an irrevocable decision to allocate $a_t \in [0, s_t]$ units of the resource to request $(s_t, z_t)$, and collect a total reward of $a_{t}\cdot r_{z_t}$. At the time of the decision, the firm has no knowledge of the type and size of future requests.

The goal is to design online allocation algorithms that maximize the cumulative reward of the firm over the course of $T$ rounds. The performance of an algorithm is evaluated by comparing it to the optimal clairvoyant solution, which has complete knowledge of the arrival sequence of the requests $(s_t, z_t)_{t\in [T]}$ in advance. Further details on the evaluation process will be provided later.

\subsection{ML Advice} Let $I= (s_t, z_t)_{t\in [T]}$ be the arrival sequence of requests, where the type and size of requests, the order of the requests, and the number of requests $T$ are chosen by an adversary. For an input sequence $I= (s_t, z_t)_{t\in [T]}$, we define the total low-reward and high-reward demand in the input sequence $I$ as $\ell(I) = \sum_{t\in [T], z_t =\ell} s_t$ and $h(I) = \sum_{t\in [T], z_t =h} s_t$, respectively. We assume that at the beginning of round $1$, the firm has access to partial knowledge about the arrival sequence $I$ and, more specifically, about $\ell(I)$ and $h(I)$. This partial knowledge, which we refer to as \emph{ML advice}, is represented by a convex region $\region \in \mathbb{R}^2$. The ML advice enforces the demand vector, i.e., $(\ell(I),h(I))$, to fall into region $\region$. Throughout the manuscript, we refer to region $\region$ as the ML region/advice. Then, the set of arrival sequences that is consistent with the ML advice (denoted by $\mathcal{S}(\region)$) is given by:
\begin{align}\mathcal{S}(\region)=\{I : (\ell(I), h(I))\in \region\}\,.\label{eq:feasible}\end{align}

We refer to the set $\mathcal{S}(\region)$ as the \emph{ML-consistent set}. For example, when $\region=\{(x, y): x\in [a_1, b_1], y\in [a_2, b_2]\}$ for some $a_1, a_2, b_1, b_2\ge 0$, the ML advice provides lower and upper bounds on the low-reward and high-reward demands. As another example, when $\region =\{(x, y); x+y\in [a, b]\}$ for some $a, b\ge 0$, the ML advice gives lower and upper bounds on the total demand. 
The ML advice does not provide any information on the order of requests. We refer the reader to Section \ref{sec:pre} for further notation related to ML advice.

The ML advice we examine is general and encompasses any convex uncertainty region. As a specific instance, the ML advice can be represented by a single point, as studied in \cite{balseiro2022single}. However, our numerical studies reveal that relying solely on a single point estimate as ML advice can result in suboptimal performance due to inherent biases in a single point estimate.

\subsection{Performance Measures} In this work, we consider two metrics to measure the performance of any online resource allocation algorithm that has access to the ML advice $\region$. These two metrics, which are called the \emph{consistent ratio} and the \emph{robust ratio}, are inspired by the fact that the ML advice may not be completely accurate.

The consistent ratio of algorithm $\mathcal A$, denoted by $\cons(\mathcal A)$, measures how well algorithm $\mathcal A$ performs when the ML advice is completely accurate. That is, it measures the performance of algorithm $\mathcal A$ on all arrival sequences in the set $\mathcal{S}(\region)$ that are consistent with the ML advice $\region$. Similarly, the robust ratio of algorithm $\mathcal A$, denoted by $\comp(\mathcal A)$, measures the performance of algorithm $\mathcal A$ on all arrival sequences. The robust ratio can evaluate the robustness of the algorithm when the ML advice is misleading.

Mathematically speaking, for any arrival sequence $I$ and an online algorithm $\mathcal A$, let $\rew(\mathcal A, I)$ be the expected cumulative reward of the algorithm under the arrival sequence $I$. Further, let $\opt(I)$ be the optimal clairvoyant solution under arrival sequence $I$. Note that the optimal clairvoyant solution, which has knowledge of the arrival sequence in advance, starts by allocating resources to the type $h$ that has the highest reward. If there is any resource remaining after that, it allocates resources to type $\ell$ requests. Then, we define:\footnote{\color{black} In our work, we assume that resources are divisible, allowing for partial acceptance of requests. While this assumption simplifies the analysis, it preserves key structural properties of the problem. The algorithms developed here also apply to the non-divisible model and remain accurate as the number of resources~$m$ grows large. Specifically, the hindsight optimal scales linearly with~$m$, while the loss in the numerator of the competitive ratio due to items being non-divisible remains constant, leading to a reduction in the competitive ratio of order~$\Theta(1/m)$.
} 
\begin{equation} \label{def:consistent}\cons(\mathcal A) = \inf_{I:I \in \mathcal{S}(\region)} \frac{\rew(\mathcal A, I)}{\opt(I)} \quad 
\text{and}\quad 
\comp(\mathcal A) = \inf_{I: I \in \mathcal{S}(\region) \cup \mathcal{S}(\region)^{\mathsf{C}}} \frac{\rew(\mathcal A, I)}{\opt(I)}\,.\end{equation}
It is worth noting that when the set $\mathcal S(\region)$ contains all possible arrival sequences, the consistent ratio is equivalent to the traditional worst-case competitive ratio notion for online algorithm design in the absence of ML advice. In this case, the algorithm's performance is evaluated based on its worst-case performance over all possible arrival sequences. Similarly, when the ML-consistent set $\mathcal S(\region)$ is empty, the robust ratio is also equivalent to the traditional competitive ratio notion, where the algorithm's performance is measured against the optimal clairvoyant solution that has full knowledge of the arrival sequence.

In the absence of ML advice, it has been shown in \cite{ball2009toward} that the optimal competitive ratio for online resource allocation algorithms is $1/(2-r_{\ell}/r_h)$, where $r_\ell$ and $r_h$ are the rewards associated with the low- and high-reward types, respectively. 

\subsection{Objectives} When the ML advice is completely accurate, maximizing the consistent ratio would lead to an optimal algorithm with the highest worst-case competitive ratio on set $\mathcal{S}(\region)$. See the details in Section \ref{sec:mlconsistent}. However, such an algorithm may perform poorly when the ML advice is inaccurate (i.e., when the ML-consistent set $\mathcal S(\region)$ does not occur). To balance this trade-off, our main goal is to present a class of parameterized \emph{Pareto optimal} algorithms. 

In a $C$-Pareto optimal algorithm, we design an algorithm that achieves the highest possible robust ratio while obtaining a consistent ratio of at least $C$ for any $C\leq \cstar$. Here, $\cstar$ is the highest possible consistent ratio for a convex ML region $\region$, in the absence of any constraint on the robust ratio. We formally characterize $\cstar$ in Section \ref{sec:mlconsistent} and suggest an $O(|\mathcal V|^3)$ 
complexity algorithm to find the value of $\cstar$ when $\region$ is a polyhedron, where $\mathcal V$ is the set containing all the vertices of $\region$. 

Mathematically, let $\Pi$ be the set of all online deterministic and randomized algorithms. Then, the $C$-Pareto optimal algorithm $\mathcal A$ solves the following optimization problem:
   \begin{align} \label{prob:original}
     \max_{\mathcal{A} \in \Pi}  \ \comp(\mathcal{A})  \quad  
s.t. \quad   \ \cons(\mathcal{A}) \geq C\,.    
\end{align}  
   Observe that by setting $C$ to $\cstar$, we can design an optimal ML-consistent algorithm under which the ML advice is fully trusted. 
When $\region =\{(x, y): x, y\ge 0\}$---i.e., the ML arrival set $\mathcal{S}(\region)$ contains all possible arrival sequences---  the protection-level algorithm  of \cite{ball2009toward} is  ML-consistent optimal. In this algorithm, type $h$  requests are always accepted while at most $\frac{m}{2-r_h/r_{\ell}}$ type $\ell$ requests are  accepted. In other words, we \emph{protect} $m- \frac{m}{2-r_{\ell}/r_{h}}$ of the resources for  high-reward requests, and by doing so, we obtain $\cstar$ of $\rho:= \frac{1}{2-r_{\ell}/r_{h}}$. Overall, the design of Pareto optimal algorithms lead to a Pareto curve (e.g., Figure \ref{fig:optpareto}) that helps us balance the trade-off between the consistent and robust ratios.

\subsection{Notation} \label{sec:pre}

In this section, we present a few definitions regarding the ML region $\region$ (refer to the figure below for illustration).
\begin{center}
    \begin{tikzpicture}[
      scale=2,
      ]
       \draw[->, black] (0, 0) -- (3, 0);
        \draw[->, black] (0, 0) -- (0, 3);
        \node[black] (x) at (2.9,-0.2) {$x$};
        \node[black] (y) at (-0.2, 2.9) {$y$};
        \node[black] (mx) at (2,-0.2) {$m$};
        \node[black] (my) at (-0.2, 2) {$m$};
         \draw[densely dotted] (0, 2) -- (2, 2);

      \fill[draw=black, fill=black!10] (0.1, 1) -- (0.5, 2.5) -- (2.3, 1) -- (2.1, 0.1) -- (1, 0.1) -- cycle;
        \node[black, above] at (1.7,1.7) {$\region$};

        \draw[orange, thick] (0.1, 1) -- (0.5, 2.5) -- (2.3, 1) -- (2.1, 0.1);
      \draw[->,orange] (2.6, 2) -- (1.6, 1.6);
      \node[orange, right] at (2.6,2) {$\bar h(.)$};
      
      \draw[brown, thick] (2.1, 0.1) -- (1, 0.1) -- (0.1, 1);
      \draw[->,brown] (-0.2, 0.5) -- (0.6, 0.5);
      \node[brown, left] at (-0.2,0.5) {$\underline h(.)$};

      \draw[black,densely dotted] (0.1, 1) -- (0.1, 0);
      \node[black] at (0.1,-0.2) {$\underline{x}$};
      \draw[black,densely dotted] (2.3, 1) -- (2.3, 0);
      \node[black] at (2.3,-0.2) {$\bar{x}$};
      \draw[black,densely dotted] (0, 0.1) -- (1, 0.1);
      \node[black] at (-0.2, 0.1) {$\underline{y}$};
      \draw[black,densely dotted] (0, 2.5) -- (0.5, 2.5);
      \node[black] at (-0.2,2.5) {$\bar{y}$};
      
      \draw[violet, thick] (0.37, 2) -- (1.1, 2);
      \node[violet, above] at (0.7,2) {$\bar{\region}$};
      
      \draw[violet, thick] (1, 0.1) -- (2.1, 0.1);
      \node[violet, above] at (1.5,0.1) {$\underline{\region}$};
      
      \draw[black] (2, 0) -- (0, 2);
      
      \node[red, inner sep=1, fill, circle] (H) at (1.1, 2) {};
       \node[red,above] at (1.15,2.05) {$H$};
       \node[red, inner sep=1, fill, circle] (L) at (1, 0.1) {};
       \node[red,above] at (1.05,0.1) {$L$};
       \node[red, inner sep=1, fill, circle] (Ro) at (1.9, 0.1) {};
       \draw[->,purple] (2.6, 0.4) -- (1.9, 0.1);
        \node[purple, right] at (2.6,0.4) {$\region_0$};
         \draw[densely dotted] (2, 0) -- (2, 2); 
    \end{tikzpicture}
  \end{center}
  Let 
\begin{align} \label{eq:barx}\underline{x}= \inf_{(x,y) \in \region}x\,, \qquad \qquad \bar{x} = \sup_{(x,y) \in \region}x\,,\qquad \qquad\underline{y}= \inf_{(x,y) \in \region}y\,, \qquad \text{and} \qquad \bar{y} = \sup_{(x,y) \in \region}y\,.\end{align}
Define 
\[
\widebar{\region} = \{(x,y) \in \region : y = \sup_{(x',y') \in \region } \min\{y',m\}  \} 
\]
as a subset of region $\region$ under which the total high-reward demand (more precisely $\min\{y', m\}$ for  any point $(x', y')\in \region$) is maximized. Similarly, we define \[
\underline{\region} = \{(x,y) \in \region : y = \inf_{(x',y') \in \region } \min\{y',m\}  \}
\] 
to be a subset of region $\region$ under which the total high-reward demand (more precisely $\min\{y', m\}$ for any point $(x', y')\in \region$) is minimized. We then define point $L=(x_{L}, y_{L})\in \underline{\region}$  as the point in set $\underline{\region}$ that has the lowest total low-reward demand.   We further define $H=(x_{H}, y_{H})\in \widebar{\region} $ as the point in set $\widebar{\region}$ that has the highest total low-reward demand. Mathematically speaking, we let
\[
L=\inf_{x}\{(x,y): (x,y) \in \underline{\region}\} \quad \text{and}\quad H=\sup_{x}\{(x,y): (x,y) \in \overline{\region} \}.
\]
Observe that point $L$ has the lowest reward among all the points in set $\underline{\region}$ and point $H$ has the highest reward among all the points in set $\widebar{\region}$. Further note that while by definition we have $y_{L}\le y_{H}$, we can have $x_{L}$ less than or greater than $x_H$.    

Then, we define the upper and lower envelop of $\region$ as $\bar h(\cdot)$ and $\underline h(\cdot)$, where 
\begin{align} \label{eq:lower_upper}\bar h(x) = \sup \{y: (x,y) \in \region \}\quad \text{and}\quad \underline h(x) = \inf\{y: (x,y) \in \region \}\end{align}for $x \in [\underline x, \bar x]$. As $\region$ is a convex set, we have $\underline h(x)$ is convex and $\bar h(x)$ is concave. Finally, let 
\[\mathcal R_{0} = \{(x, \underline h(x)): x\in [\underline x, \bar x]\} \cap \{(x,y): x+y =m\}\]
be the intersection of the lower envelop $\underline h(x)$ of the region $\region$ and line $ x+y =m$. Note than for any point  above (below respectively) this line, the total demand is greater (less respectively) than the number of resources $m$. As $\underline h(\cdot)$ is a convex function, we have $\vert \mathcal R_0 \vert \in \{0,1,2\}$ because a convex function and a line can have at most two intersection points.  


{\color{black}\subsection{Roadmap} Having presented the model,  in Section \ref{sec:adaptive}, we introduce the family of adaptive protection level algorithms  and explain their key structural properties.  
In Section \ref{sec:maxrobust}, we formulate the optimization problem that defines the \(C\)-Pareto optimal algorithm and transform the consistency and robustness constraints into more tractable forms.  
In Section \ref{sec:opt}, as illustrated in Figure \ref{fig:diagram}, we solve the problem by decomposing it into two subproblems:  
(i) the left subproblem (Figure \ref{fig:diagram}, left side), which handles sequences consistent with the ML advice; and  
(ii) the right subproblem (Figure \ref{fig:diagram}, right side), which handles sequences deviating from the ML advice.  
Finally, we combine the solutions to construct the \(C\)-Pareto optimal algorithm, achieving the highest possible robust ratio while guaranteeing a consistent ratio of at least \(C\).  
}

\section{Adaptive Protection Level Algorithms} \label{sec:adaptive}
{\color{black}In this section, we introduce the class of adaptive protection level algorithms and establish their key structural properties. These algorithms form the foundation for the optimization problem introduced in Section \ref{sec:maxrobust} and the solution framework outlined in Figure \ref{fig:diagram}.}

Adaptive protection level algorithms (PLAs) have a simple and intuitive structure, making them easy to implement. They extend the classical protection level algorithms introduced by \cite{littlewood2005special, ball2009toward} by allowing dynamic adjustment of the protection level based on observed low-reward demand. The formal definition of adaptive PLAs is provided below.

\begin{definition} [Adaptive Protection Level Algorithms] \label{def:1} A  PLA is defined by a continuous non-increasing Protection Level  (\pl{}) function  $p:\mathbb{R}_+\rightarrow  [0,m]$,  with $p'(x)\ge -1$ and $p(x)=p(\max\{m,\bar{x}\})$ for any $x \geq \max\{m,\bar{x}\}$. \footnote{For a discussion on the necessity of the validity conditions for the PL functions listed in Definition \ref{def:1}, please refer to Appendix \ref{sec:valid}.} 
Under   a PLA with a PL of $p(\cdot)$,
\begin{itemize}[leftmargin=*]
    \item high-reward requests are fully fulfilled unless we do not have enough resources left. That is, for any request $(s, z =h)$, we set the allocation $a$ to the minimum of $s$ and the remaining resources at the time of the decision.
    \item For low-reward requests, let 
    $\bar s$ be the sum of the sizes of the low-reward requests received  so far (excluding the current one) and define $\bar a$ as the total amount  of resources we have allocated to low-reward requests so far. Further, let $s$ be the size of the current low-reward request.  
     Under   a PLA with a PL of $p(\cdot)$, we allocate   $a\in [0, s]$   amount of the resource to the current low-reward request, where 
    \[a = \min\left(\widehat m,\,  \text{Proj}_{[0,s]}\big(m- p(\bar s+s) - \bar a\big)\right).\] 
    Here, $\text{Proj}_{[0,y]}(x) = \min(\max(x, 0), y)$ and $\widehat m$ is the remaining number of resources at the time of the decision. 
    Then,  it is guaranteed that $\bar a+ a\le m- p(\bar s+s)$, meaning we protect $ p(\bar s+s)$ resources for the  high-reward requests. 
\end{itemize}

\end{definition}
\medskip
To gain a better understanding of the definition of PLAs, let us consider a scenario where the size of all requests is $1$, and $m-p(x)$ is an integer. In this case, a PLA will always accept high-reward agents as long as we have the necessary resources. Moreover, the PLA will reject the $x$-th low-reward agent if either the number of low-reward agents already accepted is equal to $m-p(x)$, or there are no resources left.

It is worth noting that in a PLA, the number of units that we protect for high-reward requests is solely dependent on the total low-reward demand we have observed so far, which is represented by $\bar s+s$ in Definition \ref{def:1}. This property makes PLA a practical and easy-to-implement approach. Another important observation is that the PLA introduced in \cite{ball2009toward} can be represented by a fixed protection level function: $p(x) = m- \frac{m}{2-r_{\ell}/r_h}$ for any $x\ge 0$.

\subsection{ A Property of PLAs}
The following lemma shows that under PLAs, the consistent and robust ratios get minimized under ordered arrival sequences under which all low-reward requests arrive before high-reward requests.  See Appendix \ref{sec:prooforder} for the proof.

\begin{lemma} [Ordered Sequences]\label{lem:order1}
{\color{black}Suppose that we use a valid PLA $\mathcal A$ with a PL function $p(\cdot)$ (per Definition \ref{def:1}). 
Let ${\cal I}(x,y)$ denote the set of all arrival instances that contain a total of $x$ low-reward requests and $y$ high-reward requests (in any order). 
Let $\tilde{I}(x,y)$ be the ordered sequence in which all $x$ low-reward requests arrive first, followed by all $y$ high-reward requests. 
Then, for any $I \in {\cal I}(x,y)$,
$
\frac{\rew(\mathcal{A},I)}{\textsc{opt}(I)} \;\;\geq\;\; 
\frac{\rew(\mathcal{A},\tilde{I}(x,y))}{\textsc{opt}(\tilde{I}(x,y))}.
$
}
\end{lemma}

{\color{black}Lemma~\ref{lem:order1} justifies that, for performance analysis, it suffices to consider ordered instances. 
In particular, when we write $p(x)$ from this point onward, the argument $x$ should be understood not as the cumulative low-reward demand observed dynamically (as in Section~\ref{sec:adaptive}), but as the \emph{total} low-reward demand in an ordered instance paired with $y$ high-reward demand. 
This shift is methodological: the online nature is preserved because ordered instances capture the worst case for any online policy, and the monotonicity and slope constraints on $p(x)$ ensure that it remains a valid adaptive protection rule that could be implemented in an online fashion; see Appendix \ref{sec:valid}.}

In light of Lemma \ref{lem:order1}, we finish this section with a few definitions. 
 For any point $A= (x,y) $, let  $\CP(p;A= (x,y))$ be the \emph{compatible ratio} of  point $A= (x, y)$ under a \pl{} of $p$ when  a total of $x$ low-reward requests arrive first, followed by $y$ high-reward requests. That is, $\CP(p;(x,y))$ is the ratio of the obtained reward under \emph{ordered arrivals} associated with point $A= (x, y)$ and a fixed \pl{} of $p$ to the optimal clairvoyant solution.  When $p \geq \min\{m,y\}$, we  over-protect high-reward requests and hence 
\begin{align}
    \CP_o(p;A= (x,y)) = \frac{\min\{y,m\}r_h+\min\{x,(m-p)^{+}\}r_{\ell}}{\min\{y,m\}r_h+\min\{x, (m-y)^{+}\}r_{\ell}}, \qquad  \text{ if $p \geq \min\{m,y\}$  }\,. \label{eq:CP_over}
\end{align}
Here, the subscript ``o" in $\CP_o$ stands for over-protection. Note that in the definition of $\CP_o$,   $\min\{y,m\}r_h+\min\{x,m-p\}r_{\ell}$ is the reward at the ordered point $A=(x, y)$ with protection level $p\ge \min\{m, y\}$. Observe that we accept $\min\{y,m\}\le p$ high-reward request and $\min\{x,m-p\}$ low-reward requests.

On the other hand,  when $p <  \min\{m,y\}$, we  under-protect high-reward requests, and hence 
\begin{align}  \label{eq:CP_under}
\CP_u(p;A= (x,y)) =  
\; \frac{\max\{p, \min\{y,(m-x)^{+}\} \} r_h+\min\{x,m-p\}r_{\ell}}{\min\{y,m\}r_h+\min\{x, (m-y)^{+}\}r_{\ell}}, \qquad  \text{ if $p < \min\{m,y\}$ }\,. 
\end{align}
Here, the subscript ``u" in $\CP_u$ stands for under-protection. We allocate $\min\{x,m-p\}$ to low-reward requests. Then, if the total request is less than $m$, that is, $y<(m-x)^{+}$, we will accept all $y$ high-reward type requests. Otherwise, we will accept $\max\{p,(m-x)^{+}\}$ high-reward type requests. In summary, the total reward is $\max\{p, \min\{y,(m-x)^{+}\} \} r_h+\min\{x,m-p\}r_{\ell}$.

 When it is not clear whether we are in under- or over-protecting case, we simply use  $\CP(p;A= (x,y)) $ to denote compatible ratio of  point $A= (x, y)$.
Note that in the definition of $\CP(p;A)$, we considered ordered arrivals  where low-reward agents arrive first. This is because  of Lemma  \ref{lem:order1} where we show for any point $A$, the compatible ratio is minimized by considering its ordered sequence. 
 Then, by definition, the consistent ratio  and robust ratio of a PLA $\mathcal{A}$ with PL function $p(\cdot)$ are given by 
\begin{equation} \label{eq:defconsis}
    \cons(\mathcal{A}) = \inf_{(x,y) \in \region} \CP(p(x); (x, y)), \qquad 
    \comp(\mathcal{A}) = \inf_{(x,y) \ge 0} \CP(p(x); (x, y)).
\end{equation}
{\color{black}As noted earlier, following Lemma~\ref{lem:order1}, in Equation~\eqref{eq:defconsis} the argument $x$ of $p(x)$
is interpreted as the {total} low-reward demand in an ordered instance (rather than the
cumulative demand observed dynamically), so that the analysis captures the worst case for any
online algorithm.}

\section{Optimization Problems to Characterize Pareto Optimal Algorithms} \label{sec:maxrobust}

{\color{black}Next, we formulate the optimization problem that defines the \(C\)-Pareto optimal algorithm. As shown in Figure \ref{fig:diagram}, we will later simplify the solution process by decomposing the problem into two subproblems. To make the problem more tractable, we first transform the consistency and robustness constraints. The solution to this modified problem is presented in Section \ref{sec:opt}.}

 A $C$-Pareto optimal algorithm
   achieves the highest possible robust ratio while obtaining  a consistent ratio of at least $C$ for any $C\le \cstar$, where $\cstar$ is  the highest possible consistent ratio for a convex ML region $\region$ (in the absence of any constraint on the robust ratio); see Section \ref{sec:mlconsistent} for more details about $\cstar$. Suppose that the $C$-Pareto optimal algorithm is a PLA. Then,  to design a $C$-Pareto optimal algorithm,  by Equation \eqref{eq:defconsis}, we consider
the following optimization problem: 
\begin{align} \notag
     \max_{R\in[0,1], p(x):x \in [0,\max\{m,\bar{x}\}]} & \ R   \\  \label{eq:cpc1}
s.t.   \ & \ \CP(p(x); (x, y)) \geq C, \qquad \ (x,y) \in \region,  \\  \label{eq:cpc2}
\ & \ \CP(p(x); (x, y)) \geq R, \qquad \ x,y \geq 0,  \\  \label{eq:cpc3}
\ & \  \text{$p(x)\ge 0$ is continuous} \qquad \ x \in [0,{\max \{m, \bar x\}}],  \\ \label{eq:cpc4}
\ & \ -1 \le  p'(x) \le 0 \quad \text{ a.e.}, \qquad  \ x \in [0,{\max \{m, \bar x\}}]
\\ \tag{\text C-Pareto} \label{prob:rbmax}
\end{align}
  Recall that $\bar{x}$ is defined in Equation \eqref{eq:barx}. The first set of constraints ensures that the compatible ratio of any point $(x,y)\in \region$ at $p(x)$ is at least $C$. The second set of constraints ensures that the compatible ratio of any point $(x,y)$ in the first quadrant at $p(x)$ is at least $R$. The third and fourth sets of constraints, which we refer to as validity constraints, ensure that the protection level is valid per Definition \ref{def:1}. Note that the optimal solution to Problem \eqref{prob:rbmax} only determines the PL function $p(x)$ for $x\in[0,\max{m,\bar{x}}]$. As per the definition of PLAs in Definition \ref{def:1}, for any $x>\max\{m,\bar{x}\}$, we set $p(x)=p(\max\{m,\bar{x}\})$.

As we will show later in Theorem \ref{thm:optimal_trans}, the optimal solution to Problem \eqref{prob:rbmax} leads to a PLA that is an optimal $C$-Pareto algorithm among any deterministic and randomized non-anticipating algorithms (not only the PLAs). 

To solve Problem \eqref{prob:rbmax}, we first  transform its first and second sets of constraints (i.e., the consistency and robustness constraints). This transformation, which is done  in Sections \ref{subsec: consistrans} and \ref{subsec: robusttrans}, plays a key role in our design. We then solve Problem \eqref{prob:rbmax} using  the  properties of the transformed constraints.

\subsection{Transforming the Consistency  Constraints}  \label{subsec: consistrans}
Here, we  transform the consistency constraints in Problem \eqref{prob:rbmax} that require $\CP(p(x); (x, y)) \geq C$  for any $(x,y) \in \region$. To do so, let us fix $x\in [\underline x, \bar x]$  and its protection level $p(x)$. Then,  the worst points with the minimum consistent ratio  occur at the lower and upper  boundaries of the ML region. That is, as shown in Lemma \ref{lem:ymonotone} stated below,  for any $p\in [0,m]$,  $\min_{y\in [\underline h(x), \bar h(x)]} \CP(p; (x, y)) $ is   either  $ \CP(p; (x, \underline h(x)))$  or  $\CP(p; (x, \bar h(x)))$, where $\underline h(x)$ and $\bar h(x)$ are defined in Equation \eqref{eq:lower_upper}.  

\begin{lemma}[Monotonicity of the Compatible Ratio $\CP(p; (x, y))$ w.r.t. $y$] \label{lem:ymonotone} 
    Consider  any $x \in [\underline{x},\bar{x}]$ and $ \min\{m,y_1\}\le  \min\{m,y_2\}$. For  any  protection level $p$ with   $p \leq \min\{m,y_1\}$, we have
    $
    \CP_u(p;(x,y_1)) \geq \CP_u(p;(x,y_2))$, 
and for any  protection level $p$ with   $p \geq \min\{m,y_2\}$, we have
    $
    \CP_o(p;(x,y_2)) \geq \CP_o(p;(x,y_1))$. 
    This further implies that for any $x\in[\underline x, \bar x]$ and $p\ge 0$,  we have
    \[\min_{y\in [\underline h(x), \bar h(x)]} \{\CP(p; (x, y))\} = \min\left\{ \CP(p; (x, \underline h(x))),  \CP(p; (x, \bar h(x)))\right\}\,.\]
\end{lemma}

In light of Lemma \ref{lem:ymonotone}, we define the following two functions $\u(x; C)$ and $\ll(x;C)$. Roughly speaking,  under $\u(x; C)$, (if possible) the compatible ratio at (the worst over-protected  point) $(x, \underline h(x))$ is equal to $C$ while under $\ll(x; C)$, (if possible) the compatible ratio at (the worst under-protected  point) $(x, \bar h(x))$ is equal to $C$. 

\begin{definition}[Upper Bound Function]
   For any $C \in [0,1]$,  we define 
 \begin{align}\u(x; C)= \sup\big \{p\in [0,m]: \CP_o(p;(x,\underline{h}(x)))=C\big\}\qquad  x \in [\underline x_u,\bar x_u]\, \label{eq:u}\end{align} 
 while we set $\u(x; C) =m$ for any $x \in [0,\underline{x}_{u}] $ and $\u(x; C) =\u(\bar x_u; C)$ for any $x \in [\bar x_u,\bar x] $.
 Here, 
\begin{align} \label{eq:xmin} \bar x_u = \left\{ \begin{array}{ll}
         x_L & \quad \mbox{if $x_L+y_L \geq m$};\\
        \sup\{x \in [x_L,\bar{x}]: (1-C)\frac{r_h}{r_{\ell}}\underline{\mathcal{H}}'(x^{-})-C < 0 \} & \quad  \mbox{Otherwise}\,,\end{array} \right. \end{align}   and  
$\underline x_u=\sup\{\underline{x}<x<\bar x_u: \CP_o(m;(x,\underline{h}(x))) \geq C \}
$, where $\underline{\mathcal{H}}(x)=\min\{\underline h(x),m\}$. 
 We set $\underline x_u = \underline x$  when its defining set is empty. \footnote{We note that   $\bar x_u$ is well defined because  $L$ is the lowest point, and hence $\underline{\mathcal{H}}'(x_L^{-}) \leq 0$, and $(1-C)\frac{r_h}{r_{\ell}}\underline{\mathcal{H}}'(x_L^{-})-C < 0$.}
 \end{definition}

 \begin{definition}[Lower Bound Function] 
For any $C \in [0,1]$,  let 
\begin{align} \ll (x; C) = \inf\big \{p\in [0,m]: \CP_u(p;(x,\overline{h}(x)))=C \big\}  \qquad x \in [x_H,\bar{x}_{l}]\,,\label{eq:l}\end{align}
while we set $\ll(x; C) =\ll(x_H; C)$ for any $x \in [0,x_H] $ 
 and $\ll(x; C) =0$ for $x\in  [\bar x_l,\bar{x}]$. Here, 
 $\bar{x}_{l}=\inf\{x_H<x<\bar{x}: \CP_u(0;(x,\overline{h}(x))) \geq C \}$.
 We set $\bar x_l$  to   $\bar x$  when its defining set is empty.
\end{definition}\smallskip

\begin{figure}[H]
    \centering
    \includegraphics[width=.5\textwidth]{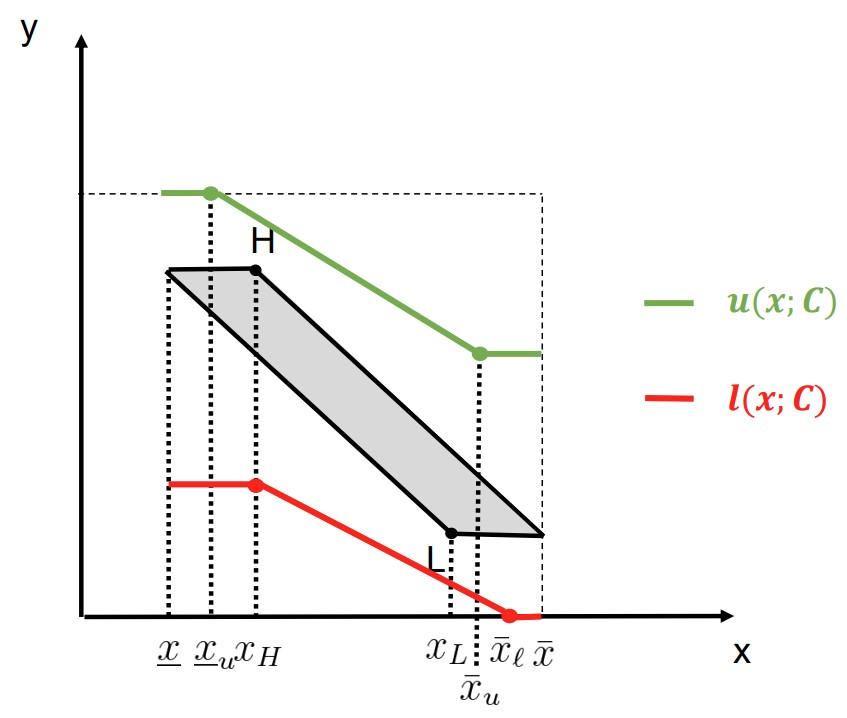}
   \caption{{ The \textbf{green line} represents the upper bound function \( u(x; C) \), and the \textbf{red line} represents the lower bound function \( l(x; C) \).
The point \( \underline{x}_u \) marks the \textbf{starting point} where the upper bound function \( u(x; C) \) begins to decrease. Intuitively, this is the largest value of \( x \) where the system can still afford to maintain the maximum protection level (i.e., \( m \)) while satisfying the compatible ratio requirement at the worst over-protected point \( (x, \underline{h}(x)) \). The point \( \bar{x}_u \) represents where the upper bound function \textbf{stops decreasing} and flattens out. This indicates that decreasing protection beyond this point is no longer necessary due to the large number of low-reward requests.
On the other hand, \( \underline{x}_l \) represents the smallest value of \( x \) where the system can afford to set protection to zero while still satisfying the compatible ratio at the worst under-protected point \( (x, \overline{h}(x)) \).\label{fig:Example_points}
}}
\end{figure}

{\color{black}We now discuss the upper and lower bounds, as illustrated in Figure \ref{fig:Example_points}.  

For the upper bound, at \( x \in [\underline x, \underline x_u) \), even if we set the protection level to \( m \)—which means rejecting all low-reward type requests—the compatible ratio at \( (x, \underline h(x)) \) exceeds \( C \). This corresponds to the first {flat part} of the green line in the figure.  

For \( x \in [\underline x_u, \bar x_u] \), setting the protection level to \( u(x; C) \) makes the compatible ratio at the worst over-protected point \( (x, \underline h(x)) \) exactly equal to \( C \). This is reflected in the {decreasing section} of the green line. As shown in Lemma \ref{lem:property_u_l}, for any $C \in [0,1]$,  $\u(x;C)$ is non-increasing  in $x\in [0, \bar x]$ and is convex for $x \in [\underline x_u, \bar x]$..  

For the lower bound, at \( x \in (\bar{x}_{l}, \bar x] \), even if the protection level is zero, the compatible ratio at \( (x, \bar h(x)) \) exceeds \( C \). For any \( x \in [\underline x_l, \bar x_l] \), setting the protection level to \( l(x; C) \) makes the compatible ratio at the worst under-protected point \( (x, \bar h(x)) \) exactly equal to \( C \). This appears in the {decreasing section} of the red line in the figure. As shown in Lemma \ref{lem:property_u_l}, for any $C \in [0,1]$, $\ll(x;C)$ is non-increasing in $x\in [0, \bar x]$ and is concave for $x \in [\underline x, \bar x_{\ell}]$. Further, \( l(x;C) \leq u(x; C) \) for any \( x \in [0, \bar x] \).}


We are now ready to present our transformation result. 

\begin{lemma}[Transforming the Consistency Constraints-I] \label{lem:feasibleregionrobust}
 For any $C\le \cstar$,  Problem \eqref{prob:rbmax} is equivalent to the following optimization problem: 
 \begin{equation} \label{prob:rbmaxtrans1}
     \max_{R\in[0,1], p(x):x \in [0,\max\{m,\bar{x}\}]}  R \qquad \text{s.t. }\quad  \ll(x;C) \leq p(x) \leq \u (x;C), ~~ x \in [0,\bar{x}], \qquad 
     \text{ and \eqref{eq:cpc2}, \eqref{eq:cpc3}, \eqref{eq:cpc4}}\,,
\end{equation}
where $\ll(x;C)$ and $\u(x;C)$ are defined in Equations \eqref{eq:l} and \eqref{eq:u}, respectively. 
\end{lemma}

The proof of all lemmas in this section can be found in Appendix \ref{sec:appendixC}. Lemma \ref{lem:feasibleregionrobust} demonstrates that enforcing lower and upper bounds on the PL function $p(\cdot)$ satisfies the consistency constraint of Problem \eqref{prob:rbmax}. These bounds, denoted by $l(\cdot;C)$ and $u(\cdot;C)$, respectively, can be easily computed and depend on the ML region. However, these bounds may not be tight since their slope can be less than $-1$. Recall that the optimal solution to Problem \eqref{prob:rbmax} should be a valid PL function, i.e., a non-increasing function with a slope greater than or equal to $-1$. With this in mind, the following lemma presents an alternative (valid) lower bound denoted by $\wll(\cdot;C)$, which is a non-increasing function with a slope greater than or equal to $-1$. This lower bound is defined as follows:
\begin{align} \label{eq:ll} \wll(x; C) = \left\{ \begin{array}{ll}
        \ll(x; C) & \quad  \mbox{$x\in [0, \nex]$ }\\
        (-(x-\nex)+\ll(\nex; C))^{+} & \quad  \mbox{$x\in [\nex, \bar{x}]$}
        .\end{array} \right. \end{align} 
Note that $x^{+}=\max\{x,0\}$ and $\nex=\sup\{x \in [x_H,\bar{x}]: \frac{\partial \ll(x^{-};C)}{\partial x} \le -1 \}$. When $\frac{\partial \ll(x^{-};C)}{\partial x} \geq -1$ for any $x\in [x_H, \bar x]$, we set $\nex$ to $\bar x$, and in this case, $\ll(x; C) = \wll(x; C) $ for any $x\in [\underline x, \bar x]$. We recall that $\frac{\partial \ll(x^{-};C)}{\partial x}= C \overline{\mathcal{H}}'(x)$ by Lemma \ref{lem:property_u_l}, and that $\overline{\mathcal{H}}(\cdot)=\min\{\bar h(\cdot),m\}$.
We observe that $\wll(x; C)$ is a non-increasing continuous function whose slope is greater than or equal to $-1$, due to the concavity of $\ll(x; C)$ in $x$ as shown in Lemma \ref{lem:property_u_l}. Moreover, we have $\ll(x; C)\le \wll(x; C)$ for any $x\in [\underline x, \bar x]$. Finally, the following lemma shows that $\wll(\cdot; C)$ is a tighter lower bound for the PL function $p(\cdot)$. 

\begin{lemma}[Transforming the Consistency Constraints-II] \label{lem:feasibleregionrobust_2}
 For any $C\le \cstar$,  Problem \eqref{prob:rbmax} is equivalent to the following optimization problem:  
 \begin{equation}  \label{prob:rbmaxtrans2}
     \max_{R\in[0,1], p(x):x \in [0,\max\{m,\bar{x}\}]}  R \qquad \text{s.t. } \quad \wll(x;C) \leq p(x) \leq \u (x;C), ~~x \in [0,\bar{x}],\qquad  \text{ and \eqref{eq:cpc2}, \eqref{eq:cpc3}, \eqref{eq:cpc4}} \,,
\end{equation}
where $\wll(x;C)$ and $\u(x;C)$ are defined in Equations \eqref{eq:ll} and \eqref{eq:u}, respectively. 

\end{lemma}


\subsection{Transforming the Robustness  Constraints} \label{subsec: robusttrans}
In this section, we aim to transform the robustness constraints of Problem \eqref{prob:rbmax} which mandate that $\CP(p(x); (x, y)) \geq R$ for any $(x,y)  
\ge 0$. Similar to the previous section, we replace this constraint with lower and upper bound constraints on the (PL) function $p(\cdot)$. This enables us to effectively address the robustness constraints while simplifying the optimization problem.

 \begin{lemma} [Transforming the Robustness Constraints] \label{lem:trapezoid}
    For any $ R\in (0,\rho]$ with $\rho=\frac{1}{2-r_{\ell}/r_h}$, and for $x \in [0,m]$, define  
    \begin{equation} \label{eq:defgR}
    \underline{g}(x;R) = \frac{m(R-r_{\ell}/r_h)}{1-r_{\ell}/r_h},\qquad 
    \bar{g}(x;R) = -Rx+m.
    \end{equation}
    For $x>m$, we have $\underline{g}(x;R)=\underline{g}(m;R)$ and $\bar{g}(x;R)=\bar{g}(m;R)$.
    For any $C\le \cstar$,  Problem \eqref{prob:rbmax} is equivalent to the following optimization problem: 
\begin{align}\notag 
     \max_{R\in[0,\rho], p(x):x \in [0,\max\{m,\bar{x}\}]} & \ R   \\  \notag 
s.t.\quad    \ & \wll(x;C) \leq p(x) \leq \u (x;C),\qquad x \in [0,\bar{x}]\,,  \\  \notag 
\ & \underline{g}(x;R) \leq p(x) \leq \bar{g}(x;R), \qquad x \in [0,\max\{m,\bar{x}\}]\,,  \\ 
&\text{Validity Constraints \eqref{eq:cpc3}, \eqref{eq:cpc4}}\tag{C-Pareto-Trans}\,. \label{prob:trans}
\end{align}
\end{lemma}

\section{Pareto Optimal Algorithms} \label{sec:opt}
{\color{black}We now solve the optimization problem from Section \ref{sec:maxrobust}. As shown in Figure \ref{fig:diagram}, we handle this by first solving the right subproblem and then the left subproblem, integrating the two solutions to derive the \( C \)-Pareto optimal algorithm.

Let an optimal solution to the transformed problem \eqref{prob:trans} be denoted by \(\pc(\cdot)\), which achieves the optimal robust ratio of \(\Rstar\). The problem decomposes into two subproblems:

\begin{minipage}[t]{0.48\textwidth}
\textbf{Left Problem:} Solve for \(x \in [0, \bar{x}]\)
\begin{align} 
\begin{aligned}
\RL = & \max_{R\in[0,\rho], p(x):x \in [0,\bar{x}]}  \ R   \\  
\text{s.t.} & \ \wll(x;C) \leq p(x) \leq \u (x;C),\quad x \in [0,\bar{x}] \\  
& \underline{g}(x;R) \leq p(x) \leq \bar{g}(x;R), \quad x \in [0,\bar{x}] \\ 
& \text{Validity Constraints \eqref{eq:cpc3}, \eqref{eq:cpc4}} \\ 
& p(\bar x) =\pright(\bar x)
\end{aligned}
\label{prob:left} \tag{C-Pareto-left}
\end{align}
\end{minipage}
\hfill
\begin{minipage}[t]{0.48\textwidth}
\raggedright
\textbf{Right Problem:} Solve for \(x \geq \bar{x}\)
\begin{align} 
 \begin{aligned}
\RR = & \max_{R\in[0,\rho], p(x):x \in [\bar{x},\max\{m,\bar{x}\}]}  \ R   \\  
\text{s.t.} \ & \ \wll(\bar{x};C)  \leq p(\bar{x}) \leq \u (\bar{x};C),  \\ 
& \underline{g}(x;R)\leq p(x) \leq \bar{g}(x;R), \qquad  x \in [\bar{x},\max\{m,\bar{x}\}]\,  \\ 
& \text{Validity Constraints \eqref{eq:cpc3}, \eqref{eq:cpc4}} 
\end{aligned}
\label{prob:right} \tag{C-Pareto-right}
\end{align}
\end{minipage}
\medskip

\noindent In the right problem, it is sufficient to satisfy the consistency lower and upper bound constraints only at \( \bar{x} \). Let \( \pright(x) \) be the optimal solution to the right problem, and define \( \RR \) as the optimal objective value. Theorem \ref{thm:optimal_trans} proves that \( \pc(x) = \pright(x) \) for any \( x \geq \bar{x} \).  

After solving the right problem, we solve the left problem by enforcing that \( p(\bar{x}) = \pright(\bar{x}) \). Let \( \pleft(x) \) be the optimal solution to the left problem and define \( \RL \) as the optimal objective value. Theorem \ref{thm:optimalrobCconsistent} shows that \( \pc(x) = \pleft(x) \) for any \( x \in [0, \bar{x}] \). Finally, the optimal objective value of Problem \eqref{prob:trans} is \( \Rstar = \min \{\RR, \RL\} \). }

\begin{theorem}[Optimal Solution to   \eqref{prob:trans} and $C$-Pareto Optimal Algorithm] \label{thm:optimal_trans}
Consider any $0 \le C\le \cstar$. 
\begin{enumerate}
\item The optimal objective value of Problem \eqref{prob:trans}, is $\Rstar=\min\{\RR , \RL\}$, where $\RR$ and $\RL$ are the optimal objective value of Problem \eqref{prob:right} and \eqref{prob:left}, respectively.  
\item Algorithm \ref{alg:trans} presents an optimal solution to Problem \eqref{prob:trans}. That is, at the optimal solution to Problem \eqref{prob:trans}, denoted by $\pc(\cdot)$, for any $x\in [0, m]$, we set
  \begin{align}   \label{eq:opt_pc}\pc(x) = \left\{ \begin{array}{ll}
         \pleft(x)  &\quad  x\in [0, \bar x]\\
        \pright(x) &\quad  x\in [ \bar x,\max\{m,\bar x\}]\,, \\
        \end{array} \right. \end{align}
  where   $\pright(\cdot)$ and $\pleft(\cdot)$ are the optimal solutions to the right and left problems, respectively.

\item   A PLA with  the PL function of $\pc(\cdot)$ is an optimal solution to Problem \eqref{prob:original}. That is, among any online algorithms $\Pi$, the aforementioned algorithm maximizes the robust ratio while ensuring its consistent ratio is at least $C$. 
\end{enumerate}
\end{theorem}

\begin{algorithm}\footnotesize
  \caption{Optimal Solution to Problem \eqref{prob:trans} and $C$-Pareto Optimal Algorithm. \label{alg:trans}
  }
  \begin{itemize}
\item[]\textbf{Input:} Convex set $\region$, resource capacity $m$, and parameter $C\in [0, \cstar]$.
\item[]\textbf{Output:} Optimal solution to Problem \eqref{prob:trans}, $\pc(\cdot): [0,\max\{m,\bar x\}]\mapsto [0,m]$, and its optimal objective $\Rstar$.
\end{itemize}
\begin{itemize}
\item[]For any $x\in [0, \max\{m,\bar x\}]$, set
\end{itemize}
  \begin{align}   \label{eq:pc}\pc(x) = \left\{ \begin{array}{ll}
         \pleft(x)  &\quad  x\in [0, \bar x]\\
        \pright(x) &\quad  x\in [ \bar x,\max\{m,\bar x\}]\,, \\
        \end{array} \right. \end{align}
\begin{itemize}
\item[]
where   $\pright(\cdot)$ and $\pleft(\cdot)$ are defined in Algorithms  \ref{alg:right} and \ref{alg:left}, respectively.
\item[]Let $\RR=\min\{\CP_o(\pright(\bar{x});(\bar{x},0)),\CP_u(\pright(m);(\max\{m,\bar x\},m))$, and set
\end{itemize}
\[\Rstar=\min\{ \RR, \inf_{x \in [0,\bar{x}]}\CP_o(\pleft(x);(x,0))\}.\] 
\end{algorithm}

\subsection{Proof Sketch  of Theorem \ref{thm:optimal_trans}}
    
The proof of  Theorem \ref{thm:optimal_trans} is stated in Appendix \ref{sec:proofmain1}. 
To show the first statement, as the main step, we need to argue that $\pc(\bar x) = \pright(\bar x)$.  To do so, we consider two cases, where in the first case $\RR \ge  \RL$, and in the second case, $\RR< \RL$. In the first case, the optimality of $\pc(\bar x)$ can be argued using 
Theorem \ref{thm:optimalextension}, where we present an optimal solution to the right problem. 
Otherwise, for the case where $\RL<\RR$, we show the result  by contradiction while using   properties  of the lower bound $\wll(\cdot;C)$. The proof of the second statement follows from  Theorems \ref{thm:optimalextension} and \ref{thm:optimalrobCconsistent} in which 
we will present an optimal solution to the left and right problems.

 To show the third statement, we first note that the optimal robust ratio $\Rstar$ is the minimum of three terms: $\CP_o(\pright(\bar{x});(\bar{x},0))$, $\CP_u(\pright(m);(\max\{m,\bar x\},m))$, and $\inf_{x \in [0,\bar{x}]}\CP_o(\pleft(x);(x,0))\}$.  Depending on which term attains the minimum, we construct worst-case arrival sequences to show that no algorithm can perform better than $\Rstar$. For the purpose of this discussion, let us assume that $\Rstar$ is equal to the compatible ratio of point $(m,m)$ (i.e., $\CP_u(\pright(m);(\max\{m,\bar x\},m))$), which is one of the three aforementioned terms.

For this case, we define two (ordered) input sequences. In the first input sequence, $I_1$, low-reward requests with $\bar{x}_u\leq \bar{x}$ arrive first, followed by high-reward requests with $\underline{h}(\bar{x}_u)$.  One can think of $I_1$ as an arrival sequence consistent with the ML advice. In the second input sequence, $I_2$, $m$ low-reward requests arrive first, followed by $m$ high-reward requests. Here, one can think of $I_2$ as an arrival sequence outside with the ML advice. Before receiving $\bar{x}_u$ low-reward requests, any deterministic or randomized algorithm cannot differentiate between the two input sequences and must decide how many low-reward requests to accept in expectation. We then show that, on these input sequences, to achieve a consistent ratio of $C$ on $I_1$, no algorithm can obtain a robust ratio greater than $\Rstar$ on $I_2$. This is demonstrated by arguing that, on the arrival sequence $I_1$, to achieve a consistent ratio of at least $C$, the algorithm must accept at least $m-u(\bar{x}_u;C)$ low-reward agents. This, in turn, prevents any algorithm from performing better than $\Rstar$ on $I_2$. For other cases, we also construct two input sequences, but the contradiction point is different in each case. See Section \ref{sec:proofmain1} for details.


\subsection{Examples}


\begin{example} \label{example}
Consider the following ML regions:
\[\region_1 = \{4 \leq x \leq 16 \} \cap \{4 \leq y \leq 16\} \cap \{20 \leq x+y\leq 25 \}\]  and \[\region_2 =\{4 \leq x \leq 16 \} \cap \{4 \leq y \leq 16\} \cap \{0 \leq y-x\leq 5 \}\,.\] 
The first region, $\mathcal{R}_1$, represents a scenario where both the low-reward and high-reward requests are between $4$ and $16$, and the total number of requests is between $20$ and $25$. The second region, $\mathcal{R}_2$, corresponds to a scenario where both the low-reward and high-reward requests are between $4$ and $16$, and the difference between the high-reward and low-reward requests is between $0$ and $5$. The regions $\mathcal{R}_1$ and $\mathcal{R}_2$ are illustrated as the shaded grey areas in the two leftmost plots of Figure~\ref{fig:optpareto}.

Here, we set $m=20$, $r_h=1$, and $r_{\ell}=1/3$. In the two leftmost plots, we consider $C \in \{0.8, 0.89\}$. These two plots in Figure~\ref{fig:optpareto} display $\pc(\cdot)$ for $C=0.8$ and $C=0.89$ under the two regions, respectively. We observe that increasing $C$ leads to an increase in $\pc(\cdot)$, which results in protecting more resources for high-reward requests.

The rightmost plot of Figure~\ref{fig:optpareto} shows the optimal robust ratio (i.e., the optimal value of Problem~\eqref{prob:rbmax} or equivalently, the original Problem~\eqref{prob:original}) as a function of $C$ for the two ML regions, $\mathcal{R}_1$ and $\mathcal{R}_2$. For each region $\mathcal{R}_i$, $i \in [2]$, we restrict to $C \leq C^{\star}(\mathcal{R}_i)$, since Problem~\eqref{prob:rbmax} becomes infeasible for any $C > C^{\star}(\mathcal{R}_i)$. As expected, for both ML regions, the optimal robust ratio decreases as $C$ increases.

\end{example}

\begin{figure} [H]
     \centering         
    \includegraphics[width=0.3\textwidth]{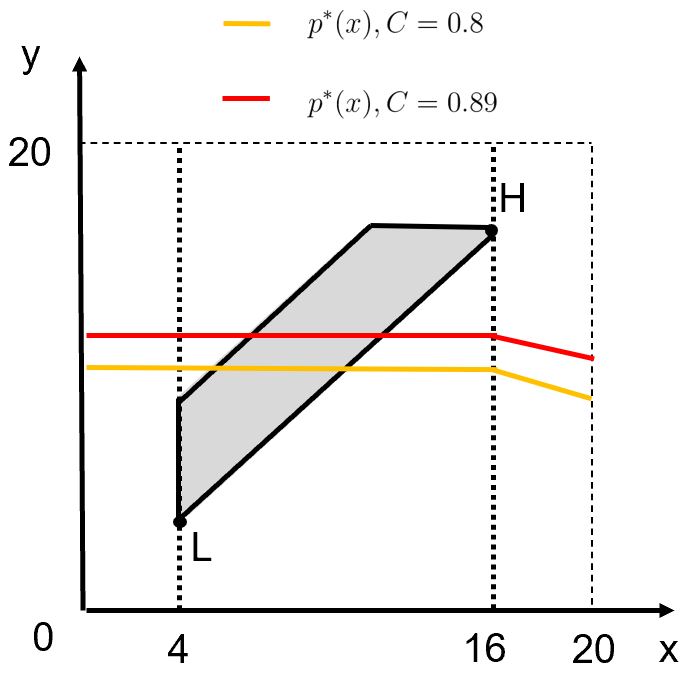}
    \includegraphics[width=0.3\textwidth]{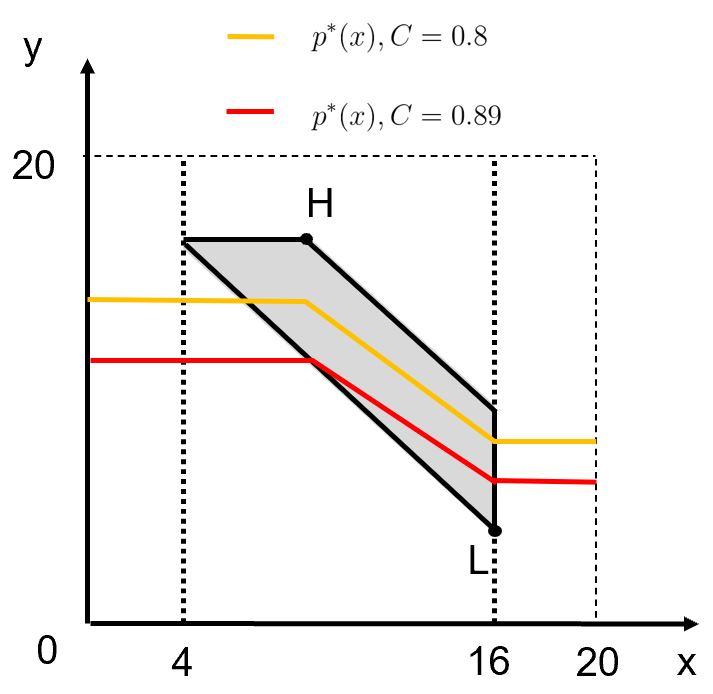}
         \includegraphics[width=0.3\textwidth]{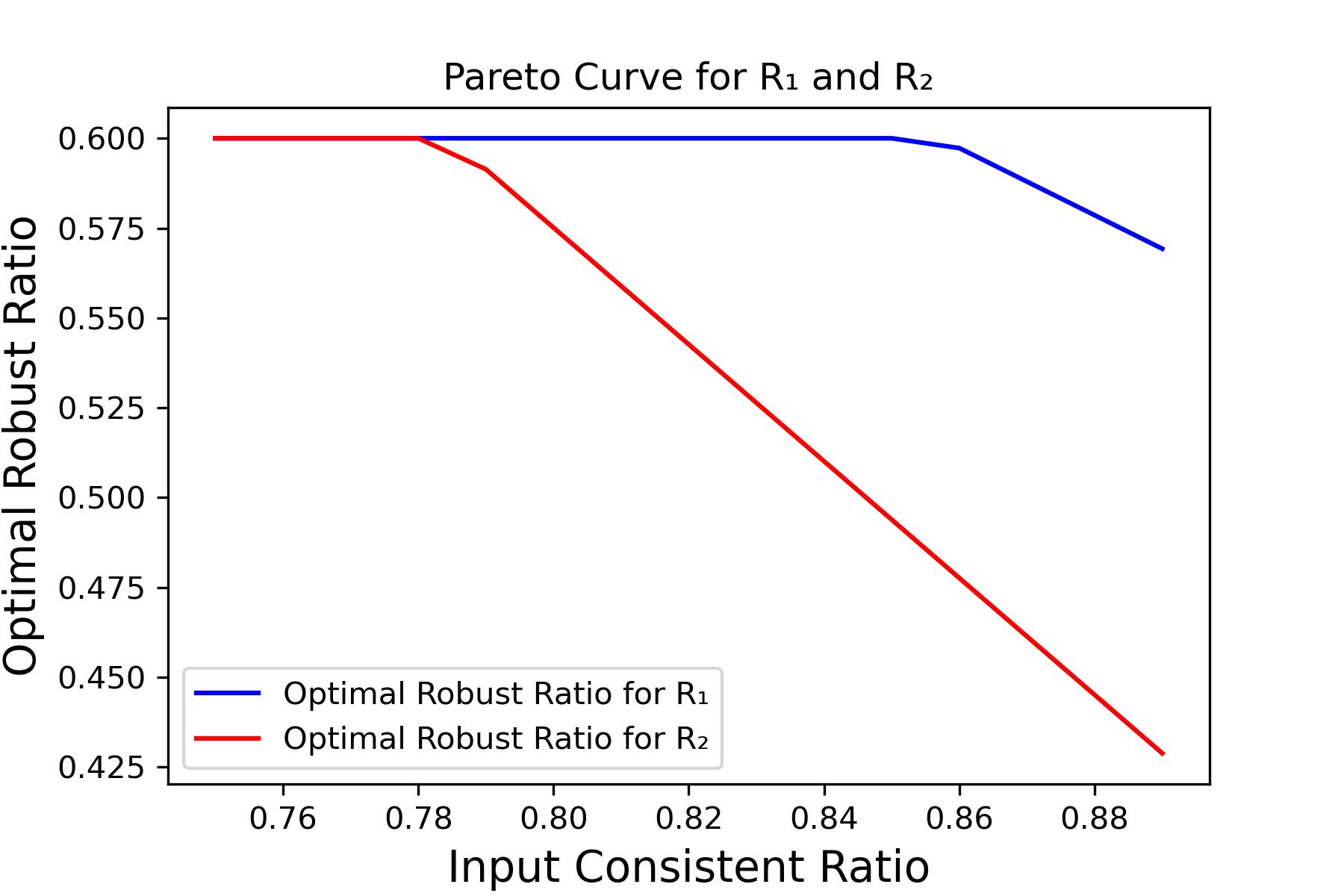}
    \caption{The figure shows the optimal robust ratio (i.e., the optimal value to Problem \eqref{prob:rbmax} or the optimal value to the original Problem \eqref{prob:original}) versus $C$ for the two ML region $\region_1$ and $\region_2$.
 Here,  $\region_1= \{4 \leq x \leq 16 \} \cap \{4 \leq y \leq 16\} \cap \{20 \leq x+y \leq 25 \}$ and $\region_2 =\{4 \leq x \leq 16 \} \cap \{4 \leq y \leq 16\} \cap \{0 \leq y-x \leq 5 \}$.}
        \label{fig:optpareto}
\end{figure}

In the following sections, we begin by presenting $\pright (\cdot)$, followed by a characterization of $\pleft(\cdot)$.

\subsection{Optimal Solution to the Right Problem \eqref{prob:right}} \label{sec:optimalright}
Algorithm \eqref{alg:right} presents the optimal solution to the right problem \eqref{prob:right}, $\pright(\cdot)$, and the optimal objective value of this problem $\RR$. First, the algorithm  presents the optimal solution at $\bar x$, i.e., $\pright(\bar x)$, which then determines $\RR$. Second, the algorithm presents the optimal solution at any  $x\in (\bar x, \max\{m,\bar{x}\}]$, using  $\pright(\bar x)$ and $\RR$. Notice that if $\bar{x} \geq m$, $(\bar x, \max\{m,\bar{x}\}]=(\bar x, \bar x]$,  is not well-defined. So, we skip the second step, and we only need to solve for $\pright(\bar x)$.

\begin{algorithm} \footnotesize
  \caption{Optimal Solution to Problem \eqref{prob:right}. \label{alg:right}}
  \label{alg:DR-SM}
  \begin{itemize}
 \item[] \textbf{Input:} Convex set $\region$, resource capacity $m$, and parameter $C\in [0, \cstar]$.
 \item[] \textbf{Output:} Optimal solution to  the right problem \eqref{prob:right}, $\pright: [\bar x, \max\{m,\bar x\}]\mapsto [0,m]$, and optimal objective value of the right problem \eqref{prob:right}, denoted by $\RR$.
 \end{itemize}
  \begin{itemize}
  \item \textbf{Optimal solution at $\bar x$} Let   \[\pright(\bar{x})= \argmin_{p \in [\wll(\bar{x};C),\u(\bar{x};C)]} \vert p - \underline{g}(\bar{x}) \vert\,,\] 
  where  
    $\underline{g}(x) = \underline{g}(x;R=\frac{1}{2-r_{\ell}/r_h}) = \frac{1-r_{\ell}/r_h}{2-r_{\ell}/r_h}m
    $, and 
  $\wll(\cdot; C)$ and  $\u(\cdot; C)$ are respectively defined in Equations  \eqref{eq:ll} and \eqref{eq:u}. Further,  define 
  \[\RR=\min\left\{\CP_o(\pright(\bar{x});(\bar{x},0)),\CP_u(\pright(\bar{x});(\max\{m,\bar x\},m))\right\}\,.\]
  \item   \textbf{Optimal Solution at $x\in (\bar x, \max\{m,\bar x\}]$.} For any $x\in (\bar x, \max\{m,\bar x\}]$, define 
  \begin{align} \label{eq:rightext}  \pright(x) = \left\{ \begin{array}{ll}
         \max\{-x+\bar{x}+\pright(\bar{x};C), \underline g(x)\}  &\quad  \pright(\bar{x}) \in [\underline g(x), \bar g(x)]\\
        \underline{g}(x; \RR) &\quad  \pright(\bar{x})<\underline{g}(\bar{x}), \\
         \bar{g}(x; \RR) &\quad  \pright(\bar{x})>\bar{g}(\bar{x}), \\
        \end{array} \right. \end{align}
        where
        $\underline g(\cdot; R)$, $\bar g(\cdot; R)$ are respectively defined in Equation \eqref{eq:defgR}.
  \end{itemize}
  \begin{itemize}
\item[]\textbf{Return.} $\pright(x)$, $x\in [\bar x, \max\{m,\bar x\}]$, and $\RR$.
  \end{itemize}
\end{algorithm}

 Recall that $\rho = \frac{1}{2-r_{\ell}/r_h}$ is the optimal consistent ratio (obtained by \cite{ball2009toward}) when the ML region is $\{(x,y): x,y\ge 0\}$. 
 With a little abuse of notation, for $x \in [\bar{x},\max\{m,\bar x\}]$, we let 
    \[
    \bar{g}(x) = \bar{g}(x;R=\rho), \quad  \text{and} \quad 
    \underline{g}(x) = \underline{g}(x;R=\rho) \,,
    \]
where $\bar{g}(x;R)$ and $\underline{g}(x;R)$ are defined in Equation \eqref{eq:defgR}.  Then, in Algorithm \ref{alg:right}, we have 
\begin{align} \label{eq:barx_opt}
\pright(\bar{x})= \argmin_{p \in [\wll(\bar{x};C),\u(\bar{x};C)]} \vert p - \underline{g}(\bar{x}) \vert\,.
\end{align}
Clearly, $\pright(\bar{x})$ satisfies the first constraint in the right problem; that is, $\pright(\bar{x})\in [\wll(\bar x;C), \u(\bar x; C)]$, as desired.  To set $\pright(\bar{x})$, the algorithm compares the feasible interval $[\wll(\bar x;C), \u(\bar x; C)]$ with $ \underline{g}(x) = \underline{g}(x;R=\rho)$. Here, $\underline{g}(x;\rho)= \frac{1-r_{\ell}/r_h}{2-r_{\ell}/r_h}m$ is the optimal PL in the setting studied in \cite{ball2009toward} when the ML region is $\{(x,y): x,y\ge 0\}$. 
 At a high level, if we can set $\pright(\bar x)$  to $\underline g(\bar x)$, the optimal (right) robust ratio $\RR$ will be equal to $\rho$ (which is the maximum robust ratio for any ML region). However, setting  $\pright(\bar x)$ to $\underline g(\bar x)$ is not always possible.  In such a case, we either set 
 $\pright(\bar{x})$ to $\u(\bar{x};C)$ or $\wll(\bar{x};C)$. The value $\pright(\bar{x})$ then determines $\RR$: 
 \begin{align}\label{eq:RR}
        \RR=\min\left\{\CP_o(\pright(\bar{x});(\bar{x},0)),\CP_u(\pright(\bar{x});(\max\{m,\bar x\},m))\right\}\,.
\end{align}
This shows that in the right problem,  points $(\bar x, 0)$ or $(\max\{m,\bar x\}, m)$ play a crucial role, determining the optimal objective value. 

Given the optimal objective value of $\RR$, if $\bar{x} < m$, for any $x\in (\bar x, m]$, we only need to make sure that $\underline{g}(x;\RR) \leq p(x) \leq \bar{g}(x;\RR)$. This is achieved by setting 
  \begin{align}   \label{eq:opt_x_beyond}\pright(x) = \left\{ \begin{array}{ll}
         \max\{-x+\bar{x}+\pright(\bar{x};C), \underline g(x)\}  &\quad  \pright(\bar{x}) \in [\underline g(x), \bar g(x)]\\
        \underline{g}(x; \RR) &\quad  \pright(\bar{x})<\underline{g}(\bar{x}), \\
         \bar{g}(x; \RR) &\quad  \pright(\bar{x})>\bar{g}(\bar{x}). \\
        \end{array} \right. \end{align}
        As it becomes more clear in the proof of Theorem \ref{thm:optimalextension}, 
 when $\pright(\bar{x}) \in [\underline{g}(\bar{x}),\bar{g}(\bar{x})]$, the optimal objective value of the right problem $\RR$ is indeed $\rho =1/(2-r_{\ell}/r_h)$. For the other cases where either $\pright(\bar{x})<\underline {g}(\bar{x})$ or  $\pright(\bar{x})>\bar{g}(\bar{x})$, we have $\RR< \rho$. There, we set $\pright(x)$ such that the compatible ratio at any points $(x, 0)$ and $(x, m)$ (with $x\in (\bar x, m]$) is greater than or equal to $\RR$, defined in Equation \eqref{eq:RR}. 

\begin{theorem}[Optimal Solution to the Right Problem] \label{thm:optimalextension} 
Algorithm \ref{alg:right} presents an optimal solution to Problem \eqref{prob:right}. That is, at the optimal solution to Problem \eqref{prob:right}, denoted by $\pright(\cdot)$, we set $\pright(x)$ based on Equations \eqref{eq:barx_opt} and \eqref{eq:opt_x_beyond}. Furthermore, the optimal objective value of  Problem \eqref{prob:right}, $\RR$, is given in Equation \eqref{eq:RR}.
\end{theorem}

\subsection{Optimal Solution to Problem \eqref{prob:left}}\label{sec:optimalleft}
In this section, we present an optimal solution to the left problem \eqref{prob:left}, denoted by $\pleft(\cdot): [0, \bar x]\mapsto [0, m]$. See Algorithm \ref{alg:left}. The algorithm shows that  at the optimal solution, we set 
\begin{align}\label{eq:pleft}
 \pleft(x)=\max\{\wll(x;C),\pright(\bar{x})\}\,, \quad x\in [0, \bar x]\,,
\end{align}
where $\pright(\bar{x})$ is the optimal solution to Problem \eqref{prob:right} at $\bar x$, and $\wll(x; C)$ is defined in Equation \eqref{eq:ll}.

\begin{algorithm}\footnotesize
  \caption{Optimal Solution to Problem \eqref{prob:left}. \label{alg:left}}
  \begin{itemize}
 \item[]\textbf{Input:} Convex set $\region$, resource capacity $m$, and parameter $C\in [0, \cstar]$.
 \item[]\textbf{Output:} Optimal solution to  the right problem \eqref{prob:left}, $\pleft: [0,\bar x]\mapsto [0,m]$.
 \end{itemize}
 \begin{itemize}
\item[]For any $x\in [0, \bar x]$, define 
\end{itemize}
\begin{align}
 \pleft(x)=\max\{\wll(x;C),\pright(\bar{x})\}\,, \quad x\in [0, \bar x]\,,
\end{align}
\begin{itemize}
\item[]where $\pright(\bar{x})$ is the optimal solution to Problem \eqref{prob:right} at $\bar x$, and $\wll(x; C)$ is  defined in Equation \eqref{eq:ll}.
\end{itemize}
\begin{itemize}
\item[]\textbf{Return.} $\pleft(x)$, $x\in [0,\bar x]$.
\end{itemize}
\end{algorithm}
The optimal solution to the left problem is obtained by one  observation. We  show in the proof of Theorem \ref{thm:optimalrobCconsistent} that in this problem, one can ignore the lower bound constraint $p(x)\ge \underline g(x; R)$, $x\in [0, \bar x]$. Given this simplification, to present an optimal solution to Problem \eqref{prob:left}, we need to find the largest value of $R$ that satisfies the conditions $\max\{\wll(x;C),\pright (\bar x)\}\le p(x) \le \bar g(x;R)$, where $\bar g(x;R)$ is a decreasing function of $R$ and $p(x)$ is a non-increasing function.  
Then, considering the fact that $\bar g(x; R)$ is decreasing in $R$, to maximize $R$ while ensuring
$\max\{\wll(x;C),\pright (\bar x)\}\le p(x) \le \bar g(x; R)$, we set $\pleft(x) = \max\{\wll(x;C),\pright (\bar x)\}$. 

\begin{theorem}[Optimal Solution to the Left Problem] \label{thm:optimalrobCconsistent} 
Algorithm \ref{alg:left} presents an optimal solution to Problem \eqref{prob:left}. That is, at the optimal solution to Problem \eqref{prob:left}, denoted by $\pleft(\cdot)$, we set $\pleft(x)$ based on Equation \eqref{eq:pleft}. The optimal objective value of Problem \eqref{prob:left} is:
\[\RL = \min \{\CP_u(\pleft(\bar{x});(\bar{x},m))),\inf_{x \in [0,\bar{x}]}\CP_o(\pleft(x);(x,0)) \}\,.\]
\end{theorem}

\section{Optimal Consistent Ratio } \label{sec:mlconsistent}
In this section, using an optimization problem, we first characterize the optimal/maximum consistent ratio ($\cstar$)  that any algorithm can achieve under ML region $\region$; see Section \ref{sec:Crstart}.  For any convex ML region $\region$, in Section \ref{subsec:bi}, we then present a simple bisection method that allows us to obtain a good approximation for $\cstar$.  Section \ref{subsec:fast}  presents a simpler algorithm to obtain the exact value of $\cstar$ when the ML region $\region$ is a polyhedron.  

\subsection{Characterizing the Optimal Consistent Ratio }\label{sec:Crstart}
We begin by the following theorem:

\begin{theorem}[Characterizing Optimal Consistent Ratio] \label{thm:cmax}
Consider any convex ML region $\region$. Then, the consistent ratio of any online algorithm under  ML region $\region$ is at most $\cstar$ where 
\begin{align*} 
\begin{aligned} 
     \cstar= &\max_{C\ge 0, p(x):x \in [0,\bar{x}]}  \ C   \\  \notag
s.t.  \ & \ \wll(x; C) \leq p(x) \leq \u(x;C), \qquad \ x \in [\underline{x},\bar{x}]  \\ &\text{Validity Constraints \eqref{eq:cpc3}, \eqref{eq:cpc4}} \quad  x \in [0,\bar{x}]\,. 
\end{aligned}
\tag{\textsc{C-max}} \label{prob:cpmax1}
\end{align*} 
Here, $\wll(x;C)$ and $\u(x;C)$ are defined in Equations \eqref{eq:ll} and \eqref{eq:u}, respectively. 
Furthermore, $\cstar \ge \rho$, where $\rho =1/(2-
r_{\ell}/r_h)$.
\end{theorem}

Theorem \ref{thm:cmax} characterizes $\cstar$ using an optimization problem \eqref{prob:cpmax1} that bears resemblance with Problem \eqref{prob:trans}. In Problem \eqref{prob:cpmax1}, we aim to characterize a valid PL function under which the consistent ratio is maximized. Note that 
by the transformation Lemma \ref{lem:feasibleregionrobust_2},  
the first  constraint (i.e., $\wll(x;C)\le p(x)\le \u(x; C)$) is equivalent  to $\CP(p(x); (x,y))\ge C$ for any $(x, y)\in \region$. The resulting PL function then leads an optimal  PLA that obtains the consistent ratio of $\cstar$, which is the maximum  consistent ratio that any online algorithms can achieve. 
\subsection{A Bisection Method to Compute $\cstar$ for a General Convex Region} \label{subsec:bi}
Here we present a simple bisection method that allows us to compute an $\epsilon$-accurate estimate of $\cstar$. The bisection method (Algorithm \ref{alg:bisection}) crucially uses the first set of constraints in Problem \eqref{prob:cpmax1}: 

\textbf{Necessary and Sufficient Conditions for $C \leq \cstar$: } By Lemma \ref{lem:feasibleregionrobust_2}, we have for any $C \in [\rho,1]$, $\cstar \ge C$ if and only if for any $x \in [\underline{x},\bar{x}]$, we have 
    \[
    \wll(x;C) \leq \u(x;C)\,,
    \] 
where we     
recall that $\wll(x;C)$ and $\u(x;C)$ are defined in Equations \eqref{eq:ll} and \eqref{eq:u}, respectively.

\begin{algorithm}[H]\footnotesize
  \caption{A Bisection Method to Compute $\cstar$. \label{alg:bisection}} 
\begin{itemize}
\item[]\textbf{Input:} Convex set $\region$, resource capacity $m$,  and accuracy parameter $\epsilon \in [0, 1/2]$.
\item[]\textbf{Output:} An $\epsilon$-accurate estimate of $\cstar$.
\end{itemize}
\begin{itemize}
\item[]Initialize $C_0 = \rho$ and $C_1 = 1$, where $\rho = 1/(2-r_{\ell}/r_h)$. 
\item[]While $|C_1-C_0|\ge \epsilon: $
\end{itemize}
\begin{itemize}
\item Compute the mid point $C_m = (C_0+C_1)/2$.
\item If $\wll(x;C_m)\le \u(x;C_m)$ for any $x\in [\underline x, \bar x]$, set $C_0$ to $C_m$.
\item  If $\wll(x;C_m)> \u(x;C_m)$ for some $x\in [\underline x, \bar x]$, set $C_1$ to $C_m$.
\end{itemize}
\begin{itemize}
\item[]\textbf{Return:} $C_0$.
\end{itemize}
\end{algorithm}

 In Algorithm \ref{alg:bisection},  we use a bisection procedure that  repeatedly checks if $\wll(x;C) \le \u(x;C)$,  $x\in [\underline x, \bar x]$, for some given $C$. Note that this condition can be easily checked considering the fact that we have a closed form solution for $\wll(x; C)$ and $\u(x;C)$. Furthermore, checking this condition
 is equivalent to check if $\min_{x \in [\underline{x},\bar{x}]}\u(x;C)-\wll(x;C)\ge 0$, where we highlight 
 this optimization can be easily solved. This is because 
 by Lemma \ref{lem:property_u_l}, we know that $\u(x;C)$ is convex in $x$, and  $\wll(x;C)$ is concave in $x$. This implies that     $\u(x;C)-\wll(x;C)$ is convex, and the aforementioned problem is a convex optimization problem. 
 The following proposition sheds light on the performance of  Algorithm \ref{alg:bisection}.

 \begin{proposition}[Bisection Method to Compute $\cstar$]\label{prop:bisection}
     Consider Algorithm \ref{alg:bisection} with an accuracy parameter $\epsilon \in [0, 1/2]$. Given a general convex set $\region$, Algorithm \ref{alg:bisection} returns a $C_0 \in [\cstar-\epsilon,\cstar+\epsilon]$, where $\cstar$ is the optimal solution to Problem \eqref{prob:cpmax1}. In addition, the computational complexity of Algorithm \ref{alg:bisection} is $O(\log(1/\epsilon))$.     
 \end{proposition}
 
 \subsection{A Faster Method to Compute $\cstar$ for Polyhedron Convex Regions } \label{subsec:fast}
 In the previous section, we presented a bisection method to estimate $\cstar$ for a general convex ML region. Here, we present a faster method to compute $\cstar$ when the ML region is a convex polyhedron. This method relies on the following theorem.

\begin{theorem} [Properties of $\cstar$ under Polyhedron ML Regions] \label{thm:property_cstar}
Let $\V$ be the set containing the $x$ value of all vertices of $\region$ and of the set $\mathcal R_{0}$, where $\mathcal R_{0} = \{(x, \underline h(x)): x\in [\underline x, \bar x]\} \cap \{(x,y): x+y =m\}$. Then, $\cstar =C$ for some $C\in [\rho, 1]$ if and only if the following two conditions hold.
    \begin{enumerate}
        \item for any $x \in \V$, $\wll(x;C) \leq \u(x;C)$.
        \item there exists $\widehat{x}\in \V$, such that $\wll(\widehat{x};C) = \u(\widehat{x};C).$
    \end{enumerate}
    Here, we     
recall that $\wll(x;C)$ and $\u(x;C)$ are defined in Equations \eqref{eq:ll} and \eqref{eq:u}, respectively. 
    \end{theorem}

    Theorem \ref{thm:property_cstar} shows that when the ML region $\region$ is a polyhedron, for any $x\in \V$, we have $\wll(x;\cstar) \leq \u(x;\cstar)$ as long as 
    there exists $\widehat x\in \V$ under which the lower bound $\wll(\widehat x; \cstar)$ is equal to the upper bound $\u(\widehat x;\cstar)$. Here, $\V$ is the set containing the $x$ value of all vertices of $\region$, and the $x$ value of $\mathcal R_{0}$. We refer to $\V$ as the set of x-vertices. 
  This theorem shows that when the ML region is a polyhedron,  we can simplify the feasibility check in Algorithm \ref{alg:bisection} by checking the condition  $\wll(x;C) \leq \u(x;C)$ only for any x-vertices $x\in \V$. While this is an improvement, we   present a faster algorithm that returns the exact value of $\cstar$ by taking advantage of properties of the lower and upper bounds $\wll(\cdot;C)$
and $\u(\cdot;C)$, presented in Lemma \ref{lem:property_u_l}.

\begin{algorithm}[htbp]\footnotesize
\caption{An Algorithm to Compute $\cstar$ for Polyhedron ML region $\region$ }\label{alg:mlconsis}
\begin{itemize}
\item[]\textbf{Input:} Polyhedron ML region  $\mathcal{R}$, resource capacity $m$.
\item[]\textbf{Output:} Optimal consistent ratio $\cstar$.
\item[]\textbf{Initialization:} Set $\mathcal S =\emptyset$.
\end{itemize}
\begin{itemize}
\item For any pair of x-vertices $x_1 , x_2 \in \V $ with $\bar h(x_1) \ge \underline h(x_2)$ and $x_2\le x_1$, find the following balancing PL $p\in [\underline h(x_2), \bar h(x_1))]$ such that 
\begin{align}\label{eq:balance1}\CP_u(p; (x_1,\bar{h}(x_1)))=\CP_o(p; (x_2,\underline{h}(x_2))\,.\end{align}
(Note that $x_1$ can be equal to $x_2$.) Add $\CP_u(p; (x_1,\bar{h}(x_1)))$ to $\mathcal S$.
\item For any pair of x-vertices $x_1 , x_2 \in \V $ with $\bar{h}(x_1)-\underline{h}(x_2) \geq x_2-x_1$ and $x_2> x_1$, find the following balancing PL $p\in [\underline{h}(x_2)+(x_2-x_1),\bar{h}(x_1)]$  such that 
\begin{align}\CP_u(p; (x_1,\bar{h}(x_1)))=\CP_o(p-(x_2-x_1); (x_2,\underline{h}(x_2)))\,. \label{eq:balance2}\end{align}
Add $\CP_u(p; (x_1,\bar{h}(x_1)))$ to $\mathcal S$. That is, update $\mathcal S$ to $\mathcal S\cup \{\CP_u(p; (x_1,\bar{h}(x_1)))\}$.
\end{itemize}
\begin{itemize}
\item[]\textbf{Return:} Return the largest $C\in \mathcal S$ under which $\wll(x; C)\le \u(x;C)$ for any $x\in \V$ as $\cstar$:
\[\cstar = \max\{C\in \mathcal S: \wll(x; C)\le \u(x;C) \quad \text{for any $x\in \V$} \}\,.\]
\end{itemize}
\end{algorithm}

In the faster algorithm, at a high level, we aim to find the $x$-vertex $\widehat x$ under which $\wll(\widehat x;\cstar) = \u(\widehat x;\cstar)$. To do so, we follow an  enumeration technique that uses the property of $\widehat x$ along with the first condition in Theorem \ref{thm:property_cstar} that allows us to only focus on $x$-vertices to determine $\cstar$.   To explain the idea behind the algorithm, 
let us recall that in defining the upper bound $\u(x; C)$, when possible, we choose the protection level $\u(x; C)$ such that  the compatible ratio at point $(x, \underline h(x))$ is equal to $C$. That is, for any $x\in [\underline x_u, \bar x_u]$, we have $\u(x; C)= \sup\big \{p\in [0,m]: \CP_o(p;(x,\underline{h}(x)))=C\big\}$.  Similarly, in defining the upper bound $\ll(x; C)$ (which we later use to define $\wll$), when possible, we choose the protection level $\ll(x; C)$ such that  the compatible ratio at point $(x, \bar h(x))$ is equal to $C$. That is, for any $x \in [x_H,\bar{x}_{l}]$, we have $\ll (x; C) = \inf\big \{p\in [0,m]: \CP_u(p;(x,\overline{h}(x)))=C \big\}$. 

Now suppose that at $\widehat x$, we have $\ll(\widehat x;\cstar) = \wll (\widehat x;\cstar) $. Then, if 
$\widehat x \in [\underline x_u, \bar x_u] \cap [x_H, \bar x_l]$, the condition $\wll(\widehat x;\cstar) = \u(\widehat x;\cstar)$ leads to \emph{balancing} two compatible ratios. That is, we need to find a protection level $p\in [\underline h(\widehat x), \bar h(\widehat x))]$ such that
\[\CP_o(p; (\widehat x, \underline h(\widehat x))) =  \CP_u(p; (\widehat x, \bar h(\widehat x)))\,.\]
This balancing idea explains Equation \eqref{eq:balance1} in Algorithm \ref{alg:mlconsis}. Now, suppose that at $\widehat x$, we have $\ll(\widehat x;\cstar) \neq \wll (\widehat x;\cstar) $, which only happens when $\widehat x \ge \nex$, where $\nex=\sup\{x \in [x_H,\bar{x}]: \frac{\partial \ll(x^{-};\cstar)}{\partial x} \le -1 \}$. By definition of $\wll$ in Equation \eqref{eq:ll}, we then know that $\wll (\widehat x;\cstar) = \ll(\nex; \cstar) - (\widehat x- \nex)$. Then,  the condition  $\wll(\widehat x;\cstar) = \u(\widehat x;\cstar)$ leads to a slightly different {balancing} procedure in which we need to find a protection level $p\in [\underline h(\widehat x), \bar h(\widehat x))]$ (i.e., $\ll(\widehat x; \cstar)$) such that 
\[\CP_u(p; (\nex, \bar h(\nex))) =  \CP_o(p-(\widehat x- \nex); (\widehat x, \underline h(\widehat x)))\,.\]
This justifies Equation  \eqref{eq:balance2} in Algorithm \ref{alg:mlconsis}. (Note that as we show in the proof of Theorem \ref{thm:MLconsis}, for any $C$, $\nex$ is a x-vertex.)

\begin{theorem}[Optimal Consistent Ratio for Convex Polyhedron ML Regions] \label{thm:MLconsis} Suppose that the ML region $\region$ is a convex  polyhedron. Algorithm \ref{alg:mlconsis} returns the optimal consistent ratio $\cstar$ for any given polyhedron $\region$ in run time $O(|\V|^3)$. 
\end{theorem}

\section{Numerical Studies}\label{sec:simulations}
In this section, we present the results of our numerical studies, which highlight the efficiency of resource allocation achieved through the integration of our algorithms with ML advice. 
We will show that thanks to our algorithm, ML advice can significantly improve the average and worst case performance, outperforming other benchmarks.

\subsection {Setup} 

\textbf{Demand/Arrival models.} 
We conduct an analysis of two demand models, one with a uniform distribution,  and the other with a normal distribution. In both models, we introduce random noise to the demand process using a uniform distribution. Specifically, denoting $x$ and $y$ as the number of high reward arrivals and low reward arrivals, respectively, in the first demand model, we have:
\begin{align} \label{distri1}  x,y \sim \left\{ \begin{array}{ll}
         \text{Uniform}(10,20) &\quad  \text{with probability 0.9}, \\
         \text{Uniform}(0,30) &\quad   \text{with probability 0.1}. \\
        \end{array} \right. \end{align}
Similarly, in the second model, we have:
\begin{align} \label{distri2}  x,y \sim \left\{ \begin{array}{ll}
         \text{Normal}(15,3) &\quad  \text{with probability 0.9}, \\
         \text{Uniform}(0,30) &\quad   \text{with probability 0.1}\,. \\
        \end{array} \right. \end{align}
We set $m=20$, $r_h=1$, and $r_{\ell}=1/3$. We analyze both worst-case and uniform arriving orders in each demand model.

\textbf{Construction of ML advice.}  In the first approach, we fix the shape of the ML advice and, given the shape, determine its parameters. 
For the uniform demand distribution described in Equation \eqref{distri1} and illustrated in Figure \ref{fig:advice3}, a box-shaped advice is naturally suited. 
Conversely, for the normal demand distribution in Equation \eqref{distri2}, an ellipsoidal shape is more appropriate, as demonstrated in Figure \ref{fig:advice3}. 
We use $n=\{10,100\}$ as the number of samples used to construct the box and ellipsoid advice. 
To construct the box (ellipsoid) advice, we identify the smallest rectangle (ellipsoid) that encompasses at least $z\%$ of the sample points, where $z\% \in \{80\%,90\%\}$. 
For the ellipsoid advice, we apply the method in \cite{gartner1997smallest} to construct such an ellipsoid.

In the second approach, we adopt a data-driven method to construct ML advice (uncertainty sets), as delineated in the work of \cite{bertsimas2018data}, particularly in Theorem 5. 
This method utilizes statistical hypothesis tests, specifically the Kolmogorov-Smirnov (KS) test, to construct uncertainty sets. 
There are two parameters, $\epsilon$, and $\theta$. 
$\epsilon$ sets the allowable sum of relative entropies across all marginals in the uncertainty set, ensuring that the total divergence from the worst-case distributions is bounded by $\log(1/\epsilon)$. Further, 
$\theta$ controls the number of vertices of the polyhedron and thus the computational complexity of the optimization. 
Figure \ref{fig:advicedb} showcases examples of such polyhedral advice. 

{\color{black} Finally, in the third approach, we leverage neural networks (NN) to construct  ML advice. 
While NNs typically provide point forecasts, recent ensemble techniques enable probabilistic predictions in the form of prediction intervals. 
We adopt the method of \cite{pearce2018high}, training separate models for the high-reward and low-reward demand types using $n$ independent samples from the demand models in Equations \eqref{distri1} and \eqref{distri2}. 
Each model produces a prediction interval that we treat as a box-shaped advice set for Algorithm~\ref{alg:trans}. 

}


To construct the ML advice, we sample $n \in \{10,100\}$ data points $K = 1000$ times from the demand models described in Equations~\eqref{distri1} and~\eqref{distri2}. 
For each sample set of size $n$, we then generate ML advice using one of the following methods:  
(i) the data-driven polyhedral approach of \cite{bertsimas2018data},  
{\color{black}(ii) NNs to produce box-shaped prediction intervals following \cite{pearce2018high},} or  
(iii) classical geometric constructions in the form of a box or an ellipsoid, depending on the underlying demand model.  
Let $S_k$, where $k \in [K]$, denote the $k$-th sample dataset and $\region_k$ the corresponding ML advice.

\begin{figure} [H]
     \centering         
     \includegraphics[width=0.47\textwidth]{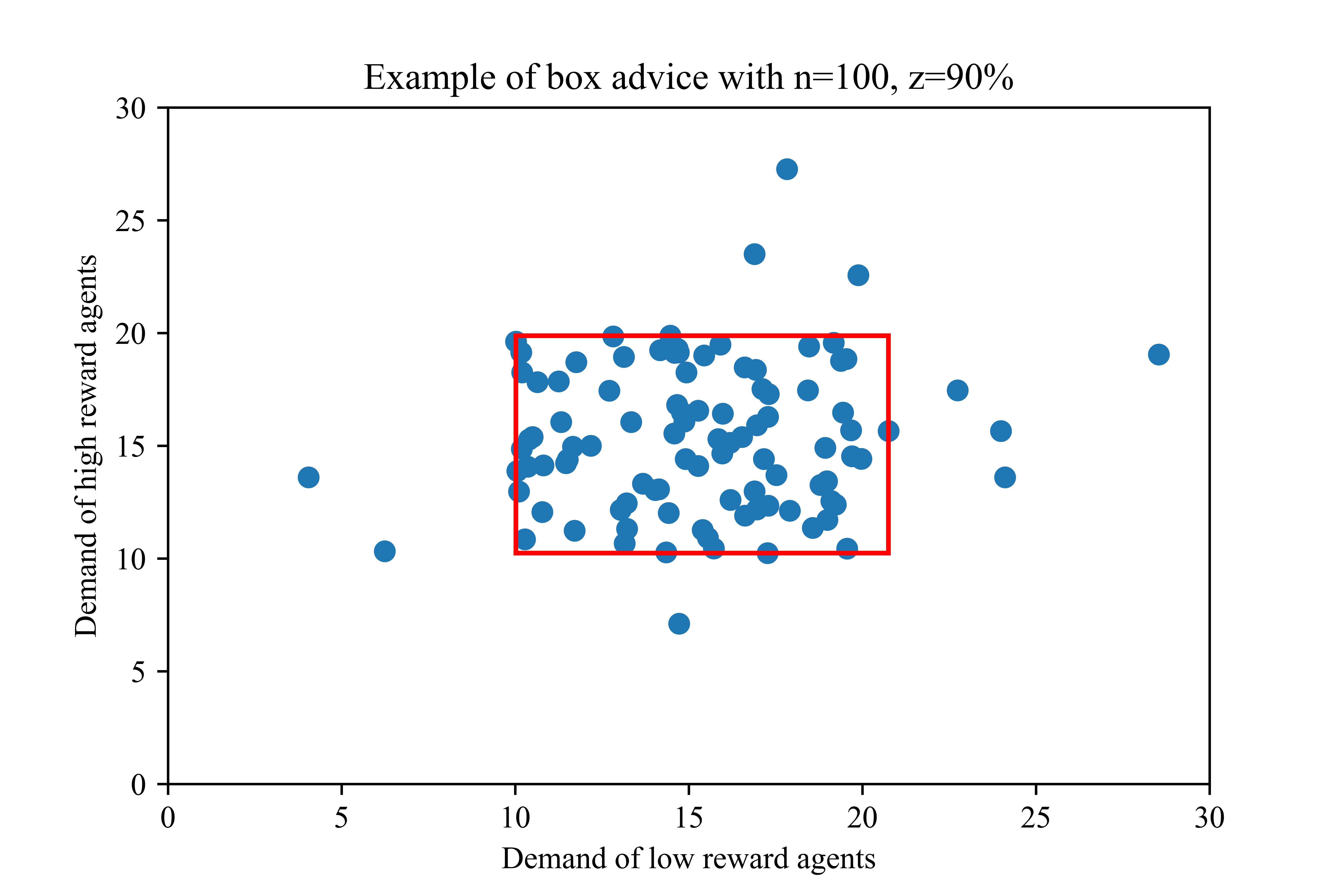}
      \includegraphics[width=0.47\textwidth]{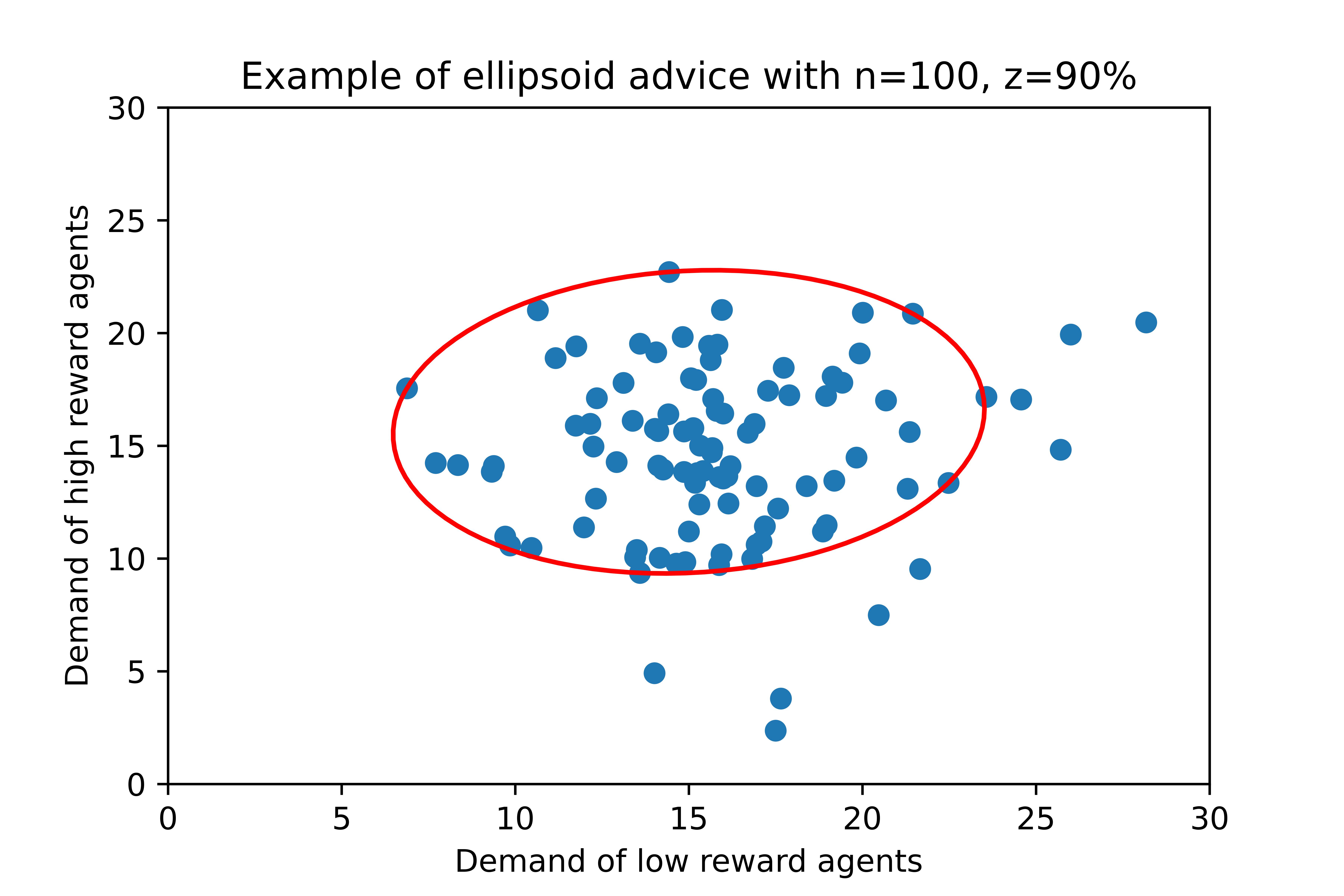}

    \caption{Constructing the box and ellipsoid advice with $n =100$ samples for the uniform and normal demand model in Equations \eqref{distri1} and \eqref{distri2} , respectively. Here,  $z= 90\%$ fraction of the samples fall into the constructed box. }
        \label{fig:advice3}
\end{figure}

\textbf{Performance evaluation.} 
To evaluate the performance of our algorithm and the benchmark algorithms that we will define shortly, 
we generate a test set which contains $100$ instances. Each test instance corresponds to a point $(x, y)$ drawn from the demand models described in Equations \eqref{distri1}  and  \eqref{distri2}. We assess the algorithm's performance using these test instances under two different scenarios: worst-case (ordered) arrival sequences and stochastic (uniform order) arrival sequences.

Let $\mathcal{T}$ denote the set of test instances. 
Then, the worst CP and Avg. CP of an  algorithm $\mathcal{A}$ are respectively defined as:
\begin{align} \label{eq:avg_worst_adv}\textsc{worst cp} = \frac{1}{K}\sum_{k=1}^{K}\min_{(x, y) \in \mathcal{T}} \textsc{CP}_{\mathcal{A}}(x, y; S_k) \quad \textsc{avg. cp} = \frac{1}{K}\sum_{k=1}^{K}\frac{1}{|\mathcal{T}|} \sum_{(x, y) \in \mathcal{T}} \textsc{CP}_{\mathcal{A}}(x, y; S_k)\,.\end{align}
Here, $\textsc{CP}_{\mathcal{A}}(x, y; S_k)$ is the compatible ratio of algorithm $\mathcal A$ under the sample set $S_k$, which is used to construct the ML advice when algorithm $\mathcal A$ is, for example,  Algorithm \ref{alg:trans}. 

For the stochastic arrival sequences, for each instance $(x, y) \in \mathcal{T}$, we generate $100$ random permutations. Similarly, we define: 
\[\textsc{worst cp} = \frac{1}{K}\sum_{k=1}^{K}\min_{(x, y) \in \mathcal{T}} \mathbb{E}[\textsc{CP}_{\mathcal{A}}(x, y; S_k)] \quad \textsc{avg. cp} = \frac{1}{K}\sum_{k=1}^{K}\frac{1}{|\mathcal{T}|} \sum_{(x, y) \in \mathcal{T}}\mathbb{E}[\textsc{CP}_{\mathcal{A}}(x, y; S_k)]\,,\]
where the expectation is taken with respect to the randomness in arrival permutations.

\textbf{Benchmarks.} To further assess the effectiveness of our algorithms, we propose several benchmarks:
\begin{figure} [H]
     \centering         
     \includegraphics[width=0.31\textwidth]{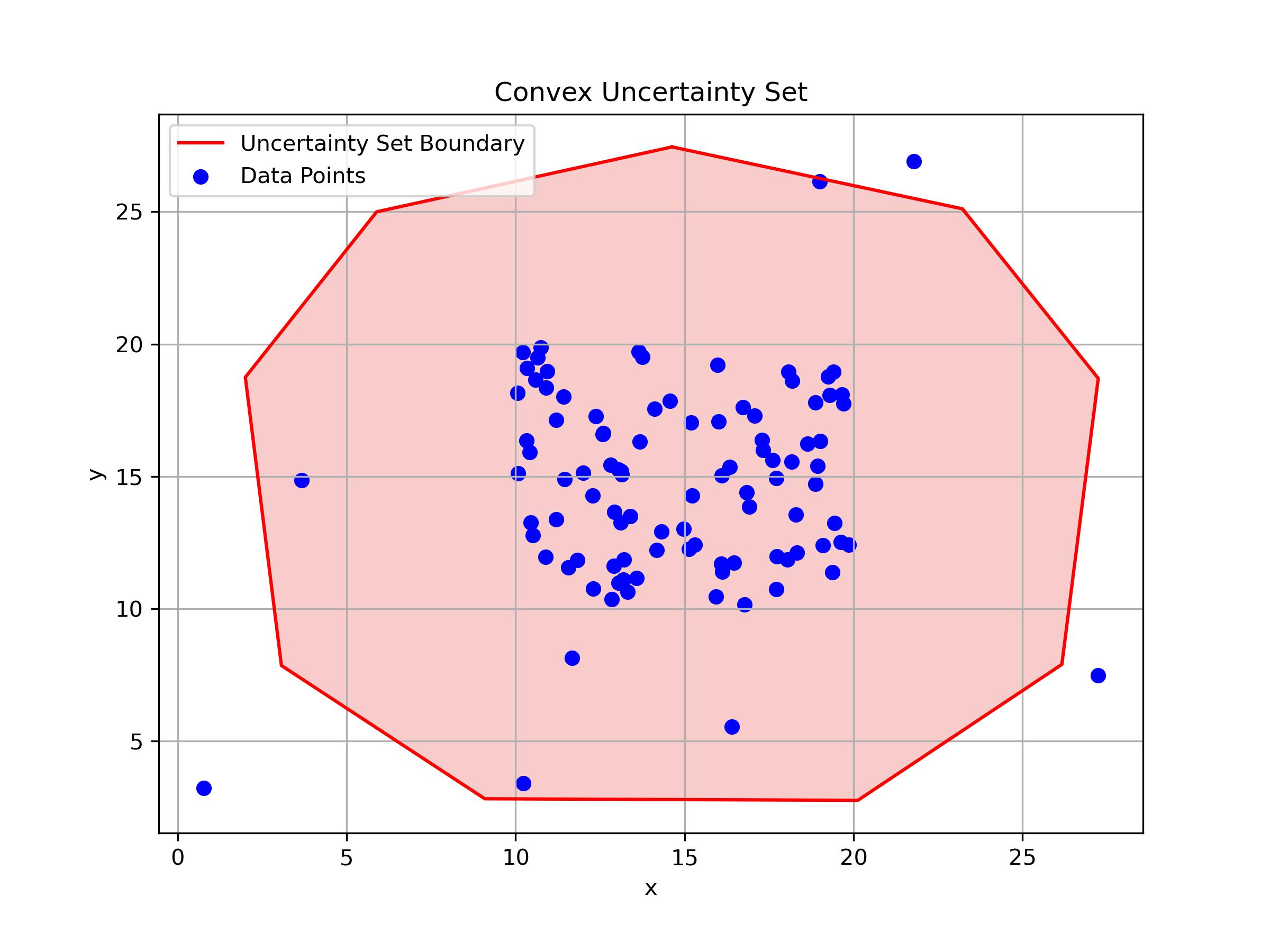}
    \includegraphics[width=0.31\textwidth]{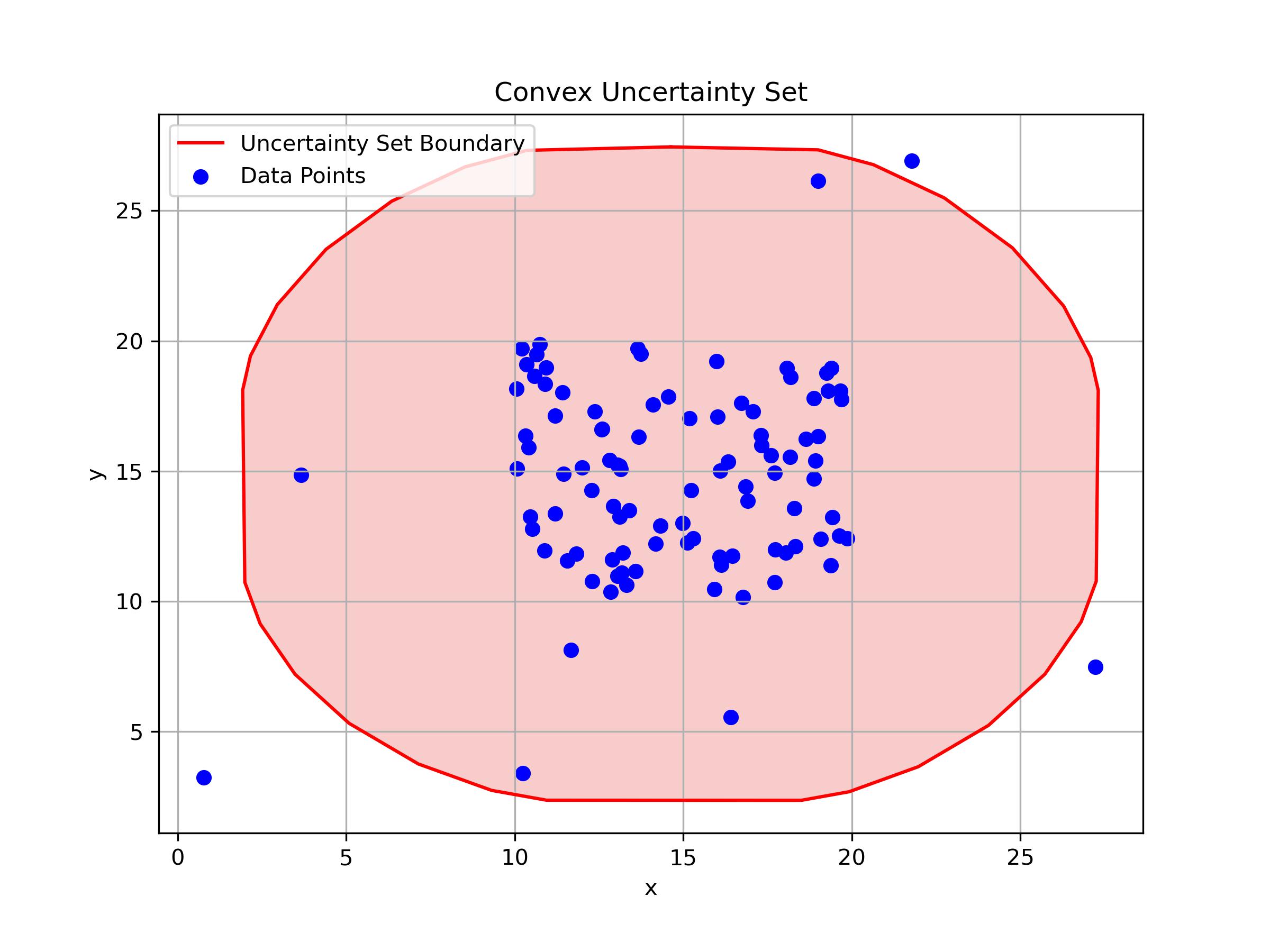}

    \includegraphics[width=0.31\textwidth]{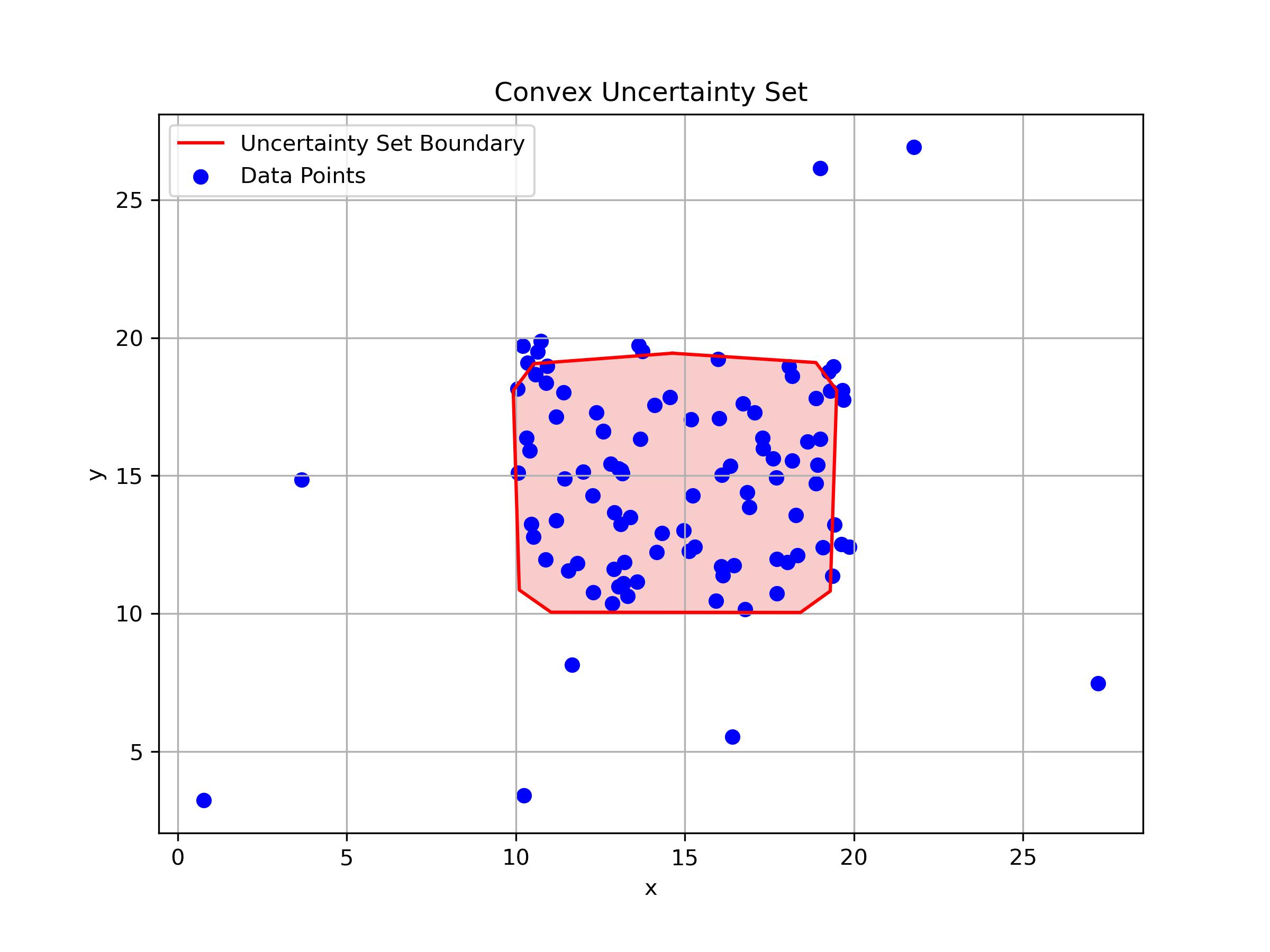}
    \includegraphics[width=0.31\textwidth]{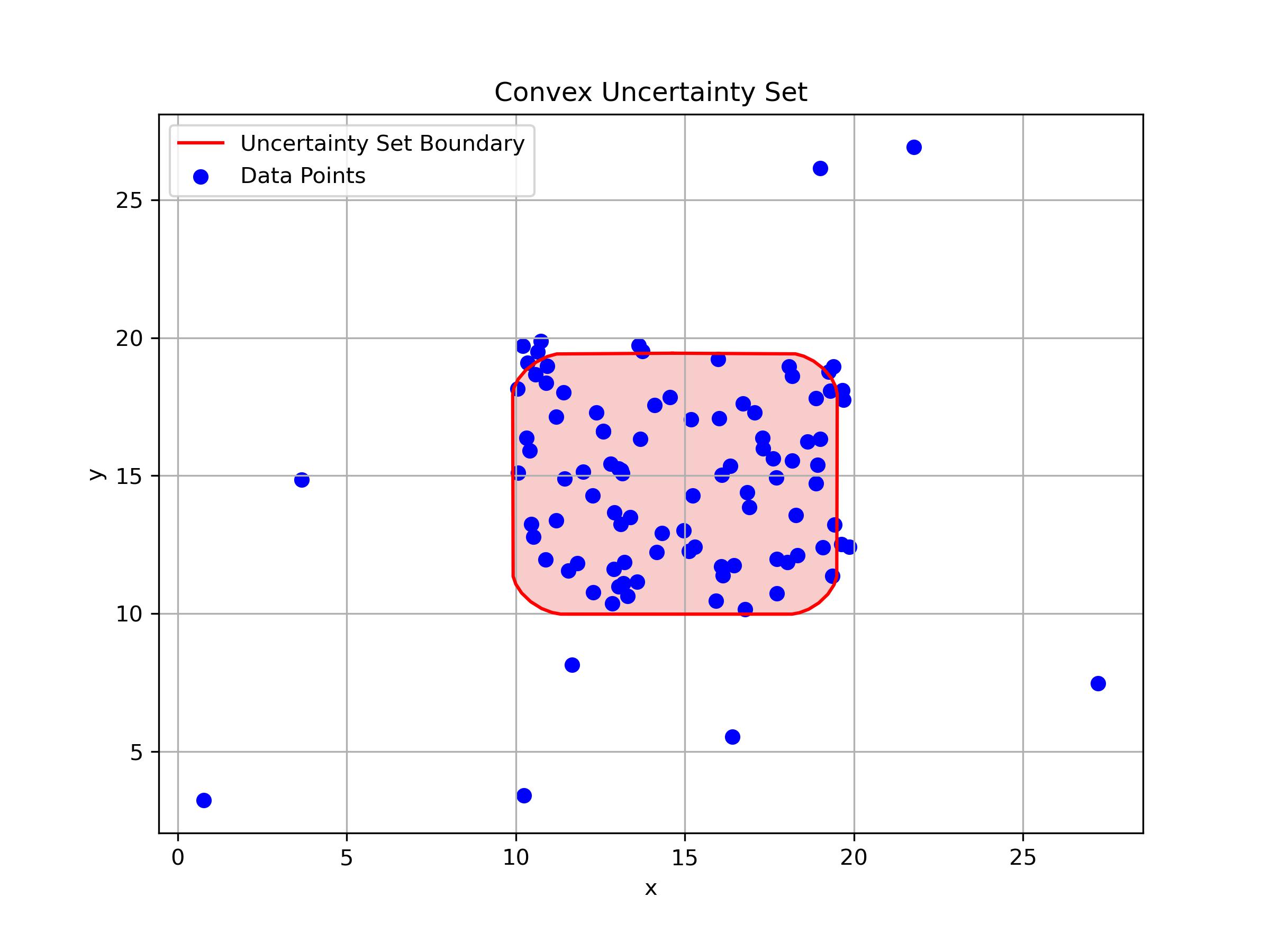}
    
    \caption{\color{black}Examples of convex sets constructed using the data-driven approach of \cite{bertsimas2018data} with 100 samples drawn from the uniform distribution specified in Equation~\eqref{distri1}. In the first row, we set \( \epsilon = 0.1 \) with \( \theta = 10 \) (left) and \( \theta = 30 \) (right). In the second row, \( \epsilon = 0.3 \) with \( \theta = 10 \) (left) and \( \theta = 30 \) (right).
}
        \label{fig:advicedb}
\end{figure}
\begin{enumerate}

\item \textit{Point estimate benchmark.} 
Under this benchmark, we consider a sample set $S_k$ with a size of $n\in \{10, 100\}$. To estimate the central location, we use the point estimate $(\widehat{x}_k, \widehat{y}_k) = \frac{1}{n}(\sum_{i\in S_k} x_i, \sum_{i\in S_k} y_i)$. We then use this point estimate to construct a convex set $\region_k$, which serves as an input for Algorithm \ref{alg:trans}. 
\footnote{\color{black} Our algorithm, when given point estimate advice \((x,y)\), mirrors the LP algorithm from \cite{balseiro2022single} at a high level. It assigns a constant protection level to high-reward agents until encountering \(x\) low-reward agents, similar to the \(Q_1\) protection level in \cite{balseiro2022single}. Both our method and the LP algorithm in \cite{balseiro2022single} incorporate adjustments when the decision-maker observes more than \(x\) low-reward agents. In \cite{balseiro2022single}, the LP algorithm may reduce \(Q_1\) to preserve capacity for additional high-reward agents. The setting of \(Q_1\) is influenced by the worst instances and the values of consistency and competitiveness parameters. In contrast, our algorithm explicitly derives a closed-form solution for the continuous  decreasing protection level function  by dynamically balancing the worst instances with the consistency and competitiveness parameters.}  

\item \textit{The BQ benchmark \citep{ball2009toward}.} This benchmark algorithm is   the PLA with a fixed PL function as proposed by \cite{ball2009toward}. In this benchmark that we refer to as the \emph{BQ}, the PL is set to $\frac{1-r_{\ell}/r_h}{2-r_{\ell}/r_h}m =8$.

\item \textit{PLA with an ML-augmented  fixed PL function \citep{perakis2010robust}.} Under this benchmark, we again have  a PLA with a fixed PL function $p(x) = p$. This benchmark that we refer to as \emph{PR} is introduced by \cite{perakis2010robust} and is specifically designed for box ML advice, assuming that the ML advice is completely accurate.
To determine the value of $p$ in this benchmark, \cite{perakis2010robust} solve a mixed integer programming (MIP) problem to obtain the optimal consistent ratio. However, their algorithm does not account for inaccurate ML advice or accommodate ellipsoid advice. Thus, we will only consider this benchmark for the uniform demand model with box ML advice. 
\end{enumerate}

\subsection{Results}

{\color{black}
Table~\ref{tab:box} reports the Avg.\ CP and Worst CP, as defined in Equation~\eqref{eq:avg_worst_adv} under adversarial arrival sequences, using polyhedron advice (constructed following \cite{bertsimas2018data}) and box ML advice derived from the uniform demand model in Equation~\eqref{distri1}.\footnote{Results under stochastic arrival sequences are shown in Table~\ref{tab:boxS} in Appendix~\ref{append:simulations}. As expected, performance improves under stochastic arrivals, but the trends remain consistent with those under adversarial arrivals.} We evaluate Algorithm~\ref{alg:trans} across various settings: $z \in \{80\%, 90\%\}$ for box advice, and $C \in \{\cstar, 0.9\cstar, 0.8\cstar\}$ for all settings. For comparison, we include benchmarks such as point ML advice, BQ, and PR (also evaluated for $z \in \{80\%, 90\%\}$). 

Regarding the choice of the convex set, \emph{both box advice and polyhedron advice perform best}. For box advice, ignoring outliers, the demand distribution---while unknown---is known to be uniform, allowing for precise calibration of the ML advice shape. However, when such information is not available, polyhedron advice (with $\epsilon = 0.3$) offers a key advantage as it does not require prior knowledge of the distribution and can adapt to any observed pattern while still achieving comparable performance to box advice. For instance, with $n = 10$ and $C = \cstar$, the average competitive ratio for polyhedron advice (with $\epsilon = 0.3$ and $\theta = 10$) is around $0.903$, close to the box advice (with $z = 90\%$) performance of approximately $0.902$.

For polyhedron advice, the number of vertices ($\theta$) has minimal impact on performance---the average and worst-case competitive ratios are nearly identical for $\theta = 10$ and $\theta = 30$. In contrast, the fraction of points included in the advice, influenced by $\epsilon$, significantly affects the outcome. Our results suggest that $\epsilon = 0.3$ is suitable for this setting, but in general, $\epsilon$ should be tuned based on the instance.

{\color{black}A third comparison point is the NN--based box advice. Its performance is consistently
better than point ML advice but slightly weaker than both the model-based box advice and
the polyhedron advice across all settings. This is expected: as
Table~\ref{tab:new_results_neu} shows, NN performance improves substantially with more
training data, whereas the model-based box and polyhedron constructions are already highly
effective even with small sample sizes. Moreover, the NN method of \cite{pearce2018high}
necessarily outputs a \emph{box}-shaped prediction interval, which limits its flexibility in
capturing the joint geometry of the demand vector. In contrast, the model-based and polyhedron
advice used in our main results are constructed directly from the empirical distribution
and can adapt to the underlying shape of the data. Finally, NN-based boxes do not
allow explicit control over the coverage level \(z\) due to its black box nature; the resulting interval may contain
substantially more or fewer than \(z\) percent of the samples, particularly when \(n\) is
small, leading to higher variance in competitive performance.} 

Turning to the effect of the number of samples ($n$), we find that for the geometric advice sets the impact is relatively limited. Increasing $n$ reduces the variance of the advice, but the PLA remains effective even with small $n$. For instance, using polyhedron advice with $\epsilon = 0.3$, $\theta = 10$, and $C = \cstar$, increasing $n$ from $10$ to $100$ improves the average CP only from $0.903$ to $0.910$ and the worst-case CP from $0.633$ to $0.674$. A similar pattern holds for box advice: with $z = 90\%$ and $C = \cstar$, increasing $n$ from $10$ to $100$ increases the average CP from $0.902$ to $0.915$ and the worst-case CP from $0.631$ to $0.676$. We discuss the impact of the number of samples ($n$) in more detail in Section~\ref{sec:number of sample}.

Performance under {point ML advice} and the {BQ benchmark} is lower than that under polyhedron and box advice. Point ML advice dominates the BQ benchmark in terms of average performance but is dominated by the BQ benchmark in terms of worst-case performance. The {PR benchmark} achieves higher average and worst-case performance than point ML advice but remains slightly weaker than polyhedron and box advice in worst-case performance.  
}

\begin{table}[htb]
\caption{Results under ML advice (using the demand model in Equation \eqref{distri1}) and adversarial order. The standard error of all reported values is less than $0.003$. In each column, the top two values of average CP are highlighted in black, and the top two values of worst-case CP are highlighted in red. Additionally, the highest average and worst-case CP across all settings for each value of $n$ are shown in bold.}
\centering
\footnotesize
\begin{tabular}{ll cccc cccc}
\toprule
 & & \multicolumn{3}{c}{\( n = 10 \)} & \multicolumn{3}{c}{\( n = 100 \)} \\
\cmidrule(lr){3-5} \cmidrule(lr){6-8}
Algorithm & Metric & $\cstar$ & $0.9 \cdot \cstar$ & $0.8 \cdot \cstar$ & $\cstar$ & $0.9 \cdot \cstar$ & $0.8 \cdot \cstar$ \\
\midrule
\multirow{2}{*}{Alg. \ref{alg:trans}, Polyhedron ($\epsilon = 0.1, \theta = 10$)} 
    & Avg. CP & 0.791 & 0.780 & 0.775 & 0.794 & 0.782 & 0.777 \\
    & Worst CP & 0.606 & 0.615 & 0.629 & 0.618 & 0.631 & 0.640 \\
\midrule
\multirow{2}{*}{Alg. \ref{alg:trans}, Polyhedron ($\epsilon = 0.1, \theta = 30$)} 
    & Avg. CP & 0.792 & 0.778 & 0.774 & 0.793 & 0.784 & 0.779 \\
    & Worst CP & 0.606 & 0.614 & 0.627 & 0.617 & 0.633 & 0.642 \\
\midrule
\multirow{2}{*}{Alg. \ref{alg:trans}, Polyhedron ($\epsilon = 0.3, \theta = 10$)} 
    & Avg. CP & \cellcolor{bestavg}{\textbf{0.903}} & \cellcolor{bestavg}{0.882} & 0.822 & 0.910 & 0.891 & 0.833 \\
    & Worst CP & \cellcolor{bestworst}{0.633} & 0.683 & 0.643 & \cellcolor{bestworst}{0.674} & \cellcolor{bestworst}{0.712} & 0.659 \\
\midrule
\multirow{2}{*}{Alg. \ref{alg:trans}, Polyhedron ($\epsilon = 0.3, \theta = 30$)} 
    & Avg. CP & \cellcolor{bestavg}{0.902} & \cellcolor{bestavg}{0.882} & \cellcolor{bestavg}{0.823} & 0.907 & 0.892 & 0.833 \\
    & Worst CP & 0.631 & 0.682 & 0.642 & \cellcolor{bestworst}{0.674} & \cellcolor{bestworst}{0.712} & 0.657 \\
\midrule
\multirow{2}{*}{Alg. \ref{alg:trans}, Box Advice ($z = 90\%$)} 
    & Avg. CP & \cellcolor{bestavg}{0.902} & {0.880} & \cellcolor{bestavg}{0.824} & \cellcolor{bestavg}{0.915} & \cellcolor{bestavg}{0.895} & \cellcolor{bestavg}{0.837} \\
    & Worst CP & 0.631 & \cellcolor{bestworst}{\textbf{0.686}} & \cellcolor{bestworst}{0.644} & \cellcolor{bestworst}{0.676} & \cellcolor{bestworst}{\textbf{0.715}} & \cellcolor{bestworst}{0.665} \\
\midrule
\multirow{2}{*}{Alg. \ref{alg:trans}, Box Advice ($z = 80\%$)} 
    & Avg. CP & 0.899 & 0.877 & \cellcolor{bestavg}{0.823} & \cellcolor{bestavg}{\textbf {0.917}} & \cellcolor{bestavg}{0.903} & \cellcolor{bestavg}{0.849} \\
    & Worst CP & \cellcolor{bestworst}{0.632} & \cellcolor{bestworst}{0.685} & \cellcolor{bestworst}{0.649} & 0.649 & 0.709 & \cellcolor{bestworst}{0.679} \\
    \midrule
\multirow{2}{*}{Alg. \ref{alg:trans}, Box Advice (NN)} 
    & Avg. CP & 0.875 & 0.841 & 0.798 & 0.884 & 0.866 & 0.820 \\
    & Worst CP & 0.617 & 0.672 & 0.638 & 0.645 & 0.694 & 0.655 \\
\midrule
\multirow{2}{*}{Point ML Advice} 
    & Avg. CP & 0.851 & 0.783 & 0.725 & 0.909 & 0.825 & 0.747 \\
    & Worst CP & 0.454 & 0.520 & 0.604 & 0.471 & 0.526 & 0.611 \\
\midrule
\multirow{2}{*}{BQ Benchmark} 
    & Avg. CP & 0.762 & - & - & 0.762 & - & - \\
    & Worst CP & 0.600 & - & - & 0.600 & - & - \\
\midrule
\multirow{2}{*}{PR Benchmark ($z = 90\%$)} 
    & Avg. CP & 0.900 & - & - & 0.909 & - & - \\
    & Worst CP & 0.626 & - & - & 0.644 & - & - \\
\midrule
\multirow{2}{*}{PR Benchmark ($z = 80\%$)} 
    & Avg. CP & 0.896 & - & - & 0.910 & - & - \\
    & Worst CP & 0.627 & - & - & 0.631 & - & - \\
\bottomrule
\end{tabular}
\label{tab:box}
\end{table}

\medskip

\begin{table}[htb]
\caption{Results under ML advice (demand model in Equation \eqref{distri2}) and adversarial order. The standard error of all the numbers is less than $0.003$. In each column, the top two values of average   CP are highlighted in black, and the top two values of worst-case CP are highlighted in red. Additionally, the highest average and worst-case CP across all settings for each value of $n$ are shown in bold.}
\centering
\footnotesize
\begin{tabular}{ll cccc cccc}
\hline
 & & \multicolumn{3}{c}{\( n = 10 \)} & \multicolumn{3}{c}{\( n = 100 \)} \\
\cline{3-5} \cline{6-8}
Algorithm & Metric & $\cstar$ & $0.9 \cdot \cstar$ & $0.8 \cdot \cstar$ & $\cstar$ & $0.9 \cdot \cstar$ & $0.8 \cdot \cstar$ \\
\hline
\multirow{2}{*}{Alg. \ref{alg:trans} (Polyhedron, $\epsilon = 0.1, \theta = 10$)} 
    & Avg. CP & 0.829 & 0.815 & 0.805 & 0.836 & 0.819 & 0.811 \\
    & Worst CP & 0.612 & 0.629 &  \cellcolor{bestworst}{0.615} & 0.620 & 0.642 & 0.623 \\
\hline
\multirow{2}{*}{Alg. \ref{alg:trans} (Polyhedron, $\epsilon = 0.1, \theta = 30$)} 
    & Avg. CP & 0.827 & 0.815 & 0.806 & 0.833 & 0.818 & 0.809 \\
    & Worst CP & 0.614 & 0.630 & 0.612 & 0.620 & 0.640 & 0.621 \\
\hline
\multirow{2}{*}{Alg. \ref{alg:trans} (Polyhedron, $\epsilon = 0.3, \theta = 10$)} 
    & Avg. CP & 0.913 & 0.890 & \cellcolor{bestavg}{0.819} & 0.921 & 0.905 & 0.829 \\
    & Worst CP & \cellcolor{bestworst}{0.638} & 0.650 & 0.604 & 0.649 &  \cellcolor{bestworst}{0.682} & 0.620 \\
\hline
\multirow{2}{*}{Alg. \ref{alg:trans} (Polyhedron, $\epsilon = 0.3, \theta = 30$)} 
    & Avg. CP & 0.911 & \cellcolor{bestavg}{0.891} & 0.818 & 0.922 & \cellcolor{bestavg}{0.906} & 0.828 \\
    & Worst CP & 0.636 & 0.647 & 0.606 & 0.650 & 0.681 & 0.621 \\
\hline
\multirow{2}{*}{Alg. \ref{alg:trans} (Ellipsoid Advice, $z = 90\%$)} 
    & Avg. CP & \cellcolor{bestavg}{\textbf{0.921}} & \cellcolor{bestavg}{0.894} & \cellcolor{bestavg}{0.830} & \cellcolor{bestavg}{\textbf{0.937}} & \cellcolor{bestavg}{0.918} & \cellcolor{bestavg}{0.853} \\
    & Worst CP & \cellcolor{bestworst}{0.646} &  \cellcolor{bestworst}{\textbf{0.661}} &  \cellcolor{bestworst}{0.613} &  \cellcolor{bestworst}{0.666} &  \cellcolor{bestworst}{\textbf{0.709}} &  \cellcolor{bestworst}{0.655} \\
\hline
\multirow{2}{*}{Alg. \ref{alg:trans} (Ellipsoid Advice, $z = 80\%$)} 
    & Avg. CP & \cellcolor{bestavg}{0.914} & 0.877 & 0.810 & \cellcolor{bestavg}{0.926} & 0.891 & \cellcolor{bestavg}{0.835} \\
    & Worst CP & 0.633 &  \cellcolor{bestworst}{0.653} & 0.605 &  \cellcolor{bestworst}{0.652} & 0.681 &  \cellcolor{bestworst}{0.649} \\
        \midrule
\multirow{2}{*}{Alg. \ref{alg:trans}, Box Advice (NN)} 
    & Avg. CP & 0.872 & 0.838 & 0.789 & 0.878 & 0.843 & 0.797 \\
    & Worst CP & 0.608 &  0.633 & 0.597 & 0.631 & 0.649 & 0.638 \\
\hline
\multirow{2}{*}{Point ML Advice} 
    & Avg. CP & 0.906 & 0.822 & 0.733 & 0.914 & 0.826 & 0.740 \\
    & Worst CP & 0.539 & 0.566 & 0.599 & 0.549 & 0.582 & 0.604 \\
\hline
\multirow{2}{*}{BQ Benchmark} 
    & Avg. CP & 0.725 & - & - & 0.725 & - & - \\
    & Worst CP & 0.600 & - & - & 0.600 & - & - \\
\hline
\end{tabular}
\label{tab:ell}
\end{table}



As in Table~\ref{tab:box}, we find that Algorithm~\ref{alg:trans}, when used with either ellipsoid or polyhedron advice, achieves the best average and worst-case performance across all settings. Ellipsoid advice performs particularly well, achieving an average CP of $0.937$ and worst-case CP of $0.709$ when $n = 100$ and $C = \cstar$. Polyhedron advice remains competitive, especially with $\epsilon = 0.3$ and $\theta = 30$, where it achieves an average CP of $0.922$ and worst-case CP of $0.650$ for $n = 100$ and $C = \cstar$. {\color{black}The NN--based box advice again performs better than point ML advice but remains below both ellipsoid and polyhedron advice.}

\begin{table}[htb] \footnotesize
\caption{Results for varying sample sizes $n$ with corresponding Avg. CP and Worst CP by Algorithm \ref{alg:trans} with input $C=\cstar$ (Polyhedron, $\epsilon=0.3$, $\theta=10$) under adversarial order for both demand models in Equations \eqref{distri1} and \eqref{distri2}. The standard error of all the numbers is less than $0.003$. }
\resizebox{\textwidth}{!}{%
\begin{tabular}{c c||c|c|c|c|c|c|c|c|c|c|}
    \cline{3-12}
    & & \multicolumn{10}{c|}{\# of samples $n$} \\
    \cline{3-12}
    & & 10 & 20 & 30 & 40 & 50 & 60 & 70 & 80 & 90 & 100 \\
    \hline
    \multicolumn{1}{|c|}{\multirow{2}{*}{Demand Model \eqref{distri1}}} 
    & Avg. CP   & 0.903 & 0.906 & 0.908 & 0.908 & 0.908 & 0.909 & 0.909 & 0.909 & 0.909 & 0.910 \\
    \multicolumn{1}{|c|}{} 
    & Worst CP & 0.633 & 0.641 & 0.648 & 0.654 & 0.659 & 0.663 & 0.666 & 0.669 & 0.672 & 0.674 \\
    \hline
    \multicolumn{1}{|c|}{\multirow{2}{*}{Demand Model \eqref{distri2}}} 
    & Avg. CP   & 0.913 & 0.915 & 0.916 & 0.917 & 0.918 & 0.919 & 0.919 & 0.920 & 0.920 & 0.921 \\
    \multicolumn{1}{|c|}{} 
    & Worst CP & 0.638 & 0.640 & 0.642 & 0.644 & 0.645 & 0.646 & 0.647 & 0.648 & 0.648 & 0.649 \\
    \hline
\end{tabular}
}
\label{tab:new_results}
\end{table}
\subsection{Impact of the Number of Samples} \label{sec:number of sample}

In this subsection, we examine how the number of training samples \(n\) affects performance under different ML-advice constructions.

\paragraph{Polyhedron advice.}
Table~\ref{tab:new_results} shows the effect of varying \(n\) when the ML advice is a polyhedron (\(\epsilon = 0.3, \theta = 10\)). The \emph{average CP remains  stable} as \(n\) increases, while the \emph{worst-case CP improves modestly}. This behavior stems from Algorithm~\ref{alg:trans}: its protection-level selection balances the upper and lower bounds of the ML advice, so variance in one bound is partially offset by variance in the other. Consequently, polyhedron advice delivers robust performance even with limited data, reducing the need for large training sets.

{\color{black}\paragraph{NN-based advice.}
We also evaluate the NN-based advice, which produces a box-shaped prediction interval. Because the NN model contains many parameters, small training sets (\(n \le 100\)) lead to underfitting and highly unstable prediction intervals. To capture this effect, we vary \(n\) over a wider range, up to \(10{,}000\). As shown in Table~\ref{tab:new_results_neu}, both average and worst-case CP improve steadily as \(n\) increases, but at a \emph{much slower rate} than with polyhedron-based advice. This is because classical convex-set (polyhedron-based) constructions enjoy statistical convergence on the order of \(O(1/\sqrt{n})\), whereas no comparable guarantee holds for NNs. This gap is reflected empirically:  Figure~\ref{fig:adviceneural}, presented in Appendix \ref{appendix:NN},  shows that for small \(n\), the NN-generated prediction intervals vary substantially, whereas for \(n \ge 2000\), the intervals stabilize and converge to a consistent range.
Despite this slower convergence, Algorithm~\ref{alg:trans} still attains competitive performance even when the NN advice is trained with as few as \(n = 10\) samples.\footnote{We expect NN-based advice to become more useful when the relationship between high-dimensional features and the target variable (i.e., demand in our setting) is complex or highly nonlinear, whereas in lower-dimensional settings classical polyhedron-based advice remains preferable and effective even with relatively few samples.
}

}

\begin{table}[htb]
\centering
\caption{Results for varying sample sizes $n$ with corresponding Avg.\ CP and Worst CP for Algorithm~\ref{alg:trans} with input $C=\cstar$ (Neural Network Box Advice \cite{pearce2018high}) under adversarial arrival sequences, for the demand models in Equations~\eqref{distri1} and~\eqref{distri2}.}
\label{tab:new_results_neu}

{\footnotesize
\begin{tabular}{c c||c|c|c|c|c|c|c|}
    \cline{3-9}
    & & \multicolumn{7}{c|}{Number of samples $n$} \\
    \cline{3-9}
    & & 10 & 100 & 1000 & 2000 & 3000 & 5000 & 10000 \\
    \hline

    \multicolumn{1}{|c|}{\multirow{2}{*}{Demand Model \eqref{distri1}}} 
    & Avg.\ CP   & 0.875  & 0.884  & 0.893 & 0.906 & 0.907 & 0.909 & 0.910 \\
    \multicolumn{1}{|c|}{} 
    & Worst CP & 0.617 & 0.645 & 0.658 & 0.667 & 0.673 & 0.681 & 0.683 \\
    \hline

    \multicolumn{1}{|c|}{\multirow{2}{*}{Demand Model \eqref{distri2}}} 
    & Avg.\ CP   & 0.872 & 0.878 & 0.881 & 0.892 & 0.898 & 0.902 & 0.904 \\
    \multicolumn{1}{|c|}{} 
    & Worst CP & 0.608 & 0.631 & 0.640 & 0.644 & 0.648 & 0.654 & 0.656 \\
    \hline
\end{tabular}
}
\end{table}

{\color{black}

{\color{black}\subsection{Guidelines for Practice}\label{subsec:guidelines}

Based on our numerical results, we offer the following practical recommendations:

\begin{itemize}[leftmargin=*]
   \item \textbf{Uncertainty Set Construction:}  
    When the underlying distribution resembles known parametric families, use a box (for uniform-like data) or an ellipsoid (for Gaussian-like data) covering \( z \in \{80\%, 90\%\} \) of the training samples.  
    If the distribution is unknown or irregular, polyhedron-based advice using the data-driven method of \cite{bertsimas2018data} is a strong default choice; tune \( \epsilon \) (and, to a lesser extent, \( \theta \)) to optimize performance.  
  {\color{black} For users who prefer a fully nonparametric approach, the NN method of \cite{pearce2018high} can also generate prediction intervals. However, because it is a black-box model, the resulting interval does not guarantee exact coverage \(z\) (the proportion of training samples that should lie inside the constructed interval), and its stability is more sensitive to sample size.}

    \item \textbf{Choosing \( C \):} 
    To \emph{maximize average performance}, use \( C = \cstar \). For better \emph{worst-case performance}, a slightly smaller value like \( C = 0.9 \cdot \cstar \) offers a good balance.

   \item \textbf{Sample Size \(n\):}  
    Increasing the number of samples reduces the variance of the ML advice, but the gains diminish quickly. Even small sample sizes (e.g., \( n = 10 \)) are sufficient for strong performance with box, ellipsoid, and polyhedron advice, making these methods practical in data-scarce settings.  
   {\color{black} In contrast, NN-based advice requires substantially larger sample sizes (typically \( n \ge 1000 \)) to produce stable and reliable prediction intervals.}
\end{itemize}
}

{\color{black}
\section{Extension to Multiple Types} \label{sec:multitype}

So far, we have shown that when the number of types is two, an optimal-Pareto algorithm is an adaptive PLA. In this section, we extend adaptive PLAs to settings with more than two types using a mixed integer programming (MIP) formulation when the ML advice is a polyhedron. 

Suppose there are \( n \) types of requests with per-unit rewards \( r_1 > r_2 > \ldots > r_n \). The adaptive PLA in this case can be characterized by \( n-1 \) functions, \( p_2(\cdot), \ldots, p_n(\cdot) \), where each function \( p_i \) maps the vector of requests of types \( i, i+1, \ldots, n \), denoted by \( \mathbf{x}_i = (x_{i}, x_{i+1}, \ldots, x_{n}) \), to the protection level for higher-reward types \( i-1, i-2, \ldots, 1 \). Here, \( x_{i} \) is the number of requests of type \( i \) seen so far.   

Consider an instance \( I = \{r_t\}_{t=1}^T \) where \( r_t = r_j \) for some \( j \in [n] \) at time \( t \in [T] \). Let \( a_\tau \in [0, s_\tau] \) represent the fraction of requests accepted by the PLA at time \( \tau \). Then, the fraction \( a_t \) of \( r_t \) accepted by a nested PLA with PL functions \( p_i(\cdot) \) for \( i \in [n] \), is given by:
\begin{align} \label{eq:allocation_nestedPLA}
a_t = \max\Big\{ z \in [0, s_t] \, \big | \, z + \sum_{\tau=1}^{t-1} a_{\tau} \mathbf{1}\{r_{\tau} \leq r_i\} \leq m - p_i(\mathbf{x}_{i,t}),\ \forall i \leq j \Big\}.
\end{align}
Here, \( \mathbf{x}_{i,t} = [x_{i, t}, x_{i+1, t}, \ldots, x_{n, t}] \), where \( x_{i, t} \) is the total number of requests of type \( i \) in the first \( t \) rounds (including round \( t \) itself). Furthermore, by convention, \( p_1(\mathbf{x}_1) = 0 \) for any \( \mathbf{x}_1 = \mathbf{x} \).  

We consider a nested policy,  where the protection levels for higher-priority (higher-reward) types are greater than or equal to the protection levels for lower-priority types. Formally, a policy is nested if:
$
p_n (x_n)\geq p_{n-1}(\mathbf x_{n-1})\geq \ldots \geq p_2(\mathbf x_2)\,.
$
Here, \( p_n (x_n) \) represents the number of resources protected for type \( 1, 2,\ldots, n-1 \) when the number of type \( n \) requests is \( x_n \). Similarly, \( p_{n-1}(\mathbf x_{n-1}) \) represents the number of resources protected for type \( 1, 2,\ldots, n-2 \) when the number of type \( n \) and \( n-1 \) requests is \( \mathbf x_{n-1} =(x_{n-1}, x_{n}) \), and so on. To ensure the validity of the policy, we require that:
\begin{align}\label{eq:new_valid}
p_i(\mathbf{x}'_i) - p_i(\mathbf{x}_i) \leq x_i' - x_i + p_{i+1}(\mathbf{x}'_{i+1}) - p_{i+1}(\mathbf{x}_{i+1}),
\end{align}
where \( \mathbf{x}_i = (x_i, x_{i+1}, \ldots, x_n) \) and \( \mathbf{x}'_i = (x_i', x_{i+1}', \ldots, x_n') \), with \( x_i' \geq x_i \). This validity constraint ensures that the protection level function decreases appropriately. For example, if \( \mathbf{x}_{i+1} = \mathbf x_{i+1}' \) and \( x'_{i} > x_{i} \), we must have: $
\frac{p_i(\mathbf x_i')-p_i(\mathbf x_i)}{x_i'-x_i}  \geq -1.
$

Under any adaptive nested PLA, the compatible ratio is minimized when the arrival sequence is ordered — that is, when requests with lower rewards arrive before requests with higher rewards. Let \( \CP(p_2(\mathbf{x}_2), \ldots, p_n(x_n); \mathbf{x}) \) denote the compatible ratio of a nested PLA with PL functions \( p_2(\mathbf{x}_2), \ldots, p_n(x_n) \) under this ordered sequence, where \( x_n \) requests of type \( n \) arrive first, followed by \( x_{n-1} \) requests of type \( n-1 \), and so on. 

Consider a convex polytope as ML advice with a set of vertices \( \mathcal{V} \). Define \( \mathcal{V}' \) as the set of vertices associated with the polytope \( [0, m]^n \). As argued in Section \ref{sec:mlconsistent}, the consistent ratio is minimized at the vertices in \( \mathcal{V} \), while the robust ratio is minimized at the vertices in \( \mathcal{V}' \). Therefore, to maximize the robust ratio while ensuring that the consistent ratio is greater than or equal to \( C \), we can set up the following optimization problem:
\begin{align} \notag
     \max_{R\in[0,1], (p_i(\cdot)_{i\in {2, \ldots, n}})} & \ R   \\  \label{eq:cpc1n}
s.t.   \ & \ \CP(p_2(\mathbf x_2),\ldots,p_n(x_n); \mathbf{x}) \geq C, \qquad \ \mathbf{x} \in \mathcal{V},  \\  \label{eq:cpc2n}
\ & \ \CP(p_2(\mathbf x_2),\ldots,p_n(x_n); \mathbf{x}) \geq R, \qquad \ \mathbf{x} \in \mathcal{V}',  \\  \label{eq:cpc3n}
\ & \ p_2(\mathbf x_2),\ldots,p_n(x_n) \text{ are valid per Eq. \eqref{eq:new_valid}} \qquad \,\qquad \mathbf{x} \in \mathcal{V}'\cup \mathcal{V}  \\
\ &p_n (x_n)\ge p_{n-1}(\mathbf x_{n-1})\ge \ldots \ge p_2(\mathbf x_2) \qquad \mathbf{x} \in \mathcal{V}'\cup \mathcal{V}
\\ \tag{\text C-Pareto-n-Types} \label{prob:rbmax_n}
\end{align} 
In this optimization problem is not convex due to the structure of the compatible ratio \( \CP \). However, we can convert it into a mixed integer linear program (MILP) using the following result:

\begin{proposition} \label{prop:variable}
    Problem \eqref{prob:rbmax_n} can be converted to an MILP with \( 3n|\mathcal V \cup \mathcal V'| + 1 \) continuous variables and \( 2n|\mathcal V \cup \mathcal V'| \) binary integer variables. 
\end{proposition}

The proof is provided in Appendix \ref{append:ntype}. After solving the MILP problem \eqref{prob:rbmax_n}, we obtain the optimal values of $p_2(\mathbf{x}_2), \ldots, p_n(\mathbf{x}_n)$ for $\mathbf{x} \in \mathcal{V}' \cup \mathcal{V}$. To the  PL functions $p_i(\mathbf{x}_i)$ to the entire space $\mathbb{N}^{n-i+1}$ for each $i \in {2,3,\ldots,n}$, we apply the interpolation heuristic in Algorithm \ref{alg:envelope} (see Appendix \ref{sec:heuristic}), where $\mathbb{N}$ denotes the set of natural numbers. At a high level, the heuristic ensures that each protection level function remains valid while maintaining the nested property.

}

{\color{black}

\subsection{Numerical Studies for the Multiple-Type Case } Here, we  present numerical experiments with three types of requests. Let \( x_1 \), \( x_2 \), and \( x_3 \) denote the number of low-reward, medium-reward, and high-reward arrivals, respectively. We assume the following demand model, consistent with Equation~\eqref{distri1}:
\begin{align} \label{distri3D}
(x_1,x_2,x_3) \sim \left\{ 
\begin{array}{ll}
         \text{Uniform}(10,20) & \text{with probability } 0.9, \\
         \text{Uniform}(0,30) & \text{with probability } 0.1. \\
\end{array} 
\right.
\end{align}
We set \( m = 20 \), with per-unit rewards \( r_1 = 1 \), \( r_2 = 1/2 \), and \( r_3 = 1/3 \), and focus on worst-case arrival orders under this demand model.

We first consider the \emph{box ML advice} approach. Following the methodology from the previous section, we construct the ML advice as the smallest-volume box containing at least \( z = 90\% \) or \( z = 80\% \) of the observed data. Examples of the resulting box ML advice are shown in Figure~\ref{fig:advice3D} in Appendix~\ref{multi-type-ML-region}.

We also implement the data-driven polyhedron construction proposed by \cite{bertsimas2018data}. In three dimensions, the construction depends on one statistical parameter and two geometric parameters. As in the two-dimensional case, $\epsilon$ controls the Kolmogorov--Smirnov (KS) divergence. However, while a single angular parameter suffices to discretize directions in 2D, in 3D two angular parameters $(\theta, \phi)$ are required to index directions on the unit sphere. The values of \(\theta\) and \(\phi\) together determine the number of vertices used in the polyhedral approximation. More generally, in an \(n\)-dimensional setting, one statistical parameter \(\epsilon\) and \(n-1\) angular parameters are needed. Examples of the resulting polyhedron ML advice are shown in Figure~\ref{fig:advice3Dpoly} in Appendix~\ref{multi-type-ML-region}.


\begin{table}[htb]
\caption{Results under box ML advice (using the demand model with three types) and adversarial order. In each column, the top two values of average CP are highlighted in black, and the top two values of worst-case CP are highlighted in red. Additionally, the highest average and worst-case CP across all settings for each value of $n$ are shown in bold.}
\centering
\footnotesize
\begin{tabular}{ll cccc cccc}
\toprule
 & & \multicolumn{3}{c}{\( n = 10 \)} & \multicolumn{3}{c}{\( n = 25 \)} \\
\cmidrule(lr){3-5} \cmidrule(lr){6-8}
Algorithm & Metric & $\cstar$ & $0.9 \cdot \cstar$ & $0.8 \cdot \cstar$ & $\cstar$ & $0.9 \cdot \cstar$ & $0.8 \cdot \cstar$ \\
\midrule
\multirow{2}{*}{Our Alg.  (Box Advice, $z = 90\%$)} 
    & Avg. CP & \cellcolor{bestavg}{0.816} & {0.793} & 0.767 & {0.829} & \cellcolor{bestavg}{0.811} & \cellcolor{bestavg}{0.787} \\
    & Worst CP & 0.332 & 0.396 & 0.458 & 0.343 & \cellcolor{bestworst}{0.400} & \cellcolor{bestworst}{0.451} \\
\midrule
\multirow{2}{*}{Our Alg.  (Box Advice, $z = 80\%$)} 
    & Avg. CP & 0.809 & \cellcolor{bestavg}{0.797} & \cellcolor{bestavg}{0.776} & 0.830 & 0.805 & \cellcolor{bestavg}{0.798} \\
    & Worst CP & 0.317 & 0.395 & {0.444} & 0.316 & 0.393 & 0.424 \\
\midrule
\multirow{2}{*}{Our Alg. (Polyhedron, $\epsilon = 0.9$, $\theta=\phi=6$)} 
    & Avg. CP & \cellcolor{bestavg}{\textbf{0.879}} & \cellcolor{bestavg}{0.806} & \cellcolor{bestavg}{0.785} & \cellcolor{bestavg}{\textbf{0.862}} & \cellcolor{bestavg}{0.826} & 0.768 \\
    & Worst CP & 0.341 & \cellcolor{bestworst}{0.420} & \cellcolor{bestworst}{0.491} & 0.350 & 0.369 & \cellcolor{bestworst}{0.464} \\
\midrule
\multirow{2}{*}{Our Alg.  (Polyhedron, $\epsilon = 0.1$, $\theta=\phi=6$)} 
    & Avg. CP & 0.730 & 0.703 & 0.645 & 0.758 & 0.711 & 0.700 \\
    & Worst CP & \cellcolor{bestworst}{0.446} & \cellcolor{bestworst}{0.465} & \cellcolor{bestworst}{0.483} & \cellcolor{bestworst}{0.471} & \cellcolor{bestworst}{0.480} & 0.445 \\
\midrule
\multirow{2}{*}{Point ML Advice} 
    & Avg. CP & 0.821 & 0.765 & 0.724 & \cellcolor{bestavg}{0.835} & 0.785 & 0.732 \\
    & Worst CP & 0.250 & 0.361 & 0.388 & 0.250 & 0.385 & 0.403 \\
\midrule
\multirow{2}{*}{BQ Benchmark} 
    & Avg. CP & 0.644 & - & - & 0.644 & - & - \\
    & Worst CP & \cellcolor{bestworst}{\textbf{0.545}} & - & - & \cellcolor{bestworst}{\textbf{0.545}} & - & - \\
\bottomrule
\end{tabular}
\label{tab:box3D}
\end{table}


Table~\ref{tab:box3D} reports the Avg. CP and Worst CP, as defined in Equation~\eqref{eq:avg_worst_adv} under adversarial arrival sequences, using polyhedron advice (constructed following \cite{bertsimas2018data}) and box ML advice derived from the uniform demand model in Equation~\eqref{distri3D}. 
For comparison, we include the BQ nested protection level and point ML advice as the benchmark.

From the results in Table~\ref{tab:box3D}, we observe that when considering the average CP, the best performance is achieved using either the box advice or the polyhedron advice with \(\epsilon = 0.9\), outperforming the point estimate advice.  
However, when focusing on the worst-case CP, it is preferable to use the polyhedron advice with \(\epsilon = 0.1\), which corresponds to a much larger polyhedron, or to simply adopt the BQ method, which disregards the sample data entirely. These findings highlight a trade-off between average and worst-case CP that can be managed by tuning the size of the ML advice.

}
\section{Concluding Remarks and Future Directions } \label{sec:extension}

In this work, we have proposed a novel online resource allocation model that addresses the challenge of integrating machine learned predictions into resource allocation decisions in a robust and efficient manner. Efficient online resource allocation is essential for entities such as hospitals, governments, and various industries that often face a trade-off between meeting low-reward and high-reward demands without precise knowledge of future demand. However, factors such as environmental fluctuations, data biases, and insufficient data points can impede consistent accuracy, making it challenging to allocate resources efficiently.

The proposed model is based on the concept of convex uncertainty sets, which use historical data to construct sets of plausible demand scenarios, allowing for flexibility and robustness in decision-making. We examine the benefits of utilizing ML advice in online resource allocation problems by proposing $C$-Pareto PLAs that balance the robust and the consistent ratios.  Compared to traditional fixed protection level algorithms, we find that adaptive PLAs often manage to obtain  high  consistent and robust ratios, highlighting the significance and advantages of adjusting the protection level.

{\color{black}This work highlights the substantial advantages of employing uncertainty set ML advice, as opposed to point estimate advice, in sequential decision-making under uncertainty. In our numerical studies, we used only a few samples to construct the uncertainty sets and demonstrated that even with limited data, our approach achieves strong performance. When the number of samples is large, a natural extension would be to consider advice that takes the form of empirical distributions rather than uncertainty sets. If the underlying distribution remains stable over time, this approach could be highly effective, as is typical in stochastic settings. However, in practical environments where distributional shifts or outliers are common, relying on ML-based uncertainty regions—as we propose—provides a much more robust foundation for decision-making. This is because, while the underlying distribution may change, the ML uncertainty region often remains stable. As our results show, the performance of our approach is not highly sensitive to small changes in the ML region.}  

\bibliographystyle{plainnat}
\bibliography{reference}
\newpage
\begin{APPENDICES}

{\color{black}\section{More on Validity Conditions of Protection Level Algorithms}\label{sec:valid}

In this section, we justify the necessity of the validity conditions stated in Definition~\ref{def:1} for PLAs.  

\textbf{Non-Increasing PL Function:}   We  argue why the protection level function \(p(x)\) must be non-increasing.  At a high level, the goal of a PLA is to reserve or “protect” a certain amount of resources for high-reward requests. However, these protection levels must be \emph{credible}—that is, any promise to reserve resources must be implementable and feasible under the given constraints. 
Suppose we are at a point where \(x\) low-reward requests have been processed. If the PLA has allocated the maximum allowable resources to low-reward requests up to this point, then the total amount allocated to low-reward requests is \(m - p(x)\), where \(m\) is the total available resource capacity. The remaining \(p(x)\) units are thus being reserved for high-reward requests.

Now, suppose that the PL function \(p(x)\) increases at \(x\), i.e., \(p'(x) > 0\). This means that as we process more low-reward requests, we plan to reserve even more resources for high-reward requests than before. But this leads to a contradiction: we have already allocated \(m - p(x)\) units to low-reward requests, so only \(p(x)\) units remain. If we now try to reserve more than \(p(x)\), we exceed the total available resource budget \(m\), making the updated protection level infeasible.

This inconsistency illustrates why the PL function must be non-increasing. By ensuring that \(p(x)\) does not grow with \(x\), we guarantee that the resources promised to high-reward requests are always consistent with past allocation decisions, thereby preserving feasibility and preventing over-commitment.

\textbf{Derivative Condition:}  
We also require that the derivative of the protection level function satisfy \(p'(x) \ge -1\). This condition ensures that protection levels do not decrease too rapidly.

To illustrate why we impose this condition, consider the same setting as before. Let \(m\) denote the total available resources, and suppose we have processed \(\bar{s}\) units of low-reward requests so far, with \(\bar{a}\) denoting the total amount of resources allocated to them. Now, a new low-reward request of size \(\epsilon > 0\), where \(\epsilon\) is very small, arrives. Assume that we have already allocated the maximum amount allowable for low-reward requests up to this point, i.e., \(\bar{a} = m - p(\bar{s} + 1)\), meaning only \(p(\bar{s} + 1)\) resources remain protected for high-reward requests.

According to the PLA rule, the amount we can allocate to the incoming \(\epsilon\)-sized request is:
\[
\min\{m - p(\bar{s} + 1 + \epsilon) - \bar{a}, \epsilon\}.
\]
Substituting \(\bar{a} = m - p(\bar{s} + 1)\), we get:
\[
\min\left\{\epsilon,\ -p(\bar{s} + 1 + \epsilon) + p(\bar{s} + 1)\right\}.
\]
If \(p'(\bar{s} + 1) < -1\), then for small \(\epsilon\), the difference \(-p(\bar{s} + 1 + \epsilon) + p(\bar{s} + 1)\) exceeds \(\epsilon\), and the PLA allows allocating the full request. In other words, the protection level drops so sharply that it remains ineffective. One can simply recover the same allocation outcome by setting \(p'(\bar{s} + 1) = -1\). Therefore, any steeper decline in the protection level function is unnecessary. Without loss of generality, we can thus assume that \(p'(x) \ge -1\).
}

\section{Proof of Lemma \ref{lem:order1}} \label{sec:prooforder}

Let $(x,y)$ be fixed, and let $I$ be the ordered sequence of arrivals such that $x$ low-reward requests arrive first, followed by $y$ high-reward requests. For any adaptive protection level algorithm $\mathcal A$, let $\tilde x(\mathcal A)$ be the total number of low-reward requests that are accepted under the ordered sequence $I$. Then, the algorithm accepts $\min\{y,m-\tilde x(\mathcal A)\}$ high-reward requests. 
Observe that algorithm $\mathcal A$ rejects the $x$-th low-reward request if and only if 
\begin{itemize}
\item condition (1): the number of low-reward requests accepted so far is not less than $m-p(x)$, where $p(\cdot)$ is the protection level function, or 
\item condition (2): all the resources have been used. 
\end{itemize}

\noindent We split the proof into two cases.

\textbf{Case 1: $y<m-\tilde x(\mathcal A)$}. In this case, we claim that under any other ordering of the arrivals, the number of high-reward requests that are accepted cannot be smaller than $y$. We also claim that under any other ordering of the arrivals, the number of low-reward requests that are accepted cannot be smaller than $\tilde x(\mathcal A)$. Showing these claims, complete the proof of this case.

We begin with the first part. First note that $\tilde x(\mathcal A)$ is an upper bound on the total number of low-reward requests that algorithm $\mathcal A$ accepts under any ordering of the arrivals. This is because condition (2) never fails as $y\le m-\tilde x(\mathcal A)$, and hence with even accepting $y$ high-reward requests, there are resources left for $\tilde x(\mathcal A)$ low-reward requests,   and condition (1) is independent of the order of high-reward requests. Moreover, since $y \leq m-\tilde x(\mathcal A)$, any arriving high-reward request is accepted. Thus, for any ordering of the arrivals, algorithm $\mathcal A$ can accept  $y$ high-reward requests.

We now show the second part. Here, we want to show that under any other ordering of the arrivals, the number of low-reward requests that are accepted is greater than or equal to $\tilde x(\mathcal A)$. Contrary to our claim, suppose that the total number of low-reward requests that algorithm $\mathcal A$ accepts under any unordered sequence is strictly less than $\tilde x(\mathcal A)$. This means that at some point, condition (2) is not satisfied.  Therefore, there exists a time $t$ such that $x(t)+y(t)=m$, where $x(t)$ and $y(t)$ denote the number of low- and high-reward requests accepted up to time $t$, respectively. However, we know that $x(t)<\tilde x(\mathcal A)$, $y(t) \leq y$, and $\tilde x(\mathcal A)+y \leq m$. Therefore, $x(t)+y(t)=m$ cannot hold, which contradicts the assumption that algorithm $\mathcal A$ has accepted fewer than $\tilde x(\mathcal A)$ low-reward requests. This completes the proof of the first case.

\textbf{Case 2: $y \geq m-\tilde x(\mathcal A)$}. In this case, we claim that under any other ordering of the requests, the number of high-reward requests accepted cannot be smaller than $m-\tilde x(\mathcal A)$. Additionally, we claim that if the number of low-reward requests accepted under any other ordering is less  than $\tilde x(\mathcal A)$, then the total reward generated by the algorithm is larger than the case where the number of low-reward requests accepted is $\tilde x(\mathcal A)$. 

We begin with the first part. We will show that  under any other ordering of the requests, the number of high-reward requests accepted cannot be smaller than $m-\tilde x(\mathcal A)$. As we know, high-reward requests are rejected only when all resources are used up. Therefore, the later a high-reward request arrives, the less chance we accept the request. Compare any unordered arrival sequence with the ordered arrival sequence, each high-reward request arrives earlier, which implies that the number of high-reward requests accepted cannot be smaller than $m-\tilde x(\mathcal A)$.

Next, we show the second part, which is if the number of low-reward requests accepted under any other ordering is less  than $\tilde x(\mathcal A)$, then the total reward generated by the algorithm is larger than the case where the number of low-reward requests accepted is $\tilde x(\mathcal A)$. If the total number of low-reward requests the algorithm $\mathcal A$ accepts is strictly less than $\tilde x(\mathcal A)$, then at some point, condition (2) is not satisfied. This means that for some time $t$, $x(t)+y(t)=m$, where $x(t)$ and $y(t)$ denote the number of low-reward and high-reward requests accepted before time $t$, respectively. Since $x(t) < \tilde x(\mathcal A)$, we have $y(t) > m-\tilde x(\mathcal A)$. Therefore, the reward generated by the algorithm is
\[
y(t)r_h+x(t)r_{\ell}=y(t)r_h+(m-y(t))r_{\ell} > (m-\tilde x(\mathcal A))r_h+\tilde x(\mathcal A)r_{\ell},
\]
which implies that the algorithm generates a larger reward.

Therefore, we conclude that the adversary should choose the first instance where all $x$ low reward requests arrive first and follow with all $y$ high reward requests.

\section{Proof of Statements in Section \ref{sec:maxrobust}} \label{sec:appendixC}

Throughout the proofs, we make use of some preliminary lemmas that are presented in Section \ref{sec:proofpre}.  

\subsection{Proof of Lemma \ref{lem:ymonotone}}
We first show that for any $p\le \min\{m,y_1\}$, we have $\CP_u(p;(x,y_1)) \geq \CP_u(p;(x,y_2))$, where we recall that $y_1\le y_2$. 
Observe that when $p\le  y_1$ and $x+y_1< m$, by Lemma \ref{lem:underbelow}, we have $\CP_u(p;(x,y_1))=1$, and hence $1=\CP_u(p;(x,y_1)) \geq \CP_u(p;(x,y_2))$ trivially holds. Now suppose that $x+y_1\ge m$. 
By definition, we have 
\begin{align*}
\CP_u(p;(x,y_1))&=\frac{\max\{p, (m-x)^{+} \} r_h+\min\{x,m-p\}r_{\ell}}{\min\{y_1,m\}r_h+\min\{x, (m-y_1)^{+}\}r_{\ell}}\\
&\geq \frac{\max\{p, (m-x)^{+} \} r_h+\min\{x,m-p\}r_{\ell}}{\min\{y_2,m\}r_h+\min\{x, (m-y_2)^{+}\}r_{\ell}} = \CP_u(p;(x,y_2))\,,
\end{align*}
where the inequality holds because $ y_2 \ge y_1$ and $y\mapsto yr_h+(m-y) r_{\ell}$ is increasing in $y$ as $r_h> r_{\ell}$.

Second, we show that any  protection level $p$ with   $p \geq \min\{m,y_2\}$, we have 
    $
    \CP_o(p;(x,y_2)) \geq \CP_o(p;(x,y_1))
    $. To show this, first consider the case where $x+y_1< m$ and $x+y_2< m$. Then,  by the definition of $\CP_o$ in Equation \eqref{eq:CP_over}, we have 
\begin{align} \label{eq:first_case}
    \CP_o(p;(x,y_2)) &= \frac{\min\{y_2,m\}r_h+\min\{x,m-p\}r_{\ell}}{\min\{y_2,m\}r_h+\min\{x, (m-y_2)^{+}\}r_{\ell}}=
    \frac{y_2r_h+\min\{x,m-p\}r_{\ell}}{y_2r_h+ x r_{\ell} }\\
   & \geq \frac{y_1r_h+\min\{x,m-p\}r_{\ell}}{y_1 r_h+xr_{\ell}} =\CP_o(p;(x,y_1))\,,
\end{align}
where the inequality holds because $y\mapsto \frac{y r_h+\min\{x,m-p\}r_{\ell}}{y r_h+x r_{\ell}}$ is increasing in $y$.

Now consider the case where  $x+y_2\ge m$. Next, we show that by fixing an $x$, for any $y$ such that $x+y>m$, $\CP_o(p;(x,y))$ is increasing in $y$. Observe that for any $y$ with $x+y \ge m$ and $p\ge y$, we have 
\[ \CP_o(p;(x,y)) = \frac{\min\{y,m\}r_h+\min\{x,m-p\}r_{\ell}}{\min\{y,m\}r_h+ \min\{x, (m-y)^{+}\}r_{\ell} },\]
If $y \geq m$, the statement is trivial. Otherwise,
\begin{align*}
    \frac{\partial \CP_o(p; (x,y))}{\partial y}&=\frac{r_h(yr_h+(m-y)r_{\ell})-(r_h-r_{\ell})(yr_h+\min\{m-p,x\}r_{\ell})}{(yr_h+(m-y)r_{\ell})^2} \\& \geq \frac{r_h(yr_h+(m-y)r_{\ell})-(r_h-r_{\ell})(yr_h+(m-p)r_{\ell})}{(yr_h+(m-y)r_{\ell})^2} \\& \geq \frac{(r_h-r_{\ell})(yr_h+(m-y)r_{\ell})-(r_h-r_{\ell})(yr_h+(m-p)r_{\ell})}{(yr_h+(m-y)r_{\ell})^2}\\
    &= \frac{(r_h-r_{\ell})(yr_h+(p-y)r_{\ell})}{(yr_h+(m-y)r_{\ell})^2}
    \geq 0,
\end{align*}
where the last inequality is because $p\ge y$. The chain of inequalities shows that $\CP_o(p; (x,y))$ is increasing in $y$ when $x+y\ge  m$ and $p\ge y$. This implies that we have  $
    \CP_o(p;(x,y_2)) \geq \CP_o(p;(x,y_1))
    $ when $x+y_i\ge m$, $i\in \{1, 2\}$, as desired.

    For the case where $x+y_2\ge m$ and $x+y_1< m$,  we have 
\begin{align*}
    \CP_o(p; (x,y_2)) \geq \CP_o(p;(x,m-x)) \geq \CP_o(p; (x,y_1))\,,
\end{align*}
where the first inequality holds because  $\frac{\partial\CP_o(p; (x,y))}{\partial y}\ge 0$ when $x+y\ge m$, and the second inequality holds because of
Equation \eqref{eq:first_case}.

Finally, we show that 
any $x\in[\underline x, \bar x]$ and $p\ge 0$,  we have
    $\min_{y\in [\underline h(x), \bar h(x)]} \{\CP(p; (x, y))\} = \min\left\{ \CP(p; (x, \underline h(x))),  \CP(p; (x, \bar h(x)))\right\}$. Suppose that $p\le  \underline h(x)$. Then, $\CP(p; (x, y)) = \CP_u(p; (x, y))$ for any $y\in [\underline h(x), \bar h(x)]$, and hence by the first result of this lemma, we have \[\min_{y\in [\underline h(x), \bar h(x)]} \{\CP(p; (x, y))\} = \CP_u(p; (x, \bar h(x)))\,,\] 
    as desired. Now, suppose that  $p\ge  \bar h(x)$. Then, $\CP(p; (x, y)) = \CP_o(p; (x, y))$ for any $y\in [\underline h(x), \bar h(x)]$, and hence by the second result of this lemma, we have \[\min_{y\in [\underline h(x), \bar h(x)]} \{\CP(p; (x, y))\} = \CP_o(p; (x, \underline h(x)))\,,\] as desired. Now, suppose the final case where $p \in (\underline h(x), \bar h(x))$. Then,
    \begin{align*}
        \min_{y\in [\underline h(x), \bar h(x)]} \{\CP(p; (x, y))\} &= \min\left\{ \min_{y\in[\underline h(x), p]} \{\CP_o(p; (x, y))\}, \min_{y\in[ p, \bar h(x)]} \{\CP_u(p; (x, y))\} \right\} \\
        &= \min\left\{ \CP_o(p; (x,\underline h(x))),  \CP_u(p; (x, \bar h(x))) \right\}\,,
    \end{align*}
    where the last inequality follows from the first and second results of this lemma.
    
\Halmos

\subsection{ Properties of functions $\u(\cdot; C)$ and $\ll(\cdot; C)$: Lemma \ref{lem:property_u_l} and its Proof} 
\begin{lemma}[Properties of functions $\u(\cdot; C)$ and $\ll(\cdot; C)$]\label{lem:property_u_l}
The functions $\u(\cdot; C)$ and $\ll(\cdot; C)$, which are respectively  defined in Equations \eqref{eq:u} and \eqref{eq:l}, have the following properties. 
\begin{enumerate}
    \item For any $x\in (\underline x_u, \bar x_u)$ and $C \in [0,1]$, let $\underline{\mathcal{H}}(x)=\min\{\underline h(x),m\}$.  When $\underline{\mathcal{H}}'(x)$ exists, we have 
    \begin{align} \label{eq:convexcase1p}  \frac{\partial \u(x;C)}{\partial x} = \left\{ \begin{array}{ll}
         ((1-C)\frac{r_h}{r_{\ell}}+C)\underline{\mathcal{H}}'(x) &\quad  \mbox{if $x+\underline{\mathcal{H}}(x) \geq m$};\\
        (1-C)\frac{r_h}{r_{\ell}}\underline{\mathcal{H}}'(x)-C &\quad  \mbox{if $x+\underline{\mathcal{H}}(x) < m$}.\end{array} \right. \end{align} 
        
        \item  For any $x\in (x_H, \bar x_l)$ and $C \in [0,1]$, let $\overline{\mathcal{H}}(x)=\min\{\bar h(x),m\}$. When $\overline{\mathcal{H}}'(x)$ exists, we have 
        \begin{equation}\label{eq:convexcase2p}
    \frac{\partial \ll(x; C)}{\partial x}=C \overline{\mathcal{H}}'(x).
    \end{equation}
    \item  For any $C \in [0,1]$,  $\u(x;C)$ is non-increasing  in $x\in [0, \bar x]$ and is convex for $x \in [\underline x_u, \bar x]$. 
    \item   For any $C \in [0,1]$, $\ll(x;C)$ is non-increasing in $x\in [0, \bar x]$ and is concave for $x \in [\underline x, \bar x_{\ell}]$.
    \item For any $C\le \cstar$ and any $x\in [0, \bar x]$ , we have
    $\ll(x; C) \le \u(x; C).$
    \item For any $x \in [0,\bar{x}]$, $\ll(x;C)$ is continuously increasing in $C$ and $\u(x;C)$ is continuously decreasing in $C$.
\end{enumerate}
\end{lemma}

\begin{proof} {Proof of Lemma \ref{lem:property_u_l}}
Here, we will show the following six properties. 

\subsubsection{Property 1}

We first show Equation \eqref{eq:convexcase1p}. We split the analysis into two cases: \textit{Case 1: $x+\underline{\mathcal{H}}(x) \geq m$} and  \textit{Case 2: $x+\underline{\mathcal{H}}(x) < m$}.

\textbf{Case 1: ($x+\underline{\mathcal{H}}(x)\ge m$). } 
By Lemma \ref{lem:Hh}, we have for any $x\in (\underline x_u, \bar x_u)$, $\underline{\mathcal{H}}(x)=\underline h(x)$, which implies that $\underline{h}(x) \leq m$. Then, we take an arbitrary point $(x_1,\underline{h}(x_1))$ with $x_1 \in (\underline{x_u},\bar x_u)$. By definition of $\u(\cdot;C)$, we should have  $\CP_o(\u(x_1;C);(x_1,\underline{h}(x_1)))=C$, and by Lemma \ref{lem:cpexists}, we have such $\u(x_1;C)$ always exists and $\u(x_1;C)\geq \underline{h}(x_1)$. By Equation \eqref{eq:CP_over},
\begin{align} \label{eq:cpo=c}
\CP_o(\u(x_1;C);(x_1,\underline{h}(x_1))) = \frac{\underline{h}(x_1)r_h+\min\{x_1,m-\u(x_1;C)\}r_{\ell}}{\underline{h}(x_1)r_h+\min\{x_1, m-\underline{h}(x_1)\}r_{\ell}} = C.
\end{align}
As in this case, if $x_1+\underline{h}(x_1) \geq m$, we have $\min\{x_1, m-\underline{h}(x_1)\}=m-\underline{h}(x_1)$. We then argue that  $\min\{x_1,m-\u(x_1;C)\}=m-\u(x_1;C)$. 
Suppose that contrary to our claim, $\min\{x_1,m-\u(x_1;C)\}=x_1$. We then have 
\[
\CP_o(\u(x_1;C);(x_1,\underline{h}(x_1))) = \frac{\underline{h}(x_1)r_h+x_1r_{\ell}}{\underline{h}(x_1)r_h+( m-\underline{h}(x_1))r_{\ell}} = \CP_o(m;(x_1,\underline{h}(x_1))),
\]
which implies that $x_1 \leq \underline x_u$. However, we define $x_1 \in (\underline x_u, \bar x_u)$, therefore, this is a contradiction, and we can only have $\min\{x_1,m-\u(x_1;C)\}=m-\u(x_1;C)$. Then, by Equation \eqref{eq:cpo=c}, we have
\[
\frac{\underline{h}(x_1)r_h+(m-\u(x_1;C))r_{\ell}}{\underline{h}(x_1)r_h+(m-\underline{h}(x_1))r_{\ell}} = C,
\]
which is equivalent to 
\[
\u(x_1;C)=((1-C)\frac{r_h}{r_{\ell}}+C)\underline{h}(x_1)+m(1-C).
\]
This implies that $\frac{\partial u(x; C)}{\partial x}=((1-C)\frac{r_h}{r_{\ell}}+C)\underline{h}'(x)$ when $x+\underline h(x)\ge m$. As we have $\underline{\mathcal{H}}(x)=\underline{h}(x)$ for $x \in [\underline{x_u},\bar x_u]$, this implies the desired result.

\textbf{Case 2 ($x+\underline{\mathcal{H}}(x)< m$). } In this case, as $x+\underline{\mathcal{H}}(x)< m$, we have $\underline{\mathcal{H}}(x) < m$, which implies that $\underline{\mathcal{H}}(x) = \underline h(x)$. Let us take a point $(x_1,\underline{h}(x_1))$ with $x_1 \in (\underline{x_u},\bar x_u)$. By definition of $\u(\cdot;C)$, we have $\CP_o(\u(x_1;C);(x_1,\underline{h}(x_1)))=C$ and by Lemma \ref{lem:cpexists}, we have such $\u(x_1;C)$ always exists and $\u(x_1;C) \geq \underline{h}(x_1)$. By Equation \eqref{eq:CP_over},
\[
\CP_o(\u(x_1;C);(x_1,\underline{h}(x_1))) = \frac{\underline{h}(x_1)r_h+\min\{x_1,m-\u(x_1;C)\}r_{\ell}}{\underline{h}(x_1)r_h+\min\{x_1, m-\underline{h}(x_1)\}r_{\ell}} = C.
\]
As in this case $x_1+\underline{h}(x_1) < m$, we have $\min\{x_1, m-\underline{h}(x_1)\}=x_1$. Here, we argue that $\min\{x_1,m-\u(x_1;C)\}=m-\u(x_1;C)$. 
Contrary to our claim, suppose that 
 $\min\{x_1,m-\u(x_1;C)\}=x_1$. We then have
\[
\CP_o(\u(x_1;C);(x_1,\underline{h}(x_1))) = \frac{\underline{h}(x_1)r_h+x_1r_{\ell}}{\underline{h}(x_1)r_h+x_1r_{\ell}} = 1 > C,
\]
which is a contradiction. Therefore, we can only have $\min\{x_1,m-\u(x_1;C)\}=m-\u(x_1;C)$. Then, we have
\[
\frac{\underline{h}(x_1)r_h+(m-\u(x_1;C))r_{\ell}}{\underline{h}(x_1)r_h+x_1r_{\ell}} = C,
\]
which is equivalent as 
\[
\u(x_1;C)=(1-C)\frac{r_h}{r_{\ell}}\underline{h}(x_1)-Cx_1+m.
\]
This implies that $\frac{\partial u(x; C)}{\partial x}=(1-C)\frac{r_h}{r_{\ell}}\underline{h}'(x)-C$
when $x+\underline h(x)< m$. As we have $\underline{\mathcal{H}}(x)=\underline{h}(x)$ for $x \in [\underline{x_u},\bar x_u]$, this implies the desired result.

\subsubsection{Property 2}
 We take a point $(x_1,\overline{\mathcal{H}}(x_1))$ with $x_1 \in (x_H,\bar{x_{\ell}})$. By definition of $\ll(\cdot;C)$, we have $\CP_u(\ll(x_1;C);(x_1,\bar{h}(x_1)))=C$. As $\overline{\mathcal{H}}(x_1)=\min\{m,\overline{h}(x_1)\}$, by Lemma \ref{lem:mout}, we have $\CP_u(\ll(x_1;C);(x_1,\overline{\mathcal{H}}(x_1)))=\CP_u(\ll(x_1;C);(x_1,\bar{h}(x_1)))=C$. 
By Lemma \ref{lem:cpexists}, such $\ll(x_1;C)$ always exists and $\ll(x_1;C) \leq \overline{\mathcal{H}}(x_1)$. By Equation \eqref{eq:CP_under},
\[
\CP_u(\ll(x_1;C);(x_1,\overline{\mathcal{H}}(x_1))) =  
 \frac{\max\{\ll(x_1;C), \min\{\overline{\mathcal{H}}(x_1),m-x_1\} \} r_h+\min\{x_1,m-\ll(x_1;C)\}r_{\ell}}{\overline{\mathcal{H}}(x_1)r_h+\min\{x_1, m-\overline{\mathcal{H}}(x_1)\}r_{\ell}} = C. 
\]

If $\min\{x_1, m-\overline{\mathcal{H}}(x_1)\}=x_1$, by Lemma \ref{lem:underbelow} and the fact that $\ll(x_1;C) \leq \overline{\mathcal{H}}(x_1)$, $\CP_u(\ll(x_1;C);(x_1,\overline{\mathcal{H}}(x_1))) = 1 \neq C $, which cannot happen, and hence $\min\{x_1, m-\overline{\mathcal{H}}(x_1)\}=m-\overline{\mathcal{H}}(x_1)$. If $\min\{x_1,m-\ll(x_1;C)\}=x_1$, then 
\[
\CP_u(\ll(x_1;C);(x_1,\overline{\mathcal{H}}(x_1))) =  
 \frac{(m-x_1) r_h+x_1r_{\ell}}{\overline{\mathcal{H}}(x_1)r_h+( m-\overline{\mathcal{H}}(x_1))r_{\ell}} = \CP_u(0;(x_1,\overline{\mathcal{H}}(x_1))), 
\]
which implies that $x_1 \geq \bar x_{\ell}$. However, we define $x_1 \in (x_H, \bar x_{\ell})$, therefore, this is a contradiction, and we can only have $\min\{x_1, m-\overline{\mathcal{H}}(x_1)\}=m-\overline{\mathcal{H}}(x_1)$ and $\min\{x_1,m-\ll(x_1;C)\}=m-\ll(x_1;C)$. Then, we have 
\[
\frac{\ll(x_1;C) r_h+(m-\ll(x_1;C))r_{\ell}}{\overline{\mathcal{H}}(x_1)r_h+(m-\overline{\mathcal{H}}(x_1))r_{\ell}} = C, 
\]
which is equivalent to 
\[
\ll(x_1;C)=C\overline{\mathcal{H}}(x_1)-\frac{(1-C)mr_{\ell}}{r_h-r_{\ell}}\,,
\]
and verifies Equation \eqref{eq:convexcase2p}.

\subsubsection{Property 3}
We first show that $\u(x;C)$ is non-increasing for $x \in [\underline x, \bar x]$. First, by the definition of $\underline x_u$, we have $\u(x;C)=m$ for $x \in [\underline x, \underline x_u]$, which is non-increasing. 

For $x\in (\underline x_u, \bar x_u)$, first recall that  \begin{align*} \bar x_u = \left\{ \begin{array}{ll}
         x_L & \quad \mbox{if $x_L+y_L \geq m$};\\
        \sup\{x \in [x_L,\bar{x}]: (1-C)\frac{r_h}{r_{\ell}}\underline{\mathcal{H}}'(x^{-})-C < 0 \} & \quad  \mbox{Otherwise}\,,\end{array} \right. \end{align*} 

\textbf{Case 1 -- $\mathbf {x_L+y_L \geq m}$.} If $x_L+y_L \geq m$, we have $\bar x_u = x_L$. Since point $L =(x_L, y_L)$ is the lowest point, $\underline{h}(x)$ decreases for $x<x_L$, increases for $x>x_L$, which implies that $\underline{h}'(x) \leq 0$ for $x<x_L$ and $\underline{h}'(x) \geq 0$ for $x>x_L$.  Now, recall Property 1 that we just showed: 
 \begin{align*} \frac{\partial \u(x;C)}{\partial x} = \left\{ \begin{array}{ll}
         ((1-C)\frac{r_h}{r_{\ell}}+C)\underline{\mathcal{H}}'(x) &\quad  \mbox{if $x+\underline{\mathcal{H}}(x) \geq m$};\\
        (1-C)\frac{r_h}{r_{\ell}}\underline{\mathcal{H}}'(x)-C &\quad  \mbox{if $x+\underline{\mathcal{H}}(x) < m$}.\end{array} \right. \end{align*} 
        where by Lemma \ref{lem:Hh}, we have $\underline{\mathcal{H}}(x)=\underline h(x)$ for $x \in (\underline x_u, \bar x_u)$.
        This property and the fact that $\underline{h}'(x) \leq 0$ for $x<x_L$ and $\underline{h}'(x) \geq 0$ for $x>x_L$ imply that  
$\u(x;C)$ decreases for $x<x_L$ and increases for $x>x_L$. As we force $\u(x;C)=\u(x_L;C)$ for $x > \bar x_u = x_L$, we have $\u(x;C)$ is always non-increasing.

{\textbf Case 2 -- $\mathbf {x_L+y_L < m}$.} If $x_L+y_L < m$, we have $\xmin=\sup\{x \in [x_L,\bar x]: (1-C)\frac{r_h}{r_{\ell}}\underline{h}'(x)-C < 0 \}$. For $x\in (\underline x_u, \bar x_u)$, by Lemma \ref{lem:Hh}, we have $\underline{\mathcal{H}}(x)=\underline h(x)$. As $\underline h(x)$ is convex, we have its subderivative is increasing. Then,  for $x<\xmin$, we have $(1-C)\frac{r_h}{r_{\ell}}\underline{\mathcal{H}}'(x)-C < 0$. Therefore, Property (1) that we just showed (i.e., Equation \eqref{eq:convexcase1p}) implies that $\u(x;C)$ decreases for $x \in [\underline x_u, \xmin]$. As we force $\u(x;C)$ to be a constant for $x \in [\xmin,\bar x]$, we have $\u(x;C)$ is non-increasing for $x \in [\underline x,\bar x]$.

Finally, we show that $\u(x;C)$ is convex for $x \in [\underline x_u, \bar x]$ by proving that its subderivative is increasing. As $\underline{h}(x)$ is convex, we have the subderivative of $\underline{h}(x)$ is increasing for $x \in [\underline x_u, \xmin]$. By Lemma \ref{lem:Hh}, we have $\underline{\mathcal{H}}(x)=\underline h(x)$ for $x\in [\underline x_u, \bar x_u]$. Then, we have the subderivative of $\underline{\mathcal{H}}(x)$ is increasing for $x \in [\underline x_u, \xmin]$. Therefore, Equation \eqref{eq:convexcase1p} implies that $\u(x;C)$ has increasing subderivative, and $\u(x;C)$ is convex for $x \in [\underline x_u, \xmin]$. As $\u(x;C)$ is non-increasing and convex  for $x \in [\underline x_u, \xmin]$ and constant for $x \in [\xmin, \bar x]$, we have $\u(x;C)$ is convex for $x \in [\underline x_u, \bar x]$.

\subsubsection{Property 4}  Because $H$ is the highest point and $\region$ is convex, $\bar{h}(x)$ increases for $x<x_H$ and  decreases for $x>x_H$, which implies that $\bar{h}'(x) \geq 0$ a.e. for $x<x_H$ and $\bar{h}'(x) \leq 0$ a.e. for $x>x_H$. (Recall that $\bar h(\cdot)$ is concave and point 
$H=(x_{H}, y_{H})\in \widebar{\region} $, which lies on the upper envelope $\bar h(\cdot)$,  is the point in set $\widebar{\region}$ that has the highest low-reward demand, where 
$
\widebar{\region} = \{(x,y) \in \region : y = \sup_{(x',y') \in \region } \min\{y',m\}  \} 
$
is a subset of region $\region$ under which the high-reward demand (more precisely $\min\{y', m\}$ for  any point $(x', y')\in \region$) is maximized.)
Therefore, Property (2) that we just showed, (i.e., $\frac{\partial \ll(x; C)}{\partial x}=C \overline{\mathcal{H}}'(x)$) implies that $\ll(x;C)$ increases for $x<x_H$. As we force $\ll(x;C)=\ll(x_H;C)$ for $x \geq x_H$, we have $\ll(x;C)$ is always non-increasing.

Next, we show that $\ll(\cdot;C)$ is concave. As stated earlier, because $\region$ is a convex set, we have $\bar{h}(x)$ is concave.  Since both $\bar{h}(x)$ and $y=m$ are concave, we have $\overline{\mathcal{H}}(x) = \min\{\bar{h}(x),m\}$ is concave. Therefore, the subderivative of $\overline{\mathcal{H}}(x)$ is decreasing. By Equation \eqref{eq:convexcase2p}, we have the subderivative of $\ll(x;C)$ decreases for $x \in [\underline x_{\ell}, \bar x_{\ell}]$, which implies that $\ll(x;C)$ is concave for $x \in [\underline x_{\ell}, \bar x_{\ell}]$. As $\ll(x;C)$ is a constant for $x \in [\underline x, \underline x_{\ell}]$ and $\ll(x;C)$ is concave and  non-increasing for $x \in [\underline x_{\ell}, \bar x_{\ell}]$, we have $\ll(x;C)$ is concave for $x \in [\underline x, \bar x_{\ell}]$.

\subsubsection{Property 5} Here, we want to show that 
for any $C\le \cstar$ and any $x\in [\underline x, \bar x]$, we have
    $\ll(x; C) \le \u(x; C).$ For $x \in [\underline{x},\bar{x}]$, let $\pbal(x)$ be a function that 
\[
\CP_o(\pbal(x);(x,\underline{h}(x)))=\CP_u(\pbal(x);(x,\bar{h}(x))).
\]
Lemma \ref{lem:inner}  shows that such $\pbal(x)$ always exists, and $\CP_o(\pbal(x);(x,\underline{h}(x)))=\CP_u(\pbal(x);(x,\bar{h}(x))) \geq \cstar$. Then, by Lemma \ref{lem:pmonotone}, for any $C \leq \cstar$, we have 
$\CP_o(\ll(x;C);(x,\underline{h}(x)))=C$ implies that $\ll(x;C) \leq \pbal(x)$, and $\CP_u(\u(x;C);(x,\bar{h}(x)))=C \leq \cstar$ implies that $\u(x;C) \geq \pbal(x)$. Therefore, we have $\ll(x;C) \leq \pbal(x) \leq \u(x;C)$.

\subsubsection{Property 6} Here, we would like to show  for any $x \in [\underline{x},\bar{x}]$, $\ll(x;C)$ is continuously increasing in $C$ and $\u(x;C)$ is continuously decreasing in $C$. This is because we have showed 
\[
\ll(x_1;C)=C\overline{\mathcal{H}}(x_1)-\frac{(1-C)mr_{\ell}}{r_h-r_{\ell}}\,,
\]
and 
\[
\u(x_1;C)=(1-C)\frac{r_h}{r_{\ell}}\underline{h}(x_1)-Cx_1+m.
\]

From these two equations, we can simply find that for any $x \in [\underline{x},\bar{x}]$, $\ll(x;C)$ is continuously increasing in $C$ and $\u(x;C)$ is continuously decreasing in $C$.
\end{proof}

\Halmos

\subsection{Proof of Lemma \ref{lem:feasibleregionrobust}}

\textbf{First Direction.} We first show that if $\ll(x;C) \leq p(x) \leq \u(x;C)$, we have $\UCP(p(x);x) \geq C$ for any $x \in [\underline{x},\bar{x}]$, where we define 
\begin{equation} \label{eq:defucpou1}
\UCP(p(x);x) = \min\{\CP_o(p(x);(x,\underline{h}(x))), \CP_u(p(x);(x,\bar{h}(x))) \}.
\end{equation}
This gives us the desired result because by Lemma \ref{lem:ymonotone}, 
for any $x\in[\underline x, \bar x]$ and $p\ge 0$,  we have
    \[\min_{y\in [\underline h(x), \bar h(x)]} \{\CP(p; (x, y))\} = \min\left\{ \CP(p; (x, \underline h(x))),  \CP(p; (x, \bar h(x)))\right\}\,.\]

\textbf{Part 1: $\boldsymbol{\CP_o(p(x);(x,\underline{h}(x)))\ge C}$ if $\boldsymbol{p(x)\le \u(x;C)}$.} Here, we show that for any $x\in [\underline x, \bar x]$,   $\CP_o(p(x);(x,\underline{h}(x)))$ is greater than or equal to $C$ as long as $p(x)\le \u(x;C)$. Let us first focus on $x\in [\underline x, \underline x_u]$ and $x\in [\bar x_u, \bar x]$. 
By Lemma \ref{lem:leftthresholdun},  
for any $\underline x<x \leq \underline x_u$, we have $
    \CP_o(m; (x,\underline{h}({x}))) \geq C$. By Lemma \ref{lem:worstrightpart}, for any  $x\in [\bar x_u,\bar{x}]$, we have $
\CP_o(m;(x,\underline{h}(x))) \geq C.$ 
By Lemma \ref{lem:pmonotone},  for any $p(x) \leq m=\u(x;C)$ for $x \in [\underline x, \underline x_u]$ and $[\bar x_u,\bar{x}]$, we have
\[
\CP_o(p(x);(x,\underline{h}(x))) \geq \CP_o(m;(x,\underline{h}(x))) \geq C.
\]

Next, we consider  $x \in [\underline x_u,\bar x_u]$.  By Lemma \ref{lem:cpexists}, we have
\[
\CP_o(\u(x;C);(x,\underline{h}(x)))=C,
\]
and by Lemma \ref{lem:pmonotone}, we have for any $p(x) \leq \u(x;C)$,
\[
\CP_o(p(x);(x,\underline{h}(x))) \geq \CP_o(\u(x;C);(x,\underline{h}(x)))=C\,, 
\]
which is the desired result. 

\textbf{Part 2: $\boldsymbol{\CP_u(p(x);(x,\bar{h}(x)))\ge C}$ if $\boldsymbol{ p(x)\ge \ll(x;C)}$.} Here, we show that for any $x\in [\underline x, \bar x]$,   $\CP_u(p(x);(x,\bar{h}(x)))$ is greater than or equal to $C$ as long as $p(x)\ge \ll(x;C)$. Let us first consider any $x\in [\underline x, x_H]$ and $x\in [\bar x_l, \bar x]$. 
By the definition of $x_H$ and Lemma \ref{lem:leftthresholdun}, for $x \in [\underline{x},x_H]$ and $[\bar x_{\ell},\bar{x}]$, we have $
\CP_u(0;(x,\bar{h}(x))) \geq C.
$
By Lemma \ref{lem:pmonotone}, for any $p(x) \geq 0=\ll(x;C)$ for $x \in [\underline{x},x_H]$ and $[\bar x_{\ell},\bar{x}]$, we then have 
\[
\CP_u(p(x);(x,\bar{h}(x))) \geq \CP_u(0;(x,\bar{h}(x))) \geq C\,,
\]
which is the desired result. 
Now, let us consider any $x \in [x_H,\bar x_{\ell}]$, By Lemma \ref{lem:cpexists},  we have
$
\CP_u(\ll(x;C);(x,\bar{h}(x)))=C,$ 
and by Lemma \ref{lem:pmonotone}, we have for any $p(x) \geq \ll(x;C)$,
\[
\CP_u(p(x);(x,\bar{h}(x))) \geq \CP_u(\ll(x;C);(x,\bar{h}(x)))=C\,,
\]
which is the desired result. 

\textbf{Second Direction.} So far we have established that if $p(x)\in [\ll(x;C), \u(x;C)]$, we have $\UCP(p(x);x)\ge C$ for any $(x, y)\in \region$. 
Next, we show that if $p(x)> \u(x;C)$ or $p(x)<\ll(x;C)$, we have $\UCP(p(x);x)<C$. 

\textbf{Part 1: $\boldsymbol{\CP_o(p(x); (x, \underline h(x)))< C}$ if $\boldsymbol{p(x)> u(x; C)}$.}
First, as $\u(x;C)=m$ for $x \in [\underline{x},\underline x_u]$ and $[\bar x_u,\bar{x}]$, and $p(x) \leq m$, we cannot have $p(x)>\u(x;C)$. Thus, we need to only consider  $x \in [\underline x_u,\bar x_u]$. For any $x \in [\underline x_u,\bar x_u]$,  by definition, $\u(x;C)$ is the largest PL value such that
\[
\CP_o(\u(x;C);(x,\underline{h}(x)))=C,
\]
and by Lemma \ref{lem:pmonotone}, if $p(x)>\u(x;C)$, we have
\[
\CP_o(p(x);(x,\underline{h}(x))) < \CP_o(\u(x;C);(x,\underline{h}(x)))=C\,, 
\]
which is the desired result. 

\textbf{Part 2: $\boldsymbol{\CP_u(p(x);(x,\bar{h}(x)))< C}$ if $\boldsymbol{ p(x)< \ll(x;C)}$.} As $\ll(x;C)=0$ for $x \in [\underline{x},x_H]$ and $[\bar x_{\ell},\bar{x}]$, and $p(x) \geq 0$, we cannot have $p(x)<\ll(x;C)$. Thus, we consider  $x \in [x_H,\bar x_{\ell}]$. For any  $x \in [x_H,\bar x_{\ell}]$, by definition, $\ll(x;C)$ is the smallest PL value such that
\[
\CP_u(\ll(x;C);(x,\bar{h}(x)))=C,
\]
and by Lemma \ref{lem:pmonotone}, if $p(x)<\ll(x;C)$, we have
\[
\CP_u(p(x);(x,\bar{h}(x))) < \CP_u(\ll(x;C);(x,\bar{h}(x)))=C\,,
\]
which is the desired result.

\subsection{Proof of Lemma \ref{lem:feasibleregionrobust_2}}
To show the result, we show the optimization problem in Equation \eqref{prob:rbmaxtrans2} is equivalent to that in Equation \eqref{prob:rbmaxtrans1}. Since the only difference between these two problems is their first set of constraints, we only need to show that the feasible regions of these two problems are identical. 
To do so, we show that any feasible solution to  Problem \eqref{prob:rbmaxtrans2} is a feasible solution to Problem \eqref{prob:rbmaxtrans1} and vice versa.

Considering a feasible solution to Problem \eqref{prob:rbmaxtrans2} with $p(x)\in [\wll(x;C), \u(x;C)]$ for any $x\in [\underline x, \bar x]$. 
By Equation \eqref{eq:ll}, we know that for any $C \in [0,1]$ and $x \in [\underline{x},\bar{x}]$, $\wll(x;C) \geq \ll(x;C)$. Therefore, $\wll(x;C) \leq p(x) \leq \u (x;C)$ implies that $\ll(x;C) \leq p(x) \leq \u (x;C)$, as desired. Recall that $\ll(x;C) \leq p(x) \leq \u (x;C)$ is the first constraint in Problem \eqref{prob:rbmaxtrans1}, and hence the above argument shows that any feasible solution to Problem \eqref{prob:rbmaxtrans2} is a feasible solution to  Problem \eqref{prob:rbmaxtrans1}.

Next, we show the opposite direction. Contrary to our claim, suppose that there exists a feasible solution $p(x)$ to  Problem \eqref{prob:rbmaxtrans1}  with $\ll(x_1;C) \leq p(x_1) < \wll(x_1;C)$ for some $x_1\in [\underline x, \bar x]$. (This shows that there exists a feasible solution to Problem \eqref{prob:rbmaxtrans1}, which is not a feasible solution to Problem \eqref{prob:rbmaxtrans2}.) By Equation \eqref{eq:ll}, we must have $x_1 > \nex $, where $\nex$ is defined in Equation \eqref{eq:ll}. As $\ll(\nex;C)=\wll(\nex;C)$, we have $p(\nex) \geq \ll(\nex;C)=\wll(\nex;C)$. Then, we have
\[
\frac{p(x_1)-p(\nex)}{x_1-\nex}< \frac{\wll(x_1;C)-\wll(\nex;C)}{x_1-\nex} = -1,
\]
where the equation holds because by definition of $\wll(x;C)$, the slope of $\wll(x;C)$ w.r.t. $x$ is $-1$ for any $x\in [\nex, x_1]$. 
That $\frac{p(x_1)-p(\nex)}{x_1-\nex}< -1$ implies that $p'(x)<-1$ on a positive measure set, and  hence, $p(x)$ is not a valid PL function, which is a contradiction.

\subsection{Proof of Lemma \ref{lem:trapezoid}} 
\textbf{First Direction.} We first show the `if' statement. That is, if $\underline{g}(x;R) \leq p(x) \leq \bar{g}(x;R)$, we have $\UCP(p(x);x) \geq R$ for any $x \in [0,\max\{m,\bar x\}]$, where with a slight abuse of notation, we define
\begin{equation} \label{eq:defucpou}
\UCP(p(x);x) = \min\{\CP_o(p(x);(x,0)), \CP_u(p(x);(x,m)) \}.
\end{equation}
Notice that by Lemma \ref{lem:mout}, we have $ \CP_u(p(x);(x,m))= \CP_u(p(x);(x,y))$ for any $y \geq m$. Then, by Lemma \ref{lem:ymonotone}, it suffices to show $\UCP(p(x);x) \ge R$ when $\underline{g}(x;R) \leq p(x) \leq \bar{g}(x;R)$. 

\textbf{Part 1: $\boldsymbol{\CP_o(p(x);(x,0))\ge R}$ if $\boldsymbol{p(x)\le \bar g(x; R)}$.} First observe that, by Definition of $\CP_o$ in Equation \eqref{eq:CP_over}, if we set $p(x) = \bar g(x;R)$, we have 
\begin{align*}
\CP_o(\bar{g}(x;R);(x,0)) &= \frac{0 \cdot r_h+\min\{x,m-\bar{g}(x;R)\}r_{\ell}}{0 \cdot r_h+\min\{x, m-0\}r_{\ell}} = \frac{\min\{x,m-\bar{g}(x;R)\}r_{\ell}}{\min\{x, m\}r_{\ell}}.
\end{align*}
If $x \leq m$, we have $\bar{g}(x;R) = -Rx+m$, and we can obtain
\[
\frac{\min\{x,m-\bar{g}(x;R)\}r_{\ell}}{\min\{x, m\}r_{\ell}}=\frac{\min\{x,m-(-Rx+m)\}r_{\ell}}{xr_{\ell}} =\frac{Rx r_{\ell}}{x r_{\ell}}=R.
\]
Otherwise, if $x>m$, we have $\bar{g}(x;R) =\bar{g}(m;R) = -Rm+m$, and we can obtain
\[
\frac{\min\{x,m-\bar{g}(x;R)\}r_{\ell}}{mr_{\ell}}=\frac{\min\{x,m-(-Rm+m)\}r_{\ell}}{mr_{\ell}}=\frac{Rm r_{\ell}}{m r_{\ell}}=R.
\]

Then, by Lemma \ref{lem:pmonotone}, we have for any $p(x) \leq \bar{g}(x;R)$, we have 
\[
\CP_o(p(x);(x,0)) \geq \CP_o(\bar{g}(x;R);(x,0)) = R\,,
\]
which is the desired result.

\textbf{Part 2: $\boldsymbol{\CP_u(p(x);(x,m))\ge R}$ if $\boldsymbol{p(x)\ge \underline g(x; R)}$.} 
By Definition of $\CP_u$ in Equation \eqref{eq:CP_under}, we have
\begin{align*}
    \CP_u(p(x);(x,m))&=\frac{\max\{p(x), \min\{m,m-x\} \} r_h+\min\{x,m-p(x)\}r_{\ell}}{mr_h+\min\{x, m-m\}r_{\ell}} \\&= \frac{\max\{p(x), m-x \} r_h+\min\{x,m-p(x)\}r_{\ell}}{mr_h}\,.
\end{align*}
We would like to show that if ${p(x)\ge \underline g(x; R)}$, we have ${\CP_u(p(x);(x,m))\ge R}$, where  
$\underline{g}(x;R) = \frac{m(R-r_{\ell}/r_h)}{1-r_{\ell}/r_h}$ for $x \in [0,\max\{m,\bar{x}\}]$. If $p(x) \geq m-x$,  we have
\[
\CP_u(p(x);(x,m)) = \frac{p(x) r_h+(m-p(x))r_{\ell}}{mr_h} \geq \frac{\frac{m(R-r_{\ell}/r_h)}{1-r_{\ell}/r_h} r_h+(m-\frac{m(R-r_{\ell}/r_h)}{1-r_{\ell}/r_h})r_{\ell}}{mr_h} = R\,,
\]
where the inequality holds because $p(x)\ge \underline{g}(x;R) = \frac{m(R-r_{\ell}/r_h)}{1-r_{\ell}/r_h}$. Otherwise, if $p(x) < m-x$, as $p(x) \geq \underline{g}(x;R)$, we have $\underline{g}(x;R)<m-x$. Then,
\[
\CP_u(p(x);(x,m)) = \frac{(m-x) r_h+xr_{\ell}}{mr_h} \geq \frac{p(x) r_h+(m-p(x))r_{\ell}}{mr_h} \geq R,
\]
where the first inequality is because $\frac{(m-x) r_h+xr_{\ell}}{mr_h}$ is decreasing in $x$ and $x<m-p(x)$. The last inequality, which is the desired result, is because \[p(x) r_h+(m-p(x))r_{\ell}=p(x)(r_h-r_{\ell})+mr_{\ell} \geq \underline{g}(x;R)(r_h-r_{\ell})+mr_{\ell}\,,\] and by some calculations, we have $\frac{\underline{g}(x;R)(r_h-r_{\ell})+mr_{\ell}}{mr_h}=R$.

\textbf{Second Direction.} 
So far, we have established that if $p(x)\in [\underline g(x;R), \bar g (x;R)]$, we have $\UCP(p(x);x)\ge R$ for any $(x, y)\in \region$. 
Next, we show that if $p(x)> \bar g(x;R)$ or $p(x)<\underline g(x;C)$, we have $\UCP(p(x);x)<R$.

If $p(x)>\bar{g}(x;R)$ for some $x \in [0,\max\{m,\bar x\}]$, by Equation \eqref{eq:defgR}, we have if $x \in [0,m]$, $\bar{g}(x;R) = -Rx+m \ge -x+m$ and hence   $p(x)+x > m$. If $x>m$, we have $\bar{g}(x;R)=\bar{g}(m;R)$.  Then, we have
\begin{align*}
\CP_o(p(x);(x,0)) &= \frac{0 \cdot r_h+\min\{x,m-p(x)\}r_{\ell}}{0 \cdot r_h+\min\{x, m-0\}r_{\ell}} =\frac{\min\{x,m-p(x)\}r_{\ell}}{\min\{x, m\}r_{\ell}}.
\end{align*}
If $x \leq m$, we have
\[
\frac{\min\{x,m-p(x)\}r_{\ell}}{\min\{x, m\}r_{\ell}} = \frac{\min\{x,m-p(x)\}r_{\ell}}{xr_{\ell}}< \frac{\left(m-\bar{g}(x;R)\right)r_{\ell}}{xr_{\ell}} =R,
\]
where the inequality is because $p(x)>\bar{g}(x;R)$. Otherwise, if $x > m$, we have
\[
\frac{\min\{x,m-p(x)\}r_{\ell}}{\min\{x, m\}r_{\ell}} = \frac{(m-p(x))r_{\ell}}{mr_{\ell}} = \frac{m-p(x)}{m} < \frac{m-\bar{g}(m;R)}{m}=R.
\]

If $p(x)<\underline{g}(x;R)$ for some $x \in [0,\max\{m,\bar x\}]$, as a valid PL function $p(x)$ is non-increasing, we have $p(\max\{m,\bar x\})<\underline{g}(x;R)$. Therefore,
\begin{align*}
\CP_u(p(\max\{m,\bar x\});(\max\{m,\bar x\},m)) &= \frac{p(\max\{m,\bar x\}) r_h+(m-p(\max\{m,\bar x\}))r_{\ell}}{mr_h} \\
&< \frac{\frac{m(R-r_{\ell}/r_h)}{1-r_{\ell}/r_h} r_h+(m-\frac{m(R-r_{\ell}/r_h)}{1-r_{\ell}/r_h})r_{\ell}}{mr_h} = R.
\end{align*}

\subsection{Proof of Theorem \ref{thm:optimalextension}}
Algorithm \ref{alg:right} presents an optimal solution to Problem \eqref{prob:right}. That is, at the optimal solution to Problem \eqref{prob:right}, denoted by $\pright(\cdot)$, we set $\pright(x)$ based on Equations \eqref{eq:barx_opt} and \eqref{eq:opt_x_beyond}. Furthermore, the optimal objective value of  Problem \eqref{prob:right}, $\RR$, is given in Equation \eqref{eq:RR}.

We split the proof into three cases, in each case, we first figure out the robust ratio under the PL function $\pright(\cdot)$ and check the feasibility and optimality of $\pright(\cdot)$: 
\begin{itemize}
    \item \textit{Case 1: $[\wll(\bar{x};C),\u(\bar{x};C)] \cap [\underline{g}(\bar{x}),\bar{g}(\bar{x})] \neq \emptyset$. } In this case, if $\wll(\bar{x};C)<\underline{g}(\bar{x})$, \[\pright(\bar{x})=\arg\min_{p \in [\wll(\bar{x};C),\u(\bar{x};C)]} \vert p - \underline{g}(\bar{x}) \vert = \underline{g}(\bar{x}).\]
    If $\wll(\bar{x};C) \geq \underline{g}(\bar{x})$, we have $\pright(\bar{x}) = \wll(\bar{x};C)$. If $\bar{x} \geq m$, then as the right problem is only defined on $\bar{x}$, by definition, we have $\underline{g}(\bar{x})=\bar{g}(\bar{x})=\frac{1-r_{\ell}/r_h}{2-r_{\ell}/r_h}m$. Then, $\pright(\bar x)=\frac{1-r_{\ell}/r_h}{2-r_{\ell}/r_h}m$, which is feasible. By Lemma \ref{lem:mout}, we have
    \begin{align*}
    \RR&=\min\left\{\CP_o(\pright(\bar{x});(\bar{x},0)),\CP_u(\pright(\bar{x});(\bar{x},m))\right\} \\&= \min\left\{\CP_o(\frac{1-r_{\ell}/r_h}{2-r_{\ell}/r_h}m;(m,0)),\CP_u(\frac{1-r_{\ell}/r_h}{2-r_{\ell}/r_h}m;(m,m))\right\}=\frac{1}{2-r_{\ell}/r_h},
    \end{align*}
    which matches the upper bound of the robust ratio in the absence of ML advice. Therefore, $\pright(\bar x)$ is optimal in this case.

    Otherwise, if $\bar{x} <m$, by Equation \eqref{eq:rightext}, $\pright(x)=\max\{-x+\bar{x}+\pright(\bar{x};C), \underline g(x)\} $ for $x \in [\bar{x},m]$. By definition, we have $\underline{g}(x) \leq \pright(x)$. For the part where  $\pright(x)=\underline{g}(x)$, we have $\pright(x) \leq \bar{g}(x)$ because by Equation \eqref{eq:defgR}, $\underline{g}(x) \leq \bar{g}(x)$ for any $x \in [0,m]$. Then, we check that $-x+\bar{x}+\pright(\bar{x};C) \leq \bar{g}(x)$ for $x \in [\bar{x},m]$. Notice that $-x+\bar{x}+\pright(\bar{x};C)$ is a line with slope $-1$ and by Equation \eqref{eq:defgR}, $\bar{g}(x)$ is a line with slope $-R \geq -1$. Moreover, $-\bar{x}+\bar{x}+\pright(\bar{x};C)=\pright(\bar{x};C)=\max\{\underline{g}(\bar{x}),\wll(\bar{x};C)\} \leq \bar{g}(x)$ since $[\wll(\bar{x};C),\u(\bar{x};C)] \cap [\underline{g}(\bar{x}),\bar{g}(\bar{x})] \neq \emptyset$. We have $\pright(x) \leq \bar{g}(x)$. By taking $R=\rho$ in Problem \eqref{prob:trans}, we can find $\pright(\cdot)$ is a feasible solution, and therefore, it achieves a robust ratio of at least $\rho$.
    By \cite{ball2009toward}, we know $\rho$ is the upper bound among all algorithms, and therefore, $\pright(\cdot)$ is optimal. In addition, notice that $\pright(m)=\underline{g}(m)$ and we can check $\CP_u(\pright(\bar{x});(m,m))=\rho$, and hence $\RR=\min\left\{\CP_o(\pright(\bar{x});(\bar{x},0)),\CP_u(\pright(\bar{x});(m,m))\right\}$.

    \item \textit{Case 2: $\u(\bar{x};C) < \underline{g}(\bar{x})$. } In this case, we first show $\pright(\cdot)$ achieves a robust ratio of $\RR = \CP_u(\pright(\bar x); (\max\{m,\bar x\},m))$ and $\CP_o(\pright(\bar x);(\bar{x},0)) \geq \CP_u(\pright(\bar x); (\max\{m,\bar x\},m))$. Then, we show $\pright(\cdot)$ is feasible, and finally, we show it is optimal among all PL functions. 

    In this case, $\pright(x)=\arg\min_{p \in [\wll(\bar{x};C),\u(\bar{x};C)]} \vert p - \underline{g}(\bar{x}) \vert = \u(\bar{x};C)$, and by definition, $\pright(x)=\pright(\bar{x})$ for $x \in [\bar{x},\max\{m,\bar x\}]$.
    By Lemmas \ref{lem:ymonotone} and \ref{lem:mout}, we know the worst case is achieved on $(x,0)$ or $(x,m)$ for some $x \in [\bar{x},\max\{m,\bar x\}]$; that is, \[\RR=\inf_{x \in [\bar{x},\max\{m,\bar x\}]}\min\{\CP_u(\pright(x);(x,m)),\CP_o(\pright(x);(x,0)) \}.\]   As we have $\pright(x) \leq m$, by Lemma \ref{lem:worstunderrob}, we have
    \[
    \CP_u(\pright(x);(x,m)) \geq \CP_u(\pright(\max\{m,\bar x\});(\max\{m,\bar x\},m)).
    \]
    As $\pright(\max\{m,\bar x\})=\pright(\bar{x})<\underline{g}(\bar{x})= \underline g(\max\{m,\bar x\}) = m(1-r_{\ell}/r_h)/(2-r_{\ell}/r_h)$, where the first inequality is because $\pright(\bar{x}) \leq \u(\bar{x};C) < \underline{g}(\bar{x})$,  by Lemma \ref{lem:pmonotone}, we have 
    \[
    \CP_u(\pright(\max\{m,\bar x\});(\max\{m,\bar x\},m)) < \CP_u(\underline{g}(\max\{m,\bar x\});(\max\{m,\bar x\},m)) = \rho.
    \]

As $\pright(\bar{x})<\underline{g}(\bar{x})$, we also have $\pright(x)<\bar{g}(\bar{x})$ for any $x \in [\bar{x},m]$. This is because by definition,  for any $x \in [\bar{x},m]$, $\underline g(x)\le \bar g(x)$. Then, by Lemma \ref{lem:pmonotone}, for $x \in [\bar{x},m]$, we have
 \[
\CP_o(\pright(x);(x,0)) \geq \CP_o(\bar{g}(x);(x,0)) = \rho > \CP_u(\pright(m);(m,m)),
\]
where the equality is because  $\bar{g}(x)=-\rho x+m$ and one can easily check $\CP_o(-\rho x +m;(x,0)) = \rho $ for any $x \in [0,m]$.

Therefore, the robust ratio of $\pright(x)$ for $x \in [\bar{x},m]$ is $\RR=\CP_u(\pright(m);(m,m))$, and by Lemma \ref{lem:trapezoid}, we have $\underline{g}(x;\RR) \leq \pright(x) \leq \bar{g}(x;\RR)$. Also, $\pright(\cdot)$ is a constant function and is valid. We obtain $\pright(\cdot)$ is feasible.  

Finally, we show that $\pright(\cdot)$ is optimal among all PL algorithms. We prove by contradiction. Suppose that a valid $p_1(x)$ can achieve a robust ratio greater than $\RR$. Then, we have
\[
\inf_{x \in [\bar{x},m]} \min\{\CP_o(p_1(x); (x, 0)), \CP_u(p_1(x); (x, m))\} > \CP_u(\pright(m);(m,m)), 
\]
which implies that $\CP_u(p_1(m);(m,m)) > \CP_u(\pright(m);(m,m))$. By Lemma \ref{lem:pmonotone}, we have $p_1(m) > \pright(m)$. As $p_1(\cdot)$ is valid, it is non-increasing, we have $p_1(\bar{x}) \geq p_1(m) > \pright(m) = \u(\bar{x};C)$, which is a contradiction because $p_1(\bar{x})>\u(\bar{x};C)$ means it is infeasible.

\item \textit{Case 3: $\wll(\bar{x};C) > \bar{g}(\bar{x})$. } In this case, $\pright(\bar{x})=\argmin_{p \in [\wll(\bar{x};C),\u(\bar{x};C)]} \vert p - \underline{g}(\bar{x}) \vert = \wll(\bar{x};C)$, and we have $\pright(x)=\bar{g}(x; \RR)$ for $x \in [\bar{x},\max\{m,\bar x\}]$ where $\RR=\CP_o(\wll(\bar{x};C);(\bar{x},0))$. Observe that $\pright(\cdot)$ is continuous at $\bar x$ because $\pright(\bar x)=\bar{g}(\bar x; \RR)=\bar{g}(\bar x;\CP_o(\wll(\bar{x};C);(\bar{x},0)))$, and if $\bar x \leq m$, we have $\bar{g}(\bar x;\CP_o(\wll(\bar{x};C);(\bar{x},0)))=-\frac{m-\wll(\bar x;C)}{\bar x}\bar x+m=\wll(\bar x;C)$. Similarly, if $\bar x >m$, we have $\bar{g}(\bar x;\CP_o(\wll(\bar{x};C);(\bar{x},0)))=-\frac{m-\wll(\bar x;C)}{m}m+m=\wll(\bar x;C)$.  
By Lemmas \ref{lem:ymonotone} and \ref{lem:mout}, for $x \in [\bar{x},\max\{m,\bar x\}]$, the worst over- and under-protected points are $(x,0)$ and $(x,m)$, respectively; that is, $\RR=\inf_{x \in [\bar{x},\max\{m,\bar x\}]}\min\{\CP_o(\pright(x);(x,0)),\CP_u(\pright(x);(x,m)) \}$.  
So, we first claim that $\pright(\cdot)$ achieves a robust ratio of $\RR=\CP_o(\wll(\bar{x};C);(\bar{x},0))$ by showing that for any $x \in [\bar{x},\max\{m,\bar x\}]$, $\CP_o(\pright(x);(x,0)) \geq \RR$ and $\CP_u(\pright(x);(x,m)) \geq \RR$. Then, we show its feasibility and optimality. 

By Equation \eqref{eq:CP_over}, we have for any $x \in [\bar{x},\max\{m,\bar x\}]$,
\begin{align*}
\CP_o(\pright(x);(x,0)) &= \frac{0 \cdot r_h+\min\{x,m-\pright(x)\}r_{\ell}}{0 \cdot r_h+\min\{x, m-0\}r_{\ell}} = \frac{\min\{x,m-\pright(x)\}r_{\ell}}{\min\{x, m\}r_{\ell}}. 
\end{align*}
If $x \leq m$, we have 
\begin{align*}
\frac{\min\{x,m-\pright(x)\}r_{\ell}}{\min\{x, m\}r_{\ell}} &= \frac{(m-\pright(x))r_{\ell}}{xr_{\ell}} \\&= \frac{(m-\bar{g}(x; \RR))r_{\ell}}{xr_{\ell}} \\&= \frac{m-\bar{g}(x; \RR)}{x} = \RR,
\end{align*}
where the second equality is because $\pright(x)=\bar{g}(x; \RR)=-\RR x+m$ and $m-(-\RR x+m)=\RR x \leq x$. If $x>m$, we have
\begin{align*}
\frac{\min\{x,m-\pright(x)\}r_{\ell}}{\min\{x, m\}r_{\ell}} &= \frac{(m-\pright(x))r_{\ell}}{mr_{\ell}} \\&= \frac{(m-\bar{g}(x; \RR))r_{\ell}}{mr_{\ell}} \\&= \frac{m-\bar{g}(m; \RR)}{m} = \RR.
\end{align*}

For any $x \in [\bar{x},\max\{m,\bar x\}]$, by Lemma \ref{lem:worstunderrob}, we have
\[
\CP_u(\pright(x);(x,m)) \geq \CP_u(\pright(\max\{m,\bar x\});(\max\{m,\bar x\},m)).
\]
As we have $\CP_u(\underline{g}(\max\{m,\bar x\});(\max\{m,\bar x\},m))=\rho$, and $\pright(\max\{m,\bar x\})=\bar{g}(\max\{m,\bar x\}) \geq \underline{g}(\max\{m,\bar x\})$, by Lemma \ref{lem:pmonotone}, we have 
\[
\CP_u(\pright(\max\{m,\bar x\});(\max\{m,\bar x\},m)) \geq \rho \geq \RR.
\]

Therefore, the robust ratio of $\pright(x)$ is $\RR$. As $\pright(x)=\bar{g}(x;\RR)$, we have $\underline{g}(x;\RR) \leq \pright(x) \leq \bar{g}(x;\RR)$. Also, since $\pright(\bar x)=\tilde{\bar x;C}$, we have $\tilde{\bar x;C} \leq \pright(\bar x) \leq \u(\bar x;C)$. In addition, $\pright(x)$ has slope $-\RR \leq -1$, which means it is valid. We have $\pright(\cdot)$ is a feasible solution.

Finally, we show that $\pright(\cdot)$ is optimal among all PL algorithms. We prove by contradiction. Suppose that a valid $p_1(x)$ can achieve a robust ratio greater than $\RR$. Then, we have
\[
\inf_{x \in [\bar{x},\max\{m,\bar x\}]} \min\{\CP_o(p_1(x); (x, 0)), \CP_u(p_1(x); (x, m))\} > \CP_o(\pright(\bar{x});(\bar{x},0)), 
\]
which implies that $\CP_o(p_1(\bar{x});(\bar{x},0))>\CP_o(\pright(\bar{x});(\bar{x},0))$. By Lemma \ref{lem:pmonotone}, we have $p_1(\bar{x})<\pright(\bar{x})$. However, as we have $\pright(\bar{x})=\wll(\bar{x};C)$, we obtain $p_1(\bar{x})<\wll(\bar{x};C)$, which means $p_1(\cdot)$ is infeasible and forms a contradiction.

\end{itemize}

\Halmos

\subsection{Proof of Theorem \ref{thm:optimalrobCconsistent}}
We first show that $(\pleft(\cdot), \RL)$ is a feasible solution to Problem \eqref{prob:left}, where  $\pleft(x) =\max\{\wll(x;C),\pright(\bar{x})\}\,, x\in [0, \bar x]$ and  $\RL = \min \{\CP_u(\pleft(\bar{x});(\bar{x},m))),\inf_{x \in [0,\bar{x}]}\CP_o(\pleft(x);(x,0)) \}$. 
Under the PL $\pleft(\cdot)$, we first note that by Lemma \ref{lem:ymonotone}, we only need to consider the points $(x, 0)$ and $(x, m)$ for any $x\in [0, \bar x]$. That is, 
\[\RL = \min\left\{\inf_{x \in [0,\bar{x}]}\CP_o(\pleft(x);(x,0)), \inf_{x \in [0,\bar{x}]}\CP_u(\pleft(x);(x,m))\right\}\,,\]
By Lemma \ref{lem:worstunderrob}, we then have 
\[
\CP_u(\pleft(x);(x,m)) \geq \CP_u(\pleft(\bar{x});(\bar{x},m)) \quad \Rightarrow  \quad \inf_{x \in [0,\bar{x}]}\CP_u(\pleft(x);(x,m)) = \CP_u(\pleft(\bar x);(\bar x,m)).
\]
This implies that $\RL = \min \{\CP_u(\pleft(\bar x);(\bar x,m)),\inf_{x \in [0,\bar{x}]}\CP_o(\pleft(x);(x,0)) \}$, as desired.  

$\RL$ is feasible because the range of compatible ratio is $[0,1]$. To show that $\pleft(x)$ is feasible, first, as we have shown it achieves a robust ratio of $\RL$, by Lemma \ref{lem:trapezoid}, we have $\underline{g}(x;\RL) \leq \pleft(x) \leq \bar{g}(x;\RL)$. Second, we show that $\pleft(\cdot)$ is a valid PL function and $\pleft(x)\in [\wll(x;C), \u(x;C)]$. 
Observe that $\wll(x;C)$ is a valid PL function by definition, and hence  we have $\pleft(x)$ is also valid. 
Third, we show that $\wll{x;C} \leq \pleft(x) \leq u(x;C)$ for $x \in [0,\bar x]$. As $\pleft(x) =\max\{\wll(x;C),\pright(\bar{x})\}$, we have $\wll(x;C) \leq \pleft(x)$ for any $x \in [0,\bar{x}]$. In addition, we have $\pleft(\bar{x}) =\pright(\bar x)\leq \u(\bar{x};C)$, and by Lemma \ref{lem:property_u_l}, we have $\u({x};C)$ is a non-increasing function in $x$. This implies that $\u(x;C) \geq \pleft(\bar{x})$ for any $x \in [0,\bar{x}]$. Also, as $C$ is assumed to be less than $\cstar$, by Lemma  \ref{lem:feasibleregionrobust_2}, we have $\wll(x;C) \leq \u(x;C)$ for $x \in [0,\bar{x}]$. Therefore, we have $\pleft(x)=\max\{\pleft(\bar{x}),\wll(x;C) \} \leq \u(x;C)$, which is the desired result.

Second, we show that $\pleft(\cdot)$ is optimal. 
To do so, we  argue that  (i) by Lemma \ref{lem:trapezoid}, we have 
\[\underline g(x; \RL)\le \pleft(x)\le \bar g(x; \RL)\,,\]
and (ii) there does not exist any other valid PL function that achieves a higher robust ratio than $\RL$ while satisfying the consistency lower and upper bounds.

Recall that $\RL = \min \{\CP_u(\pleft(\bar x);(\bar x,m)),\inf_{x \in [0,\bar{x}]}\CP_o(\pleft(x);(x,0)) \}$. As Problem \eqref{prob:left} restricts the value of $p(\bar{x})$, we have any PL algorithm has the same worst under-protected ratio, i.e. $\CP_u(\pleft(\bar x);(\bar x,m))$. Then, we show that any PL algorithm cannot get a larger worst over-protected ratio. For over-protected case, we let
\[
\widetilde \RL=\inf_{x \in [0,\bar{x}]}\CP_o(\pleft(x);(x,0)),
\]
and let the infimum is achieved on $x_1$, i.e. $\CP_o(\pleft(x_1);(x_1,0))=\widetilde  \RL$. If there exists a PL algorithm $p(x)$ such that $\inf_{x \in [0,\bar{x}]}\CP_o(p(x);(x,0))>\RL$, then $\CP_o(p(x_1);(x_1,0))>\CP_o(\pleft(x_1);(x_1,0))$. By Lemma \ref{lem:pmonotone}, we have $p(x_1)<\pleft(x_1)$. This implies that either $p(x_1)<\pright(\bar x)$ or $p(x_1)<\tilde{\ell}(x;C)$. However, if $p(x_1)<\pright(\bar x)$, as $p(\cdot)$ is non-increasing, we have $p(\bar x)<\pright(\bar x)$, which implies $p(\cdot)$ is infeasible. If $p(x_1) < \tilde{\ell}(x;C)$, this immediately contradicts to $\tilde{\ell}(x;C) \leq p(x) \leq u(x;C)$. Therefore, such $p(x)$ does not exist.

\Halmos

\section{Proof of Theorem \ref{thm:optimal_trans}} \label{sec:proofmain1}
The proof is naturally divided into three parts. 

\subsection{Result 1: $\Rstar=\min\{\RR , \RL\}$}
First, we show that $\Rstar=\min\{\RR , \RL\}$, where Theorems \ref{thm:optimalextension} and  \ref{thm:optimalrobCconsistent} show that  \[\RL = \min \{\CP_u(\pleft(\bar x);(\bar x,m)),\inf_{x \in [0,\bar{x}]}\CP_o(\pleft(x);(x,0)) \}\,.\]
and 
\[
    \RR=\min\left\{\CP_o(\pright(\bar{x});(\bar{x},0)),\CP_u(\pright(\bar{x});(\max\{m,\bar x\},m))\right\}\,.
\]
Let us denote $\widetilde \RL = \inf_{x \in [0,\bar{x}]}\CP_o(\pleft(x);(x,0)) $, and note that  by Lemma \ref{lem:worstunderrob}, we have $\CP_u(\pright(\bar x);(\bar x,m)) \geq \CP_u(\pright(\bar x);(\max\{m,\bar x\},m))$, where by construction, we have $\CP_u(\pright(\bar x);(\bar x,m))= \CP_u(\pleft(\bar x);(\bar x,m))$. Therefore, we have
\begin{align*}
&\min\{\RR , \RL\} \\&= \min\Big\{\CP_o(\pright(\bar{x});(\bar{x},0)),\CP_u(\pright(\bar{x});(\max\{m,\bar x\},m)), \CP_u(\pleft(\bar x);(\bar x,m)),\inf_{x \in [0,\bar{x}]}\CP_o(\pleft(x);(x,0)) \Big\} \\&= \min\Big\{\CP_o(\pright(\bar{x});(\bar{x},0)),\CP_u(\pright(\bar{x});(\max\{m,\bar x\},m)), \CP_u(\pright(\bar x);(\bar x,m)),\inf_{x \in [0,\bar{x}]}\CP_o(\pleft(x);(x,0)) \Big\} \\&= \min\Big\{\CP_o(\pright(\bar{x});(\bar{x},0)),\CP_u(\pright(\bar{x});(\max\{m,\bar x\},m)), \inf_{x \in [0,\bar{x}]}\CP_o(\pleft(x);(x,0)) \Big\} = \min\{\RR, \widetilde \RL\},
\end{align*}
where the second equality is because $\pright(\bar{x})=\pleft(\bar{x})$, and the third equality is because $\CP_u(\pright(\bar x);(\bar x,m)) \geq \CP_u(\pright(\bar x);(\max\{m,\bar x\},m))$.

Therefore, showing $\Rstar=\min\{\RR,\RL\}$ is equivalent to show that $\Rstar = \min\{\RR ,\widetilde \RL\}$.
As Theorems \ref{thm:optimalextension} and \ref{thm:optimalrobCconsistent} show, if we set $p(\bar{x})$ optimally, our $\pc(\cdot)$, presented in Algorithm \ref{alg:trans}, achieves an optimal robust ratio. Therefore, it suffices to show that for any valid $\widehat p(x)$ such that $\widehat p(\bar{x}) \neq \pright(\bar{x})$, we have $\comp(p(\cdot)) \leq \comp(\pc(\cdot))$. Here, with a slight abuse of notation, $\comp(p(\cdot))$ is the robust ratio of a PLA with PL of $p(\cdot)$. 
We split the analysis into two cases based on   the value of $\widetilde \RL$ and $\RR$, where in case 1, we have $\RR \leq \widetilde \RL$, and in case 2, we have $\RR > \widetilde \RL$.

\begin{itemize}
    \item $\RR \leq \widetilde \RL$. In this case, $\comp(\pc(\cdot)) = \min\{\widetilde \RL, \RR\} =\RR$. By Theorem \ref{thm:optimalextension}, we know that no PL can achieve a robust ratio greater than $\RR$ for $x \in [\bar{x},\max\{m,\bar x\}]$. Therefore, we have $\comp(\widehat p(\cdot)) \leq \comp(\pc(\cdot)) = \RR$, which is the desired result.

    \item $\widetilde \RL < \RR$. As is shown in Theorem \ref{thm:optimalrobCconsistent}, by fixing $\pc(\bar{x})=\pright(\bar{x})$, no PL can achieve a robust ratio greater than $\widetilde \RL$. Then, in this part, we show that if a valid and feasible PL function $\hat{p}(x)$ for $x \in [0,\bar{x}]$ does not have restriction on $\bar{x}$, it can still not achieve a robust ratio greater than $\widetilde \RL$. To show this, we define $\widehat{x} \in [0,\bar{x}]$ as such that $\wll(\widehat{x};C)=\pright(\bar{x})$. We start with showing that such $\widehat{x}$ always exists, and then we show that $\pleft(\cdot)$ achieves $\widetilde \RL$ in $[0,\widehat{x}]$ and no PL function $\hat{p}(x)$ can outperform $\widetilde \RL$ in $[0,\widehat{x}]$. 

    To show the existence, we use a contradiction argument. Contrary to our claim, suppose that there does not exist any $\widehat{x}\in [0, \bar x]$ such that $\wll(\widehat{x};C)=\pright(\bar{x})$.  
Observe that $\pright(\bar{x})$ is a constant and is greater than $\wll(\bar{x};C)$ by feasibility of $\pright(\cdot)$. Further, note that  by Lemma \ref{lem:propertywll}, $\wll(x;C)$ is non-increasing in $x$. Then, when $\widehat x$ does not exist, we must have $\pright(\bar{x}) > \wll(x;C)$ for any $x \in [0,\bar{x}]$. By Theorem \ref{thm:optimalrobCconsistent}, in this case, $\pleft(x)=\pright(\bar{x})$ is a constant function.  (Recall that $\pleft(x)=\max\{\wll(x;C),\pright(\bar{x})\}$ for any $x\in [0, \bar x]$.) By Lemma \ref{lem:worstkeep}, we have
    \[
    \widetilde \RL = \inf_{x \in [0,\bar{x}]}\CP_o(\pright(\bar{x});(x,0)) = \CP_o(\pright(\bar{x});(\bar{x},0)).
    \]
    However, as $\CP_o(\pright(\bar{x});(\bar{x},0)) \geq \min\{\CP_o(\pright(\bar{x});(\bar{x},0)),\CP_u(\pright(\bar{x});(\max\{m,\bar x\},m))\}$, we have $\widetilde \RL \geq \RR$, which is a contradiction. Therefore, such $\widehat{x}$ exists.
    
    Then, $\pleft(x)=\pright(\bar{x})$ for $x \in [\widehat{x},\bar{x}]$ and $\pleft(x)=\wll(x;C)$ for $x \in [0,\widehat{x}]$. By definition of $\widetilde \RL$, we have
    \[
    \widetilde \RL = \min\{\inf_{x \in [0,\widehat{x}]}\CP_o(\wll(x;C);(x,0)),\inf_{x \in [\widehat{x},\bar{x}]}\CP_o(\pright(\bar{x});(x,0))  \}.
    \]
    By Lemma \ref{lem:worstkeep}, we have
    \[
    \inf_{x \in [\widehat{x},\bar{x}]}\CP_o(\pright(\bar{x});(x,0)) = \CP_o(\pright(\bar{x});(\bar{x},0)) \geq \RR.
    \]
    Given that $\widetilde \RL < \RR$, we have
    \[
    \widetilde \RL = \inf_{x \in [0,\widehat{x}]}\CP_o(\wll(x;C);(x,0)).
    \]

    Finally, we show that no valid and feasible PL $\hat{p}(x)$ can outperform $\widetilde \RL$ in $[0,\widehat{x}]$. Suppose that $x_1=\argmin_{x \in [0,\widehat{x}]}\CP_o(\wll(x;C);(x,0))$, which means that $\CP_o(\wll(x_1;C);(x_1,0))=\widetilde \RL$. If $\hat{p}(x)$ outperforms $\widetilde \RL$, we have
    \[
    \inf_{x \in [0,\widehat{x}]}\CP_o(\hat{p}(x);(x,0)) > \widetilde \RL,
    \]
    which implies that $\CP_o(\hat{p}(x_1);(x_1,0))>\widetilde \RL$. By Lemma \ref{lem:pmonotone}, we have $\hat{p}(x_1)<\wll(x_1;C)$, which shows that $\hat{p}$ is not a feasible PL function and forms a contradiction.

\end{itemize}

\Halmos  
\subsection{Result 2: $\pc$ is an Optimal Solution}
Second, it is trivial that Algorithm \ref{alg:trans} presents an optimal solution to Problem \eqref{prob:trans}. The reason is by Theorem \ref{thm:optimalextension}, we have $\pright(x)$ achieves $\RR$ for Problem \eqref{prob:right} and by Theorem \ref{thm:optimalrobCconsistent}, we have $\pleft(x)$ achieves $\RL$ for Problem \eqref{prob:left}. By Equation \eqref{eq:opt_pc}, we know that $\comp(\pc(\cdot))=\min\{\RR,\RL\}$, and we just showed that $\min\{\RR,\RL\}=\Rstar$. The remaining thing is to show that $\pc(\cdot)$ is feasible and valid. As we have shown in Theorems \ref{thm:optimalextension} and \ref{thm:optimalrobCconsistent} that $\pright(\cdot)$   and $\pleft(\cdot)$ are both feasible and valid for any $x\ge \bar x$ and $x\in[0, \bar x]$, respectively.  Further, $\pright(\bar{x})=\pleft(\bar{x})$, which means that $\pc(\cdot)$ is continuous. Therefore, $\pc(\cdot)$ is feasible. 

\subsection{Result 3: No Algorithm Can Outperform $\pc$} \label{subsec:noalgp}

Here, we show that a PLA with  the PL function of $\pc(\cdot)$ is an optimal solution to Problem \eqref{prob:original}. That is, among any online algorithms $\Pi$, the aforementioned algorithm maximizes the robust ratio while ensuring its consistent ratio is at least $C$.  Recall that we just showed  
\begin{equation} \label{eq:conditionmin}
\Rstar=\min\{\RR,\RL\} = \min\{\CP_u(\pc(\max\{m,\bar x\});(\max\{m,\bar x\},m)),\CP_o(\pc(\bar{x}),(\bar{x},0)), \inf_{x \in [0,\bar{x}]}\CP_o(\pc(x);(x,0)) \}.
\end{equation}
Then, we split the proof into three parts, where in each parts, we discuss each term in Equation \eqref{eq:conditionmin} is the minimum value.

\textit{Part 1: $\Rstar= \CP_u(\pc(\max\{m,\bar x\});(\max\{m,\bar x\},m))$.} Here, we show that no deterministic or randomized algorithm can achieve a robust ratio more than $\CP_u(\pc(\max\{m,\bar x\});(\max\{m,\bar x\},m))$. 
We define two (ordered) input sequences: In the first input  sequence, $I_1$, 
 $\xmin \le \bar x$ low-reward requests arrive first, followed with $\underline{h}(\xmin)$ high-reward requests. In the second input  sequence, $I_2$, 
 $\max\{m,\bar x\}$ low-reward requests arrive first,  followed  $m$ high-reward requests. Before receiving $\xmin$ low reward requests, any deterministic or randomized algorithm cannot differentiate the two input sequences and has to decide to accept how many low-reward requests in expectation. If there exists a deterministic or randomized algorithm $\mathcal{A}$, which can achieve a consistent ratio of at least $C$, and a robust ratio higher than $\CP_u(\pc(\max\{m,\bar x\});(\max\{m,\bar x\},m))$, it should satisfy
 \[
 \frac{\mathbf{E}[\text{Rew}(\mathcal{A},I_1)]}{\textsc{opt}(I_1)} \geq C,
 \]
 Let the expected total amount of high-reward, low-reward requests being accepted by $\mathcal{A}$ be $h(\mathcal{A},I_1)$, $\ell(\mathcal{A},I_1)$, respectively. Then, 
  \[
  \frac{\mathbf{E}[\text{Rew}(\mathcal{A},I_1)]}{\textsc{opt}(I_1)} = \frac{h(\mathcal{A},I_1)r_h+\ell(\mathcal{A},I_1)r_{\ell}}{\underline{h}(\xmin)r_h+\min\{\xmin,m-\underline{h}(\xmin)\}r_{\ell}} \geq C.
  \]
  By the definition of $\u(\bar x_u;C)$, we have $\u(\bar x_u;C) = \sup\{p: \CP_o(p;(\bar x_u,\underline{h}(\xmin) ))=C \}$, and by Equation \eqref{eq:CP_over}, $\frac{\min\{\underline{h}(\xmin),m\}r_h+\min\{\xmin,m-\u(\bar x_u;C)\}r_{\ell}}{\min\{\underline{h}(\xmin),m\}r_h+\min\{\xmin,(m-\underline{h}(\xmin))^{+}\}r_{\ell}} = C$, and by Lemma \ref{lem:specialpoints}, we have $\frac{\min\{\underline{h}(\xmin),m\}r_h+(m-\u(\bar x_u;C))r_{\ell}}{\min\{\underline{h}(\xmin),m\}r_h+\min\{\xmin,(m-\underline{h}(\xmin))^{+}\}r_{\ell}} = C$. Therefore, we have
  \[
  \frac{h(\mathcal{A},I_1)r_h+\ell(\mathcal{A},I_1)r_{\ell}}{\min\{\underline{h}(\xmin),m\}r_h+\min\{\xmin,(m-\underline{h}(\xmin))^{+}\}r_{\ell}} \geq \frac{\min\{\underline{h}(\xmin),m\}r_h+(m-\u(\bar x_u;C))r_{\ell}}{\min\{\underline{h}(\xmin),m\}r_h+\min\{\xmin,(m-\underline{h}(\xmin))^{+}\}r_{\ell}},
  \]
  as $h(\mathcal{A},I_1) \leq \min\{\underline{h}(\xmin),m\}$, we have $\ell(\mathcal{A},I_1) \geq m-\u(\xmin;C)$, which implies that we should accept at least $m-\u(\xmin;C)$ low-reward requests in expectation.

 However, to achieve a robust ratio higher than $\CP_u(\pc(\max\{m,\bar x\});(\max\{m,\bar x\},m))$ for sequence $I_2$, 
 we have 
 \[
 \frac{\mathbf{E}[\text{Rew}(\mathcal{A},I_2)]}{\textsc{opt}(I_2)} > \CP_u(\pc(\max\{m,\bar x\});(\max\{m,\bar x\},m)).
 \]
 By Equation \eqref{eq:CP_under}, we have 
 \begin{align*}
 \frac{\mathbf{E}[\text{Rew}(\mathcal{A},I_2)]}{\textsc{opt}(I_2)}&=\frac{h(\mathcal{A},I_2)r_h+\ell(\mathcal{A},I_2)r_{\ell}}{mr_h}>\CP_u(\pc(\max\{m,\bar x\});(\max\{m,\bar x\},m))\\&=\frac{\pc(\max\{m,\bar x\})r_h+(m-\pc(\max\{m,\bar x\}))r_{\ell}}{mr_h},
 \end{align*}
 as $h(\mathcal{A},I_2) \leq (m-\ell(\mathcal{A}),I_2)$, we have $\ell(\mathcal{A},I_2)<m-\pc(\max\{m,\bar x\})$, which implies that we should accept less than $m-\pc(\max\{m,\bar x\})$ low-reward requests in expectation. However, when $\Rstar=\CP_u(\pc(\max\{m,\bar x\});(\max\{m,\bar x\},m))$, Part 1 of the proof of Theorem \ref{thm:optimalextension} shows that $\pc(\max\{m,\bar x\})=\pright(\max\{m,\bar x\})=\u(\bar{x};C)$. By Equation \eqref{eq:u}, we have $\u(\bar{x};C)=\u(\xmin;C)$. Therefore, under sequence $I_2$, we have $\ell(\mathcal{A},I_2)<m-\u(\xmin;C)$, which is a contradiction because once we observed $\xmin$ low-reward requests, we have already accept at least $m-\u(\xmin;C)$ of them in expectation.

\Halmos
\endproof

\textit{Part 2: $\Rstar = \CP_o(\pc(\bar{x}),(\bar{x},0))$.} To show the result, we split the analysis into two sub parts based on the value of  $\pc(\bar{x})$ and $\wll(\nex;C)$. First, we show that  when $\ll(\nex;C)>\pc(\bar{x})$, no deterministic or randomized algorithm can achieve a robust ratio more than $\CP_o(\pc(\bar{x}),(\bar{x},0))$. Here, we recall that $\nex=\sup\{x \in [x_H,\bar{x}]: \frac{\partial \ll(x^{-};C)}{\partial x} \le -1 \}$. Then, we show the same result when $\ll(\nex;C)\le \pc(\bar{x})$.

For the case where $\ll(\nex;C)>\pc(\bar{x})$,   we 
 again define two (ordered) input sequences: In the first input  sequence, $I_1$, 
 $\nex$ low-reward requests arrive first, followed with $\bar{h}(\nex)$ high-reward requests. In the second input  sequence, $I_2$,  $\bar{x}$ low-reward requests arrive first,  followed  by $0$ high-reward requests. Before receiving $\nex$ low reward requests, any deterministic or randomized algorithm cannot differentiate the two input sequences and has to decide to accept how many low-reward requests in expectation. If there exists a deterministic or randomized algorithm $\mathcal{A}$, which can achieve a consistent ratio of at least $C$, and a robust ratio more than $\CP_o(\pc(\bar{x}),(\bar{x},0))$, it should satisfy
 \[
 \frac{\mathbf{E}[\text{Rew}(\mathcal{A},I_1)]}{\textsc{opt}(I_1)} \geq C,
 \]
Let the expected total amount of high-reward, low-reward requests being accepted by $\mathcal{A}$ be $h(\mathcal{A},I_1)$, $\ell(\mathcal{A},I_1)$, respectively. Then, 
  \[
  \frac{\mathbf{E}[\text{Rew}(\mathcal{A},I_1)]}{\textsc{opt}(I_1)} = \frac{h(\mathcal{A},I_1)r_h+\ell(\mathcal{A},I_1)r_{\ell}}{\bar{h}(\nex)r_h+\min\{\nex,m-\bar{h}(\nex)\}r_{\ell}} \geq C.
  \]
  By Equation \eqref{eq:l}, we have $\CP_u(\ll(\nex;C);(\nex,\bar{h}(\nex))=C$. By Equation \eqref{eq:CP_under}, we have
  \[
  \CP_u(\ll(\nex;C);(\nex,\bar{h}(\nex)) = \frac{\max\{\ll(\nex;C), \min\{\bar{h}(\nex),m-\nex\} \} r_h+\min\{\nex,m-\ll(\nex;C)\}r_{\ell}}{\bar{h}(\nex)r_h+\min\{\xmin,m-\bar{h}(\nex)\}r_{\ell}}.
  \]
  By Lemma \ref{lem:specialpoints}, we have $\max\{\ll(\nex;C), \min\{\bar{h}(\nex),m-\nex\} \} = \ll(\nex;C)$ and $\min\{\nex,m-\ll(\nex;C)\}=m-\ll(\nex;C)$. Therefore, we have
  \[
  \CP_u(\ll(\nex;C);(\nex,\bar{h}(\nex)) = \frac{\ll(\nex;C) r_h+(m-\ll(\nex;C))r_{\ell}}{\bar{h}(\nex)r_h+\min\{\xmin,m-\bar{h}(\nex)\}r_{\ell}}=C,
  \]
  which implies that
  \[
  \frac{h(\mathcal{A},I_1)r_h+\ell(\mathcal{A},I_1)r_{\ell}}{\bar{h}(\nex)r_h+\min\{\xmin,m-\bar{h}(\nex)\}r_{\ell}} \geq \frac{\ll(\nex;C) r_h+(m-\ll(\nex;C))r_{\ell}}{\bar{h}(\nex)r_h+\min\{\xmin,m-\bar{h}(\nex)\}r_{\ell}}.
  \]
  As $h(\mathcal{A},I_1) \leq m-\ell(\mathcal{A},I_1)$, we have $\ell(\mathcal{A},I_1) \leq m-\ll(\nex;C)$, which implies that it should accept no more than $m-\ll(\nex;C)$ low-reward requests in expectation. 
  However, to achieve a robust ratio more than $\CP_o(\pc(\bar{x}),(\bar{x},0))$ for sequence $I_2$, 
 we have 
 \[
 \frac{\mathbf{E}[\text{Rew}(\mathcal{A},I_2)]}{\textsc{opt}(I_2)} > \CP_o(\pc(\bar{x}),(\bar{x},0)).
 \]
 Let the expected total amount of low-reward requests being accepted under $I_2$ by $\mathcal{A}$ be $\ell(\mathcal{A},I_2)$, respectively. Then, 
  \[
  \frac{\mathbf{E}[\text{Rew}(\mathcal{A},I_1)]}{\textsc{opt}(I_1)} = \frac{\ell(\mathcal{A})}{\bar{x}} > \CP_o(\pc(\bar{x}),(\bar{x},0)).
  \]
  By Equation \eqref{eq:CP_over}, we have 
  \[
  \CP_o(\pc(\bar{x}),(\bar{x},0)) = \frac{\min\{\bar{x},m-\pc(\bar{x}) \}}{\bar{x}}.
  \]
  If $\bar{x} \leq m-\pc(\bar{x})$, we have $\CP_o(\pc(\bar{x}),(\bar{x},0)) =1 $, which means no algorithm can achieve a higher robust ratio. Otherwise, we have
  \[
  \CP_o(\pc(\bar{x}),(\bar{x},0)) = \frac{m-\pc(\bar{x})}{\bar{x}},
  \]
  and we have
  \[
   \frac{\mathbf{E}[\text{Rew}(\mathcal{A},I_1)]}{\textsc{opt}(I_1)} = \frac{\ell(\mathcal{A},I_2)}{\bar{x}} > \frac{m-\pc(\bar{x})}{\bar{x}},
  \]
  which implies that $\ell(\mathcal{A})>m-\pc(\bar{x})$, and
  we should accept more than $m-\pc(\bar{x})$ low-reward requests in expectation. However, when $\Rstar=\CP_o(\pc(\bar{x}),(\bar{x},0))$, by the proof of  Theorem \ref{thm:optimalextension}, we have in this case $\pc(\bar{x})=\pright(\bar{x})=\wll(\bar{x};C)$.  
  However, given that $\ll(\nex;C)>\pc(\bar{x})=\wll(\bar{x};C)$, we have $\nex<\bar{x}$. Otherwise, $\ll(\nex;C)=\ll(\bar{x};C)=\wll(\bar{x};C)$. Between $\nex$ and $\bar{x}$, there are at most $\bar{x}-\nex$ low-reward requests arriving, and any algorithm can accept at most $\bar{x}-\nex$ low-reward requests. But we accept no more than $m-\ll(\nex;C)$ low-reward requests in expectation under $I_1$ and more than $m-\pc(\bar{x})$ low-reward requests in expectation under $I_2$, and since $\wll(x;C)$ is a line with slope $-1$ for $x \geq \nex$, we have  $m-\pc(\bar{x})-(m-\ll(\nex;C))=\wll(\nex;C)-\wll(\bar{x};C) =\bar{x}-\nex$, which is a contradiction.

 Now, let us consider the case where $\ll(\nex;C)\le \pc(\bar{x})$, where we recall that here by assumption $\Rstar = \CP_o(\pc(\bar{x}),(\bar{x},0))$. 
By the proof of Theorem \ref{thm:optimalextension}, when $\RR=\CP_o(\pc(\bar{x}),(\bar{x},0))$, we have $\pc(\bar{x})=\pright(\bar{x})=\wll(\bar{x};C)$. In addition, we have $\ll(\nex;C)=\wll(\nex;C)$, and $\ll(\nex;C)\le \pc(\bar{x})$ is equivalent to $\wll(\nex;C)\le \wll(\bar{x};C)$. Due to $\wll(\cdot;C)$ is non-increasing, we have $\nex=\bar{x}$ in this case.

To show the result, we again  define two (ordered) input sequences: In the first input  sequence, $I_1$, 
 $\bar{x}$ low-reward requests arrive first, followed by $\bar{h}(\bar{x})$ high-reward requests. In the second input  sequence, $I_2$, 
 $\bar{x}$ low-reward requests arrive first,  followed  by  $0$ high-reward requests. If there exists a deterministic or randomized algorithm $\mathcal{A}$, which can achieve a consistent ratio of at least $C$, and a robust ratio more than $\CP_o(\pc(\bar{x}),(\bar{x},0))$, it should satisfy
 \[
 \frac{\mathbf{E}[\text{Rew}(\mathcal{A},I_1)]}{\textsc{opt}(I_1)} \geq C.
 \]
 Let the expected total amount of high-reward, low-reward requests being accepted by $\mathcal{A}$ be $h(\mathcal{A},I_1)$, $\ell(\mathcal{A},I_1)$, respectively. Then, 
  \[
  \frac{\mathbf{E}[\text{Rew}(\mathcal{A},I_1)]}{\textsc{opt}(I_1)} = \frac{h(\mathcal{A},I_1)r_h+\ell(\mathcal{A},I_1)r_{\ell}}{\bar{h}(\bar{x})r_h+\min\{\bar{x},m-\bar{h}(\bar{x})\}r_{\ell}} \geq C.
  \]
  By Equation \eqref{eq:ll}, we have $\CP_u(\ll(\bar{x};C);(\bar{x},\bar{h}(\bar{x})))=C$. By Equation \eqref{eq:CP_under},
  \[
  \CP_u(\ll(\bar{x};C);(\bar{x},\bar{h}(\bar{x})))=\frac{\max\{\ll(\bar{x};C),\min\{\bar{h}(\bar{x}),m-\bar{x} \}r_h+\min\{\bar{x},m- \ll(\bar{x};C)\}r_{\ell}}{\bar{h}(\bar{x})r_h+\min\{\bar{x},m-\bar{h}(\bar{x})\}r_{\ell}}.
  \]
  As we have shown $\bar{x}=\nex$ at the beginning of this case, by Lemma \ref{lem:specialpoints}, we have 
  
  $\max\{\ll(\bar{x};C),\min\{\bar{h}(\bar{x}),m-\bar{x} \}=\ll(\bar{x};C)$ and $\min\{\bar{x},m- \ll(\bar{x};C)\}=m- \ll(\bar{x};C)$. Therefore, we have
  \[
  \frac{\ll(\bar{x};C)r_h+(m- \ll(\bar{x};C))r_{\ell}}{\bar{h}(\bar{x})r_h+\min\{\bar{x},m-\bar{h}(\bar{x})\}r_{\ell}}=C,
  \]
  and this implies
  \[
  \frac{h(\mathcal{A},I_1)r_h+\ell(\mathcal{A},I_1)r_{\ell}}{\bar{h}(\bar{x})r_h+\min\{\bar{x},m-\bar{h}(\bar{x})\}r_{\ell}} \geq \frac{\ll(\bar{x};C)r_h+(m- \ll(\bar{x};C))r_{\ell}}{\bar{h}(\bar{x})r_h+\min\{\bar{x},m-\bar{h}(\bar{x})\}r_{\ell}}.
  \]
  As $h(\mathcal{A},I_1)=m-\ell(\mathcal{A},I_1)$, we have $\ell(\mathcal{A},I_1) \leq m-\ll(\bar{x};C)$, which implies that we should accept no more than $m-\ll(\bar{x};C)$ low-reward requests in expectation. However, to achieve a robust ratio more than the $\CP_o(\pc(\bar{x}),(\bar{x},0))$ for sequence $I_2$, 
 we have 
 \[
 \frac{\mathbf{E}[\text{Rew}(\mathcal{A},I_2)]}{\textsc{opt}(I_2)} > \CP_o(\pc(\bar{x}),(\bar{x},0)).
 \]
This implies that we should accept more than $m-\pc(\bar{x})$ low-reward requests in expectation, which is shown in the previous case. However, as is shown at the beginning of this case, here, $\bar{x}=\nex$, and $\pc(\bar{x})=\ll(\bar{x};C)$, which shows that upon receiving $\bar{x}$ low-reward requests, we should accept no more than $m-\pc(\bar{x})$ low-reward requests and more than $m-\pc(\bar{x})$ low-reward requests, which is obviously a contradiction.

\textit{Part 3: $\Rstar= \inf_{x \in [0,\bar{x}]}\CP_o(\pc(x);(x,0))$.} Define $\widehat{x} = \text{argmin}_{x \in [0,\bar{x}]}\CP_o(\pc(x);(x,0))$. If $\widehat{x}$ is not unique, we randomly  the one with smallest $x$ value. 
Let us first consider the case where $\ll(\nex;C)>\pc(\bar{x})$ and $\widehat{x} \in [\nex, \nex+\ll(\nex;C)-\pc(\bar{x})]$.
  In this case, we replace $I_2$ in the proof of Part 2 under the case where $\ll(\nex;C)>\pc(\bar{x})$ by: we define $I_2$ as 
 $\widehat{x}$ low-reward requests arrive first,  followed  $0$ high-reward requests, and we can have $\pc(x)$ is optimal among all algorithms in this case. 

For any other cases, we replace $I_1$, $I_2$ in the proof of Part 2 under the case where $\ll(\nex;C) \leq \pc(\bar{x})$ by: we define $I_1$ as 
 $\widehat{x}$ low-reward requests arrive first,  followed  $\underline{h}(\widehat{x})$ high-reward requests; we define $I_2$ as 
 $\widehat{x}$ low-reward requests arrive first,  followed  $0$ high-reward requests, and we can have $\pc(x)$ is optimal among all algorithms in this case.

Therefore, we have $\pc(x)$ for $x \in [0,m]$ is optimal among all deterministic and randomized algorithms.

\Halmos

\subsection{Lemma \ref{lem:worstunderrob} and its Proof}
\begin{lemma}\label{lem:worstunderrob}
    Given an arbitrary valid PL function $p(\cdot)$, we have for any $x \in [0,\max\{m,\bar{x}\}]$, $\CP_u(p(x);(x,m))$ is a non-incresing function in $x$. That is, for any $x \in [0,\max\{m,\bar{x}\}]$, with $p(x) \leq m$, 
    \[
    \CP_u(p(m);(\max\{m,\bar{x}\},m)) \leq \CP_u(p(x);(x,m)).
    \]
\end{lemma}

\proof{Proof of Lemma \ref{lem:worstunderrob}}
Take any valid PL function $p(x)$ for $x \in [0,m]$, by Equation \eqref{eq:CP_under},
\begin{align*}
\CP_u(p(x);(x,m)) &= \frac{\max\{p(x), \min\{m,m-x\} \} r_h+\min\{x,m-p(x)\}r_{\ell}}{mr_h+\min\{x, m-m\}r_{\ell}} \\&= \frac{\max\{p(x), m-x \} r_h+\min\{x,m-p(x)\}r_{\ell}}{mr_h}.
\end{align*}

Let $x_1 = \sup\{x: p(x)+x \leq m \}$, first, as a valid PL function satisfies that $p'(x) \geq -1$ a.e., and the line $x+y=m$ has slope $-1$, once $x+p(x) \geq m$, we have for all $x'>x$, we have $x'+p(x') \geq m$. Therefore, for $x \in [0,x_1]$, we have $p(x)+x \leq m$, and 
\[
\CP_u(p(x);(x,m)) = \frac{(m-x) r_h+xr_{\ell}}{mr_h},
\]
which is a monotone decreasing function with $x$. 

For $x \in [x_1,m]$, 
\[
\CP_u(p(x);(x,m)) = \frac{p(x)r_h+(m-p(x))r_{\ell}}{mr_h},
\]
which is also a non-increasing function with $x$ due to $p(x)$ is non-increasing. Therefore, it is monotone decreasing, and we have for any valid PL function $p(x)$, 
\[
\CP_u(p(m);(m,m)) \leq \CP_u(p(x);(x,m)),
\]
for $x \in [0,m]$. If $\bar{x}>m$, by Lemma \ref{lem:mout}, we have $\CP_u(p(\bar x);(\bar x,m)) = \CP_u(p(x);(x,m))$. Therefore, we have 
\[
    \CP_u(p(m);(\max\{m,\bar{x}\},m)) \leq \CP_u(p(x);(x,m)),
\]
for any $x \in [0,\max\{m,\bar{x}\}]$.
\Halmos
\endproof

\section{Proof of Statements in Section \ref{sec:mlconsistent}} \label{appendix:mlcon}

In this section, we provide the proof of statements in Section \ref{sec:mlconsistent}.In Section \ref{subsec:propbisection}, we prove Proposition \ref{prop:bisection} and shows the computational complexity and accuracy of the bisection algorithm. In Section \ref{subsec:propertycstar}, we prove Theorem \ref{thm:property_cstar}, which provides several properties of an optimal $\cstar$. In Section \ref{subsec:mlconsis}, we prove Theorem \ref{thm:MLconsis}, which states that Algorithm \ref{alg:mlconsis} returns $\cstar$. Finally, in Section \ref{subsec:cmax}, we prove Theorem \ref{thm:cmax}, which shows that the optimal PL function is optimal among all deterministic and randomized algorithms.

\subsection{Proof of Proposition \ref{prop:bisection}} \label{subsec:propbisection}

We first show that the feasibility (i.e., determining if for any  $x \in [\underline{x},\bar{x}]$, and a given $C$, we have $\u(x;C) \geq \wll(x;C)$) check can be performed by a polynomial time algorithm. For any $C\in [0,1]$, checking whether for any $x \in [\underline{x},\bar{x}]$, $\u(x;C) \geq \wll(x;C)$ is equivalent to checking  the following condition 
\begin{equation} \label{eq:optpoly}
\min_{x \in [\underline x,\bar x]} \u(x;C)-\wll(x;C) \ge 0.
\end{equation}
Notice that by definition,  for $x \in [\underline{x},\underline x_u]$, $\u(x;C)=m$ and $\u(x;C)-\wll(x;C)$ is always non-negative and for $x \in [\tilde x_{\ell}, \bar x]$, $\wll(x;C)=0$ and $\u(x;C)-\wll(x;C)$ is always non-negative, where $\tilde x_{\ell}=\inf\{\nex<x<\bar{x}: \wll(x;C)=0 \}$. Therefore, in Equation \eqref{eq:optpoly}, we can ignore these two intervals. For any $x\in [\underline x_u, \tilde x_{\ell}]$, by Lemmas \ref{lem:property_u_l} and \ref{lem:propertywll}, $\u(\cdot;C)$ is convex and $\wll(\cdot;C)$ is concave, which implies that $\u(x;C)-\wll(x;C)$ is convex. Therefore, Problem \eqref{eq:optpoly} is a convex optimization problem on a compact set, which can be solved by polynomial-time algorithms.

With this result, Algorithm \ref{alg:bisection} has exactly the same structure with the classical bisection method  which is  a root-finding method that applies to any continuous function. Here, we make an analogy to the classical bisection method for finding the root (a single zero point). Compare to the classical bisection method for finding a single zero point, we can treat $\cstar$ as the zero point, and treat the feasibility check as whether the value is positive or negative. It is well known that the classical bisection method can return a $x \in [x_0-\epsilon, x_0+\epsilon]$ in $O(\log(1/\epsilon))$ time, where $x_0$ is the zero point. Therefore, we have Algorithm \ref{alg:bisection} can return a $C_0 \in [\cstar-\epsilon,\cstar+\epsilon]$ in $O(\log(1/\epsilon))$ time.

\subsection{Proof of Theorem \ref{thm:property_cstar}} \label{subsec:propertycstar}

We split the proof into three parts. In part 1, we show that if for any $x \in \V$, we have $\wll(x;C) \leq \u(x;C)$, then such a $C$ is feasible. That is, $\cstar \le C$.  In part 2, we show that a feasible $C =\cstar$ is optimal if and only if there exists $\widehat{x}\in [\underline{x},\bar{x}]$, such that $\wll(\widehat{x};\cstar) = \u(\widehat{x};\cstar)$. In part 3, we show that such $\widehat{x}$ must belong to $\V$.

\textbf{Part 1: Feasibility of $C$.}  By Lemma \ref{lem:feasibleregionrobust_2}, $C$ is feasible if and only if $\u(x;C) \geq \wll(x;C)$ for any $x \in [\underline{x},\bar{x}]$. Here, we show that by only checking $\u(x;C)$ and $\wll(x;C)$ on $x \in \V$ is enough to check the feasibility of $C$ given that $\region$ is a polyhedron. As $\region$ is a polyhedron, we have both $\bar{h}(\cdot)$ and $\underline{h}(\cdot)$ are piece-wise linear functions. By the first two properties of Lemma \ref{lem:property_u_l} and Equation \eqref{eq:ll}, we have both $\wll(\cdot;C)$ and $\u(\cdot;C)$ are also piece-wise linear functions. We want to show that  $\wll(x;C) \leq \u(x;C)$  for any $x \in \V$, implies that  $\wll(x;C) \leq \u(x;C)$ for  any $x \in [\underline{x},\bar{x}]$

Suppose that for all $x \in \V$, $\wll(x;C) \leq \u(x;C)$. We take any $x_1 \in [\underline{x},\bar{x}]-\V$. Let $[\underline{x}_1,\bar{x}_1]$ be the smallest interval contains $x_1$ such that $\underline{x}_1,\bar{x}_1 \in \V$. As it is the smallest interval, it does not contain any other $x$-vertices, which implies that for any $x \in [\underline{x}_1,\bar{x}_1]$, both $\wll(x;C)$ and $\u(x;C)$ are linear. As $\wll(\underline{x}_1;C) \leq \u(\underline{x}_1;C)$ and $\wll(\bar{x}_1;C) \leq \u(\bar{x}_1;C)$, we have $\wll(x_1;C) \leq \u(x_1;C)$ due to linearity and continuity of $\wll(x;C)$ and $\u(x;C)$.

\textbf{Part 2: Optimality of $C$.} We first prove the `if' statement. That is, if there exists $\widehat{x}\in [\underline{x},\bar{x}]$ and $C\in [\rho, 1]$, such that $\wll(\widehat{x};C) = \u(\widehat{x};C)$, then we have $C =\cstar$ is optimal. We prove by contradiction. Contrary to our claim, suppose that  there exists a feasible $\widehat C>C$ under which  for any $x \in [\underline{x},\bar{x}]$, we have $\wll(x;\widehat C) \leq \u(x;\widehat C)$. 

By the sixth property of Lemma \ref{lem:property_u_l} and Lemma \ref{lem:propertywll}, we have $\wll(\widehat{x};\widehat C) \geq \wll(\widehat{x};C) = \u(\widehat{x};C) \geq \u(\widehat{x};\widehat C)$.

First, consider the case where $\widehat{x} \in [\underline x, \underline x_u]$, where we recall that  $\underline x_u=\sup\{\underline{x}<x<\bar x_u: \CP_o(m;(x,\underline{h}(x))) \geq C \}$; that is,  we have $\u(\widehat{x};C)=m$ if and only if $x \in [\underline x, \underline x_u]$.
We then have \[\wll(\widehat{x};\widehat C) \geq \wll(\widehat{x};C) = \u(\widehat{x};C) = m \geq \u(\widehat{x};\widehat C)\,,\] which implies that $\wll(\widehat{x};\widehat C) = \wll(\widehat{x};C)=m$. However,  by definition, we have $\CP_u(\wll(\widehat{x};\widehat C);(\widehat{x},\bar{h}(\widehat{x}))) \geq \widehat C$ and $\CP_u(\wll(\widehat{x};C);(\widehat{x},\bar{h}(\widehat{x}))) \geq C$, and to be in the  under-protecting case, given that  
$\wll(\widehat{x};\widehat C) = \wll(\widehat{x};C)=m$, we must  have $\bar{h}(\widehat{x})\geq m$. 
By putting $p=m$ and $y=m$ into Equation \eqref{eq:CP_under}, we have $\CP_u(\wll(\widehat{x};\widehat C);(\widehat{x},\bar{h}(\widehat{x})))= \CP_u(\wll(\widehat{x};C);(\widehat{x},\bar{h}(\widehat{x})))=1$, which implies that $C=\widehat C=1$, which contradicts to $C>\\widehat C$.

Now consider the case where $\widehat{x} \in [\underline x_u, \bar{x}]$. We have $\u(\widehat{x};C)<m$ and $\wll(\widehat{x};\widehat C) \geq \wll(\widehat{x};C) = \u(\widehat{x};C) > \u(\widehat{x};\widehat C)$, where the last strict inequality is by Lemma \ref{lem:munique} which states that $\u(\widehat{x};\widehat C) = \u(\widehat{x};C)$ only happens when both of them equals $m$. This forms a contradiction
because the chain of inequalities implies  $\wll(\widehat{x};\widehat C)> \u(\widehat{x};\widehat C)$.

Next, we prove the `only if' statement. That is, if $\wll(x;C) < \u(x;C)$ for all $x \in [\underline{x},\bar{x}]$, then $C$ is not optimal. That is, there exists $\widehat C>C$ such that $\wll(x;C) \leq \u(x;C)$ for all $x \in [\underline{x},\bar{x}]$

In this case, by the sixth property of Lemma \ref{lem:property_u_l} and Lemma \ref{lem:propertywll}, for any $x \in [\underline{x},\bar{x}]$, $\wll(x;C)$ is continuously increasing in $C$ and $\u(x;C)$ is continuously decreasing in $C$. Therefore, by continuity of 
$C$, there exists $\delta>0$ such that with $\widehat C=C+\delta$,  we have $\wll(x;\widehat C) < \u(x;\widehat C)$ for all $x \in [\underline{x},\bar{x}]$. By Lemma \ref{lem:feasibleregionrobust_2}, we have $C$ is feasible and $C$ is not optimal.

\textbf{Part 3: $\widehat{x} \in \V$.}  .
As we defined,
$\V$ is the set containing the $x$ value of all vertices of $\region$  (i.e., $\underline h(\cdot)$ and $\bar h(\cdot)$) and the set $\mathcal R_{0}$, where $\mathcal R_{0} = \{(x, \underline h(x): x\in [\underline x, \bar x]\} \cap \{(x,y): x+y =m\}$. In Lemma \ref{lem:piecewise}, we show that all of $x$-vertices of $\u(\cdot;C)$ and $\wll(\cdot;C)$ are a subset of $\V$ and   the elements of $\mathcal{R}_0$ might be $x$-vertices of $\u(\cdot;C)$. 
By Lemma \ref{lem:worstvertexpl}, we have there exists $\widehat{x} \in \V$ such that $\u(\widehat{x};\cstar) \geq \wll(\widehat{x};\cstar)$, which completes the proof. 

\Halmos

\subsection{Proof of Theorem \ref{thm:MLconsis}}\label{subsec:mlconsis}
Given a polyhedron $\region$, let $\cstar$ be the optimal consistent ratio among all PLAs. First, it is easy to see that the computational complexity is $O(|\V|^3)$ because we enumerate at most $|\V|(|\V|-1)/2$ pairs of $x$-vertices, and recall that 
\[\cstar = \max\{C\in \mathcal S: \wll(x; C)\le \u(x;C) \quad \text{for any $x\in \V$} \}\,,\]
for each pair of vertices, if $\wll(x; C)\le \u(x;C)$, for any $x\in \V$, we will add $C$ into the set $\mathcal S$. By doing this, for each pair of vertices, we compare the value of $\wll(x; C)$ and $\u(x;C)$ for $x \in \V$ with at most $|\V|$ complexity. To summarize, 
the total complexity is bounded by $|\V|(|\V|-1)/2 \cdot |\V|$, which is $O(|\V|^3)$.

Second, Algorithm \ref{alg:mlconsis} cannot return an output $C>\cstar$. This is because it will return \[\cstar = \max\{C\in \mathcal S: \wll(x; C)\le \u(x;C) \quad \text{for any $x\in \V$} \}\,.\]  and by Theorem \ref{thm:property_cstar}, any $C$ such that $\wll(x; C)\le \u(x;C)$ for any $x\in \V$ implies $C \leq \cstar$.

Finally, we show that by enumerating vertices as in Algorithm \ref{alg:mlconsis}, we can find $\cstar$. To show this, we need to split the proof into three cases according to the location of $\widehat{x}$, where $\widehat{x}$ is the intersection point of $\u(\cdot;\cstar)$ and $\wll(\cdot;\cstar)$. By Theorem \ref{thm:property_cstar}, such $\widehat{x}$ always exists. If there are multiple intersection points, we take the one with the smallest $x$ value. Before we divide into three cases, we highlight that enumerating vertices is only for reducing computational complexity, and all the following statements not related to $\V$ are correct for any general convex set $\region$.

\begin{itemize}
    \item \textit{Case 1: $\widehat{x} \leq x_H$ } In this case, we first show that $\widehat{x} = \xmin$ and $\xmin \leq x_H$. Second, as $H$ is a vertex, we have $x_H \in \V$, and by Lemma \ref{lem:xminvertex}, we have $\xmin \in \V$. We show that by balancing $H$ and $(\xmin,\underline{h}(\xmin))$, we obtain a $C_1$ in Algorithm \ref{alg:mlconsis} according to Equation \eqref{eq:balance1}.

    For any $C_1 \leq \cstar$, by the definition of $\wll(\cdot;C_1)$, we have $\wll(x;C_1)=\wll(x_H;C_1)$ for $x \in [\underline{x},x_H]$. As $\u(\widehat{x};C_1)=\wll(\widehat{x};C_1)=\wll(x_H;C_1)$ and $\u(x;C_1) \geq \wll(x;C_1)$ for any $x \in [\underline{x},\bar{x}]$, we have  
    \[
    \widehat{x} \in \text{argmin}_{x \in [\underline{x},\bar{x}]} \u(x;C_1).
    \]
     By Lemma \ref{lem:barxusmall}, we have $\widehat{x} \in [\xmin,x_H]$. As we have defined that if there are multiple intersection points, we take $\widehat{x}$ as the one with the smallest $x$ value. That is,  $\widehat{x}=\xmin$. Moreover, we must have $\xmin \leq x_H$ in this case because otherwise, if $\xmin>x_H$, we have $\u(x;C_1)=m$ for any $x \in [\underline x, x_H]$, and as $\wll(x;C_1)<m$, we have there does not exist $\widehat x$ such that $\u(\widehat x;C_1)=\wll(\widehat x;C_1)$, which is a contradiction.

     Next, we show that by balancing $H$ and $(\xmin,\underline{h}(\xmin))$ according to Equation \eqref{eq:balance1}, we can get $C_1$. Recall that in Equation \eqref{eq:balance1}, given two points $x_1$ and $x_2$ (here $x_H$ and $\bar x_u$), we find $p$ (here $\wll(x_H; C_1)$) such 
      $\CP_u(p; (x_1,\bar{h}(x_1)))=\CP_o(p; (x_2,\underline{h}(x_2))$. 
     This is because, by definition, we have $\CP_u(\wll(x_H;C_1);H)=C_1$, and $\CP_o(\wll(\xmin;C_1);(\xmin,\underline{h}(\xmin)))=C_1$, which implies that 
     \[\CP_u(\wll(x_H;C_1);H)=\CP_o(\wll(\xmin;C_1);(\xmin,\underline{h}(\xmin)))=C_1.\]

     Finally, as $\wll(\widehat x; C_1) = \wll(x_H; C_1)$, we have $\wll(\widehat x; C_1)=\u(\widehat x; C_1)$, and by Theorem \ref{thm:property_cstar}, we know $C_1$ is optimal.

    \item \textit{Case 2: $x_H \leq \widehat{x} \leq \nex$ }  In this case, by Theorem \ref{thm:property_cstar}, we have $\widehat{x} \in \V$. We first show that by balancing $(\widehat{x},\bar{h}(\widehat{x}))$ and $(\widehat{x},\underline{h}(\widehat{x}))$, we obtain a $C_1$ in Algorithm \ref{alg:mlconsis} according to Equation \eqref{eq:balance1}. Moreover, by definition, we have $\u(\widehat{x};C_1)=\wll(\widehat{x};C_1)$, which shows the optimality of $C_1$ and the algorithm can return such a $C_1$. 
    That is, let 
    \[\widehat p = \{p:\CP_u(p;(\widehat{x},\bar{h}(\widehat{x})))=\CP_o(p; (\widehat{x},\underline{h}(\widehat{x})))\}\, \quad \text{and} \quad C_1 = \CP_u(\widehat p;(\widehat{x},\bar{h}(\widehat{x}))) .\]
    We will show that   $\widehat p = \wll(\widehat x; C_1)$.

    As $x_H \leq \widehat{x} \leq \nex$, by definition, we have $\CP_u(\wll(\widehat{x};C_1); (\widehat{x},\bar{h}(\widehat{x})))=C_1$. In addition, by Lemma \ref{lem:hatxinter}, we have $\widehat{x} \in [\underline x_u, \bar x_u]$, and by definition, $\CP_o(\u(\widehat{x};C_1);(\widehat{x},\underline{h}(\widehat{x})))=C_1$, which implies that
    \[
    \CP_u(\wll(\widehat{x};C_1); (\widehat{x},\bar{h}(\widehat{x})))=\CP_o(\u(\widehat{x};C_1);(\widehat{x},\underline{h}(\widehat{x})))=C_1.
    \]
   This shows that $\widehat p = 
 \wll(\widehat{x};C_1) = \u(\widehat{x};C_1)$, as desired.

    \item \textit{Case 3: $\widehat{x} > \nex$ } In this case, by Theorem \ref{thm:property_cstar}, we have $\widehat{x} \in \V$, and by Lemma \ref{lem:nexvertex}, we have $\nex \in \V$. First, we show that by balancing $(\widehat{x},\underline{h}(\widehat{x}))$ and $(\nex,\bar{h}(\nex))$, we obtain a $C_1$ in Algorithm \ref{alg:mlconsis} according to Equation \eqref{eq:balance2}. Finally, we show that under this $C_1$, $\u(\widehat{x};C_1)=\wll(\widehat{x};C_1)$, which shows the optimality of $C_1$ and the algorithm can return such a $C_1$. Let 
    \[\widehat p = \{p:\CP_u(p; (\nex,\bar{h}(\nex)))=\CP_o(p-(\widehat x-\nex); (\widehat x,\underline{h}(\widehat x)))\} \quad \text{and} \quad  \CP_u(\widehat p; (\nex,\bar{h}(\nex))) = C_1\,.\]
    We will show that $\widehat p = \u(\widehat x; C_1)+ (\widehat x- \nex)=\wll(\widehat x; C_1)+ (\widehat x- \nex) =\wll(\nex; C_1 )$. 

    By Lemma \ref{lem:hatxinter}, we have $\widehat{x} \in [\underline{x}_u,\bar x_u]$, and by definition, $\CP_o(\u(\widehat{x};C_1);(\widehat{x},\underline{h}(\widehat{x})))=C_1$. 
  We futher observe that, by Equation \eqref{eq:ll}, we have $\wll(\nex;C_1)=\ll(\nex;C_1)$, and we have
    \[
    \CP_u(\ll(\nex;C_1); (\nex,\bar{h}(\nex)))=\CP_u(\wll(\nex;C_1); (\nex,\bar{h}(\nex))) = C_1 ,
    \]
    which implies that 
    \[
    \CP_u(\wll(\nex;C_1); (\nex,\bar{h}(\nex)))=\CP_o(\u(\widehat{x};C_1);(\widehat{x},\underline{h}(\widehat{x})))=C_1.
    \]

    Finally, as $\wll(\widehat x ;C_1)) = \wll(\nex;C_1) - (\widehat x - \nex) = \u(\widehat x; C_1)$, we have $\widehat p = \u(\widehat x; C_1)+ (\widehat x- \nex)=\wll(\widehat x; C_1)+ (\widehat x- \nex) =\wll(\nex; C_1 )$.

\end{itemize}

\subsection{Proof of Theorem \ref{thm:cmax}} \label{subsec:cmax}
To show that the optimal consistent ratio among all PLAs is optimal among any deterministic and randomized algorithm, we still split the proof into three cases which are the same three cases as in the Proof of Theorem \ref{thm:MLconsis}. Before presenting  the proof, we highlight that although we used many properties from Theorem \ref{thm:property_cstar}, here  in the proof of Theorem \ref{thm:MLconsis}, we do {NOT} assume $\region$ is a polyhedron. 

\textit{Case 1: $\widehat{x} \leq x_H$ } In this case, in case 1 of the Proof of Theorem \ref{thm:MLconsis}, we have established that $\widehat{x}=\xmin$, and $\cstar=\CP_u(p; H)=\CP_o(p;(\widehat{x},\underline{h}(\widehat{x})))$, where $p=\ll(x_H;\cstar)=\u(\xmin;\cstar)$. Our goal here is to show that there does not exist any deterministic or randomized algorithm with a consistent ratio greater than $\cstar$. Let us define two (ordered) input sequences: In the first input  sequence, $I_1$, 
 $\xmin$ low-reward requests arrive first, followed with $\underline h(\xmin)$ high-reward requests. In the second input  sequence, $I_2$,
 $x_H$ low-reward requests arrive first,  followed  $y_H$ high-reward requests. We note that 
 $\widehat{x}=\bar x_u <x_H$ and  $\underline h(\widehat{x})<y_H$ because $H$ is the highest point of $\region$. 
 
 Observe that before receiving $\widehat{x}$ low-reward requests, any algorithm cannot differentiate the two input sequences. Hence, 
any algorithm should decide how many low-reward requests to accept in expectation among the first $\widehat{x}$ ones. 

 Next, we prove by contradiction. Suppose that there exists an algorithm $\mathcal{A}$, which can be either deterministic or randomized, and has a consistent ratio larger than $\cstar$. Then, we have
 \[
 \frac{\mathbf{E}[\text{Rew}(\mathcal{A},I_1)]}{\textsc{opt}(I_1)} > \cstar,
 \]
 where the expectation is taken on the randomization of the algorithm.  Let the expected total amount of high-reward, low-reward requests being accepted by $\mathcal{A}$ be $h(\mathcal{A},I_1)$, $\ell(\mathcal{A},I_1)$, respectively. Replace $C$ by $\cstar$ in Part 1 of Section \ref{subsec:noalgp} where we show that any algorithm facing this instance should accept at least $m-\u(\xmin;C)$ low-reward requests, we have $\ell(\mathcal{A},I_1) > m-\u(\xmin;\cstar)$.

 Next, if $\mathcal{A}$ achieves a consistent ratio greater than $\cstar$, it should also satisfy that
 \[
 \frac{\mathbf{E}[\text{Rew}(\mathcal{A},I_2)]}{\textsc{opt}(I_2)} > \cstar.
 \]
 Let the expected total amount of high-reward, low-reward requests being accepted by $\mathcal{A}$ be $h(\mathcal{A},I_2)$, $\ell(\mathcal{A},I_2)$, respectively. Then, 
  \[
  \frac{\mathbf{E}[\text{Rew}(\mathcal{A},I_2)]}{\textsc{opt}(I_2)} = \frac{h(\mathcal{A},I_2)r_h+\ell(\mathcal{A},I_2)r_{\ell}}{y_Hr_h+\min\{x_H,m-y_H\}r_{\ell}} >\cstar.
  \]
Recall that by the definition of $\ll(\cdot;\cstar)$, we have $\CP_u(\ll(x_H;\cstar);H)=\cstar$ because $x_H \in [x_H,\nex]$. Then, by Equation \eqref{eq:CP_under}, we have
\[
\CP_u(\ll(x_H;\cstar);H)=\frac{\max\{\ll(x_H;\cstar), \min\{y_H,m-x_H\} \} r_h+\min\{x_H,m-\ll(x_H;\cstar)\}r_{\ell}}{y_Hr_h+\min\{x_H,m-y_H\}r_{\ell}},
\]
and by Lemma \ref{lem:specialpoints}, we have $\min\{x_H,m-\ll(x_H;\cstar)\}=m-\ll(x_H;\cstar)$ and 
$\max\{\ll(x_H;\cstar), \min\{y_H,m-x_H\}=\ll(x_H;\cstar)$. Therefore, we have
\[
\frac{\ll(x_H;\cstar) r_h+(m-\ll(x_H;\cstar))r_{\ell}}{y_Hr_h+\min\{x_H,m-y_H\}r_{\ell}}=\cstar.
\]
This implies that
\[
\frac{h(\mathcal{A},I_2)r_h+\ell(\mathcal{A},I_2)r_{\ell}}{y_Hr_h+\min\{x_H,m-y_H\}r_{\ell}} > \frac{\ll(x_H;\cstar) r_h+(m-\ll(x_H;\cstar))r_{\ell}}{y_Hr_h+\min\{x_H,m-y_H\}r_{\ell}}.
\]
As $h(\mathcal{A},I_2) \leq m-\ell(\mathcal{A},I_2)$, we have $\ell(\mathcal{A},I_2)<m-\ll(x_H;\cstar)$. Recall that as mentioned at the beginning, we have $\ll(x_H;\cstar)=\u(\xmin;\cstar)$. Therefore, we have 
\[
\ell(\mathcal{A},I_2)<m-\ll(x_H;\cstar)=m-\u(\xmin;\cstar)<\ell(\mathcal{A},I_1),
\]
which is a contradiction, because before receiving $\widehat{x}$ low-reward requests, any algorithm cannot differentiate the two input sequences, and this implies that $\ell(\mathcal{A},I_2) \geq \ell(\mathcal{A},I_1)$.

\textit{Case 2: $x_H \leq \widehat{x} \leq \nex$.} We replace the instances in the proof of case 1 to get the proof for this case. We replace the first instance by: $\widehat{x}$ low-reward requests arrive first, followed with $\underline{h}(\widehat{x})$ high-reward requests. We replace the second instance by: $\widehat{x}$ low-reward requests arrive first, followed with $\bar{h}(\widehat{x})$ high-reward requests. Then, we use similar arguments in case 1 to show the result.

\textit{Case 3: $\widehat{x} > \nex$.} In this case, as we showed in case 3 of the Proof of Theorem \ref{thm:MLconsis}, we  have that \[\cstar=\CP_u(p; (\nex,\bar{h}(\nex)))=\CP_o(p-(\widehat{x}-\nex);(\widehat{x},\underline{h}(\widehat{x}))),\] where $p=\ll(\nex;\cstar)$, and $\u(\widehat{x};\cstar)=p-(\widehat{x}-\nex)$.
Let us define two (ordered) input sequences: In the first input  sequence, $I_1$,  
 $\nex$ low-reward requests arrive first, followed with $\bar h(\nex)$ high-reward requests. In the second input  sequence, $I_2$,
 $\widehat{x}$ low-reward requests arrive first,  followed  $\underline{h}(\widehat{x})$ high-reward requests. In this case, we have 
 $\nex<\widehat{x}$. Before receiving $\nex$ low-reward requests, any deterministic or randomized algorithm cannot differentiate the two input sequences. Hence, 
any algorithm should decide how many low-reward requests to accept among the first $\nex$ ones in expectation. 

 Next, we prove by contradiction. Suppose that there exists an algorithm $\mathcal{A}$ which has a consistent ratio larger than $\cstar$. As
 \[
 \frac{\mathbf{E}[\text{Rew}(\mathcal{A},I_1)]}{\textsc{opt}(I_1)} > \cstar,
 \]
 we have $\mathcal{A}$ should accept less than $m-\wll(\nex;\cstar)$ low-reward requests in expectation. Let the expected total amount of high-reward, low-reward requests being accepted by $\mathcal{A}$ be $h(\mathcal{A},I_1)$, $\ell(\mathcal{A},I_1)$, respectively. Replace $C$ by $\cstar$ in Part 2 of Section \ref{subsec:noalgp} which shows that any algorithm facing this instance should accept less than $m-\ll(\nex;C)$ low-reward requests, we have $\ell(\mathcal{A},I_1) < m-\ll(\nex;\cstar)$.
 
 Next, if $\mathcal{A}$ achieves a consistent ratio greater than $\cstar$, it should also satisfy that
 \[
 \frac{\mathbf{E}[\text{Rew}(\mathcal{A},I_2)]}{\textsc{opt}(I_2)} > \cstar.
 \]
 Let the expected total amount of high-reward, low-reward requests being accepted by $\mathcal{A}$ be $h(\mathcal{A},I_2)$, $\ell(\mathcal{A},I_2)$, respectively. Then, 
  \[
  \frac{\mathbf{E}[\text{Rew}(\mathcal{A},I_2)]}{\textsc{opt}(I_2)} = \frac{h(\mathcal{A},I_2)r_h+\ell(\mathcal{A},I_2)r_{\ell}}{\underline{h}(\widehat{x})r_h+\min\{\widehat{x},m-\underline{h}(\widehat{x})\}r_{\ell}} >\cstar.
  \]
As is mentioned at the beginning of this case, we have $\CP_o(\u(\widehat{x};\cstar);(\widehat{x},\underline{h}(\widehat{x})))=\cstar$. By Equation \eqref{eq:CP_over}, we have
\[
\CP_o(\u(\widehat{x};\cstar);(\widehat{x},\underline{h}(\widehat{x})))=\frac{\underline{h}(\widehat{x})r_h+\min\{\widehat{x},m-\u(\widehat{x};\cstar) \}r_{\ell}}{\underline{h}(\widehat{x})r_h+\min\{\widehat{x},m-\underline{h}(\widehat{x})\}r_{\ell}}.
\]
Here, we must have $\min\{\widehat{x},m-\u(\widehat{x};\cstar) \}=m-\u(\widehat{x};\cstar)$ because otherwise, if $\min\{\widehat{x},m-\u(\widehat{x};\cstar) \}=\widehat{x}$, then due to $\u(\widehat{x};\cstar) \geq \underline{h}(\widehat{x})$, we have $\min\{\widehat{x},m-\underline{h}(\widehat{x})\}=\widehat{x}$, and 

$\CP_o(\u(\widehat{x};\cstar);(\widehat{x},\underline{h}(\widehat{x})))=1 \neq \cstar$. Therefore, by taking $\min\{\widehat{x},m-\u(\widehat{x};\cstar) \}=m-\u(\widehat{x};\cstar)$, we obtain
\[
\frac{\underline{h}(\widehat{x})r_h+(m-\u(\widehat{x};\cstar) )r_{\ell}}{\underline{h}(\widehat{x})r_h+\min\{\widehat{x},m-\underline{h}(\widehat{x})\}r_{\ell}}=\cstar.
\]
This implies that 
\[
\frac{h(\mathcal{A},I_2)r_h+\ell(\mathcal{A},I_2)r_{\ell}}{\underline{h}(\widehat{x})r_h+\min\{\widehat{x},m-\underline{h}(\widehat{x})\}r_{\ell}} >\frac{\underline{h}(\widehat{x})r_h+(m-\u(\widehat{x};\cstar) )r_{\ell}}{\underline{h}(\widehat{x})r_h+\min\{\widehat{x},m-\underline{h}(\widehat{x})\}r_{\ell}}.
\]
As $h(\mathcal{A},I_2) \leq \underline{h}(\widehat{x})$, we have $\ell(\mathcal{A},I_2) > m-\u(\widehat{x};\cstar)$. However, in this case $\widehat x>\nex$, between $\nex$ and $\widehat{x}$, there are at most $\widehat{x}-\nex$ low-reward requests arriving, and this implies that $\ell(\mathcal{A},I_2)-\ell(\mathcal{A},I_1) \leq \widehat{x}-\nex$. 

However, as is mentioned at the beginning, we have $\u(\widehat{x};\cstar)=\ll(\nex;\cstar)-(\widehat{x}-\nex)$. This implies that 
\begin{align}
\ell(\mathcal{A},I_2)-\ell(\mathcal{A},I_1) &> m-\u(\widehat{x};\cstar) - (m-\ll(\nex;\cstar)) \\&= m-(\ll(\nex;\cstar)-(\widehat{x}-\nex)) - (m-\ll(\nex;\cstar))
 \\&=\widehat{x}-\nex,
\end{align}
which is a contradiction.

{\color{black}
\section{Appendix of Section \ref{sec:multitype}}
\subsection{Proof of Proposition \ref{prop:variable}}\label{append:ntype}
To  show Proposition \ref{prop:variable}, consider  constraints listed in \eqref{eq:cpc1n} and \eqref{eq:cpc2n}.  We can express the $\CP$ function in these constraints as follows:
\[
\CP(p_2(\mathbf x_2),\ldots,p_n(x_n); \mathbf{x}) = \frac{\sum_{i=1}^nA_i(\mathbf{x})r_i}{\textsc{OPT}(\mathbf{x})},
\]
where $\textsc{OPT}(\mathbf{x})$ is the offline optimal value, which is a fixed quantity based on $\mathbf{x}$ and does not depend on the protection level functions. Here, \( A_i(\mathbf{x}) \) represents the effective demand served to the $i$-th highest reward request based on the protection level functions and available capacity. The value of \( A_i(\mathbf{x}) \) is computed recursively:
\[
A_n(\mathbf{x}) = \min\{x_n, (m - p_{n-1}(\mathbf{x}))^+\},
\]
where the term \( (m - p_{n-1}(\mathbf{x}))^+ \) ensures that the demand served to the $n$-th highest reward request (i.e., lowest reward request) is based on the available capacity after considering the commitments for higher reward requests. For any \( i \in [n-1] \), we have:
\[
A_i(\mathbf{x}) = \min\{x_i, (m - p_{i-1}(\mathbf{x}) - \sum_{j=i+1}^{n}A_j(\mathbf{x}))^+\} =  \min\{x_i,\widetilde{A_i}(\mathbf{x})\}.
\]
Here, $\widetilde{A_i}(\mathbf{x}) =(m - p_{i-1}(\mathbf{x}) - \sum_{j=i+1}^{n}A_j(\mathbf{x}))^+$.
\begin{figure} [htb]
     \centering         
  \includegraphics[width=0.49\textwidth]{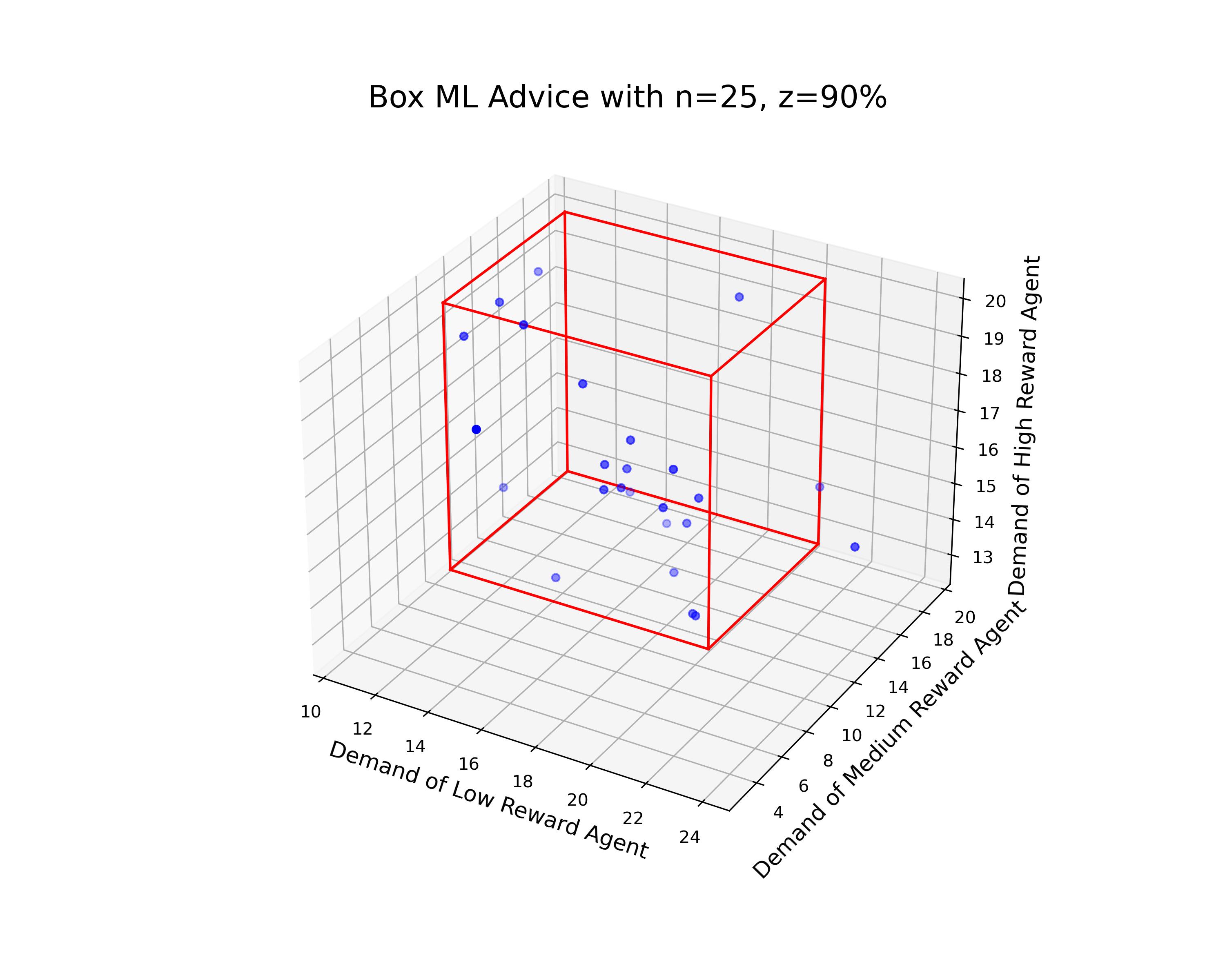}  \includegraphics[width=0.49\textwidth]{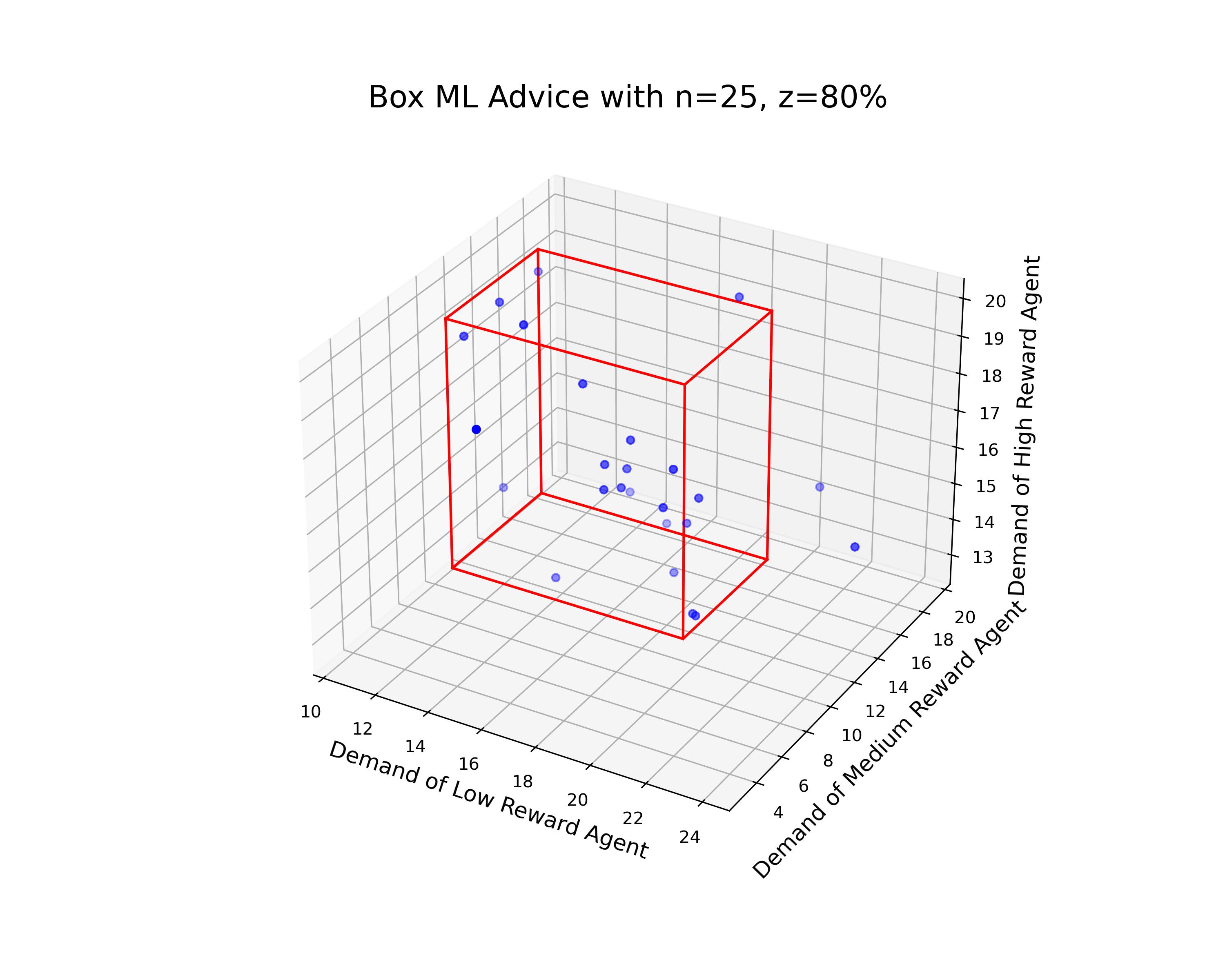}
    \caption{Constructing the box advice with $n =25$ samples with $z= 90\%$ and $z= 80\%$ for the uniform demand model in Equations \eqref{distri3D}, respectively.}
        \label{fig:advice3D}
\end{figure}
We convert Problem \eqref{prob:rbmax_n} into a mixed integer program (MIP) by writing the constraints in a unified form:
\begin{align*}
    & \frac{\sum_{i=1}^nA_i(\mathbf{x})r_i}{\textsc{OPT}(\mathbf{x})} \geq C, \qquad \mathbf{x} \in \mathcal{V} \cup \mathcal{V}', \\
    & A_i(\mathbf{x}) \leq x_i, \qquad A_i(\mathbf{x}) \leq \widetilde{A_i}(\mathbf{x}), \qquad i \in [n], \\
    & A_i(\mathbf{x}) \geq \widetilde{A_i}(\mathbf{x}) + (x_i - m) w_i(\mathbf{x}), \qquad i \in [n], \\
    & A_i(\mathbf{x}) \geq x_i w_i(\mathbf{x}), \qquad i \in [n], \\
    & \widetilde{A_i}(\mathbf{x}) \geq m - p_{i-1}(\mathbf{x}) - \sum_{j=i+1}^{n}A_j(\mathbf{x}), \qquad i \in [n], \\
    & \widetilde{A_i}(\mathbf{x}) \geq 0, \qquad i \in [n], \\
    & \widetilde{A_i}(\mathbf{x}) \leq m(1 - \widetilde{w_i}(\mathbf{x})), \qquad i \in [n], \\
    & \widetilde{A_i}(\mathbf{x}) \leq m - p_{i-1}(\mathbf{x}) - \sum_{j=i+1}^{n}A_j(\mathbf{x}) + n m \widetilde{w_i}(\mathbf{x}), \qquad i \in [n].
\end{align*}
Here,  $w_i(\mathbf{x}) \in \{0,1\}, $ and $\widetilde{w_i}(\mathbf{x})\in \{0, 1\}$. When $w_i(\mathbf{x}) = 1$, we have 
$
A_i(\mathbf{x}) =\min\{x_i,\widetilde{A_i}(\mathbf{x}) \} = x_i,
$
indicating that the full demand of type $i$-th request is served. When $w_i(\mathbf{x}) = 0$, we have
$
A_i(\mathbf{x}) =\widetilde{A_i}(\mathbf{x})
$.
Similarly, when $\widetilde{w_i}(\mathbf{x}) = 1$, we have 
$
\widetilde{A_i}(\mathbf{x}) = (m - p_{i-1}(\mathbf{x}) - \sum_{j=i+1}^{n}A_j(\mathbf{x}))^+ = 0,
$
and when $\widetilde{w_i}(\mathbf{x}) = 0$, 
$
\widetilde{A_i}(\mathbf{x}) = m - p_{i-1}(\mathbf{x}) - \sum_{j=i+1}^{n}A_j(\mathbf{x})$. A similar transformation applies to 
$
\frac{\sum_{i=1}^nA_i(\mathbf{x})r_i}{\textsc{OPT}(\mathbf{x})} \geq R, ~ \mathbf{x} \in \mathcal{V}'.
$

The formulation contains \( 2n|\mathcal{V} \cup \mathcal{V}'| \) continuous variables: \( A_i(\mathbf{x}), \widetilde{A_i}(\mathbf{x}) \). It also contains \( 2n|\mathcal{V} \cup \mathcal{V}'| \) integer variables: \( w_i(\mathbf{x}), \widetilde{w_i}(\mathbf{x}) \). In addition, there are \( n|\mathcal{V} \cup \mathcal{V}'| \) continuous variables for the protection levels \( p_i(\mathbf{x}) \). Finally, \( R \) is a single continuous variable. Therefore, the total number of continuous variables is \( 3n|\mathcal{V} \cup \mathcal{V}'| + 1 \), and the total number of binary integer variables is \( 2n|\mathcal{V} \cup \mathcal{V}'| \).

\subsection{Algorithm \ref{alg:envelope}: Nested Interpolation for Adaptive PLAs}\label{sec:heuristic}

As stated earlier, after solving the MILP problem \eqref{prob:rbmax_n}, we obtain the optimal values of $p_2(\mathbf{x}_2), \ldots, p_n(\mathbf{x}_n)$ only at the points $\mathbf{x} \in \mathcal{V}' \cup \mathcal{V}$. To extend the adaptive protection level functions $p_i(\mathbf{x}_i)$ across the full space $\mathbb{N}^{n-i+1}$ for each $i \in {2,3,\ldots,n}$, we use the interpolation heuristic outlined in Algorithm \ref{alg:envelope}, where $\mathbb{N}$ denotes the set of natural numbers.
\begingroup
 \floatname{algorithm}{Algorithm}
\begin{algorithm} \footnotesize
\caption{\textbf{Nested Interpolation for Adaptive PLAs}}
\label{alg:envelope}

\begin{itemize}[leftmargin=*]
\item[] \textbf{Input:} Sets of vertices $\mathcal{V}$ (ML polytope) and $\mathcal{V}'$ (vertices of $[0,m]^n$), and the optimal values $\{p_i(\mathbf{x}_i)\}_{i=2}^{n}$ available only for $\mathbf{x} \in \mathcal{V} \cup \mathcal{V}'$.

\item[] \textbf{Output:} Extended protection level functions $\tilde{p}_i: \mathbb{N}^{n-i+1} \to [0,m]$ for every $i \in \{2,\dots,n\}$.
\end{itemize}

\begin{itemize}[leftmargin=*]
\item[\textbf{1.}] \textbf{Initialize $\tilde{p}_{n}$.}
      \begin{itemize}[leftmargin=1.2em]
      \item For every integer $x_n \in \mathbb{N}$, define the interpolation:
            \[
              \tilde{p}_n(x_n) = \min_{v_n \in \mathcal{U}_n}
              \Bigl\{ p_n(v_n) + [v_n - x_n]_+ \Bigr\},
              \qquad 
              \mathcal{U}_n = \{ {v}_n \mid {\mathbf v}=(v_1, \ldots, v_n)  \in \mathcal{V} \cup {V}' \}.
            \]
            (Here, $[z]_+ := \max\{0,z\}$.)
      \end{itemize}

\item[\textbf{2.}] \textbf{Iterate downward for $i = n-1, \dots, 2$.}
      \begin{itemize}[leftmargin=1.2em]
      \item[\textbf{(a)}] \textit{Raw envelope.}\;
            For each integer vector $\mathbf{x}_i = (x_i, \dots, x_n) \in \mathbb{N}^{n-i+1}$, set
            \[
              p_i^{\mathrm{env}}(\mathbf{x}_i) = \min_{\mathbf{v}_i \in \mathcal{U}_i}
              \Bigl\{
                p_i(\mathbf{v}_i) + \sum_{j=i}^{n} [v_j - x_j]_+
              \Bigr\},
              \quad 
              \mathcal{U}_i = \{ \mathbf{v}_i = (v_i, \ldots, v_n) \mid  \mathbf{v} =(v_1, v_2, \ldots, v_n) \in  \mathcal{V} \cup \mathcal{V}' \}.
            \]
      \item[\textbf{(b)}] \textit{Enforce nesting with the next‐higher PLA.}\;
            \[
              \tilde{p}_i(\mathbf{x}_i) = 
              \min\Bigl\{
                    p_i^{\mathrm{env}}(\mathbf{x}_i),\;
                    \tilde{p}_{i+1}(x_{i+1}, \dots, x_n)
                  \Bigr\}
              \qquad \forall\, \mathbf{x}_i \in \mathbb{N}^{n-i+1}.
            \]
      \end{itemize}

\item[\textbf{3.}] \textbf{Return} $\{\tilde{p}_i\}_{i=2}^{n}$.
\end{itemize}
\end{algorithm}

\endgroup

\subsection{Multi-type ML Regions} \label{multi-type-ML-region}
In this section, we present examples of three-dimensional ML advice constructed from 25 random samples drawn from the distribution in Equation~\eqref{distri3D}. Figure~\ref{fig:advice3D} illustrates the box advice under different parameter settings: the left plot captures 90\% of the samples, while the right plot captures 80\% of the samples. In Figure~\ref{fig:advice3Dpoly}, we construct three-dimensional polyhedra based on the approach in~\cite{bertsimas2018data}, varying the parameters \(\epsilon\), \(\theta\), and \(\phi\). We observe that as \(\theta\) and \(\phi\) increase, the number of vertices increases. Additionally, as \(\epsilon\) decreases, the ML region expands, becoming more conservative and encompassing more outliers.


\begin{figure} [htb]
     \centering           \includegraphics[width=0.4\textwidth]{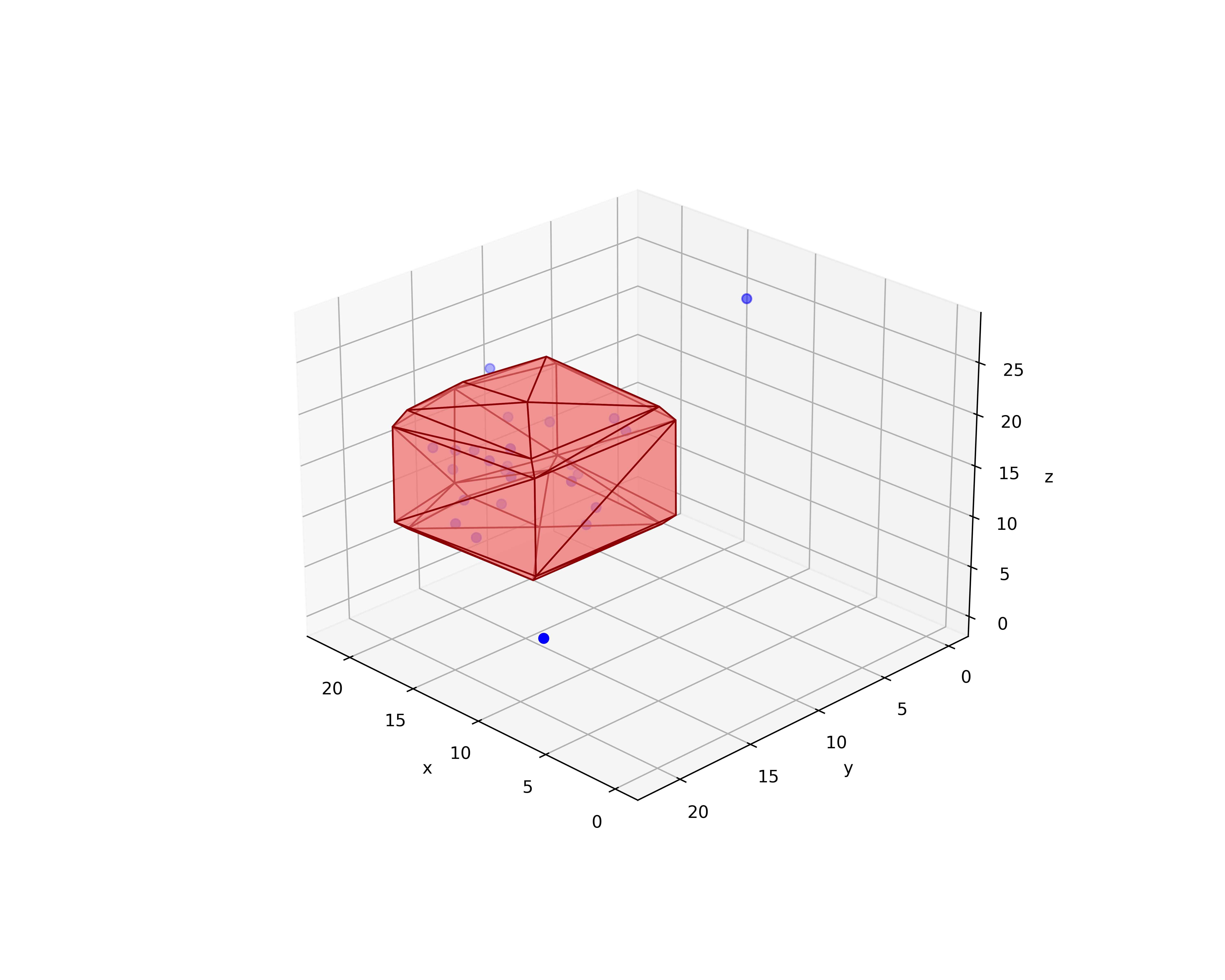}
     \includegraphics[width=0.4\textwidth]{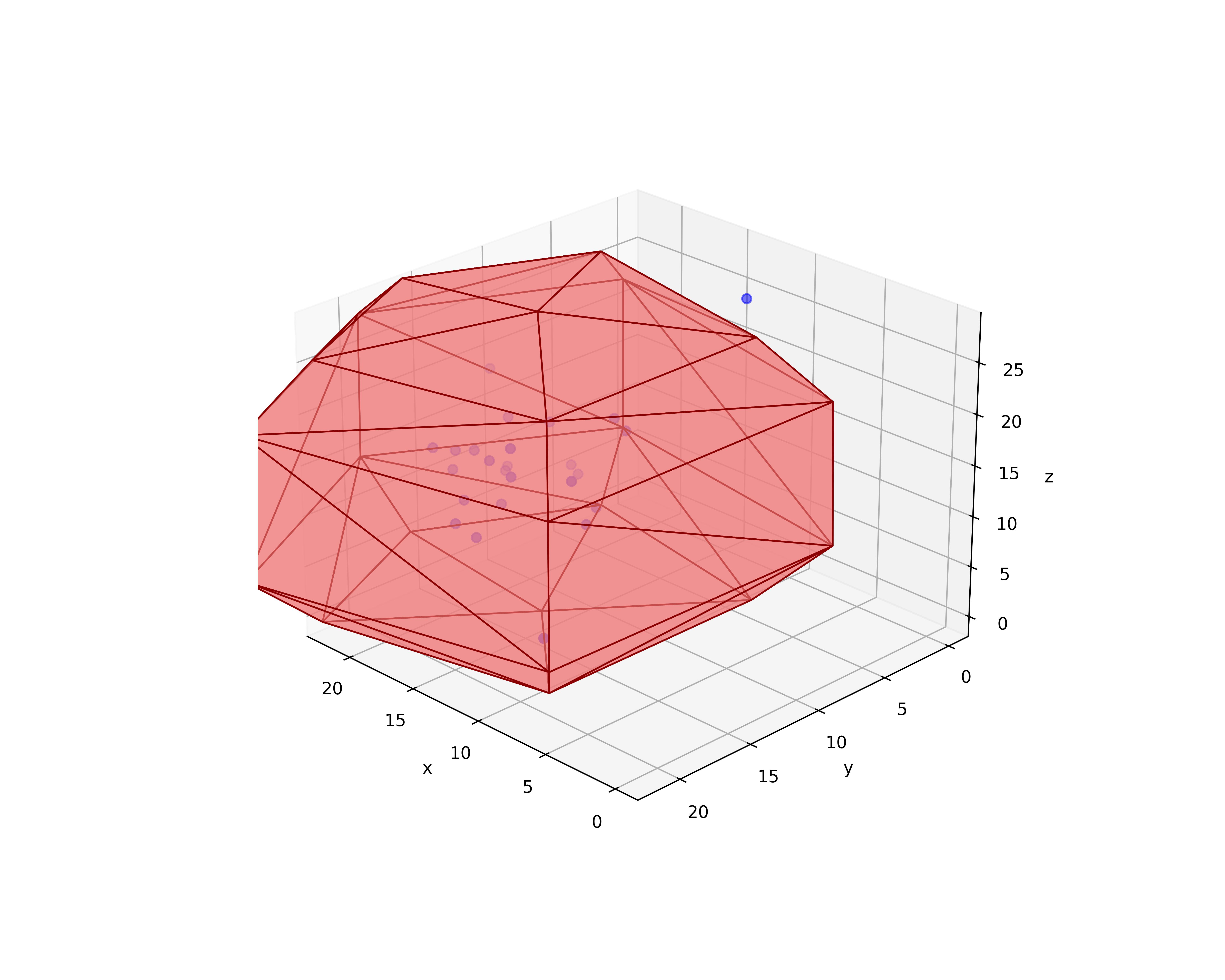} 
\includegraphics[width=0.4\textwidth]{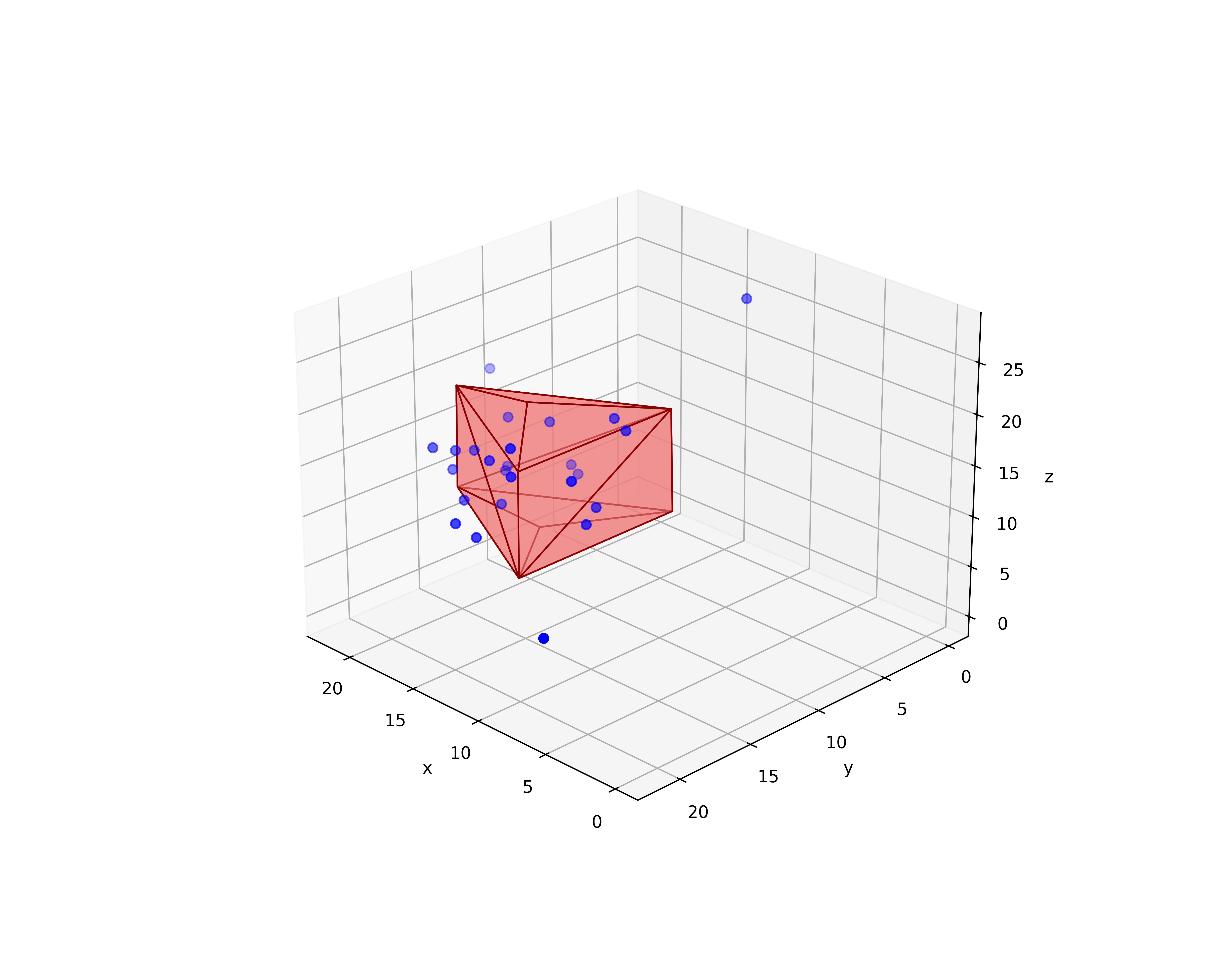}
      \includegraphics[width=0.4\textwidth]{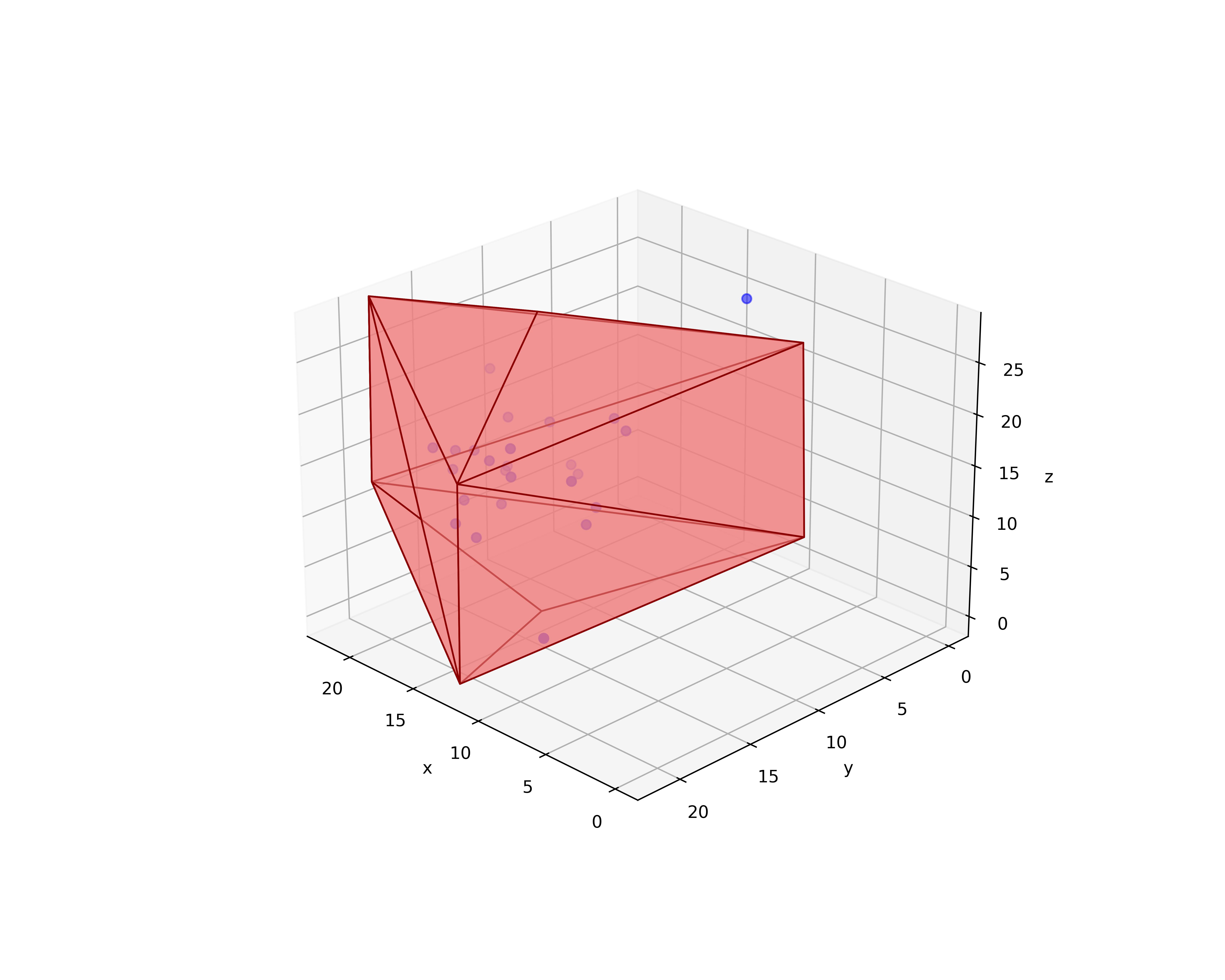}

    \caption{Constructing the box advice with $n =25$ samples for the uniform demand model in Equations \eqref{distri3D}, respectively. \textbf{Upper-left: } $\epsilon=0.9$, $\theta=\phi=6$; \textbf{Upper-right: } $\epsilon=0.1$, $\theta=\phi=6$; \textbf{Lower-left: } $\epsilon=0.9$, $\theta=\phi=4$; \textbf{Lower-right: } $\epsilon=0.1$, $\theta=\phi=4$. }
        \label{fig:advice3Dpoly}
\end{figure}
}

\section{Preliminary Lemmas} 
\label{sec:proofpre}
Here, we introduce several small lemmas which will be used in  several  proofs in this paper.

\subsection{Lemma \ref{lem:underbelow} and its Proof}
\begin{lemma}[Not Enough Demand] \label{lem:underbelow}
    For any $A=(x,y)$ with $x+y<m$ and  $p<y$, we have $\CP_u(p;A=(x,y))=1$.
\end{lemma}

\proof{Proof of Lemma \ref{lem:underbelow}}
By Equation \eqref{eq:CP_under}, we have \begin{align*} 
\CP_u(p;A= (x,y)) &=  
\; \frac{\max\{p, \min\{y,m-x\} \} r_h+\min\{x,m-p\}r_{\ell}}{yr_h+\min\{x, m-y\}r_{\ell}}\\
&= \frac{\max\{p, y \} r_h+\min\{x,m-p\}r_{\ell}}{yr_h+xr_{\ell}} = \frac{y r_h+xr_{\ell}}{yr_h+xr_{\ell}} = 1\,,
\end{align*}
where in the first equation, we used the assumption that $x+y< m$, and in the second equation, we used the assumption that $x+y< m$ and $p< y$.  The last equation is the desired result. 

\endproof

\subsection{Lemma \ref{lem:pmonotone} and its Proof}
\begin{lemma}[Monotonicity of the Compatible Ratio $\CP(p; (x, y))$ w.r.t. $p$] \label{lem:pmonotone} For any $0 \le y \le p_1 \le p_2$, and any $x\ge 0$, we have
    \[
    \CP_o(p_1;(x,y)) \geq \CP_o(p_2;(x,y)).
    \]
    Further, for any $p_1 \le p_2 \le y$, and any $x\ge 0$,  we have
    \[
    \CP_u(p_2;(x,y)) \geq \CP_u(p_1;(x,y)).
    \]
    That is, the compatible ratio of a point $(x, y)$ increases when the gap between the protection level $p$ and $y$, i.e., $|p-y|$,  gets smaller. In addition, if $p>m-x$, we have the strong monotonicity for both of $\CP_o(p;(x,y))$ and $\CP_u(p;(x,y))$.
    
\end{lemma}

\proof{Proof of Lemma \ref{lem:pmonotone}}
By the definition of $\CP_o$ in Equation \eqref{eq:CP_over} and the assumption that $y\le p_1\le p_2$, we have
\[
\CP_o(p_1;(x,y)) = \frac{\min\{y,m\}r_h+\min\{x,m-p_1\}r_{\ell}}{\min\{y,m\}r_h+\min\{x, (m-y)^{+}\}r_{\ell}} \geq \frac{\min\{y,m\}r_h+\min\{x,m-p_2\}r_{\ell}}{\min\{y,m\}r_h+\min\{x, (m-y)^{+}\}r_{\ell}}=\CP_o(p_2; (x,y))\,.
\]
The last inequality is the desired result. Moreover, if $p_2 >p_1>m-x$, we have the strong monotonicity.

By the definition of $\CP_u$ in Equation \eqref{eq:CP_under} and the assumption that $ p_1\le p_2 \le y$, we have
\begin{align*}
    \CP_u(p_2; (x,y)) &= \frac{\max\{p_2, \min\{y,(m-x)^{+}\} \} r_h+\min\{x,m-p_2\}r_{\ell}}{\min\{y,m\}r_h+\min\{x, (m-y)^{+}\}r_{\ell}} \\&\geq \frac{\max\{p_1, \min\{y,(m-x)^{+}\} \} r_h+\min\{x,m-p_1\}r_{\ell}}{\min\{y,m\}r_h+\min\{x, (m-y)^{+}\}r_{\ell}} = \CP_u(p_1; (x,y)),
\end{align*}
where the inequality holds because for $p_2 \geq p_1$, $\max\{p_2, \min\{y,m-x\} \} \geq \max\{p_1, \min\{y,m-x\} \}$ and $\min\{x,m-p_2\} \geq \min\{x,m-p_1\}$. Moreover, if $p_2 >p_1>m-x$, we have the strong monotonicity.

\endproof

\subsection{Lemma \ref{lem:mout} and its Proof}
\begin{lemma} \label{lem:mout}
    For any $y > m$ and $p \in [0,m]$, we have
    \[
    \CP_u(p;(x,y))=\CP_u(p;(x,m)).
    \]
    For any $x>m$ and $p \in [0,m]$, we have
    \[
    \CP_o(p;(x,y))=\CP_o(p;(m,y)),
    \]
    and
    \[
    \CP_u(p;(x,y))=\CP_u(p;(m,y)).
    \]
\end{lemma}

\proof{Proof of Lemma \ref{lem:mout}}
We first prove that $\CP_u(p;(x,y))=\CP_u(p;(x,m))$ for any $y>m$. By Equation \eqref{eq:CP_under}, we have
\begin{align*}
    \CP_u(p;(x,y)) &= \frac{\max\{p, \min\{y,(m-x)^{+}\} \} r_h+\min\{x,m-p\}r_{\ell}}{\min\{y,m\}r_h+\min\{x, (m-y)^{+}\}r_{\ell}} \\&= \frac{\max\{p, \min\{m,(m-x)^{+}\} \} r_h+\min\{x,m-p\}r_{\ell}}{mr_h} = \CP_u(p;(x,m)).
\end{align*}

Next, we show that $\CP_o(p;(x,y))=\CP_o(p;(m,y))$ for any $x>m$. By Equation \eqref{eq:CP_over}, we have
\begin{align*}
    \CP_o(p;(x,y))&= \frac{\min\{y,m\}r_h+\min\{x,m-p\}r_{\ell}}{\min\{y,m\}r_h+\min\{x, (m-y)^{+}\}r_{\ell}} \\&= \frac{\min\{y,m\}r_h+\min\{m,m-p\}r_{\ell}}{\min\{y,m\}r_h+\min\{m, (m-y)^{+}\}r_{\ell}} = \CP_o(p;(m,y)).
\end{align*}
By Equation \eqref{eq:CP_under}, we have
\begin{align*}
    \CP_u(p;(x,y)) &= \frac{\max\{p, \min\{y,(m-x)^{+}\} \} r_h+\min\{x,m-p\}r_{\ell}}{\min\{y,m\}r_h+\min\{x, (m-y)^{+}\}r_{\ell}} \\&= \frac{\max\{p, \min\{y,0\} \} r_h+\min\{m,m-p\}r_{\ell}}{\min\{y,m\}r_h+\min\{m, (m-y)^{+}\}r_{\ell}} = \CP_u(p;(m,y)).
\end{align*}

\Halmos
\endproof

\subsection{Lemma \ref{lem:Hlow} and its Proof}
\begin{lemma} \label{lem:Hlow}
    Recall that $\underline h(x)$, defined in Equation \eqref{eq:lower_upper}, is the lower envelop of $\region$. Let $\underline{\mathcal{H}}(x)=\min\{\underline h(x),m\}$. If there exists $x_0 \in [\underline x, \bar x]$ such that $\underline h(x_0)<m$, then either $\underline{\mathcal{H}}(x)=\underline h(x)$ for all $x \in [\underline x, \bar x]$ or there exists $\underline x \leq x_1 < x_2 \leq \bar x$ such that $\underline{\mathcal{H}}(x)=m$ for $x \in [\underline x, x_1] \cup [x_2, \bar x]$ and $\underline{\mathcal{H}}(x)=\underline h(x)$ for $x \in [x_1,x_2]$. In other words, $\underline{\mathcal{H}}(x)=\underline h(x)$ on a connected interval.  
\end{lemma}

\proof{Proof of Lemma \ref{lem:Hlow}}
We will prove this statement by contradiction. Let us suppose that $\underline{\mathcal{H}}(x)=\underline h(x)$ on $[x_3,x_4] \cup [x_5,x_6]$. Then, let us randomly pick $x_7 \in (x_4,x_5)$. We recall that $\underline{\mathcal{H}}(x)=\min\{\underline h(x),m\}$.

Now, we observe that $\underline{h}(x_4)<m$, $\underline{h}(x_5)<m$, and $\underline{h}(x_7) \geq m$. However, this contradicts the fact that $\underline{h}(\cdot)$ is a convex function. Therefore, our initial assumption that $\underline{\mathcal{H}}(x)=\underline h(x)$ on $[x_3,x_4] \cup [x_5,x_6]$ is false

\Halmos
\endproof

\subsection{Lemma \ref{lem:Hh} and its Proof}
\begin{lemma} \label{lem:Hh}
    Recall that $\bar x_u$ is defined in Equation \eqref{eq:xmin}, and $\underline x_u=\sup\{\underline{x}<x<\bar x_u: \CP_o(m;(x,\underline{h}(x))) \geq C \}$. Then, for $x \in (\underline x_u, \bar x_u)$, we have $\underline{\mathcal{H}}(x)=\underline h(x)$.
\end{lemma}

\proof{Proof of Lemma \ref{lem:Hh}}
By Lemma \ref{lem:Hlow}, we have either $\underline{\mathcal{H}}(x)=\underline h(x)$ for all $x \in [\underline x, \bar x]$ or there exists $\underline x \leq x_1 < x_2 \leq \bar x$ such that $\underline{\mathcal{H}}(x)=m$ for $x \in [\underline x, x_1] \cup [x_2, \bar x]$ and $\underline{\mathcal{H}}(x)=\underline h(x)$ for $x \in [x_1,x_2]$. In the first case, the statement is trivial. 

In the second case, we prove by contradiction. Consider any $x\in (\underline x_u, \bar x_u)$ and assume that  $x \in [\underline x, x_1]$. We then have $\underline{\mathcal{H}}(x)=m$, which implies that $\underline{h}(x) \geq m$. By Equation \eqref{eq:CP_over}, we have
\[
\CP_o(m;(x,\underline{h}(x))) =1 \geq C,
\]
which implies that $x < \underline x_u$, which contradicts the fact that $x \in (\underline x_u, \bar x_u)$. 

Now, consider any $x\in (\underline x_u, \bar x_u)$ and assume that   $x \in [x_2,\bar x]$. We then have $\underline{\mathcal{H}}(x)=m$, which implies that $\underline{h}(x) \geq m$. Recall that 
\begin{align} \bar x_u = \left\{ \begin{array}{ll}
         x_L & \quad \mbox{if $x_L+y_L \geq m$};\\
        \sup\{x \in [x_L,\bar{x}]: (1-C)\frac{r_h}{r_{\ell}}\underline{\mathcal{H}}'(x^{-})-C < 0 \} & \quad  \mbox{Otherwise}\,,\end{array} \right. \end{align}
In this case, as $\underline{h}(x)<m$ for $x \in [x_1,x_2]$, we have $x_L \in [x_1,x_2]$. Therefore, $\bar x_u \neq x_L$ because otherwise, we have $\bar x_u = x_L \leq x_2$, and there does not exist $x\in (\underline x_u, \bar x_u)$ and $x \in [x_2,\bar x]$. Then, we have $\bar x_u = \sup\{x \in [x_L,\bar{x}]: (1-C)\frac{r_h}{r_{\ell}}\underline{\mathcal{H}}'(x^{-})-C < 0 \}$. As we assume that $\underline{h}(x) \geq m$ for $x \in [x_2,\bar x]$, then we have $\underline{\mathcal{H}}'(\bar x^{-}) = 0$ because $\underline{\mathcal{H}}(x)=m$ for $x \in [x_2,\bar x]$. Therefore, $\bar x_u=\bar x$. In addition, as we have $\underline{h}(\bar x) \geq m$, by Equation \eqref{eq:CP_over}, we have
\[
\CP_o(m;(\bar x,\underline{h}(\bar x))) =1 \geq C,
\]
which implies that $\underline x_u = \bar x_u$. Therefore, there does not exist $x\in (\underline x_u, \bar x_u)$, which is a contradiction.
\Halmos
\endproof

\subsection{Lemma \ref{lem:leftthresholdun} and its Proof}
\begin{lemma} \label{lem:leftthresholdun} 
    Fix any $C \in (0,1)$. Recall that $\underline x_u=\sup\{\underline{x}<x<\bar x_u: \CP_o(m;(x,\underline{h}(x))) \geq C \}$, and $\bar{x}_{l}=\inf\{x_H<x<\bar{x}: \CP_u(0;(x,\overline{h}(x))) \geq C \}$.  Then,  
    for $x \in [\bar x_{\ell},\bar x]$, we have 
    \[
    \CP_u(0; (x,\overline{h}({x}))) \geq C.
    \]
    Similarly, for any $\underline x<x \leq \underline x_u$, we have
    \[
    \CP_o(m; (x,\underline{h}({x}))) \geq C.
    \]
\end{lemma}

\proof{Proof of Lemma \ref{lem:leftthresholdun}}
\textit{Part 1: } We first define the region $\region_1$ such that for any $(x,y) \in \region_1$, $\CP_u(p(x)=0;(x,y)) = C$.  
Obviously, $\{x+y \leq m\} \not\subset \region_1$ because if $x+y \leq m$ and $p(x)=0$, by Lemma \ref{lem:underbelow}, $\CP_u(p(x)=0;(x,y))=1$. Next, we find $(x,y) \in \{x+y > m\}$, which belongs to $\region_1$. 

We solve the following equation: for $x+y>m$
\[
\CP_u(0; (x,y))=C.
\]
Notice that for $x>x_H$, by the definition of $H$, we have $y<m$, so we can obtain
\[
\frac{(m-x)r_h+xr_{\ell}}{yr_h+(m-y)r_{\ell}}=C,
\]
which is equivalent as 
\begin{equation} \label{eq:temp15}
(m-x)r_h+xr_{\ell}=C(yr_h+(m-y)r_{\ell}).
\end{equation}

We take derivative on both side of Equation \eqref{eq:temp15}, and we get
\[
y'(x)=-1/C.
\]

Let $\mathcal{L}(x)$ be the line with slope $-1/C$ and across $(\bar{x}_{\ell},\bar{h}(\bar{x}_{\ell}))$. We have $\CP_u(0;(x,\mathcal{L}(x))) = C$. Recall that $\bar{x}_{\ell}=\inf\{x_H<x<\bar{x}: \CP_u(0;(x,\overline{h}(x))) \geq C \}$. Then, for any $x<\bar{x}_{\ell}$, we have $\CP_u(0;(x,\bar{h}(x))) < C$. By Lemma \ref{lem:ymonotone}, we have $\bar{h}(x)>\mathcal{L}(x)$, and by Lemma \ref{lem:leftconcave}, we have for $x > \bar{x}_{\ell}$, $\bar{h}(x)<\mathcal{L}(x)$. By Lemma \ref{lem:ymonotone} again, we have
    $
    \CP_u(0; (x,\bar{h}(x)) \geq C.$

\textit{Part 2: }
We define the region $\region_2$ such that for any $(x,y) \in \region_2$, $\CP_o(p(x)=m;(x,y)) = C$. For $(x,y) \in \{x+y \geq m\}$, we solve 
\[
\CP_o(m;(x,y)) = C.
\]
Notice that we must have $y<m$ because otherwise, $\CP_o(m;(x,y)) =1 \neq C$. Then, we can obtain 
\[
\frac{yr_h}{yr_h+(m-y)r_{\ell}}=C,
\]
and we can get $y=\frac{C r_{\ell}}{r_h-C(r_h-r_{\ell})}m$. Therefore, in the area $\{x+y \geq m\}$, $\region_2$ is a line with slope $0$.

Next, we explore the part $(x,y) \in \{x+y<m\}$, this time
\[
\CP_o(m;(x,y)) = C,
\]
implies that
\[
\frac{yr_h}{yr_h+xr_{\ell}}=C,
\]
which is equivalent as
\begin{equation} \label{eq:temp16}
y(1-C)r_h=C r_{\ell} x.
\end{equation}
Then, we take derivative on both side of Equation \eqref{eq:temp16}, and we get
\[
y'(x)=\frac{C r_{\ell}}{(1-C) r_h}, 
\]
which means that, in the area $\{x+y<m\}$, $\region_2$ is a line with positive slope $\frac{C r_{\ell}}{(1-C) r_h}$, for $x < m-\frac{C r_{\ell}}{r_h-C(r_h-r_{\ell})}m$, and across $(m-\frac{C r_{\ell}}{r_h-C(r_h-r_{\ell})}m,\frac{C r_{\ell}}{r_h-C(r_h-r_{\ell})}m)$.

Combined with the results above, we start to prove this lemma. Let $\mathcal{L}(x)$ be a piecewise linear line segment such that $(x,\mathcal{L}(x)) \in \region_2$ for all $x \in [0,m]$. That is, for $x \in[0, m-\frac{C r_{\ell}}{r_h-C(r_h-r_{\ell})}m]$, $\mathcal{L}(x)$ has slope $\frac{C r_{\ell}}{(1-C) r_h}$ and across $(m-\frac{C r_{\ell}}{r_h-C(r_h-r_{\ell})}m,\frac{C r_{\ell}}{r_h-C(r_h-r_{\ell})}m)$. For $x \in (m-\frac{C r_{\ell}}{r_h-C(r_h-r_{\ell})}m,m]$, $\mathcal{L}(x)=\frac{C r_{\ell}}{r_h-C(r_h-r_{\ell})}m$. Moreover, for any $x \in [0,m]$, we have $\CP_o(m;(x,\mathcal{L}(x)))=C$.

Recall that $\underline x_u=\sup\{\underline{x}<x<\bar x_u: \CP_o(m;(x,\underline{h}(x))) \geq C \}$. We first show the case where $\underline x_u \leq x_L$. As $\mathcal{L}$ is an non-decreasing function, and for $x \leq x_L$, $\underline{h}(\cdot)$ is a decreasing function. Therefore, as we have $\underline{h}(\underline x_u)=\mathcal{L}(\underline x_u)$, we have $\underline{h}(x)>\mathcal{L}(x)$ for any $x \in [\underline x, \underline x_u]$. Therefore, by Lemma \ref{lem:ymonotone}, we have for any $x \in [\underline x, \underline x_u]$, 
\[
\CP_o(m; (x,\underline{h}({x}))) \geq C.
\]

Second, we show the case where $\underline x_u > x_L$. In this case, $\underline{h}(\cdot)$ is decreasing for $x<x_L$, and increasing for $x>x_L$. If there exists $x_1 \in [\underline x, \underline x_u]$ such that $\CP_o(m;(x_1,\underline{h}(x_1)))<C$, then, by Lemma \ref{lem:ymonotone}, this implies that $\underline{h}(x_1)<\mathcal{L}(x_1)$. As the lowest point $L$ is below the line $\mathcal{L}$ and $\underline{h}(x)$ is convex, we have there are at most two intersections of $\underline{h}(x)$ and $\mathcal{L}$. One is in the left of $L$ and the other is $(\underline x_u,\underline{h}(\underline x_u)$. As $\underline{h}(x)$ is convex increasing for $x \in [x_L,\bar x_u]$ and given that $\mathcal{L}(x)$ is concave increasing for $x \in [x_L,\bar x_u]$, we have 
\[
\underline{h}'(\underline x_u^{+}) \geq \underline{h}'(\underline x_u^{-}) \geq \mathcal{L}(\underline x_u^{+}),
\]
which implies that there exists $\epsilon>0$, such that $\underline{h}(\underline x_u+\epsilon) \geq \mathcal{L}(\underline x_u+\epsilon)$. By Lemma \ref{lem:ymonotone}, we have 
\[
\CP_o(m;(\underline x_u+\epsilon,\underline{h}(\underline x_u+\epsilon))) \geq \CP_o(m;(\underline x_u+\epsilon,\mathcal{h}(\underline x_u+\epsilon)))=C,
\]
which contradicts to the definition of $\underline x_u$. Therefore, for any $x \in [\underline x, \underline x_u]$, 
\[
\CP_o(m; (x,\underline{h}({x}))) \geq C.
\] 
\Halmos
\endproof

\subsection{Lemma \ref{lem:cpexists} and its Proof} 

\begin{lemma} \label{lem:cpexists}
    Fix any $C>0$. For any $x\in [\underline x_u, \bar x_u]$, $\CP_o(p;(x,\underline{h}(x)))=C$ has a solution $p \ge \min\{\underline h(x),m\}$. For any $x \in [x_H,\bar x_{\ell}]$, $\CP_u(p;(x,\bar{h}(x)))=C$ has a solution $p \le \min\{\bar h(x),m\}$.
\end{lemma}

\proof{Proof of Lemma \ref{lem:cpexists}}
Fix any $C>0$. Take $p=\min\{\underline{h}(x),m\}$, by Equation \eqref{eq:CP_over}, we have $\CP_o(p;(x,\underline{h}(x)))=1$. If $\underline h(x) \geq m$, then we claim that $x \notin [\underline x_u, \bar x_u]$. This is because take $p=m$, by Lemma \ref{lem:mout}, we have $\CP_o(m;(x,\underline{h}(x)))=\CP_o(m;(x,m))=1$, which contradicts the definition of $ \underline x_u$. Otherwise, if $\underline{h}(x)<m$, we have $\CP_o(m;(x,\underline{h}(x)))<C$ and $\CP_o(\underline{h}(x);(x,\underline{h}(x)))=1$. Therefore, by mean value theorem, we have there exists $p_1 \in [\underline{h}(x),m]$ such that $\CP_o(p_1;(x,\underline{h}(x)))=C$. 

Similarly, take $p=\min\{\bar{h}(x),m\}$, by Equation \eqref{eq:CP_under}, we have $\CP_u(p;(x,\bar{h}(x)))=1$. Take $p=0$, by the definition of $\underline{x}_{\ell}=\sup\{\underline{x}<x<x_H:\CP_u(0;(x,\overline{h}(x)))=C\}, \quad
 \bar{x}_{l}=\inf\{x_H<x<\bar{x}: \CP_u(0;(x,\overline{h}(x))) \geq C \}$, we have $\CP_u(p;(x,\bar{h}(x)))<C$. From Equation \eqref{eq:CP_under}, we can simply check that $\CP_u(p;(x,\bar{h}(x)))$ is continuous in $p$ for any $x$. Therefore, by mean value theorem, we have there exists $p_1 \in [0, \min\{\bar{h}(x),m\}]$ such that $\CP_u(p_1;(x,\bar{h}(x)))=C$.

\subsection{Lemma \ref{lem:worstrightpart} and its Proof}
\begin{lemma} \label{lem:worstrightpart}
    Recall that \begin{align} \bar x_u = \left\{ \begin{array}{ll}
         x_L & \quad \mbox{if $x_L+y_L \geq m$};\\
        \sup\{x \in [x_L,\bar{x}]: (1-C)\frac{r_h}{r_{\ell}}\underline{\mathcal{H}}'(x^{-})-C < 0 \} & \quad  \mbox{Otherwise}.\end{array} \right. \end{align}
     
     Then, for any  $x\in [\bar x_u,\bar{x}]$, we have $
\CP_o(\u(x;C);(x,\underline{h}(x))) \geq C$.     
\end{lemma}

\proof{Proof of Lemma \ref{lem:worstrightpart}}
Define  $\hat{\u}(x;C)$ as the original $\u(x;C)$ without forcing to be constant after $\bar x_u$.  More precisely, 
for any $C \in [0,1]$,  we define 
 \begin{align}\hat \u(x; C)= \sup\big \{p\in [0,m]: \CP_o(p;(x,\underline{h}(x)))=C\big\}\qquad  x \in [\underline x_u,\bar x]\, \end{align} 
 while we set $\hat \u(x; C) =m$ for any $x \in [0,\underline{x}_{u}] $. Note while $\u(x;C)\neq \hat \u(x; C)$ for any $x \in (\bar x_u,\bar x] $, we have $\u(x;C) =\hat \u(x;C)$ for any $x\in [0, \bar x_u]$.  
Then, we can check that  for any $x\in [\underline x_u, \bar x]$,   Equation \eqref{eq:convexcase1p} is satisfied,  which means
for any $x\in [\underline x_u, \bar x)$ and $C \in [0,1]$, when $\underline{\mathcal{H}}'(x)$ exists, we have 
    \begin{align} \label{eq:convexcase1p}  \frac{\partial \hat{\u}(x;C)}{\partial x} = \left\{ \begin{array}{ll}
         ((1-C)\frac{r_h}{r_{\ell}}+C)\underline{\mathcal{H}}'(x) &\quad  \mbox{if $x+\underline{\mathcal{H}}(x) \geq m$};\\
        (1-C)\frac{r_h}{r_{\ell}}\underline{\mathcal{H}}'(x)-C &\quad  \mbox{if $x+\underline{\mathcal{H}}(x) < m$}\,,\end{array} \right. \end{align} 
 where  $\underline{\mathcal{H}}(x)=\min\{\underline h(x),m\}$.        
By the definition of $\bar x_u$, as $\frac{\partial \hat{\u}(x;C)}{\partial x}\vert_{x= \bar x_u^{-}}<0$ and $\frac{\partial \hat{\u}(x;C)}{\partial x}\vert_{\bar x_u^{+}} \geq 0$, we obtain
\[
\inf_{x \in [\underline x, \bar x]} \hat{\u}(x;C) = \hat{\u}(\bar x_u;C),
\]
which is because the derivative of $\hat{\u}(x;C)$ is linear, and $\frac{\partial \hat{\u}(x;C)}{\partial x}\vert_{x= \bar x_u^{-}}<0$ and $\frac{\partial \hat{\u}(x;C)}{\partial x}\vert_{\bar x_u^{+}} \geq 0$ implies that $\frac{\partial \hat{\u}(x;C)}{\partial x} \geq 0$ for all $x \geq \bar x_u$. 
Therefore, by any $x \geq \bar x_u$, $\hat{\u}(x;C) \geq \u(x;C)$ as we force $\u(x;C)$ to be a constant value of  $\u(\bar x_u;C)$. The definition of $\hat{\u}(\cdot;C)$ implies  that $\CP_o(\hat{\u}(x;C);(x,\underline{h}(x))) = C$. By Lemma \ref{lem:pmonotone}, for $x \geq \bar x_u$, as $\hat{\u}(x;C) \geq \u(x;C)$, we have
\[
\CP_o(\u(x;C);(x,\underline{h}(x))) \geq \CP_o(\hat{\u}(x;C);(x,\underline{h}(x))) \geq C.
\]


\endproof

\subsection{Lemma \ref{lem:propertywll} and its Proof}

\begin{lemma} \label{lem:propertywll}
    For any $C \in [0,1]$, $\wll(x;C)$ is decreasing and is concave for $x \in [\underline{x},\tilde x_{\ell}]$, where $\tilde x_{\ell}=\inf\{\nex<x<\bar{x}: \wll(x;C)=0 \}$. For any $x \in [\underline{x},\bar{x}]$, $\wll(x;C)$ is continuously increasing in $C$.
\end{lemma}

\proof{Proof of Lemma \ref{lem:propertywll}}
By Equation \eqref{eq:ll}, as $\ll(x;C)$ is decreasing for any $C \in [0,1]$, we have $\wll(x;C)$ is also decreasing for any $C \in [0,1]$. Next, by Lemma \ref{lem:property_u_l}, we have $\ll(x;C)$ is concave for $x \in [\underline x, \bar x_{\ell}]$. Therefore, $\frac{\partial \ll(x^{+};C)}{\partial x}$ is non-increasing. As $\wll(x;C)=\ll(x;C)$ for $x \in [\underline x, \nex]$, we have $\frac{\partial \wll(x^{+};C)}{\partial x}$ is non-increasing for $x \in [\underline{x},\nex]$. By the definition of $\nex$, we have $\frac{\partial \wll(\nex^{-};C)}{\partial x} \geq \frac{\partial \wll(\nex^{+};C)}{\partial x} = -1$, and for $\nex<x\leq \tilde{x}$, we have $\wll(x;C)$ is a line with slope $-1$ and $\frac{\partial \wll(x^{+};C)}{\partial x}=-1$, which is non-increasing. Therefore, $\wll(x;C)$ is concave for $x \in [\underline{x},\tilde x_{\ell}]$.

Fix any $x_1 \in [\underline{x},\bar{x}]$, by Lemma \ref{lem:property_u_l}, we know that $\ll(x_1;C)$ is a continuous increasing function in $C$. If $x_1 \in [\underline{x},\nex]$, as $\wll(x_1;C)=\ll(x_1;C)$, we have $\wll(x_1;C)$ is also continuously increasing in $C$. Otherwise, define $\mathcal{L}(x;C)=(-x+\nex+\ll(\nex;C),0)^{+}$. Then, $\wll(x_1;C)=\max\{\mathcal{L}(x_1;C),\ll(x_1;C) \}$. As both $\mathcal{L}(x_1;C)$ and $\ll(x_1;C)$ are increasing continuous in $C$, we have $\wll(x_1;C)$ is increasing and continuous in $C$.

\endproof

\subsection{Lemma \ref{lem:inner} and its Proof }

\begin{lemma} \label{lem:inner}
    Define $\pbal(x)$ as a function which balances the compatible ratio of two points $(x, \underline h(x))$ and $(x, \bar h(x))$ for any $x\in [\underline x, \bar x]$; that is, 
\[\CP_o(\pbal(x);(x,\underline{h}(x)))=\CP_u(\pbal(x);(x,\bar{h}(x)))\,.\]
Then, $\pbal(x)$ exists for any $x\in [\underline x, \bar x]$, and 
\[\CP_o(\pbal(x);(x,\underline{h}(x)))=\CP_u(\pbal(x);(x,\bar{h}(x))) \geq \cstar\,,\]
where $\cstar$ is the maximum  consistent ratio among all PLAs given that the ML advice $\region$.

\end{lemma}

\proof{Proof of Lemma \ref{lem:inner}}
Fix any $x \in [\underline x,\bar x]$, define $f(p)=\CP_o(p;(x,\underline{h}(x)))-\CP_u(p;(x,\bar{h}(x)))$. As $\CP_o(p;(x,\underline{h}(x)))$ and $\CP_u(p;(x,\bar{h}(x)))$ are both continuous in $p$, we have $f(p)$ is continuous in $p$. Next, take $p=\min\{m,\underline{h}(x)\}$, then $\CP_o(p;(x,\underline{h}(x)))=1$ and $\CP_u(p;(x,\bar{h}(x))) \leq 1$, and hence we have $f(\min\{m,\underline{h}(x)\}) \geq 0$. Take $p=\min\{m,\bar{h}(x)\}$, then $\CP_u(p;(x,\bar{h}(x)))=1$ and $\CP_o(p;(x,\underline{h}(x))) \leq 1$, and hence we have $f(\min\{m,\bar{h}(x)\}) \leq 0$. Therefore, by mean value theorem, there must exist $p \in [\min\{m,\underline{h}(x)\},\min\{m,\bar{h}(x)\}]$ such that $f(p)=0$, i.e. $\CP_o(\pbal(x);(x,\underline{h}(x)))=\CP_u(\pbal(x);(x,\bar{h}(x)))$.

Next, we prove that $\CP_o(\pbal(x);(x,\underline{h}(x)))=\CP_u(\pbal(x);(x,\bar{h}(x))) \geq \cstar$ for any $ x \in [\underline x,\bar x]$. We prove by contradiction. Suppose that there exists $x_1 \in [\underline x,\bar x]$ such that $\CP_o(\pbal(x);(x,\underline{h}(x)))=\CP_u(\pbal(x);(x,\bar{h}(x))) < \cstar$, then, if we set any $p<\pbal(x)$, by Lemma \ref{lem:pmonotone}, we have
$$\CP_u(p;(x,\bar{h}(x))) <\CP_u(\pbal(x);(x,\bar{h}(x))) < \cstar\,.$$ 
Similarly, if we set any $p>\pbal(x)$, by Lemma \ref{lem:pmonotone}, we have $$\CP_o(p;(x,\underline{h}(x)))<\CP_o(\pbal(x);(x,\underline{h}(x)))<\cstar\,.$$ Therefore, there does not exist a PL function such that the consistent ratio is $\cstar$, which is a contradiction. This is because here $\cstar$ in  is defined as an upper bound on the consistent ratio of any PLA. 
\endproof

\subsection{Lemma \ref{lem:worstkeep} and its Proof}
\begin{lemma} \label{lem:worstkeep}
    For any fixed PL function $p$, we have $\CP_o(p;(x,0))$ is a decreasing function in $x$. That is, for any $x_1 \leq x_2$, we have
    \[
    \CP_o(p;(x_1,0)) \geq \CP_o(p;(x_2,0)).
    \]
\end{lemma}

\proof{Proof of Lemma \ref{lem:worstkeep}}
If $x_1 \leq x_2 \leq m$, by Equation \eqref{eq:CP_over}, we have
\[
\CP_o(p;A= (x,0))=\frac{0 \cdot r_h+\min\{x,m-p\}r_{\ell}}{0 \cdot r_h+\min\{x, m-0\}r_{\ell}} = \frac{\min\{x,m-p\}}{x}.
\]
Take any $x_1 \leq x_2$, if $\min\{x_2,m-p\}=x_2$, then $\min\{x_1,m-p\}=x_1$, and we have
\[
\frac{\min\{x_1,m-p\}}{x_1}=\frac{\min\{x_2,m-p\}}{x_2}=1.
\]
If $\min\{x_2,m-p\}=m-p$ and $\min\{x_1,m-p\}=m-p$, we have
\[
\frac{\min\{x_2,m-p\}}{x_2} = \frac{m-p}{x_2} \leq \frac{m-p}{x_1} = \frac{\min\{x_1,m-p\}}{x_1}. 
\]
If $\min\{x_2,m-p\}=m-p$ and $\min\{x_1,m-p\}=x_1$, we have
\[
\frac{\min\{x_2,m-p\}}{x_2} = \frac{m-p}{x_2} \leq 1= \frac{x_1}{x_1}=\frac{\min\{x_1,m-p\}}{x_1}.
\]
Therefore, we have for any $x_1 \leq x_2 \leq m$, we have
    \[
    \CP_o(p;(x_1,0)) \geq \CP_o(p;(x_2,0)).
    \]

Next, for $x_1 \leq m \leq x_2$, by Lemma \ref{lem:mout}, we have
\[
\CP_o(p;(x_1,0)) \geq \CP_o(p;(m,0)) = \CP_o(p;(x_2,0)).
\]
For $m \leq x_1 \leq x_2$, again, by Lemma \ref{lem:mout}, we have
\[
\CP_o(p;(m,0)) =\CP_o(p;(x_1,0)) = \CP_o(p;(x_2,0)).
\]
\endproof

\subsection{Lemma \ref{lem:specialpoints} and its Proof}
\begin{lemma} \label{lem:specialpoints}
    For any $C \in (0,1)$. Let $\xmin$ be defined in Equation \eqref{eq:xmin}, we have \[\min\{\xmin,m-\u(\xmin;C)\}=m-\u(\xmin;C).\] Recall that $H$ is the point in set $\widebar{\region}$ that has the highest low-reward demand, where $\widebar{\region} = \{(x,y) \in \region : y = \sup_{(x',y') \in \region } \min\{y',m\}  \} $ Further,  we have \[\min\{x_H,m-\ll(x_H;\cstar)\}=m-\ll(x_H;\cstar).\]    
    Let $\nex=\sup\{x \in [x_H,\bar{x}]: \frac{\partial \ll(x^{-};C)}{\partial x} \le -1 \}$. We have \[\nex \geq m-\ll(\nex;C)\]. 
\end{lemma}

\proof{Proof of Lemma \ref{lem:specialpoints}}
We prove by contradiction. If $\min\{\xmin,m-\u(\xmin;C)\}=\xmin$, then by Equation \eqref{eq:CP_over}, we have 
\[
\CP_o(\u(\xmin;C);(\xmin,\underline{h}(\xmin))) = \frac{\min\{\underline{h}(\xmin),m\}r_h+\min\{m-\u(\xmin;C) ,\xmin,\}r_{\ell}}{\min\{\underline{h}(\xmin),m\}r_h+\min\{(m-\underline{h}(\xmin))^{+},\xmin\}r_{\ell}}.
\]
Given that $\min\{\xmin,m-\u(\xmin;C)\}=\xmin$, as $\u(\xmin;C) \geq \min\{\underline{h}(\xmin),m\}$, we have $\min\{(m-\underline{h}(\xmin))^{+},\xmin\}=\xmin$, and we have
\[
\CP_o(\u(\xmin;C);(\xmin,\underline{h}(\xmin))) = \frac{\min\{\underline{h}(\xmin),m\}r_h+\xmin r_{\ell}}{\min\{\underline{h}(\xmin),m\}r_h+\xmin r_{\ell}} =1>C,
\]
which contradicts to the definition of $\u(\cdot;C)$.

Next, to show the second statement, We still prove by contradiction. If $\min\{x_H,m-\ll(x_H;\cstar)\}=x_H$, then by Equation \eqref{eq:CP_under}, we have 
\[
\CP_u(\ll(x_H;\cstar);H)=\frac{\max\{\ll(x_H;\cstar), \min\{y_H,(m-x_H)^{+}\} \} r_h+\min\{x_H,m-\ll(x_H;\cstar)\}r_{\ell}}{y_Hr_h+\min\{x_H,m-y_H\}r_{\ell}}.
\]
Given that $\min\{x_H,m-\ll(x_H;\cstar)\}=x_H$,  we have
\[
\CP_u(\ll(x_H;\cstar);H) = \frac{\min\{y_H,(m-x_H)^{+}\}r_h+x_H r_{\ell}}{y_Hr_h+\min\{x_H,m-y_H\}r_{\ell}} = \CP_u(0;H),
\]
which implies that $x_H \geq \bar x_{\ell}$, and this contradicts to the definition of $\bar x_{\ell}$ since $x_H \leq \nex < \bar x_{\ell}$.

Finally, we show that $\nex \geq m-\ll(\nex;C)$. We still prove by contradiction. Suppose that $\nex<m-\ll(\nex;C)$. By Equation \eqref{eq:CP_under}, we have
\[
\CP_u(\ll(\nex;C);(\nex,\bar{h}(\nex))) =  \frac{\max\{\ll(\nex;C), \min\{\bar{h}(\nex),(m-\nex)^{+}\} \} r_h+\min\{\nex,m-\ll(\nex;C)\}r_{\ell}}{\min\{ \bar{h}(\nex),m\}r_h+\min\{\nex, (m-\bar{h}(\nex))^{+}\}r_{\ell}}.
\]
If we have  $\nex<m-\ll(\nex;C)$, then we have
\[
C=\CP_u(\ll(\nex;C);(\nex,\bar{h}(\nex))) = \frac{\min\{\bar{h}(\nex),(m-\nex)^{+}\}r_h+\nex r_{\ell}}{\min\{ \bar{h}(\nex),m\}r_h+\min\{\nex, (m-\bar{h}(\nex))^{+}\}r_{\ell}} = \CP_u(0;(\nex,\bar{h}(\nex))),
\]
which implies that $\nex \geq \bar x_{\ell}$, which is a contradiction.

\endproof

\subsection{Lemma \ref{lem:munique} and its Proof}

\begin{lemma} \label{lem:munique}
    Let $\cstar$ be the optimal consistent ratio of $\region$. Take any $C>\cstar$, then $\u(\hat{x};C) = \u(\hat{x};\cstar)$ only if both of them are equal to $m$.
\end{lemma}

\proof{Proof of Lemma \ref{lem:munique}}
If $\u(\hat{x};\cstar)<m$, then by definition of $\underline x_u$, we have $\hat{x} \in [\underline x_u, \bar x]$. For $\hat x \in [\underline x_u, \bar x_u]$, by Equation \eqref{eq:u}, we have $\u(\hat x; \cstar)= \sup\big \{p\in [0,m]: \CP_o(p;(\hat x,\underline{h}(\hat x)))=\cstar\big\}$. For $C>\cstar$, 
if $\u(\hat{x};C) = \u(\hat{x};\cstar)$, then we have 
\begin{align*}
\sup\big \{p\in [0,m]: \CP_o(p;(\hat x,\underline{h}(\hat x)))=C\big\}&=\u(\hat x; C)=\u(\hat x; \cstar)\\&= \sup\big \{p\in [0,m]: \CP_o(p;(\hat x,\underline{h}(\hat x)))=\cstar\big\},
\end{align*}
which implies that $\CP_o(\u(\hat x; C);(x,\underline{h}(\hat x)))=C$ and $\CP_o(\u(\hat x; C);(x,\underline{h}(\hat x)))=\cstar$, which is a contradiction. For $\hat x \in [\bar x_u,\bar x]$, we have $\u(\hat x; \cstar)=\u(\bar x_u;\cstar)$ and  $\u(\hat x; C)=\u(\bar x_u;C)$. By the similar statement above, we have 
\begin{align*}
\sup\big \{p\in [0,m]: \CP_o(p;(\hat x,\underline{h}(\hat x)))&=C\big\}=\u(\bar x_u; C)=\u(\bar x_u; \cstar)\\&= \sup\big \{p\in [0,m]: \CP_o(p;(\hat x,\underline{h}(\hat x)))=\cstar\big\},
\end{align*}
which implies that $\CP_o(\u(\bar x_u; C);(\hat x,\underline{h}(\hat x)))=C$ and $\CP_o(\u(\bar x_u; C);(\hat x,\underline{h}(\hat x)))=\cstar$, which is a contradiction.
\endproof

\subsection{Lemma \ref{lem:barxusmall} and its Proof} 
\begin{lemma} \label{lem:barxusmall}
    Recall that $\u(\cdot;C)$ is defined in Equation \eqref{eq:u}, and $\bar x_u$ is defined in Equation \eqref{eq:xmin}. We have $\u(x;C)$ gets its minimum value at $x \in [\bar x_u, \bar x]$.
\end{lemma}

\proof{Proof of Lemma \ref{lem:barxusmall}}
By the first property of Lemma \ref{lem:property_u_l}, we have 
\begin{align}   \frac{\partial \u(x;C)}{\partial x} = \left\{ \begin{array}{ll}
         ((1-C)\frac{r_h}{r_{\ell}}+C)\underline{\mathcal{H}}'(x) &\quad  \mbox{if $x+\underline{\mathcal{H}}(x) \geq m$};\\
        (1-C)\frac{r_h}{r_{\ell}}\underline{\mathcal{H}}'(x)-C &\quad  \mbox{if $x+\underline{\mathcal{H}}(x) < m$}.\end{array} \right. \end{align} 
Recall that \begin{align}  \bar x_u = \left\{ \begin{array}{ll}
         x_L & \quad \mbox{if $x_L+y_L \geq m$};\\
        \sup\{x \in [x_L,\bar{x}]: (1-C)\frac{r_h}{r_{\ell}}\underline{\mathcal{H}}'(x^{-})-C < 0 \} & \quad  \mbox{Otherwise}.\end{array} \right. \end{align}

Therefore, if $x_L+y_L \geq m$, we have $\bar x_u=x_L$. As $L$ is the lowest point, which implies that $\underline{h}'(x_L^{-})<0$ and $\underline{h}'(x_L^{+})>0$. Recall that $\underline{\mathcal{H}}(x)=\min\{\underline{h},m\}$, we have $\underline{\mathcal{H}}'(x_L^{-}) \leq 0$ and $\underline{\mathcal{H}}'(x_L^{+}) \geq 0$. For $x+\underline{\mathcal{H}}(x) \geq m$, as we have $\frac{\partial \u(x;C)}{\partial x} = ((1-C)\frac{r_h}{r_{\ell}}+C)\underline{\mathcal{H}}'(x)$, we have $\frac{\partial \u(x;C)}{\partial x}<0$ for $x < x_L$ and $\frac{\partial \u(x;C)}{\partial x}>0$ for $x>x_L$, which implies that $\bar x_u = x_L$ is the point such that $\u(x;C)$ achieves its lowest value. As we force $\u(x;C)=\u(\bar x_u;C)$, we have $\u(x;C)$ gets its minimum value at $x \in [\bar x_u, \bar x]$. 

In other cases, by taking $\xmin = \sup\{x \in [x_L,\bar{x}]: (1-C)\frac{r_h}{r_{\ell}}\underline{\mathcal{H}}'(x^{-})-C < 0 \}$, we have $\frac{\partial \u(x;C)}{\partial x}<0$ for $x < \xmin$ and $\frac{\partial \u(x;C)}{\partial x}>0$ for $x>\xmin$, which implies that $\bar x_u$ is the point such that $\u(x;C)$ achieves its lowest value. As we force $\u(x;C)=\u(\bar x_u;C)$, we have $\u(x;C)$ gets its minimum value at $x \in [\bar x_u, \bar x]$.

\endproof

\subsection{Lemma \ref{lem:xminvertex} and its Proof}

\begin{lemma} \label{lem:xminvertex}
    Recall that $\xmin$ is defined in Equation \eqref{eq:xmin}, and $\V$ is the $x$-vertices set of a polyhedron $\region$ plus all elements of $\mathcal{R}_0$, where $\mathcal R_{0} = \{(x, \underline h(x): x\in [\underline x, \bar x]\} \cap \{(x,y): x+y =m\}$, then $\xmin \in \V$.
\end{lemma}

\proof{Proof of Lemma \ref{lem:xminvertex}}
Recall that 
\begin{align}  \bar x_u = \left\{ \begin{array}{ll}
         x_L & \quad \mbox{if $x_L+y_L \geq m$};\\
        \sup\{x \in [x_L,\bar{x}]: (1-C)\frac{r_h}{r_{\ell}}\underline{\mathcal{H}}'(x^{-})-C < 0 \} & \quad  \mbox{Otherwise}\,,\end{array} \right. \end{align} 
That is,  $\bar x_u$ equals to either $x_L$ or $\sup\{x \in [x_L,\bar{x}]: (1-C)\frac{r_h}{r_{\ell}}\underline{\mathcal{H}}'(x^{-})-C < 0 \}$. As $L$ is a vertex of $\region$, we have $x_L \in \V$.

Then, we claim that $\xmin=\sup\{x \in [x_L,\bar{x}]: (1-C)\frac{r_h}{r_{\ell}}\underline{\mathcal{H}}'(x^{-})-C < 0 \} \in \V$. As $\region$ is a polyhedron, we have $\underline{h}(\cdot)$ is a piecewise linear function. As $\underline{\mathcal{H}}(x)=\min\{\underline{h}(x),m\}$, we have $\underline{\mathcal{H}}(\cdot)$ is also a piecewise linear function. If $\xmin$ is not a x-vertex, then there exists $\epsilon>0$ such that $\underline{\mathcal{H}}'(x^{-})=\underline{\mathcal{H}}'((x+\epsilon)^{-})$, and we have 
\[
(1-C)\frac{r_h}{r_{\ell}}\underline{\mathcal{H}}'((x+\epsilon)^{-})-C < 0,
\]
which contradicts  the definition of $\xmin$.

\endproof

\subsection{Lemma \ref{lem:hatxinter} and its Proof}

\begin{lemma} \label{lem:hatxinter}
    If $\wll(\hat{x};\cstar)=\u(\hat{x};\cstar)$ for $x_H \leq \hat{x} \leq \nex$, we have $\hat{x} \in [\underline x_u, \bar x_u]$.
\end{lemma}

\proof{Proof of Lemma \ref{lem:hatxinter}}
As $\u(\hat{x};\cstar)=\wll(\hat{x};\cstar)<m$, we have $\hat{x} > \underline{x_u}$. Then, we show that $\hat{x} \leq \bar x_u$ by contradiction. If $\hat{x}>\bar x_u$, as $u(\cdot;\cstar)$ is constant between $[\bar x_u,\hat{x}]$ and $\wll(\cdot;\cstar)$ is a decreasing function for $x \in [x_H,\hat{x}]$, we have there exists $\epsilon>0$ such that $\wll(\hat{x}-\epsilon;\cstar)>\u(\hat{x}-\epsilon;\cstar)$, which is a contradiction. Therefore, we have $\hat{x} \in [\underline{x_u}, \bar x_u]$,  

\endproof

\subsection{Lemma \ref{lem:nexvertex} and its Proof}
\begin{lemma} \label{lem:nexvertex}
    Recall that $\nex=\sup\{x \in [x_H,\bar{x}]: \frac{\partial \ll(x^{-};\cstar)}{\partial x} \le -1 \}$. Given a polyhedron $\region$, recall that $\V$ is the $x$-vertices set of $\region$ plus all elements in $\mathcal R_0$, then $\nex \in \V$.
\end{lemma}

\proof{Proof of Lemma \ref{lem:nexvertex}}
As $\region$ is a polyhedron, we have $\bar{h}(\cdot)$ is a piecewise linear function. By the second property of Lemma \ref{lem:property_u_l}, we have $\ll(\cdot;C)$ is also a piecewise linear function. Then, we prove by contradiction. Suppose that $\nex \notin \V$. Then, there exists $\epsilon>0$ such that $\frac{\partial \ll(\nex^{-};\cstar)}{\partial x}=\frac{\partial \ll((\nex+\epsilon)^{-};\cstar)}{\partial x}$. Then, we have $\frac{\partial \ll((\nex+\epsilon)^{-};\cstar)}{\partial x} \le -1 $, which contradicts to the definition of $\nex$.

\endproof

\subsection{Lemma \ref{lem:piecewise} and its Proof}

\begin{lemma} \label{lem:piecewise}
    Suppose that $\region$ is a polyhedron. Recall that $\u(\cdot;C)$ and $\wll(\cdot;C)$ are defined in Equations \eqref{eq:u} and \eqref{eq:ll} respectively. Then, we have $\u(\cdot;C)$ and $\wll(\cdot;C)$ are piecewise linear functions and all of their $x$-vertices are a subset of $\V$. In addition,    the elements of $\mathcal{R}_0$ are  $x$-vertices of $\u(\cdot;C)$, where $\mathcal R_{0} = \{(x, \underline h(x): x\in [\underline x, \bar x]\} \cap \{(x,y): x+y =m\}$. 
\end{lemma}

\proof{Proof of Lemma \ref{lem:piecewise}}
As $\region$ is a polyhedron, we have both $\bar{h}(\cdot)$ and $\underline{h}(\cdot)$ are piecewise linear. Since $\overline{\mathcal{H}}(x) = \min\{m,\bar h(x)\}$ and $\underline{\mathcal{H}}(x)=\min\{m,\underline h(x)\}$, we have both $\overline{\mathcal{H}}(\cdot)$ and $\underline{\mathcal{H}}(\cdot)$ are piecewise linear and both $\overline{\mathcal{H}}'(\cdot)$ and $\underline{\mathcal{H}}'(\cdot)$ are piecewise constant. 

Recall that by Lemma \ref{lem:property_u_l}, we have for any $x\in (x_H, \bar x_l)$ and $C \in [0,1]$, 
\[
\frac{\partial \ll(x; C)}{\partial x}=C \overline{\mathcal{H}}'(x).
\]
Therefore, $\frac{\partial \ll(x; C)}{\partial x}$ is piecewise constant and $\ll(x; C)$ is piecewise linear. In addition, we can find that the x-vertices of $\ll(\cdot;C)$ are also ones of $\bar{h}(\cdot)$. By Equation \eqref{eq:ll} and Lemma \ref{lem:nexvertex}, we have $\wll(\cdot;C)$ is also piecewise linear with $x$-vertices belong to ones of $\bar{h}(\cdot)$.

Recall that by Lemma \ref{lem:property_u_l}, we have for any $x\in (\underline x_u, \bar x_u)$ and $C \in [0,1]$,
\begin{align}  \frac{\partial \u(x;C)}{\partial x} = \left\{ \begin{array}{ll}
         ((1-C)\frac{r_h}{r_{\ell}}+C)\underline{\mathcal{H}}'(x) &\quad  \mbox{if $x+\underline{\mathcal{H}}(x) \geq m$};\\
        (1-C)\frac{r_h}{r_{\ell}}\underline{\mathcal{H}}'(x)-C &\quad  \mbox{if $x+\underline{\mathcal{H}}(x) < m$}.\end{array} \right. \end{align} 

Therefore, we have $\frac{\partial \u(x; C)}{\partial x}$ is piecewise constant and $\u(x; C)$ is piecewise linear. In addition, the x-vertices of $\u(x; C)$ for $x+\underline{\mathcal{H}}(x) \geq m$ is a subset to x-vertices of $\overline{\mathcal{H}}(x)$. The x-vertices of $\u(x; C)$ for $x+\underline{\mathcal{H}}(x) < m$ is also a subset to x-vertices of $\overline{\mathcal{H}}(x)$. Moreover, any point $x$ such that $x+\underline{\mathcal{H}}(x) = m$, which is we defined as an element of $\mathcal R_0$, is also a vertex, and by our definition, $\V$ contains all elements of $\mathcal R_0$.

\endproof

\subsection{Geometric Lemmas}

\begin{lemma} \label{lem:leftconcave}
Suppose that we have a line $\mathcal{L}(x)$ with any negative slope on an interval $\mathcal{I}$. $f(x)$ is a concave decreasing function which intersects $\mathcal{L}(x)$ at $(x_0,f(x_0))$. If there exists $x_1<x_0$ such that $f(x_1)>f(x_0)$, then for all $x>x_0$, we have $f(x)<\mathcal{L}(x)$.
\end{lemma}

\proof{Proof of Lemma \ref{lem:leftconcave}}
Suppose that there exists $x_2>x_0$ such that $f(x_2) \geq \mathcal{L}(x_2)$. As $(x_1,f(x_1)$, $(x_2,f(x_2))$ are both above the line $\mathcal{L}$, if we connect $(x_1,f(x_1)$, $(x_2,f(x_2))$ by a line $\mathcal{L}_1(x)$, we have $\mathcal{L}_1(x)>\mathcal{L}(x)$ for $x \in [x_1,x_2]$. Therefore, $\mathcal{L}_1(x_0)>\mathcal{L}(x_0)=f(x_0)$ since $x_0 \in [x_1,x_2]$.

However, as $f(x)$ is a concave function, if we take $x_1<x_2$ and connect $(x_1,f(x_1)$, $(x_2,f(x_2))$ by a line $\mathcal{L}_1(x)$, we should always have $f(x) \geq \mathcal{L}_1(x)$, which is a contradiction to $\mathcal{L}_1(x_0)>\mathcal{L}(x_0)=f(x_0)$.

\Halmos
\endproof

\begin{lemma} \label{lem:worstvertexpl}
    Let $f(x)$ and $g(x)$ be two piecewise linear functions defined on an interval $I$, with $f(x) \geq g(x)$ for any $x \in I$. If $\{x:f(x)=g(x)\}$ is not empty, we have there exists $x_0 \in \V$, which is an $x$-vertex of either $f(x)$ or $g(x)$, such that $f(x_0)=g(x_0)$.
\end{lemma}

\proof{Proof of Lemma \ref{lem:worstvertexpl}}
We prove by contradiction. Suppose that any $x_1 \in \{x:f(x)=g(x)\}$ is not an $x$-vertex of either $f(x)$ or $g(x)$. Then, we have $f'(x_1^{+})=f'(x_1^{-})$ and $g'(x_1^{+})=g'(x_1^{-})$. As $f(x) \geq g(x)$ everywhere and $f(x_1)=g(x_1)$, we have $f'(x_1^{-}) \leq g'(x_1^{-})$. 

If $f'(x_1^{-}) < g'(x_1^{-})$, we have $f'(x_1^{+}) < g'(x_1^{+})$ and by $f(x_1)=g(x_1)$, we have there exists $\epsilon>0$ such that $f(x_1+\epsilon)<g(x_1+\epsilon)$, which is a contradiction to $f(x) \geq g(x)$ everywhere.

If $f'(x_1^{-}) = g'(x_1^{-})$, we define $x_2$ as $x_2=\inf\{x: f(x)=g(x) \text{ for any $x \in [x,x_1]$}\}$. Then, we have $f'(x_2^{+}) = g'(x_2^{+})$. By the definition of $x_2$, we know that for any $\epsilon_1>0$, we have $f(x_2-\epsilon_1)>g(x_2-\epsilon_1)$, which implies that $f(x_2^{-}) \neq g(x_2^{-})$. As $f'(x_2^{+}) = g'(x_2^{+})$, we have either $f(x_2^{-}) \neq f'(x_2^{+})$ or $g(x_2^{-}) \neq g'(x_2^{+})$, which implies that $x_2$ is an $x$-vertex for $f$ or $g$, which is a contradiction.

\Halmos

\endproof

\section{Supplementary Materials for Section \ref{sec:simulations}} \label{append:simulations}

\subsection{Stochastic Arrivals}
{\color{black}Tables~\ref{tab:boxS} and~\ref{tab:ellS} present the results under different forms of ML advice, using the demand models defined in Equations~\eqref{distri1} and~\eqref{distri2}, as well as under the stochastic order assumption. Consistent with our earlier findings, the average compatible ratio remains relatively stable as \( n \) increases, while the worst-case ratio shows more noticeable improvement. Furthermore, we observe that the best performance—both in terms of average and worst-case ratios—is achieved when using either the polyhedron advice with \( \epsilon = 0.3 \) or the box/ellipsoid  advice.
}

\begin{table}[htb]
\caption{Results under ML advice (using the demand model in Equation \eqref{distri1}) and stochastic order. The standard error of all reported values is less than $0.003$. In each column, the top two values of average CP are highlighted in black, and the top two values of worst-case CP are highlighted in red. Additionally, the highest average and worst-case CP across all settings for each value of $n$ are shown in bold.}
\centering
\footnotesize
\begin{tabular}{ll cccc cccc}
\toprule
 & & \multicolumn{3}{c}{\( n = 10 \)} & \multicolumn{3}{c}{\( n = 100 \)} \\
\cmidrule(lr){3-5} \cmidrule(lr){6-8}
Algorithm & Metric & $\cstar$ & $0.9 \cdot \cstar$ & $0.8 \cdot \cstar$ & $\cstar$ & $0.9 \cdot \cstar$ & $0.8 \cdot \cstar$ \\
\midrule
\multirow{2}{*}{Alg. \ref{alg:trans} (Polyhedron, $\epsilon = 0.1, \theta = 10$)} 
    & Avg. CP & 0.852 & 0.840 & 0.831 & 0.855 & 0.842 & 0.836 \\
    & Worst CP & 0.639 & 0.683 & 0.687 & 0.646 & 0.688 & 0.693 \\
\midrule
\multirow{2}{*}{Alg. \ref{alg:trans} (Polyhedron, $\epsilon = 0.1, \theta = 30$)} 
    & Avg. CP & 0.855 & 0.841 & 0.832 & 0.854 & 0.844 & 0.838 \\
    & Worst CP & 0.637 & 0.683 & 0.686 & 0.647 & 0.687 & 0.691 \\
\midrule
\multirow{2}{*}{Alg. \ref{alg:trans} (Polyhedron, $\epsilon = 0.3, \theta = 10$)} 
    & Avg. CP & \cellcolor{bestavg}{0.941} & {0.927} & \cellcolor{bestavg}{0.900} & {0.949} & {0.935} & {0.902} \\
    & Worst CP & \cellcolor{bestworst}{0.674} & \cellcolor{bestworst}{0.730} & \cellcolor{bestworst}{0.737} & {0.681} & \cellcolor{bestworst}{0.736} & {0.744} \\
\midrule
\multirow{2}{*}{Alg. \ref{alg:trans} (Polyhedron, $\epsilon = 0.3, \theta = 30$)} 
    & Avg. CP & {0.940} & \cellcolor{bestavg}{0.928} & \cellcolor{bestavg}{0.900} & {0.949} & {0.937} & \cellcolor{bestavg}{0.903} \\
    & Worst CP & {0.672} & {0.729} & \cellcolor{bestworst}{0.737} & \cellcolor{bestworst}{0.682} & {0.735} & \cellcolor{bestworst}{0.745} \\
\midrule
\multirow{2}{*}{Alg. \ref{alg:trans} (Box Advice, $z = 90\%$)} 
    & Avg. CP & \cellcolor{bestavg}{\textbf{0.943}} & \cellcolor{bestavg}{0.929} & \cellcolor{bestavg}{0.901} & \cellcolor{bestavg}{\textbf{0.953}} & \cellcolor{bestavg}{0.939} & \cellcolor{bestavg}{0.905} \\
    & Worst CP & \cellcolor{bestworst}{0.676} & \cellcolor{bestworst}{0.733} & \cellcolor{bestworst}{\textbf{0.741}} & \cellcolor{bestworst}{0.685} & \cellcolor{bestworst}{0.744} & \cellcolor{bestworst}{\textbf{0.746}} \\
\midrule
\multirow{2}{*}{Alg. \ref{alg:trans} (Box Advice, $z = 80\%$)} 
    & Avg. CP & {0.936} & 0.922 & 0.893 & \cellcolor{bestavg}{0.950} & \cellcolor{bestavg}{0.943} & 0.907 \\
    & Worst CP & 0.667 & 0.726 & 0.732 & 0.687 & 0.738 & 0.744 \\
\midrule
\multirow{2}{*}{Point ML Advice} 
    & Avg. CP & 0.910 & 0.853 & 0.814 & 0.918 & 0.860 & 0.819 \\
    & Worst CP & 0.555 & 0.584 & 0.673 & 0.559 & 0.590 & 0.676 \\
\midrule
\multirow{2}{*}{BQ Benchmark} 
    & Avg. CP & 0.823 & - & - & 0.823 & - & - \\
    & Worst CP & 0.674 & - & - & 0.671 & - & - \\
\midrule
\multirow{2}{*}{PR Benchmark ($z = 90\%$)} 
    & Avg. CP & {0.933} & - & - & {0.941} & - & - \\
    & Worst CP & 0.666 & - & - & 0.688 & - & - \\
\midrule
\multirow{2}{*}{PR Benchmark ($z = 80\%$)} 
    & Avg. CP & 0.926 & - & - & 0.935 & - & - \\
    & Worst CP & 0.658 & - & - & 0.668 & - & - \\
\bottomrule
\end{tabular}
\label{tab:boxS}
\end{table}

\begin{table}[htb]
\caption{Results under ML advice (using the demand model in Equation \eqref{distri2}) and stochastic order. The standard error of all reported values is less than $0.005$. In each column, the top two values of average CP are highlighted in black, and the top two values of worst-case CP are highlighted in red. Additionally, the highest average and worst-case CP across all settings for each value of $n$ are shown in bold.}
\centering
\footnotesize
\begin{tabular}{ll cccc cccc}
\toprule
 & & \multicolumn{3}{c}{\( n = 10 \)} & \multicolumn{3}{c}{\( n = 100 \)} \\
\cmidrule(lr){3-5} \cmidrule(lr){6-8}
Algorithm & Metric & $\cstar$ & $0.9 \cdot \cstar$ & $0.8 \cdot \cstar$ & $\cstar$ & $0.9 \cdot \cstar$ & $0.8 \cdot \cstar$ \\
\midrule
\multirow{2}{*}{Alg. \ref{alg:trans} (Polyhedron, $\epsilon = 0.1, \theta = 10$)} 
    & Avg. CP & 0.858 & 0.850 & 0.842 & 0.867 & 0.857 & 0.848 \\
    & Worst CP & 0.630 & 0.645 & 0.650 & 0.634 & 0.648 & 0.655 \\
\midrule
\multirow{2}{*}{Alg. \ref{alg:trans} (Polyhedron, $\epsilon = 0.1, \theta = 30$)} 
    & Avg. CP & 0.859 & 0.852 & 0.842 & 0.868 & 0.859 & 0.849 \\
    & Worst CP & 0.628 & 0.645 & 0.649 & 0.633 & 0.647 & 0.654 \\
\midrule
\multirow{2}{*}{Alg. \ref{alg:trans} (Polyhedron, $\epsilon = 0.3, \theta = 10$)} 
    & Avg. CP & \cellcolor{bestavg}{0.939} & {0.910} & \cellcolor{bestavg}{0.883} & {0.942} & {0.915} & 0.890 \\
    & Worst CP & 0.649 & 0.666 & 0.669 & 0.653 & 0.672 & 0.675 \\
\midrule
\multirow{2}{*}{Alg. \ref{alg:trans} (Polyhedron, $\epsilon = 0.3, \theta = 30$)} 
    & Avg. CP & {0.938} & 0.908 & 0.882 & {0.941} & {0.915} & 0.889 \\
    & Worst CP & 0.650 & 0.665 & 0.668 & 0.652 & 0.671 & 0.674 \\
\midrule
\multirow{2}{*}{Alg. \ref{alg:trans} (Ellipsoid Advice, $z = 90\%$)} 
    & Avg. CP & \cellcolor{bestavg}{\textbf{0.947}} & \cellcolor{bestavg}{0.922} & \cellcolor{bestavg}{0.891} & \cellcolor{bestavg}{\textbf{0.955}} & \cellcolor{bestavg}{0.939} & \cellcolor{bestavg}{0.912} \\
    & Worst CP & \cellcolor{bestworst}{0.679} & \cellcolor{bestworst}{0.698} & \cellcolor{bestworst}{\textbf{0.703}} & \cellcolor{bestworst}{0.692} & \cellcolor{bestworst}{0.716} & \cellcolor{bestworst}{\textbf{0.724}} \\
\midrule
\multirow{2}{*}{Alg. \ref{alg:trans} (Ellipsoid Advice, $z = 80\%$)} 
    & Avg. CP & 0.934 & \cellcolor{bestavg}{0.918} & 0.880 & \cellcolor{bestavg}{0.952} & \cellcolor{bestavg}{0.931} & \cellcolor{bestavg}{0.895} \\
    & Worst CP & \cellcolor{bestworst}{0.664} & \cellcolor{bestworst}{0.683} & \cellcolor{bestworst}{0.692} & \cellcolor{bestworst}{0.672} & \cellcolor{bestworst}{0.696} & \cellcolor{bestworst}{0.703} \\
\midrule
\multirow{2}{*}{Point ML Advice} 
    & Avg. CP & 0.909 & 0.852 & 0.819 & 0.914 & 0.862 & 0.823 \\
    & Worst CP & 0.551 & 0.583 & 0.639 & 0.559 & 0.599 & 0.655 \\
\midrule
\multirow{2}{*}{BQ Benchmark} 
    & Avg. CP & 0.834 & - & - & 0.834 & - & - \\
    & Worst CP & 0.670 & - & - & 0.670 & - & - \\
\bottomrule
\end{tabular}
\label{tab:ellS}
\end{table}
{\color{black}
\subsection{Additional Details on the Neural Network Method}\label{appendix:NN}
This appendix presents the figures associated with the NN-based ML advice.  
Following the approach of \cite{pearce2018high}, the NN is trained to produce a prediction interval for each customer type. Figure~\ref{fig:adviceneural} displays results for four training-sample sizes: \(n = 10, 100, 1000,\) and \(2000\). Each plot reports the training loss and shows a histogram of 500 test samples, together with dashed lines indicating the learned prediction interval (PI) for the high-reward type.

The figure highlights a clear trend: when the training set is very small (\(n < 1000\)), the prediction intervals fluctuate substantially, reflecting underfitting and sensitivity to sampling noise. As \(n\) increases, the intervals become more stable; for \(n \ge 2000\), we observe that the endpoints converge to a consistent range. These visual patterns align with the quantitative results in Table~\ref{tab:new_results_neu} and help explain why larger sample sizes are important for reliable NN-based ML advice.

}

\begin{figure}[ht]
    \centering
    
    \begin{subfigure}{0.48\textwidth}
        \centering
        \includegraphics[width=\linewidth]{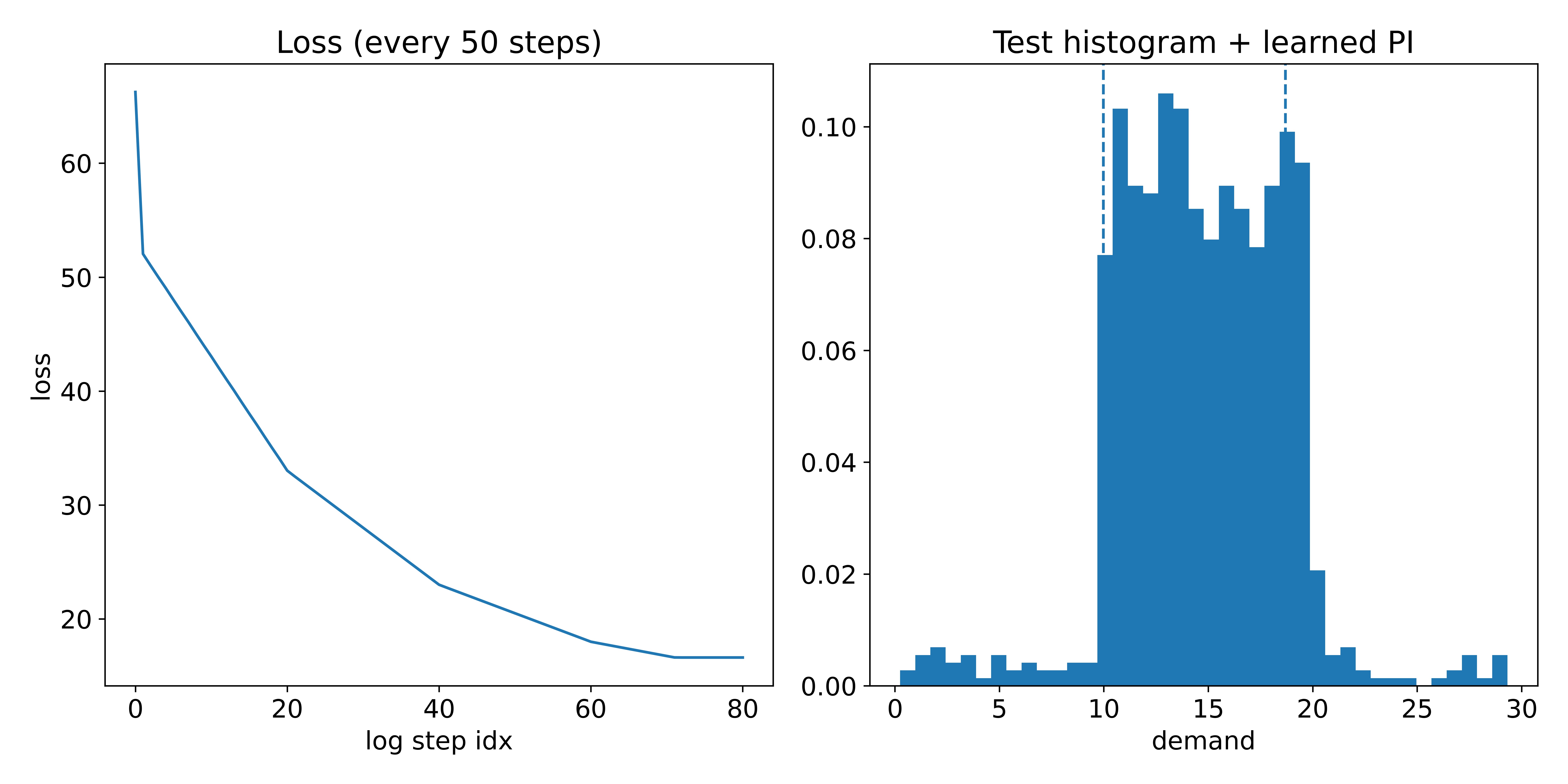}
        \caption{$n = 10$}
    \end{subfigure}
    \hfill
    \begin{subfigure}{0.48\textwidth}
        \centering
        \includegraphics[width=\linewidth]{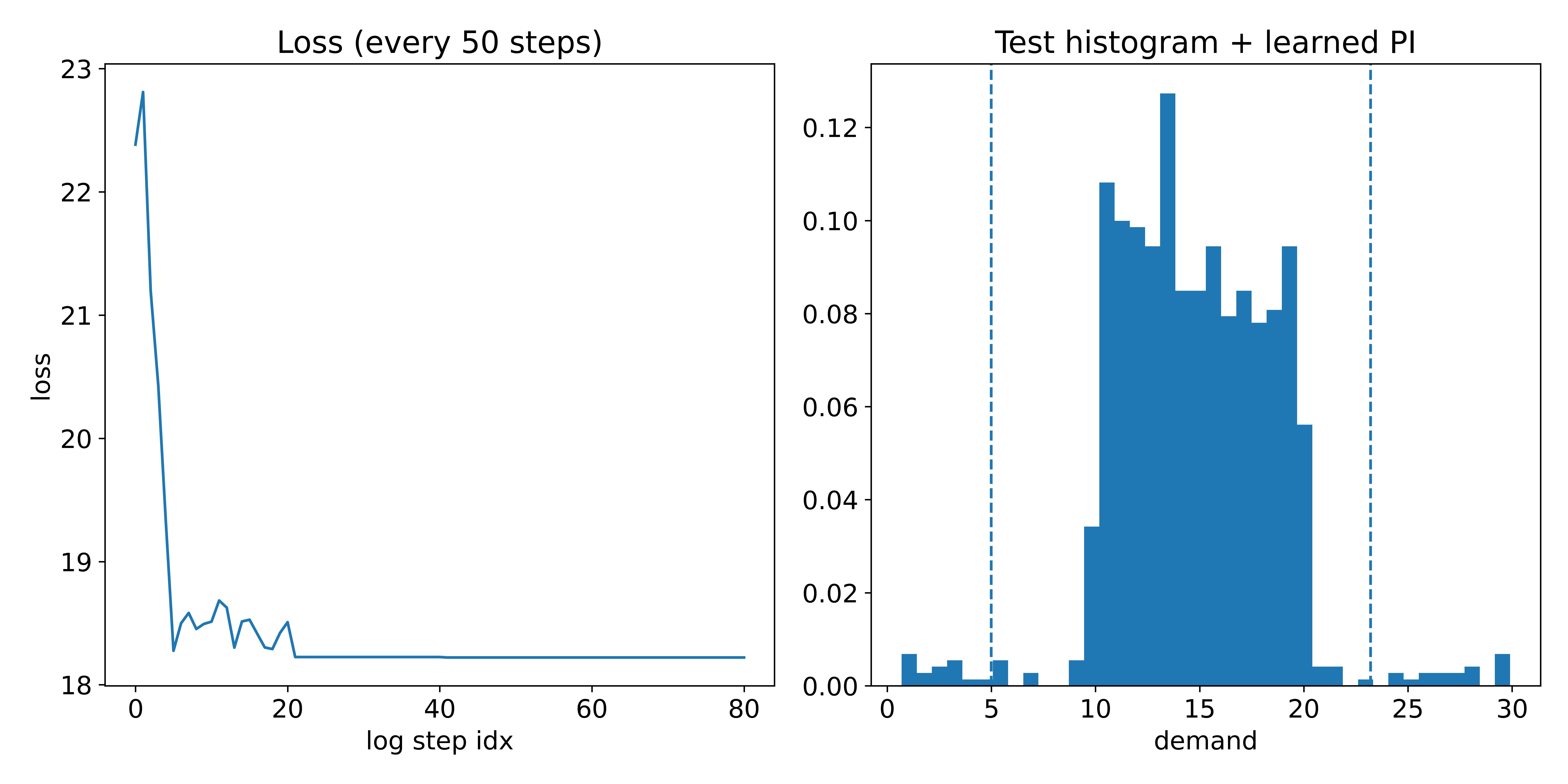}
        \caption{$n = 100$}
    \end{subfigure}

    \begin{subfigure}{0.48\textwidth}
        \centering
        \includegraphics[width=\linewidth]{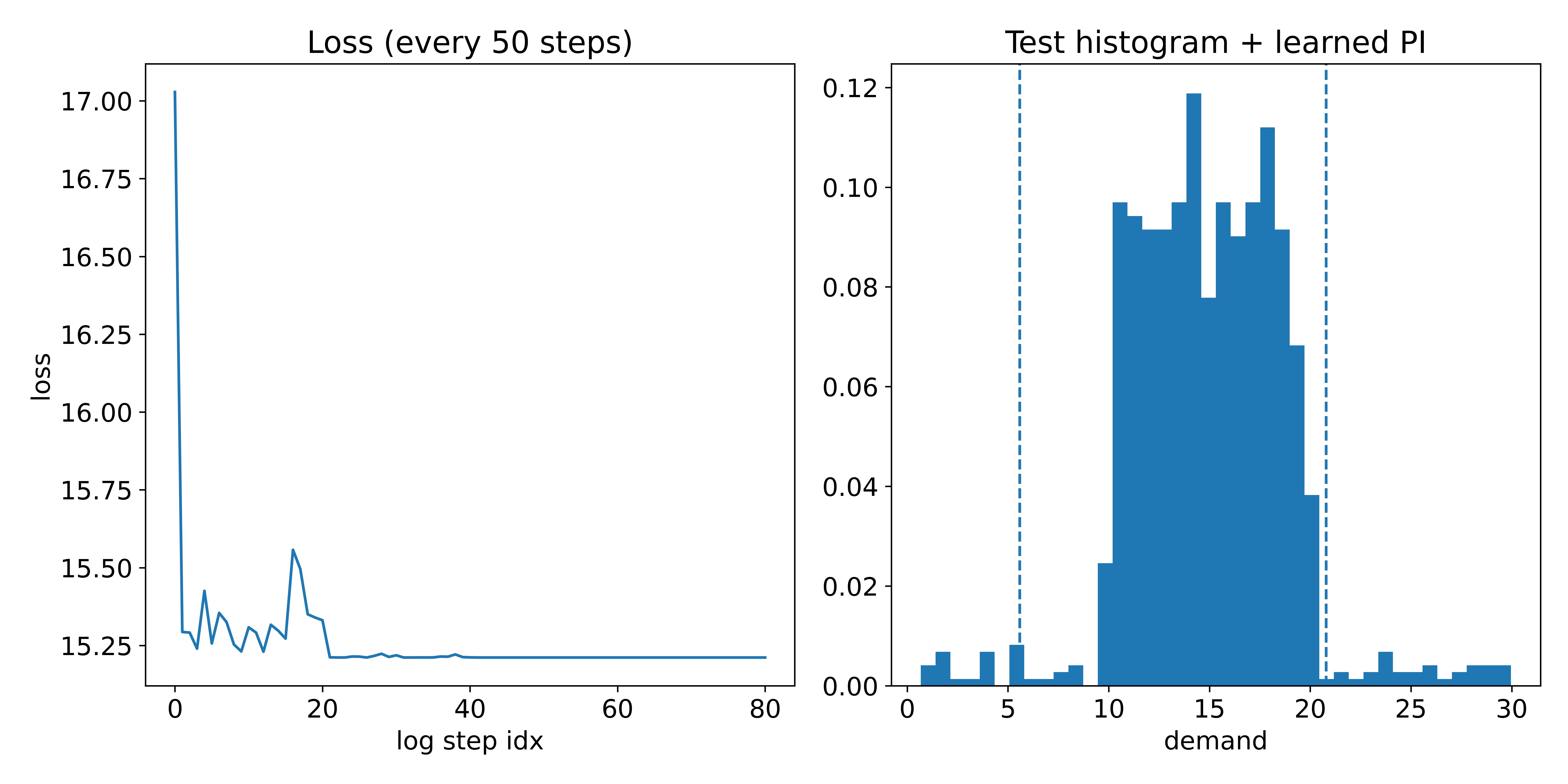}
        \caption{$n = 1000$}
    \end{subfigure}
    \hfill
    \begin{subfigure}{0.48\textwidth}
        \centering
        \includegraphics[width=\linewidth]{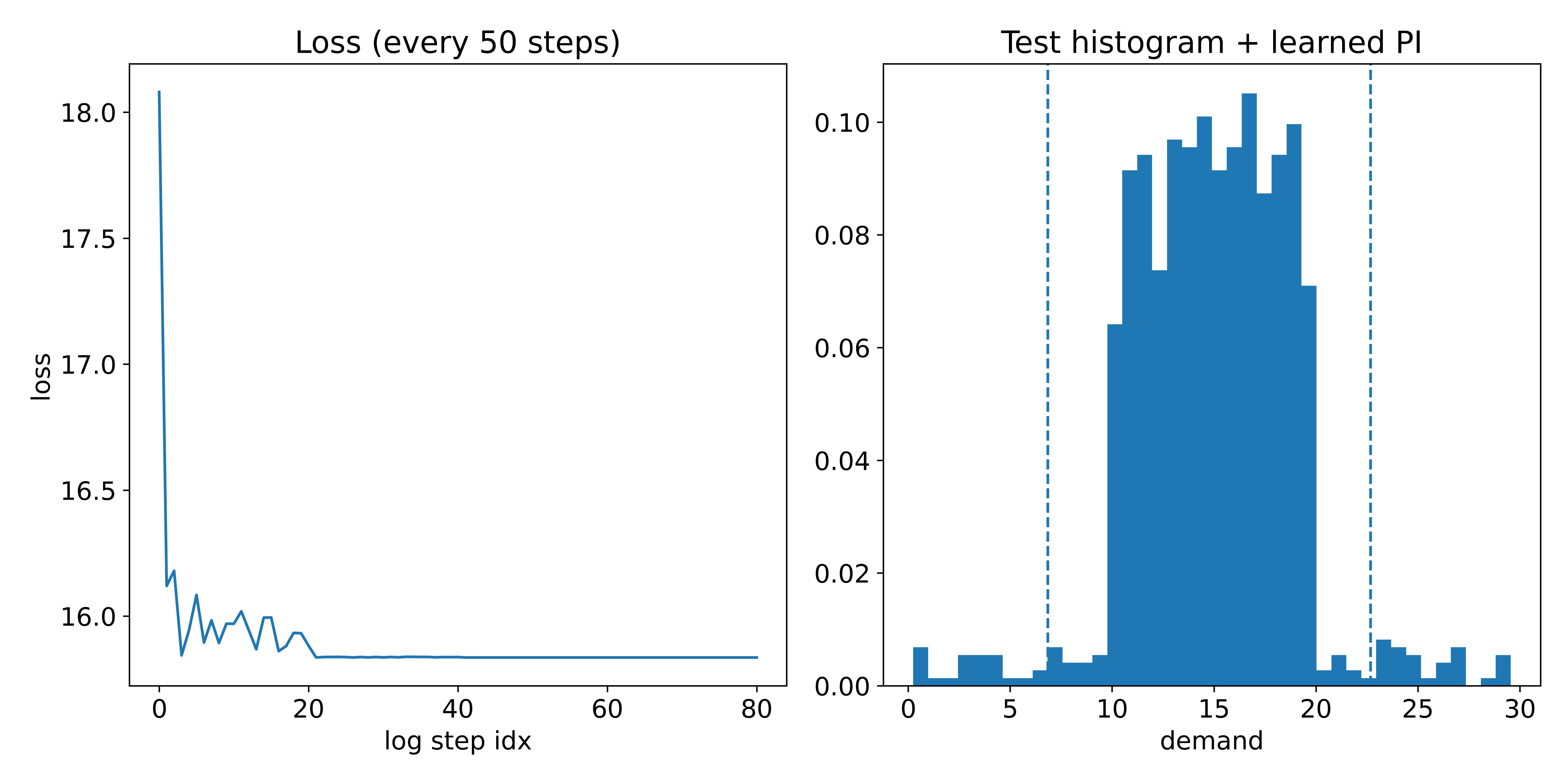}
        \caption{$n = 2000$}
    \end{subfigure}

    \caption{Training loss (left) and prediction intervals for the high-reward type (right), constructed using the neural network approach of \cite{pearce2018high} with training sample sizes $n \in \{10, 100, 1000, 2000\}$. The right-hand curve in each subfigure shows the learned interval (dashed lines) together with the histogram of test samples.}
    \label{fig:adviceneural}
\end{figure}

\end{APPENDICES}

\end{document}